\documentclass{article}
\usepackage[utf8]{inputenc} % allow utf-8 input
\usepackage[T1]{fontenc}    % use 8-bit T1 fonts
\usepackage{graphicx} % Required for inserting images
\usepackage{amsmath}
\usepackage{array}

\newcommand{\cE}{\mathcal{E}}

\usepackage{mdframed}

\usepackage{a4}
\usepackage{array}
\usepackage{tabularx}
\usepackage{graphicx}
\usepackage{amsmath,amssymb,amsthm,mathtools}
\usepackage{paralist}
\usepackage{bm}
\usepackage{bbm}
\usepackage{xspace}
\usepackage{url}
\usepackage{fullpage, prettyref}
\usepackage{boxedminipage}
\usepackage{wrapfig}
\usepackage{ifthen}
\usepackage{color}
\usepackage{xcolor}
\usepackage{framed}
\usepackage[pagebackref,colorlinks=true,urlcolor=blue,linkcolor=blue,citecolor=blue,pdfstartview=FitH]{hyperref}
\usepackage[nameinlink]{cleveref}
\usepackage{fullpage}
\usepackage{esvect}
\usepackage{mdframed}
\usepackage{ifthen}

\usepackage{thmtools}
\usepackage{thm-restate}
\newtheorem{theorem}{Theorem}[section]

\newtheorem{lemma}[theorem]{Lemma}
\newtheorem{claim}[theorem]{Claim}
\newtheorem{corollary}[theorem]{Corollary}
\newtheorem{definition}[theorem]{Definition}

\newtheorem{fact}[theorem]{Fact}

\newenvironment{proofof}[1]{\noindent \textbf{Proof of #1.}}{\qed}

\newcommand{\ignore}[1]{}

\newcommand{\cC}{{\cal C}}
\newcommand{\cD}{\mathcal{D}}

\newcommand{\cX}{{\cal X}}
\newcommand{\domain}{\cX}

\newcommand{\R}{\mathbb R}

\newcommand{\eps}{\varepsilon}

\newcommand{\poly}{\mathrm{poly}}
\newcommand{\polylog}{\mathrm{polylog}}

\newcommand{\calH}{{\cal H}}
\newcommand{\calM}{{\cal M}}

\newcommand{\calT}{{\cal T}}

\newcommand{\calX}{{\cal X}}
\newcommand{\calY}{{\cal Y}}

\newcommand{\barI}{\overline{I}}

\newcommand{\mbone}{\mathbbm{1}}

\newcommand{\norm}[1]{\left\lVert #1 \right\rVert}
\newcommand{\set}[1]{\left\{ #1 \right\}}
\newcommand{\ceil}[1]{\lceil#1\rceil}
\newcommand{\Exp}{\EX}
\newcommand{\Var}{\mathrm{Var}}

\newcommand{\given}{\textrm{ s.t. }}
\newcommand{\ifT}{\textrm{ if }}

\newcommand{\otherwise}{\textrm{ o/w }}

\newcommand{\EX}{\mathbb{E}}

\newcommand{\distribution}{D}
\newcommand{\p}{\boldsymbol{p}}
\newcommand{\q}{\boldsymbol{q}}
\newcommand{\coin}{\boldsymbol{c}}
\newcommand{\rD}{\boldsymbol{r}}

\newcommand{\ind}[1]{\mbone\left[ #1 \right]}

\newcommand{\Poi}{\mathrm{Poi}}
\newcommand{\PoiS}{\mathrm{PoiS}}

\newcommand{\Bern}{\mathrm{Bern}}

\newcommand{\tvd}{\mathrm{tvd}}

\newcommand{\Sec}[1]{\hyperref[sec:#1]{\Cref*{sec:#1}}} %section
\newcommand{\Eqn}[1]{\hyperref[eq:#1]{(\ref*{eq:#1})}} %equation
\newcommand{\Fig}[1]{\hyperref[fig:#1]{Fig.\,\ref*{fig:#1}}} %figure
\newcommand{\Tab}[1]{\hyperref[tab:#1]{Tab.\,\ref*{tab:#1}}} %table
\newcommand{\Thm}[1]{\hyperref[thm:#1]{Theorem\,\ref*{thm:#1}}} %theorem
\newcommand{\Fact}[1]{\hyperref[fact:#1]{Fact\,\ref*{fact:#1}}} %fact
\newcommand{\Lem}[1]{\hyperref[lem:#1]{Lemma\,\ref*{lem:#1}}} %lemma
\newcommand{\Prop}[1]{\hyperref[prop:#1]{Prop.~\ref*{prop:#1}}} %property
\newcommand{\Cor}[1]{\hyperref[cor:#1]{Corollary~\ref*{cor:#1}}} %corollary
\newcommand{\Conj}[1]{\hyperref[conj:#1]{Conjecture~\ref*{conj:#1}}} %conjecture
\newcommand{\Def}[1]{\hyperref[def:#1]{Definition~\ref*{def:#1}}} %definition
\newcommand{\Alg}[1]{\hyperref[alg:#1]{Alg.~\ref*{alg:#1}}} %algorithm
\newcommand{\Obs}[1]{\hyperref[obs:#1]{Obs.~\ref*{obs:#1}}} %observation
\newcommand{\Ex}[1]{\hyperref[ex:#1]{Ex.~\ref*{ex:#1}}} %example
\newcommand{\Clm}[1]{\hyperref[clm:#1]{Claim~\ref*{clm:#1}}} %example
\newcommand{\Step}[1]{\hyperref[step:#1]{Step~\ref*{step:#1}}} %example

\newcommand{\bigO}[1]{O\left(#1\right)}
\newcommand{\bigtO}[1]{\tilde{O}\left(#1\right)}
\newcommand{\bigOm}[1]{\Omega\left(#1\right)}
\newcommand{\bigTh}[1]{\Theta\left(#1\right)}
\newcommand{\litO}[1]{o\left(#1\right)}

\newcommand{\innerAlg}{\mathcal{A}}

\crefname{algocf}{alg.}{algs.}
\Crefname{algocf}{Alg.}{Algs.}

\newcommand{\flatten}{{\rm flat}}
\newcommand{\testing}{{\rm test}}

\usepackage[ruled,linesnumbered]{algorithm2e}
\SetKwComment{Comment}{/* }{ */}

\usepackage{algorithmic}

\hypersetup{linktocpage = true, linkcolor = purple}

\renewcommand{\Comment}[1]{ \color{teal} #1 \color{black} }

\title{Distribution Testing in the Presence of Arbitrarily Dominant Noise with Verification Queries}

\author{
Hadley Black\thanks{Supported by NSF HDR TRIPODS Phase II grant 2217058 (EnCORE Institute).} \vspace{-0.4em} \\
\emph{University of California, San Diego} \vspace{-0.4em} \\
\texttt{hablack@ucsd.edu}
\and
Christopher Ye\footnotemark[1] \vspace{-0.4em} \\
\emph{University of California, San Diego} \vspace{-0.4em} \\
\texttt{czye@ucsd.edu}
}

\begin{document}

\maketitle

\thispagestyle{empty}

\begin{abstract} 
We study distribution testing without direct access to a source of relevant data, but rather to a highly contaminated one, from which only a tiny fraction (e.g. $1\%$) is relevant. To enable this, we introduce the following \emph{verification query} model. The goal is to perform a statistical task on distribution $\p$ given sample access to a mixture $\boldsymbol{r} = \lambda \boldsymbol{p} + (1-\lambda)\boldsymbol{q}$ and the ability to \emph{query} whether a sample $x \sim \boldsymbol{r}$ was generated by $\boldsymbol{p}$ (relevant) or by $\boldsymbol{q}$ (irrelevant). This captures scenarios where it is cheap to acquire data from a massive pool, but expensive to verify whether it is of interest for the specific task. In general, if $m_0$ clean samples from $\p$ suffice for a task, then $O(m_0/\lambda)$ samples and verification queries trivially suffice in our model. We ask, \emph{are there tasks for which the number of queries can be significantly reduced?}

We show that for the canonical problems in distribution testing (uniformity, identity, and closeness), the answer is yes. In fact, we obtain matching upper and lower bounds that reveal smooth trade-offs between sample and query complexity. For all $m \leq n$, we obtain (i) a \emph{uniformity and identity} tester using $\smash{O(m  + \frac{\sqrt{n}}{\eps^2 \lambda})}$ samples and $\smash{O(\frac{n}{m \eps^4 \lambda^2})}$ queries, and (ii) a \emph{closeness} tester using $\smash{O(m + \frac{n^{2/3}}{\eps^{4/3} \lambda} + \frac{1}{\eps^4 \lambda^3})}$ samples and $\smash{O(\frac{n^2}{m^2 \eps^4 \lambda^3})}$ queries. 
%\begin{itemize}
%    \item a \emph{uniformity} (and identity) tester using $O(m  + \frac{\sqrt{n}}{\eps^2 \lambda})$ samples and $O(\frac{n}{m \eps^4 \lambda^2})$ queries, and
%    \item a \emph{closeness} tester using $O(m + \frac{n^{2/3}}{\eps^{4/3} \lambda} + \frac{1}{\eps^4 \lambda^3})$ samples and $O(\frac{n^2}{m^2 \eps^4 \lambda^3})$ queries. 
%\end{itemize}
Moreover, we show that these query complexities are tight for all testers using $m \ll n$ samples. In particular, obtaining the optimal sample complexity for standard testing requires every sample to be verified, while if one takes $O(n)$ samples, then only $\mathrm{poly}(1/\eps \lambda)$ queries suffice.  

Next, we show that for testing closeness (and hence also identity and uniformity) using $m = \widetilde{O}(n/\eps^2\lambda)$ samples we can achieve query complexity $\widetilde{O}(1/\eps^2\lambda)$ which is nearly optimal even for the basic task of \emph{bias estimation} (uniformity testing when $n=2$) with unbounded samples. Our uniformity testers work even in the more challenging setting where the contaminated samples are generated by an adaptive adversary (at the cost of a $\log n$ factor). Finally, we show that the sample-query trade-off lower bounds can be circumvented if the algorithm is provided with the PDF of the mixture $\rD$.

\end{abstract}

\newpage

\tableofcontents

\thispagestyle{empty}

\newpage

\clearpage
\pagenumbering{arabic}

\newcommand{\todo}{1}

\ifthenelse{\equal{\todo}{1}}{

}

\section{Introduction}

%Over the last two decades the design of noise-robust algorithms has been a central focus in statistics, learning theory, and algorithm design broadly. There are various ways of modeling noise, for example by manipulation of data \emph{labels} (random classification, Massart, etc), or by adding or subtracting data points entirely (Huber contamination), resulting in a data-set which violates the underlying distributional assumptions. The latter has led notably to the growing field of \emph{robust statistics} (see e.g. \cite{DBLP:books/cu/20/DiakonikolasK20}), as well as other models such as \emph{tolerant testing} (see e.g. \cite{DBLP:journals/jcss/ParnasRR06}) from the property testing community. However, not surprisingly, these models face the fundamental information-theoretic limitation that the fraction of noisy data cannot exceed the desired testing or learning error, $\eps$. 

In machine learning and algorithm design, direct access to data that is representative of the phenomena we wish to learn or understand is an assumption that is typically taken for granted. In other words, the process of acquiring and verifying data is completely decoupled from the role of algorithms. We imagine this process as follows: first data is gathered from an undiscerning source (a mixture of good and bad data), then verified and filtered by an expert to produce the input (purely good data) for an algorithm. If a $\lambda$-fraction of examples are good, and $m$ good examples are needed, then one will need to acquire and verify roughly $m/\lambda$ samples. If verification is expensive, then one may wish to improve this. Is this possible? We initiate the study of a new model for solving statistical problems when an arbitrarily large fraction of data is bad by giving the algorithm access to a \emph{verification oracle}, therefore \emph{coupling algorithm design with data verification}. Our central question is the following:

%\begin{center}
%\emph{By filtering data on the fly during an algorithm's execution, is it possible to significantly reduce the amount of data verification that is needed to solve a problem?}
%\end{center}

\begin{center}
\emph{By letting an algorithm decide the most informative samples to verify on the fly, is it possible to significantly reduce the total amount of data verification that is needed to solve a problem?}
\end{center}

We show that for certain basic statistical tasks, namely distribution testing, the answer is yes. 
%The idea is that by allowing the algorithm to carefully choose which samples are verified, it is able to extract more information per verification. 
Our aim is to address scenarios, such as the following examples, with a large fraction of irrelevant data:

%due to either a dominant source of contamination or a lack of direct access to the relevant data: % but where one can perform a costly operation to filter relevant samples. 
%The following are a few examples:

%Nevertheless, practical scenarios exist in which a large fraction, or even a majority of the data is corrupted or not relevant, but where it is possible (but expensive or undesirable) to check whether a data-point is relevant, and one may want to minimize the number of checks. Consider the following examples:
\begin{enumerate}
    \item (\emph{Respecting privacy of a protected group:}) You are interested in the income distribution of a relatively small protected group (e.g. by gender and/or ethnicity, etc), but you only have sample access to the broader general population. It is possible to ask (query) an individual to disclose whether they are a part of the protected group, but this is time consuming and may compromise privacy.
    %they may decline the request out of privacy concerns. %You would like to minimize the number of such requests since they take time and can compromise privacy. 
    \item (\emph{Being robust to fake news:}) You are conducting a survey of news and research articles, a significant fraction of which contain fabricated statistics or non-reproducible results. You can employ (query) an expert to investigate the trustworthiness or quality of an article.%, but this is prohibitively costly to perform en masse.
    \item (\emph{Astronomy and particle physics:}) In experiments like those at the Large Hadron Collider, or in telescope data, a large majority of recorded events are due to background noise. However, patterns involving rare events must be identified and analyzed. Expert verification can distinguish signal vs. background noise, but this can only be done sparingly due to the volume of data.
\end{enumerate}

%A trivial solution is to simply gather samples, while checking each and every one until enough relevant data has been acquired to solve the problem of interest. However, this may be extremely costly, and we are interested in knowing to what extent this can be avoided. The key question is, \emph{is it possible to check only a tiny number of the needed samples and still solve the problem of interest on the relevant data?} We show that for certain fundamental statistical tasks, namely distribution testing, the answer is \emph{yes}. 

%In particular, we obtain distribution testing algorithms with sublinear sample complexity while only querying a tiny fraction of samples.

\paragraph{The Verification Query Model.} Formalizing the above discussion, we initiate the study of a new \emph{verification query} model for statistical tasks. We are given sample access to a mixture $\rD = \lambda \p + (1-\lambda) \q$ where $\p$ is the input distribution of interest over a domain $\cX$, $\q$ is an arbitrary contamination distribution over $\cX$, and $\lambda \in (0,1)$ is an unknown mixture parameter which we think of as a small constant, but can be arbitrarily small. The distribution $\p$ represents a pure source of relevant data lying within a larger, perhaps contaminated source, $\rD$. Given a sample $x \sim \rD$, the \emph{verification query oracle} returns "relevant" or "irrelevant", indicating whether $x$ was generated by $\p$ or by $\q$, respectively. In general, we view queries as more expensive then samples. In fact, we obtain algorithms whose query complexity decreases smoothly as a function of the number of samples, and thus can be adapted to the relative cost of queries vs. samples in a given application. 

%Our main aim is to minimize the query complexity, as we view queries as expensive, and our secondary goal is to minimize samples. %We note 

%\paragraph{Event queries.} 
Verification queries can also be interpreted as revealing whether a particular event occurred for a sample. %, e.g. whether the sample belongs to the protected group in Example 1 above. 
I.e., consider a distribution $\cD$ over domain $\Omega$, and an event $\cE \subset \Omega$. We can express $\cD$ as a mixture \[\cD = \Pr[\cE] \cdot \cD|_{\cE} + (1-\Pr[\cE]) \cdot \cD|_{\Omega \setminus \cE}\] where $\cD|_{\cdot}$ denotes the distribution conditioned on $(\cdot)$. The verification query model then captures the problem of performing a statistical task for the distribution $\cD|_{\cE}$, given sample access to $\cD$ and the ability to query whether $\cE$ occurred for a given sample. As mentioned, we imagine that this event query is expensive, invasive, or otherwise undesirable to perform extensively. 

If $m$ samples suffice for a given task in the standard model, then $O(m/\lambda)$ samples and queries are always sufficient in the verification query model since %First, it is possible obtain a constant factor approximation of $\lambda$ with $O(1/\lambda)$ samples and queries, as querying a sample simulates a $\lambda$-probability coin. 
this many samples from $\rD$ will reveal at least $m$ samples from $\p$ with high probability, allowing for the standard algorithm to be applied as a black-box. When can we do significantly better than this?

\paragraph{Comparison with Semi-verified Learning.} To the best of our knowledge, there have been very few works which consider settings where a majority of data is irrelevant.
For example, list-decodable algorithms return a list of answers containing at least one accurate answer, but provably cannot produce a single correct output \cite{DBLP:conf/stoc/CharikarSV17}.
The closest model to ours is \emph{semi-verified learning} \cite{DBLP:conf/stoc/CharikarSV17, DBLP:conf/colt/MeisterV18, DBLP:conf/aistats/Zeng023} introduced by \cite{DBLP:conf/stoc/CharikarSV17}. In this model, like ours, only a small fraction of data is good (generated by $\p$), and to remedy this the algorithm is provided with a small set of $k$ pre-verified samples from $\p$ as side information. In our verification \emph{query} model one can obtain this side information using $O(k/\lambda)$ samples and queries, and so up to this multiplicative $\frac{1}{\lambda}$ term, our model is stronger. 

The key difference between these models is that \emph{verification queries} allow the algorithm the power to carefully \emph{choose samples to query} in order to extract as much signal as possible, as opposed to seeing $k$ independent samples from $\p$. In other words, the algorithm has some control over the good data it pays to see, and our results rely heavily on this additional power. Conceptually, our results show that strengthening the semi-verified model in this small and natural way significantly improves the capability to perform certain basic statistical tasks, namely distribution testing. (See \Cref{sec:related-work} for a more formal discussion on the relationship with semi-verified learning.) 

\paragraph{On Robust Statistics and Tolerant Testing.} Over the last two decades the design of noise-robust algorithms has been a central focus in statistics, learning theory, and algorithm design broadly. 
There are various ways of modeling noise, for example by manipulation of data \emph{labels} (random classification, Massart, etc), or by adding or subtracting data points entirely (Huber contamination), resulting in a data-set which violates the underlying distributional assumptions. 
The latter has led notably to the growing field of \emph{robust statistics} (see e.g. \cite{DBLP:books/cu/20/DiakonikolasK20} for a survey), as well as related models such as robust distribution testing \cite{canonne2023full} and \emph{tolerant testing} (see e.g. \cite{DBLP:journals/jcss/ParnasRR06}) from the property testing community. 
However, not surprisingly, these models face the fundamental information-theoretic limitation that the fraction of noisy data cannot exceed the desired testing or learning error, $\eps$. 
Our work aims to circumvent this limitation by adding the power to distinguish relevant data from contamination.

%We view this as the main conceptual contribution of our work 

%Our results demonstrate that this small amount of added power \emph{significantly improves} the capability to perform certain basic statistical tasks, namely distribution testing. (See \Cref{remark:constrast}.) 

%This model has been studied for problems such as mean estimation \cite{DBLP:conf/stoc/CharikarSV17}, ... \cite{DBLP:conf/aistats/Zeng023}, and ... \cite{DBLP:conf/colt/MeisterV18}, but not for distribution testing as far as we know, which is the focus of our work.

 %We believe this to be a more natural model in the sense that the point is to not allow the algorithm direct access to $\p$, but giving clean samples as side information circumvents this. 

\subsection{Our Results: Distribution Testing with Verification Queries}

As an initial foray into the verification query model, we investigate basic testing questions for discrete probability distributions. Namely, we study the canonical problems in the area of \emph{distribution testing}: uniformity and closeness testing, which have been extremely well-studied in the literature (see e.g. \cite{goldreich1997property,DBLP:journals/jacm/GoldreichGR98,DBLP:conf/focs/BatuFFKRW01,DBLP:phd/us/Batu01,DBLP:journals/tit/Paninski08,DBLP:journals/jacm/BatuFRSW13, DBLP:conf/soda/ChanDVV14,diakonikolas2014testing,diakonikolaskane2016,valiant2017automatic,diakonikolas2018sample,DBLP:journals/cjtcs/DiakonikolasGPP19} for a sample of relevant works and \cite{Canonne:Survey:ToC} for a survey). We use the following definitions and terminology: 

%We begin with the necessary technical definitions before stating our results. 

%Namely, distribution testing, a classic topic at the intersection statistics, learning theory, and property testing and obtain tight bounds for the canonical problems in this area: testing uniformity and closeness, as well as bias estimation. We begin with the necessary technical definitions before stating our results. 

%\hadley{Add bit about why queries are necessary in distribution testing. I.e. simple example showing how we can get $\rD = U_n$ in both yes and no case.}

\begin{itemize}
    %\itemindent=-13pt
    \item A distribution $\rD$ is a $(1 - \lambda)$-contaminated source for $\p$ if $\rD = \lambda \p + (1 - \lambda) \q$ for some distribution $\q$.
    \item A statistical task with data domain $\calX$ and output space $\calY$ is a family of pairs $\calT = \set{(\distribution, G_{\distribution})}$ where $G_{\distribution} \subset \calY$ is the set of good outputs for distribution $\distribution$ over $\calX$.
    \item A (randomized) algorithm $\innerAlg$ is $\lambda$-robust against distributional contamination for statistical task $\calT$ with sample complexity $m$ and query complexity $q$ if for every $(\distribution_0, G_{\distribution_0}) \in \calT$ and every $(1 - \lambda)$-contaminated source for $\distribution_0$ (denoted $\distribution$),
    \begin{equation*}
        \Pr_{S \sim \distribution^{m}, \innerAlg} \left( \innerAlg(S) \in G_{\distribution_0} \right) > \frac{2}{3} \text{.}
    \end{equation*}
    Succinctly, we say $\innerAlg$ is a $\lambda$-distributionally robust algorithm with complexity $(m, q)$. We say $\innerAlg$ is $\delta$-correct if the output lies in $G_{\distribution_{0}}$ with probability $1 - \delta$.
\end{itemize}

Perhaps the most basic statistical task is bias estimation: estimating the bias of a coin up to error $\eps$.
It is well known that $\Theta(\eps^{-2})$ samples are necessary and sufficient.
As mentioned, there is a naive algorithm using $\bigO{\eps^{-2} \lambda^{-1}}$ samples and queries.
We show that this is in fact optimal.

\begin{restatable}{theorem}{BiasEstimationLB}
    \label{thm:bias-estimation-lb}
    \emph{(Bias estimation lower bound.)}
    Any $\lambda$-distributionally robust $\eps$-bias estimator requires $\bigOm{\frac{1}{\eps^2 \lambda}}$ samples and $\bigOm{\frac{1}{\eps^2 \lambda}}$ queries.
\end{restatable}

While verification queries do not allow a bias estimation algorithm to reduce the amount of data verification required, we show that this is in fact possible for testing uniformity and closeness over $[n]$.

\paragraph{Tight Sample vs. Query Tradeoffs.} 

Uniformity testing is perhaps the most fundamental question in distribution testing. The goal is to determine whether an input distribution $\p$ supported over $[n]$ is \emph{uniform} or $\eps$-far from uniform in total variation distance\footnote{The total variation distance between $\p,\q$ supported over $[n]$ is defined as $d_{\mathrm{TV}}(\p,\q) = \frac{1}{2}\norm{\p - \q}_1 = \frac{1}{2}\sum_{i=1}^n |\p[i]-\q[i]|$. 
We say that $\p$ is $\eps$-far from uniform if $d_{\mathrm{TV}}(\p,U_n) > \eps$ where $U_n$ denotes the uniform distribution over $[n]$.} given \emph{sample-access} to $\p$ (in the standard model). 
It is well-known that $\Theta(\sqrt{n}/\eps^2)$ samples are both necessary and sufficient \cite{DBLP:journals/tit/Paninski08}. 
One reason for its central importance is that it is at least as hard as testing whether $\p = \p^{\ast}$ for any \emph{given and fixed} $\p^{\ast}$ (identity testing) \cite{goldreich2020uniform}. 
In our setting, samples and queries can be simulated by an identical reduction, so that all of our uniformity testing algorithms imply identity testing algorithms with identical sample and query complexity up to constant factors (see \Cref{thm:identity-testing-reduction}).
In the verification query model, sample access is given instead to a mixture $\rD = \lambda \p + (1-\lambda)\q$ where $\p$ is the distribution we want to test, and $\q$ is an adversarially chosen source of contamination. 
In particular, it is possible that $\rD = U_n$ while either $d_{\mathrm{TV}}(\p,U_n) > \eps$ or $\p = U_n$, and so it is impossible to distinguish with samples alone. As mentioned, an algorithm using $\smash{O(\frac{\sqrt{n}}{\eps^2 \lambda})}$ samples and queries is trivial to obtain using known algorithms for standard testing. We show that by increasing the number of samples $m$, there is an algorithm with query complexity that improves by a factor $\frac{1}{m}$.

%\begin{restatable}{theorem}{SimpleUniformityAlg}
%    \label{thm:simple-uniformity-alg} (\emph{Uniformity testing upper bound.})
%    There is a $\lambda$-distributionally robust $\varepsilon$-uniformity tester with sample complexity $\bigO{m + \frac{\sqrt{n}}{\eps^2 \lambda}}$ and query complexity $\bigO{\frac{n}{m \eps^{4} \lambda^2} + \frac{1}{\eps^{4} \lambda^2}}$.
%\end{restatable}

\begin{restatable}{theorem}{SimpleUniformityAlg}
    \label{thm:simple-uniformity-alg} \emph{(Uniformity testing upper bound.)}
    For all $m \leq O(n), \eps,\lambda \in (0,1)$, there is a $\lambda$-distributionally robust $\varepsilon$-uniformity (resp. identity) tester with sample complexity $\bigO{m + \frac{\sqrt{n}}{\eps^2 \lambda}}$ and query complexity 
    $\bigO{\frac{n}{m \eps^{4} \lambda^2}}$.
    %$\bigO{\frac{n}{m \eps^{4} \lambda^2} + \frac{1}{\eps^{4} \lambda^2}}$.
\end{restatable}

When $m = \Theta(\frac{\sqrt{n}}{\eps^2\lambda})$, we recover the trivial algorithm, whereas when $m = \Theta(n)$ we obtain a query complexity $O(\frac{1}{\eps^4 \lambda^2})$, \emph{independent of} $n$. 
In fact, we show that this sample-query tradeoff is optimal in the sub-linear regime i.e., $m \leq cn$ where $c > 0$ is a sufficiently small constant.
Our lower bound holds even when the algorithm is given $\lambda$, which we do not assume in our upper bounds.

\begin{restatable}{theorem}{UniformityTestingLB}
    \label{thm:dist-contamination-lb-eps} \emph{(Uniformity testing lower bound.)}
    Let $\lambda \leq \frac{1}{2}$ and $\eps < 0.1$.
    Any $\lambda$-distributionally robust $\eps$-uniformity tester with $m \leq n/100$ samples requires $\bigOm{\frac{n}{m \eps^{4} \lambda^{2}}}$ queries.
\end{restatable}

Next, we study the problem of \emph{closeness testing}: given sample access to two distributions $\p_1,\p_2$ supported over $[n]$, distinguish between the case of $\p_1 = \p_2$ vs. $d_{\mathrm{TV}}(\p_1,\p_2) > \eps$. In the standard model, the sample-complexity of this question is well-known to be $\smash{\Theta(\frac{n^{2/3}}{\eps^{4/3}} + \frac{\sqrt{n}}{\eps^2})}$ (\cite{DBLP:journals/jacm/BatuFRSW13, DBLP:conf/soda/ChanDVV14, diakonikolaskane2016}). In the verification query model, we have sample and verification query access to \emph{two mixtures} $\rD_1 = \lambda \p_1 + (1-\lambda)\q_1$ and $\rD_2 = \lambda \p_2 + (1-\lambda)\q_2$ and the objective\footnote{Note that we are assuming here that the two sources share a common mixture parameter $\lambda$, which may not be the case in some settings. We give a reduction in \Cref{lem:unequal-mixture-params} showing that we can reduce from the case of two different mixture parameters $\lambda_1,\lambda_2$ to a common one obtained by their product, $\lambda_1\lambda_2$. Therefore, this is without loss of generality, up to a potentially quadratic blowup in the dependence on $\lambda$. An interesting question is whether there is a more efficient reduction to common mixture parameter $\min(\lambda_1,\lambda_2)$ instead of the product, which would be best possible.} is to test closeness of $\p_1,\p_2$. Again, $\smash{\Theta((\frac{n^{2/3}}{\eps^{4/3}} + \frac{\sqrt{n}}{\eps^2})\frac{1}{\lambda})}$ samples and queries is trivial to obtain. Similar to our results for uniformity testing, we show that the query complexity of testing closeness improves as the number of samples is increased, this time by a factor $\frac{1}{m^2}$. 

\begin{restatable}{theorem}{ClosenessTestingUB}
    \label{thm:simple-closeness-alg} \emph{(Closeness testing upper bound.)}
   For all $m \leq O(n), \eps,\lambda \in (0,1)$, there is a $\lambda$-distributionally robust $\eps$-closeness tester with sample complexity $\bigO{m + \frac{n^{2/3}}{\eps^{4/3}\lambda} + \frac{1}{\eps^4 \lambda^3}}$ and query complexity $\bigO{\frac{n^2}{m^2 \eps^{4}\lambda^3}}$.
\end{restatable}

%\begin{definition}[$\eps$-Closeness Testing]
%    \label{def:closeness}
%    Given sample access to $\p, \q$ over $[n]$, determine if $\p = \q$ or $\tvd(\p, \q) > \eps$.
%\end{definition}

We establish that the sample-query trade-off achieved by the above algorithm is tight for $m \leq cn$. 
Again, this lower bound holds even when the algorithm is given $\lambda$.

\begin{restatable}{theorem}{ClosenessTestingLB}
    \label{thm:dist-contam-closeness-lb} \emph{(Closeness testing lower bound.)}
    Let $\lambda \leq \frac{1}{2}$ and $\eps < 0.1$.
    Any $\lambda$-distributionally robust $\eps$-closeness tester with $m \leq n/100$ samples requires $\bigOm{\frac{n^2}{m^2 \eps^{4} \lambda^{3}}}$ queries.
\end{restatable}

%\begin{definition}[$\eps$-Bias Estimation]
%    \label{def:bias-estimation}
%    Given sample access to $\Bern(p)$, output $\hat{p}$ such that $|\hat{p} - p| < \eps$.
%\end{definition}

\paragraph{Algorithms with Optimal Query Complexity.} Notice that the upper bounds on the query complexity of uniformity and closeness testing achieved by \Cref{thm:simple-uniformity-alg} and \Cref{thm:simple-closeness-alg} are $O(\eps^{-4}\lambda^{-2})$ and $O(\eps^{-4}\lambda^{-3})$ at best, respectively. (Achieved using $O(n)$ samples). We have established in the corresponding lower bounds that this is optimal given a budget of $O(n)$ samples, but it is natural to ask whether the query complexity can be improved further using $m \gg n$ samples. We show in the following theorem that this is indeed the case. In fact, we give a closeness tester (and hence also an identity tester, and uniformity tester) using $\widetilde{O}(\eps^{-2}\lambda^{-1})$ queries, matching our lower bound for bias estimation (\Cref{thm:bias-estimation-lb}) up to $\poly\log(\eps^{-1}\lambda^{-1})$ factors. This tester follows a different strategy which we outline in \Cref{sec:tech-sup}.

\begin{restatable}{theorem}{ClosenessTestingUB-high}
    \label{thm:simple-closeness-alg-high} \emph{(Query-optimal closeness testing, \Cref{thm:closenessUB-high} informal.)}
   For all $\eps,\lambda \in (0,1)$, there is a $\lambda$-distributionally robust $\eps$-closeness tester with sample complexity $\widetilde{O}(n/\eps^2\lambda)$ and query complexity $\widetilde{O}(1/\eps^2\lambda)$.
\end{restatable}

\paragraph{Algorithms Against Adaptive Contamination.} 
So far, we have considered only \emph{distributional} sources of contamination. 
That is, although the contamination distribution $\q$ is worst-case (chosen adversarially in response to $\p$), the contaminated samples are drawn i.i.d. from $\q$. Alternatively, one could consider a more challenging model where these samples are chosen by an adaptive adversary (also called Huber contamination), which leads to further technical challenges.

\begin{itemize}
    %\itemindent=-13pt
    \item An adversarially $(1 - \lambda)$-contaminated source for $\p$ is sampled from as follows: 
    \begin{enumerate}
        \item $m$ samples are drawn iid from $\p$.
        \item An adversary arbitrarily corrupts up to a $(1 - \lambda)$-fraction of the samples drawn from $\p$ and returns the corrupted dataset.\footnote{Note that the adversary may corrupt less than $(1 - \lambda)$-fraction of samples.}
        % \hadley{Is this what we want? This means that there is always exactly $\lambda m$ good samples, which differs from the distributional setup. Should it be that each sample is independently corrupted with probability $1-\lambda$?} 
        % \chris{no, because it is ``up to" so you could corrupt less.}
    \end{enumerate} 
    % is an oracle $\rD_{A}$ which when sampled returns a sample from $\p$ with probability $\lambda$ and otherwise returns an arbitrary, adversarially chosen element from the support of $\p$. 
    \item The notion of an \emph{adversarially robust algorithm} is precisely analogous to the distributional definition given earlier in the section, but defined with respect to adversarially contaminated sources.
\end{itemize}

%\begin{definition}[Adversarially Contaminated Source]
%    \label{def:adver-contamination-source}
%    An adversarial $(1 - \lambda)$-contaminated source for $\p$ is an oracle $\rD_{A}$ which when sampled returns a sample from $\p$ with probability $\lambda$ and otherwise returns an arbitrary, adversarially chosen element from the support of $\p$. 
%\end{definition}

%The notion of an \emph{adversarially robust algorithm} is precisely analogous to \Cref{def:dist-contamination-alg}, but defined with respect to adversarially contaminated sources (\Cref{def:adver-contamination-source}). 
Unsurprisingly, obtaining adversarially robust algorithms is more technically challenging. Nevertheless, we are able to strengthen our upper bounds for uniformity testing to hold in this more difficult setting, with similar sample and query complexity (up to log-factors).

\begin{restatable}{theorem}{UniformityAdvAlg}
    \label{thm:adversarially-robust-uniformity-alg}  \emph{(Adversarial uniformity.)}
    For all $m \in \mathbb{N}, \eps,\lambda \in (0,1)$, there is a $\lambda$-adversarially robust $\varepsilon$-uniformity (resp. identity) tester with sample complexity $\bigO{m + \frac{\sqrt{n}}{\eps^2 \lambda}}$ and query complexity $\bigO{\frac{n \log n}{m \eps^{4} \lambda^2} + \frac{1}{\eps^{4} \lambda^2}}$.
\end{restatable}

%We note that our algorithm has optimal sample-query trade-off up to a logarithmic factor (since any adversarially robust algorithm is distributionally robust).
%Furthermore, if one is willing to take $\widetilde{O}(n/\eps^2)$ samples (enough for \emph{learning}), then it is possible to use only $\widetilde{O}(1/\eps^2))$. %\chris{$\tO{1/\eps^2}$ looks cleaner I think} queries. 

Similarly, we give a strengthening of \Cref{thm:simple-closeness-alg-high} for the special case of testing uniformity.

\begin{restatable}{theorem}{AdversarialUniformityAlg}
    \label{thm:adversarial-uniformity-alg} \emph{(Query-optimal adversarial uniformity, \Cref{thm:n-sample-1-query-adversary} informal.)}
    For all $\eps,\lambda \in (0,1)$, there is a $\lambda$-adversarially robust $\varepsilon$-uniformity (resp. identity) tester with sample complexity $\widetilde{O}(n/\eps^2\lambda)$ and query complexity $\widetilde{O}(1/\eps^2\lambda)$.
\end{restatable}

This query complexity is nearly optimal in light of our lower bound for bias estimation, \Cref{thm:bias-estimation-lb}.

\paragraph{Circumventing Query Lower Bounds with Mixture Knowledge.} 
Next, we ask what happens when the algorithm is provided with an explicit description of the mixture distribution $\rD$ as input. 
We show that this added information enables us to circumvent the lower bound of \Cref{thm:dist-contamination-lb-eps} and achieve near-optimal (up to $\poly(\log n,\eps^{-1},\lambda^{-1})$ factors) dependence on $n$ for both sample and query complexity, simultaneously.

% \hadley{What is the dependence on $\lambda$?}

\begin{restatable}{theorem}{UniformityAlgMixtureKnowledge}
    \label{thm:uniformity-alg-mixture-informal} \emph{(Uniformity testing with mixture knowledge, \Cref{thm:dist-mix-knowledge-alg} informal.)} For all $\lambda, \eps \in (0,1)$, given complete knowledge of the PDF of the distributionally contaminated source, there is a $\lambda$-distributionally robust $\varepsilon$-uniformity (resp. identity) tester with sample complexity $\bigtO{\frac{\sqrt{n}}{\poly(\eps \lambda)}}$ and query complexity $\poly\left(\frac{\log n}{\eps \lambda}\right)$.
\end{restatable}

We remark that in the interest of simplicity we have not optimized our dependence on $\eps, \lambda$ in \Cref{thm:uniformity-alg-mixture-informal}.

\begin{table}[ht]
\label{table:results-alt}
\centering
\begin{tabular}{|c|c|c|c|c|c|}
\hline
 & \textbf{Samples} & \textbf{Queries} & \textbf{Query Lower Bound} & \textbf{Adversarial?} & \textbf{Ref}  \\
\hline
\textbf{Uniformity} &
$m + \frac{\sqrt{n}}{\eps^2 \lambda}$ &
$\frac{n}{m \eps^{4} \lambda^{2}}$ &
$\frac{n}{m \eps^{4} \lambda^{2}}$ \text{ if $m \ll n$} &
\textbf{Y} [\ref{thm:adversarially-robust-uniformity-alg}] & 
 \ref{thm:simple-uniformity-alg}, \ref{thm:dist-contamination-lb-eps} \\[5pt]
\cline{2-6}
&
$\frac{n}{\eps^2 \lambda}$ &
$\frac{1}{\eps^2 \lambda}$ &
$\frac{1}{\eps^2 \lambda}$ \text{ for all $m$}&
\textbf{Y} &
\ref{thm:bias-estimation-lb}, \ref{thm:adversarial-uniformity-alg} \\[5pt]
\hline
\textbf{Closeness} &
$m + \frac{n^{2/3}}{\eps^{4/3} \lambda} + \frac{1}{\eps^{4} \lambda^{3}}$ &
$\frac{n^{2}}{m^{2} \eps^{4} \lambda^{3}}$ &
$\frac{n^{2}}{m^{2} \eps^{4} \lambda^{3}}$ \text{ if $m \ll n$} &
&
\ref{thm:simple-closeness-alg}, \ref{thm:dist-contam-closeness-lb}\\[5pt]
\cline{2-6}
&
$\frac{n}{\eps^2 \lambda}$ &
$\frac{1}{\eps^2 \lambda}$ &
$\frac{1}{\eps^2 \lambda}$ \text{ for all $m$}&
&
\ref{thm:bias-estimation-lb}, \ref{thm:simple-closeness-alg-high}\\[5pt]
\hline
\end{tabular}
\caption{\small{Table of our results. Polylogarithmic terms are omitted for simplicity. When $m \ll n$, all query complexities are optimal up to constant factors. All lower bounds hold in the distributional setting, and therefore in the adversarial setting as well. Our uniformity testing algorithms also hold more generally for identity testing.}}
\end{table}

%We leave establishing tight dependence on all parameters as interesting future work.

%\paragraph{Bias estimation.} Finally, we consider perhaps the most basic of all statistical tasks: estimating the bias of a coin. It is well known that $\eps$-bias estimation (even with uncontaminated sample access) requires $\bigTh{1/\eps^2}$ samples.
%By querying all samples, one again obtains a simple algorithm to estimate bias up to $\eps$-accuracy with $O(1/\eps^2)$ samples and queries.
%We show that this is in fact optimal.

%\begin{restatable}{theorem}{BiasEstimationLB}
%    \label{thm:bias-estimation-lb}
%    Any $\frac{1}{2}$-distributionally robust $\eps$-bias estimator requires $\bigOm{\frac{1}{\eps^2}}$ samples and $\bigOm{\frac{1}{\eps^2}}$ queries.
%\end{restatable}

%Perhaps the most fundamental distribution testing problem is uniformity testing, where given sample access to some unknown distribution $\p$ over $[n]$, the algorithm must determine if $\p$ is the uniform distribution over $[n]$, denoted $U_n$, or far from uniform (in total variation distance).

%\begin{definition}[$\eps$-Uniformity Testing]
 %   \label{def:uniformity}
 %   Given sample access to $\p$ over $[n]$, determine if $\p = U_n$ or $\tvd(\p, U_n) > \eps$.
%\end{definition}

\paragraph{Paper organization.} For the sublinear sample regime ($m\leq O(n)$) we prove our upper bounds for testing uniformity (\Cref{thm:simple-uniformity-alg} and \Cref{thm:adversarially-robust-uniformity-alg}) in \Cref{sec:uniformity-UB} and closeness (\Cref{thm:simple-closeness-alg}) in \Cref{sec:closeness-UB}. We prove our lower bounds (\Cref{thm:dist-contamination-lb-eps} and \Cref{thm:dist-contam-closeness-lb}) in \Cref{sec:LB}. We obtain our query-optimal testers for closeness (\Cref{thm:simple-closeness-alg-high}) and uniformity (\Cref{thm:adversarial-uniformity-alg}) in \Cref{sec:closeness-UB-highsample} and \Cref{sec:highsample}. Our uniformity tester with perfect mixture knowledge (\Cref{thm:uniformity-alg-mixture-informal}) is presented in \Cref{sec:uniformity-UB-MK}. We provide technical overviews for the proofs of our main results on bias estimation, uniformity, and closeness testing in \Cref{sec:tech}.

\subsection{Discussion and Open Questions}

We propose a novel algorithmic framework for statistical analysis on heavily contaminated data sources: verification queries which reveal the trustworthiness of a sample.
Within this framework, we have given nearly tight upper and lower bounds for fundamental problems in distribution testing: bias estimation, uniformity testing, and closeness testing.

A natural open question is to investigate the power of verification queries for other statistical problems. 
For example, one can consider mean estimation in high dimensional settings, or more general learning tasks such as PAC learning.
As usual, the naive algorithm is to run the standard algorithm and query all samples. 
Is this optimal, or is there a better algorithm?

Our work also leaves open the question of whether there is a separation in our model between sources of distributional vs. adversarial contamination. 
% In particular, do our upper bounds for uniformity and closeness testing also hold when the contamination distribution $\q$ is replaced by an adversary?
For uniformity, we have shown that our algorithms can be adapted to the adversarial setting with near optimal query complexity.
Can the same be done for any statistical task? 
More generally, is there a generic way to make a distributionally robust algorithm adversarially robust?
Is there a separation (perhaps via some other problem) between the two contamination models?

%Finally, another concrete question left open by our work is to resolve the query complexity of closeness testing in the "high sample regime". Using $O(n)$ samples we obtain an $O(\eps^{-4}\lambda^{-3})$ query tester, and this is optimal given linear sample complexity. Is $\widetilde{O}(\eps^{-2} \mathrm{poly}(1/\lambda))$ queries possible, or is closeness testing harder than uniformity in the verification query model?

\subsection{Related Work and Comparison with Other Models} \label{sec:related-work}

\paragraph{On Tolerant Testing.} An alternative approach to addressing distributional contamination is the concept of \emph{tolerant testing} \cite{DBLP:journals/jcss/ParnasRR06}. 
In the tolerant uniformity testing problem, an algorithm must decide if $\p$ is $\eps$-close to uniform or $2 \eps$-far from uniform.
In the verification query model setting, if $\p$ is uniform, then $\rD$ is $(1 - \lambda)$-close to uniform, while if $\p$ is $\eps$-far from uniform, $\rD$ is $\eps \lambda$-far from uniform. 
Then, if $\lambda \geq \frac{1}{1 + \eps}$ (low contamination), we can in fact test uniformity without any verification queries.

However, this approach immediately runs into two significant obstacles.
First, tolerant testing algorithms succeed only when contamination is low ($\lambda \geq \frac{1}{2}$), while our algorithms succeed for arbitrary contamination parameters $\lambda$.
Second, even for constant\footnote{See \cite{CJKL:22} for the dependency on tolerance parameters for the full landscape of the problem.} $\eps$ the sample complexity of tolerant testing is barely sublinear $\Theta(n/\log n)$ \cite{valiant2010clt,valiant2017estimating,CJKL:22}.
In contrast, our sample-query trade-off (\Cref{thm:simple-uniformity-alg}) yields an algorithm with strongly sub-linear sample complexity at the cost of querying a tiny fraction of the samples (e.g. $n^{0.9}$ samples and $n^{0.1}$ queries).

\paragraph{Contrasting with Semi-verified Learning.} Our results show that verification queries are significantly more powerful than semi-verified learning \cite{DBLP:conf/stoc/CharikarSV17} for the problems we study. Consider the uniformity testing problem with $\lambda = 1/2$. Using known lower bounds for uniformity testing in the standard model, we can construct distributions $\p_{\mathrm{yes}},\p_{\mathrm{no}},\q_{\mathrm{yes}},\q_{\mathrm{no}}$ where the following hold: 
\begin{enumerate}
    \item $\p_{\mathrm{yes}} = U_n$, $d_{\mathrm{TV}}(\p_{\mathrm{no}},U_n) > \varepsilon$.
    \item A set of $k = o(\sqrt{n}/\eps^2)$ samples from $\p_{\mathrm{yes}}$ or $\p_{\mathrm{no}}$ are statistically close, i.e. $d_{\mathrm{TV}}(\p_{\mathrm{yes}}^k,\p_{\mathrm{no}}^k) = o(1)$.
    \item $\rD_{\mathrm{yes}} = \frac{1}{2}(\p_{\mathrm{yes}} + \q_{\mathrm{yes}}) = \frac{1}{2}(\p_{\mathrm{no}} + \q_{\mathrm{no}}) = \rD_{\mathrm{no}}$ by defining $\q_{\mathrm{yes}},\q_{\mathrm{no}}$ appropriately.
\end{enumerate}
By (1), a uniformity tester must be able to distinguish these two cases. However, by (2) unless $k = \Omega(\sqrt{n}/\eps^2)$, the side information is not useful, and by (3) samples from $\rD_{\mathrm{yes}},\rD_{\mathrm{no}}$ are identically distributed (and in fact can even be simulated). 
Therefore, in semi-verified learning, essentially the entire burden is placed on the side information, i.e. we must have $k \geq \Omega(\sqrt{n}/\eps^2)$ pre-verified samples regardless of the number of corrupted samples.
In contrast, our \Cref{thm:simple-uniformity-alg} shows that in the verification query model it is possible to distinguish these cases using only $\mathrm{poly}(1/\eps)$ verified samples.

\paragraph{Additional Related Work.} One of the primary themes in our work is the use of queries to overcome uncertainty about sampled data. Although we are the first to study the verification query model we consider, other models have been studied with a similar theme. For instance, the \emph{huge object} \cite{DBLP:journals/theoretics/0001R23} and the \emph{confused collector} models \cite{DBLP:conf/innovations/PintoH24,DBLP:journals/corr/abs-2304-01374} for distribution testing. 

Another high-level theme of our work is to study algorithms equipped with both a weak oracle that is cheap, as well as a stronger oracle that is expensive. 
This is similar in spirit to \emph{active testing} \cite{DBLP:conf/focs/BalcanBBY12} and dual-access models considered in the context of metric clustering, e.g. \cite{DBLP:journals/siamcomp/BatuDKR05,DBLP:conf/soda/GuhaMV06}. 

Finally, we consider distribution testing of $\p$ with sample access to a mixture $\lambda \p + (1-\lambda)\q$. It is worth noting that the work of \cite{DBLP:conf/colt/Aliakbarpour0R19} considers the problem of testing whether a distribution $\p$ can be written as $\p = \lambda \q_1 + (1-\lambda)\q_2$ for some $\lambda$, given sample access to $\p,\q_1,\q_2$.  

We have primarily studied the problems of uniformity and closeness testing.
These fundamental distribution testing problems have been studied in a variety of other settings, see e.g. \cite{chakraborty2013power, canonne2015testing, acharya2018differentially, aliakbarpour2018differentially, aliakbarpour2019private, liu2024replicable, CKO24, aamand2025structure, diakonikolas2025replicable, DBLP:journals/corr/abs-2503-12518, blanc2025instance, goldreich2025location}.

\section{Technical Overview} \label{sec:tech}

%\hadley{I think this should be explained more slowly so that people appreciate the ideas. As in, even the fact that it suffices to just query collisions was something that didn't occur to us right away, and its the core reason why we can get $o(m)$ queries. We should also point out the challenge with this that we struggled with initially, and that flattening conveniently helps with. Like we can point out that (a) uniformity/closeness testers rely on counting collisions, and (b) if $\norm{\q}_{2}^2$ is not much larger than $\norm{\p}_2^2$ then we can easily simulate this, and (c) conveniently there is this flattening technique that makes this assumption "without loss of generality". Then the second stage is this bias estimation step (estimating the fraction of collisions that are $\p$-$\p$) that you wrote about. I think its cool how these things fit together (and lead to the sample-query trade-off) and we want people to appreciate that as well as opposed to giving the impression that our results follow easily from the standard model, and it helps to take people through the chain of challenges and ideas that we went through even if it seems trivial in retrospect. For the closeness testing discussion we should also point out the added technical annoyance of getting different number of samples from $\p_1,\p_2$.}

In the interest of emphasizing our main ideas, we assume $\lambda = \frac{1}{2}$ throughout the overview.
In reality, we design algorithms that succeed given a lower bound on $\lambda$ and spend $\bigO{\frac{1}{\lambda}}$ samples and queries to obtain this lower bound.
Note that this term never dominates as $\bigO{\frac{m}{\lambda}}$ samples are required to obtain sufficiently many samples from $\p$ and $\bigO{\frac{1}{\lambda}}$ queries are required just to detect a single sample from $\p$.

\subsection{Upper Bounds in the Sublinear Sample Regime} \label{sec:tech-sub}

\paragraph{Uniformity Testing.}

In uniformity testing, we are given sample access to $\rD = \frac{1}{2}(\p + \q)$ and asked to determine if $\p$ is uniform or far from uniform.
If we have an $m$ sample algorithm $\innerAlg_{0}$ for uniformity testing (without noise), then we have a distributionally robust algorithm with $O(m)$ samples and $O(m)$ queries.
For example, with $O(m)$ samples we guarantee that at least $m$ samples come from $\p$, and we can identify these samples with $O(m)$ queries and run $\innerAlg_{0}$.
Of course, $\Omega(m)$ samples are necessary.
At first glance, it is not obvious that there is even an algorithm with $o(m)$ queries.
In fact in the bias estimation problem, we prove that $\Omega(\eps^{-2})$ queries are necessary (\Cref{thm:bias-estimation-lb}).

For uniformity testing, we begin with the crucial observation that there are very few influential samples (although it may take many samples to observe them).
In particular, it suffices to count collisions to determine if $\p$ is uniform or not, and indeed this approach yields an optimal tester in the standard setting \cite{DBLP:journals/cjtcs/DiakonikolasGPP19}. 
Given $m$ samples, a uniform distribution produces roughly $\frac{m^2}{n}$ collisions while a distribution $\eps$-far from uniform produces at least $\frac{(1 + \eps^2) m^2}{n}$ collisions with high probability.
Thus, with $m = \bigTh{n^{1/2} \eps^{-2}}$ samples, a uniform distribution only produces $\bigTh{\eps^{-4}}$ collisions.
By only querying collisions, we can estimate how many were produced by $\p$, and this is sufficient to test uniformity of $\p$.
In an ideal setting, we would like to directly query collisions from $\p$, in which case we can even hope for query complexity independent of $n$! 
%In particular, this avoids querying samples where the target distribution $\p$ produces singletons and obtain large savings in query complexity.

However, we can only query collisions to $\rD$.
If $\norm{\rD}_{2}^{2} \gg \norm{\p}_{2}^{2}$, most collisions in the observed samples from $\rD$ will be produced by $\q$ and not $\p$. 
For example, if $\q$ is uniform over a subdomain of size $10m$, then $\q$ will produce $\Theta(m)$ collisions but we will not be able to detect the domain of $\q$ without using many queries.
Still, when $\norm{\rD}_{2}^{2} \leq t \norm{\p}_{2}^{2}$ is not too large, we obtain an efficient algorithm: repeatedly query \emph{random} collisions in the set of samples in order to estimate the number of collisions produced by $\p$. 
Since we find a $\p$-collision with every $O(t)$ queries, we obtain an algorithm using $O(t / \poly(\eps))$ queries.

The final ingredient is to bound $t$.
Intuitively, if $\rD$ has significant mass on any element, we can detect this in $m$ samples, i.e. we can learn the mass of any element with mass greater than $\frac{1}{m}$.
This intuition is captured by the flattening technique \cite{diakonikolaskane2016} which states that using $O(m)$ samples, we can assume without loss of generality that $\norm{\rD}_{2}^{2} \leq \frac{1}{m}$ (and therefore $\norm{\q}_{2}^{2} \leq \frac{1}{m}$).

This leads to our final algorithm.
We use $O(m)$ samples to flatten (\cite{diakonikolaskane2016}) the input distribution $\rD$ so that every bucket has mass at most $\frac{1}{m}$.
Under this assumption, if we take another $m$ samples the observed number of collisions is in expectation at most $\binom{m}{2} \norm{\rD}_{2}^{2} = O(m)$.
Now, consider the two cases.
If $\p$ is uniform, it produces at most $\binom{m}{2} \norm{\p}_{2}^{2} \leq \frac{m^2}{n}$ collisions.
If $\p$ is $\eps$-far from uniform, a standard calculation (see e.g. \cite{Canonne:Survey:ToC}) shows that it produces at least $\binom{m}{2} \norm{\p}^{2}_2 \geq \frac{m^2(1 + \eps^2)}{n}$ collisions.
If we query a random collision, we find a collision between two elements sampled from $\p$ with probability at most $\frac{m}{n}$ in the former case and at least $\frac{(1 + \eps^2)m}{n}$ in the latter case.
Since this reduces the problem of uniformity testing to differentiating a coin with bias $\frac{m}{n}$ from one with bias $\frac{(1 + \eps^2)m}{n}$, it suffices to query $\frac{n}{m \eps^4}$ samples.

There is still one remaining issue: $\p$ can produce many collisions for two reasons: 1) $\p$ is $\eps$-far from uniform, or 2) more samples come from $\p$ than $\q$.
We must ensure that the signal from 1) outweighs the noise from 2).
To address this issue, we compute an estimate $\hat{m}_{\p}$ of $m_{\p}$, where $m = m_{\p} + m_{\q}$ denote the number of samples drawn from $\p, \q$ respectively.\footnote{Alternatively, we may increase the number of samples to address this issue. In particular, $m_{\p} = \frac{m}{2} \pm O(\sqrt{m})$ so that in the completeness case we have the number of $\p$-collisions is at most $\frac{m_{\p}^{2}}{n} \ll \frac{m^2 + O(m^{1.5})}{n}$ while in the soundness case the number of $\p$-collisions is at least $\frac{m_{\p}^2(1 + \eps^2)}{n} \gg \frac{m^2 + m^2 \eps^2}{n}$. Thus, $m \gg \frac{1}{\eps^{4}}$ ensures that $\eps^2 m^2 \gg m^{1.5}$.}
Clearly, a randomly queried sample reveals $\p$ with probability $\frac{m_{\p}}{m}$. 
With $O(\eps^{-4})$ queries, we can estimate $m_{\p}$ up to accuracy $\eps^2 m$.
In particular, we can estimate $\frac{m_{\p} (1 + 0.5 \eps^{2})}{n}$ up to accuracy $\ll \frac{\eps^2 m}{n}$ which suffices to ensure that the signal from 1) dominates the noise from 2).

\paragraph{Adversarially Robust Uniformity Testing.}
We can make our uniformity testing algorithm adversarially robust (\Cref{thm:adversarially-robust-uniformity-alg}).
A close inspection of our distributionally robust uniformity tester reveals that it relies on the following two facts:
\begin{enumerate}
    \item There are enough samples from $\p$ to guarantee that there are many $\p$-collisions when $\p$ is far from uniform and few $\p$-collisions when $\p$ is uniform.
    This is guaranteed by taking $\gg \frac{\sqrt{n}}{\lambda \eps^2}$ samples.
    \item There are not too many collisions in $\rD$.
    This is guaranteed by the flattening procedure \cite{diakonikolaskane2016}.
    As discussed above, we use this to bound the number of queries, since this ensures we can efficiently find a $\p$-collision by querying a random collision.
\end{enumerate}
In the adversarial setting, we can no longer first apply the flattening algorithm to the first half of samples and then test uniformity on the flattened distribution with the remaining samples, since we are no longer guaranteed sample access to the same distribution.
However, if we randomly subsample our data into a flattening dataset and a testing dataset, we can still ensure that there are not too many collisions in the testing dataset (since random sub-sampling ensures that any bucket with high frequency in the testing dataset must also have high frequency in the flattening dataset).
To meet the first requirement, note that we still collect sufficiently many samples from $\p$ to ensure that the number of collisions are separated in the completeness and soundness cases.

\paragraph{Closeness Testing.}
In closeness testing, we are given sample access to $\rD_1 = \frac{1}{2}(\p_1 + \q_1)$ and $\rD_2 = \frac{1}{2}(\p_2 + \q_2)$ and asked to determine if $\p_1, \p_2$ are identical or far from each other.
Our algorithm aims to use queries to simulate a collision based closeness testing algorithm \cite{DBLP:journals/cjtcs/DiakonikolasGPP19}.
In the standard setting, \cite{DBLP:journals/cjtcs/DiakonikolasGPP19} show that there is an optimal collision based estimator, which computes $Z = C_{\p_1} + C_{\p_2} - \frac{m - 1}{m} C_{\p_1\p_2}$ where $C_{\p_1}, C_{\p_2}, C_{\p_1\p_2}$ denote the $\p_1$-collisions, $\p_2$-collisions and collisions between $\p_1, \p_2$.
It is easy to see that $\EX[Z] = \binom{m}{2} \norm{\p_1 - \p_2}_{2}^{2}$ so that $Z$ gives a closeness tester in $\ell_{2}$ norm.

Given an upper bound $\norm{\p_{1}}_{2}^{2}, \norm{\p_{2}}^{2} \leq b$, \cite{DBLP:journals/cjtcs/DiakonikolasGPP19} show that $Z$ does not deviate from its mean by more than $\simeq m^2 \max(\alpha, \eps^{2})$ where $\alpha = \norm{\p_1 - \p_2}_{2}^{2}$.
Thus, our goal is to estimate the collision counts up to this error.
For instance, we hope to estimate $\hat{C}_{\p_1, \p_2}$ such that $|\hat{C}_{\p_1, \p_2} - C_{\p_1, \p_2}| \ll m^{2} \eps^{2}$.
Since there are at most $O(m^2 b) = O(m)$ collisions (as flattening ensure $b = O(1/m)$) this is equivalent to bias estimation up to accuracy $m \eps^{2}$ which requires $\bigO{\frac{1}{m^{2} \eps^{4}}}$ queries.
Since $\norm{\p_{1} - \p_{2}}_{1} \geq \eps$ implies that $\norm{\p_{1} - \p_{2}}_{2} \geq \frac{\eps}{\sqrt{n}}$, we obtain an algorithm using $\bigO{\frac{n^2}{m^2 \eps^{4}}}$ queries.
$\hat{C}_{\p_1}, \hat{C}_{\p_2}$ can be estimated similarly.

However, it is still not obvious that we can simply compare the estimated test statistic $\hat{Z} = \hat{C}_{\p_1} + \hat{C}_{\p_2} - \frac{m - 1}{m} \hat{C}_{\p_1 \p_2}$ to some predetermined threshold. 
In particular, we do not have control over the number of samples from $\rD_1, \rD_2$ that are drawn by $\p_1,\p_2$ and we do not have a query-efficient way to compute these numbers to sufficient accuracy, unlike in the uniformity testing case, where cruder estimates are sufficient. 
%\chris{is this true? I mean we can query random samples to estimate these numbers, maybe say that this would be prohibitively expensive in query complexity}); 
(Of course, we do know that these numbers are well-concentrated around their expectation $\lambda m$, but this is not enough on its own.) Thus, another technical issue is that differing numbers of samples from $\p_1,\p_2$ can cause the collision counts $\hat{C}_{\p_1}, \hat{C}_{\p_2}$ to look very different, even when $\p_1 = \p_2$, causing the tester to reject. Therefore, bounding the difference $|Z - \hat{Z}|$ requires carefully bounding the effects of these "extra" collisions (see the proof of \Cref{thm:closenessUB} for details). We remark that this technical issue would be avoided in a simpler, but less realistic model, where an \emph{exactly} $\lambda$-fraction of samples are drawn from the target distributions $\p_1,\p_2$. However, we feel that this would significantly sacrifice the applicability and appeal of the model.
%\chris{feel like I'm not really explaining this technical issue well.}

\subsection{Lower Bounds}

\paragraph{Bias Estimation Lower Bound.}
We begin with a simple lower bound for bias estimation.
Our construction for general $\lambda$ is slightly more involved.
In the completeness case $\p, \q$ are both fair coins (i.e. uniform), while in the soundness case $\p$ has bias $\frac{1}{2} + \eps$ while $\q$ has bias $\frac{1}{2} - \eps$ (i.e. far from uniform).
Since in both cases, $\rD$ is a fair coin, any given sample reveals no information about $\p$.
In general, our hard instances will ensure that $\rD$ is identical in both completeness and soundness cases so that samples alone reveal no information.
Consider a sample that shows heads.
If both $\p, \q$ are fair coins, a query on this sample reveals $\p$ with probability $\frac{1}{2}$.
On the other hand, if $\p$ has bias $\frac{1}{2} + \eps$, a query reveals $\p$ with probability $\frac{1}{2} + \eps$.
In particular, each query simulates a coin flip.
A standard bias estimation lower bound then shows that $\Omega(\eps^{-2})$ queries are necessary.

\paragraph{Uniformity Testing Lower Bound.}
We begin by showing that any algorithm that takes $m \leq n^{0.99}$ samples requires $\bigOm{\frac{n}{m}}$ queries (see \Cref{thm:dist-contamination-lb}).
Consider the following instance (see \Cref{def:uniformity-hard-instance}).
There are $m$ heavy buckets where $\p$ has mass $\frac{1}{n}$ and $\q$ has mass $\frac{4\eps}{n} + \frac{1 - 4 \eps}{m}$.
In the completeness case, $\p$ has mass $\frac{1}{n}$ on the remaining $n - m$ light buckets while $\q$ has mass $\frac{4 \eps}{n}$.
In the soundness case, $\p, \q$ have $\frac{1 + 2 \eps}{n}$ and $\frac{2 \eps}{n}$ respectively on half of the remaining light buckets and mass $\frac{1 - 2 \eps}{n}$ and $\frac{6 \eps}{n}$ respectively on the other half.
In both cases, the light buckets collectively produce at most $\bigTh{\frac{m^2}{n}}$ collisions while $\rD$ produces at least $m^2 \norm{\rD}_{2}^{2} = \Theta(m)$ collisions.
In fact, for any fixed frequency $c$, there are $\Theta(m)$ $c$-collisions among the heavy buckets.
Since $\p$ is identical on heavy buckets, these queries should not reveal any information and any correct algorithm must query at least one light bucket.
However, any collision query lands on a light bucket with probability at most $\bigO{\frac{m^2/n}{m}} = \bigO{\frac{m}{n}}$, so any uniformity tester requires $\bigOm{\frac{n}{m}}$ queries.

Examining this instance more closely also yields the $\bigOm{\frac{n}{m \eps^{4}}}$ lower bound. 
Consider an algorithm that randomly queries collisions.
The distribution produces $\Theta(m)$ collisions (mostly on heavy buckets).
On the other hand $\p$ produces $\frac{m^2}{n}$ collisions when $\p$ is uniform and $\geq \frac{m^2(1 + \eps^{2})}{n}$ collisions when $\p$ is far from uniform.
Thus, a randomly chosen collision is a $\p$ collision with probability $\frac{m}{n}$ in the former case and $\frac{m(1 + \eps^2)}{n}$ in the latter.
A standard lower bound for bias estimation shows that $\bigOm{\frac{n}{m \eps^{4}}}$ queries are necessary to distinguish the two cases.
Here, note that since $\frac{m}{n} \ll 1$, only $\bigO{\frac{n}{m \eps^{4}}}$ queries are necessary to estimate the bias of a coin with small entropy in contrast to the $\bigO{\frac{n^2}{m^2 \eps^{4}}}$ queries necessary to estimate the bias of a coin with high entropy up to the same accuracy.

To formalize our intuition (and obtain a lower bound for arbitrary algorithms), we employ a mutual information based framework for obtaining the tight $\bigOm{\frac{n}{m \eps^{4}}}$ lower bound, which first appeared in the distribution testing context in \cite{diakonikolaskane2016}.
Let $\innerAlg$ be a uniformity tester with $m$ samples and $q$ queries.
First, we make two standard assumptions 1) $\innerAlg$ takes $\Poi(m)$ samples, and 2) each bucket is heavy with probability $\frac{m}{n}$ and light with probability $\frac{n - m}{n}$.\footnote{In the soundness case, each bucket is each type of light bucket with equal probability.}
For 1) it is easy to see that $\Poi(m) = O(m)$ with high probability.
For 2), while $\p, \q$ chosen in this way may not be distributions, they are non-negative measures with total mass close to $1$.
Normalizing by $\norm{\p}_{1}, \norm{\q}_{1}$ therefore should not significantly change the distribution.
These assumptions allow us to assume that samples observed from each bucket are independent, which will be crucial to our analysis.
In particular, we may assume that the algorithm observes $Y_i \sim \Poi(m \rD[i])$ samples from the $i$-th bucket, where $Y_i = A_i + B_i$ with $A_i \sim \Poi(m \p[i]/2), B_i \sim \Poi(m \q[i]/2)$ are sampled independently.

Suppose an algorithm $\innerAlg$ takes $m$ samples and $q$ queries.
We will in fact assume that a single query reveals all the samples on the queried bucket, and prove a lower bound on the number of queried buckets.
Note that this assumption only strengthens $\innerAlg$.
Let $Q \subset [n]$ be the set of queried buckets.
Let $T$ be the input to the algorithm.
In particular, let $T = (T_1, \dotsc, T_n)$ where $T_i = (A_i, B_i)$ if $i \in Q$ is queried and $T_i = Y_i$ otherwise.
We design an adversary that flips an unbiased bit $X$: and produces a completeness instance ($\p$ uniform) if $X = 0$ and a soundness instance ($\p$ far from uniform) if $X = 1$.
To prove a lower bound on $q$, we require two facts:
\begin{enumerate}
    \item If $\innerAlg$ is correct, then $I(T:X) = \Omega(1)$.
    \label{it:overview:uniformity-mi-lb}
    \item If $\innerAlg$ makes $q$ queries, then $I(T:X) = \bigO{q \cdot \frac{\eps^{4} m}{n}}$.
    \label{it:overview:uniformity-mi-ub}
\end{enumerate}
The lower bound on $q$ immediately follows from the above two observations.
\Cref{it:overview:uniformity-mi-lb}, which states that any correct algorithm must observe samples correlated with the desired output, follows essentially from the data processing inequality.
Since $\innerAlg$ is a randomized function such that $\innerAlg(T) = X$ with probability $\geq \frac{2}{3}$, $\innerAlg(T)$ and therefore $T$ must be correlated with $X$ (see \Cref{lemma:mutual-info-bound}).
\Cref{it:overview:uniformity-mi-ub} says that few queries do not have significant correlation with the desired output $X$.
To prove the bound \Cref{it:overview:uniformity-mi-ub}, we observe that if $i \not\in Q$ is not queried, $T_i = Y_i \sim \Poi(m \rD[i])$ is independent of $X$.
Furthermore, since the samples of bucket $i$ are independent (even conditioned on $X$) it suffices to bound the mutual information gained from a single bucket.
This follows from a technical calculation (see \Cref{lemma:uniformity-one-bucket-mi} for details).

\paragraph{Closeness Testing Lower Bound.}

Consider the following instance (see \Cref{def:closeness-hard-instance}).
As before, we assume $\innerAlg$ takes $\Poi(m)$ samples.
Each bucket is independently heavy with probability $\frac{m}{n}$, light with probability $\frac{n - m}{n}$.
Light buckets are type 1 and 2 with equal probability.
On heavy buckets $\p_1, \q_1, \p_2, \q_2$ have mass $\frac{1 - \eps}{m}$.
In the completeness case, on type 1 light buckets $\p_1, \p_2$ have mass $\frac{3 \eps}{2(n - m)}$ and $\q_1, \q_2$ have mass $\frac{\eps}{2(n - m)}$, and on type 2 buckets, $\q_1, \q_2$ have mass $\frac{3 \eps}{2(n - m)}$ and $\p_1, \p_2$ have mass $\frac{\eps}{2(n - m)}$.
Note that $\p_1 = \p_2$.
In the soundness case, on type 1 light buckets $\p_1, \q_2$ have mass $\frac{3 \eps}{2(n - m)}$ and $\q_1, \p_2$ have mass $\frac{\eps}{2(n - m)}$, and on type 2 buckets, $\q_1, \p_2$ have mass $\frac{3 \eps}{2(n - m)}$ and $\p_1, \q_2$ have mass $\frac{\eps}{2(n - m)}$.
Note $\p_1, \p_2$ are $\eps$-far apart with high probability.
Now, consider the number of collisions produced between every pair of distributions in the two scenarios.

\begin{table}[h!]
    \centering
    \begin{tabular}{|| c | c c c c ||} 
         \hline
         & $\p_1$ & $\q_1$ & $\p_2$ & $\q_2$ \\ [0.5ex] 
         \hline\hline
         $\p_1$ & $\frac{m}{2} + \frac{5 \eps^2 m^2}{2 n}$ & $\frac{m}{2} + \frac{3 \eps^2 m^2}{2 n}$ & $\frac{m}{2} + \frac{5 \eps^2 m^2}{2 n} {\color{red} - \frac{\eps^2 m^2}{n}}$ & $\frac{m}{2} + \frac{3 \eps^2 m^2}{2 n} {\color{red} + \frac{\eps^2 m^2}{n}}$ \\ 
         $\q_1$ & & $\frac{m}{2} + \frac{5 \eps^2 m^2}{2 n}$ & $\frac{m}{2} + \frac{3 \eps^2 m^2}{2 n} {\color{red} + \frac{\eps^2 m^2}{n}}$ & $\frac{m}{2} + \frac{5 \eps^2 m^2}{2 n} {\color{red} - \frac{\eps^2 m^2}{n}}$ \\
         $\p_2$ & & & $\frac{m}{2} + \frac{5 \eps^2 m^2}{2 n}$ & $\frac{m}{2} + \frac{3 \eps^2 m^2}{2 n}$ \\
         $\q_2$ & & & & $\frac{m}{2} + \frac{5 \eps^2 m^2}{2 n}$ \\
         \hline
    \end{tabular}
    \caption{The expected number of collisions between every pair of distributions. 
    For simplicity of expressions, we assume we draw exactly $m$ samples from $\p, \q$, $\binom{m}{2} \simeq \frac{m^2}{2}$, $1 - \eps \simeq 1$, and $n - m \simeq n$. 
    The expected number of collisions shown in black is in the completeness case. 
    Any changes in the soundness case are shown in red.
    No red indicates that the expected number of collisions is identical in both cases.}
    \label{table:closeness-lb-overview}
\end{table}

As before, consider an algorithm $\innerAlg$ that queries random collisions.
Suppose for example that this algorithm attempts to test closeness by counting the number of $(\p_1, \p_2)$-collisions.
On average, $\rD_1, \rD_2$ produces $\Theta(m)$ collisions.
In the completeness case, there are $\frac{m}{2} + \frac{5 \eps^2 m^2}{2n}$ $(\p_1, \p_2)$-collisions while in the soundness case there are only $\frac{m}{2} + \frac{3 \eps^2 m^2}{2n}$.
In particular, a random collision yields $(\p_1, \p_2)$ with probability $\Theta(1) + \frac{5 \eps^2 m}{n}$ in the former case and $\Theta(1) + \frac{3 \eps^2 m}{n}$ in the latter.
Estimating the bias of a coin (with bias $\Omega(1)$) up to accuracy $\frac{\eps^2 m}{n}$ requires $\bigO{\frac{n^2}{m^2 \eps^{4}}}$ queries.
Note that in contrast to uniformity testing, we obtain a stronger lower bound as we now have to estimate the bias of a coin with high entropy.

To formalize this (again for arbitrary algorithms), we employ a similar mutual information framework to the uniformity testing problem.
Much of the formal argument is similar, except we obtain a stronger result that states that querying a bucket $i$ only reveals $\bigO{\frac{\eps^{4} m^{2}}{n^{2}}}$ information (\Cref{lemma:closeness-one-bucket-mi}).
The query lower bound then follows identically.

\subsection{Query-Optimal Testing with Superlinear Samples} \label{sec:tech-sup}

It is a standard fact that $O(n/\eps^2)$ samples are sufficient to \emph{learn} a discrete distribution $\p$ up to total variation distance $\eps$. When this many samples are allowed, we exploit this fact to obtain testers with query complexity $\widetilde{O}(1/\eps^2\lambda)$, which is optimal up to $\mathrm{polylog}(1/\eps\lambda)$ factors. Our testers in this regime follow a different approach from the sublinear sample setting, outlined in \Cref{sec:tech-sub}, which we now describe. %Recall that in this technical overview, we are assuming mixture parameter $\lambda = 1/2$. %Also note that our uniformity tester in this setting (\Cref{thm:adversarial-uniformity-alg}) works even when the contaminated samples are produced \emph{adversarially} (as opposed to being produced by a distribution $\q$). 

\paragraph{Closeness Testing.} First, using $O(n \log n)$ samples up-front, a flattening argument (\Cref{thm:strong-mixture-flat}) allows us to assume our distributions have bounded infinity norm, namely $\norm{\rD_1}_{\infty},\norm{\rD_2}_{\infty} \leq O(1/n)$ and $\norm{\p_1}_{\infty},\norm{\p_2}_{\infty} \leq O(1/\lambda n)$. Another simple reduction (\Cref{lem:unifier}) allows us to further assume that $\rD_b[i] = \Theta(1/n)$ for every $i \in [n]$, and so we will assume this throughout the discussion.

We begin by taking $m = \widetilde{O}(n/\eps^2\lambda)$ samples (enough for learning) from each mixture $\rD_b$ (for $b \in \{1,2\}$) and let $Z_1^{(b)},\ldots,Z_n^{(b)}$ denote the partition on these samples by domain element. For each $i \in [n]$, let $P_b[i],Q_b[i],R_b[i] = P_b[i] + Q_b[i]$ denote the number of samples in $Z_i^{(b)}$ which were drawn according to $\p_b,\q_b,\rD_b$, respectively. Note that the tester knows $R_b[i]$ without making any queries, but not $P_b[i],Q_b[i]$. By the learning bound, we have 
\[
\frac{P_b[i]}{\lambda m} \approx \p_b[i] ~\text{ and }~ \frac{R_b[i]}{m} \approx \rD_b[i] 
\approx \frac{1}{n} \text{.}
\]
(In particular, we need a small additive error estimate of $\p_b[i]$, which relies on our bound on the infinity norms guaranteed by the flattening reduction.) These estimates would clearly be useful for testing closeness of $\p_1,\p_2$, but as mentioned we do not know $P_b[i]$. However, by performing a verification query on a uniform random sample in $Z^{(b)}_i$, we simulate a $\rho_b[i] := P_b[i]/R_b[i]$ probability coin. Thus, the idea is to use queries to simulate these coin flips in order to estimate $P_b[i]/\lambda m \approx \p_b[i]$. If in doing so we can detect that $\p_1[i] \neq \p_2[i]$ for some $i \in [n]$, then we can safely reject. %Therefore, our strategy is to sample random domain elements $i \in [n]$ and attempt to estimate $\p_b[i]$ by simulating trials to a $\rho[i]$-probability coin.
%First, let us ignore for a moment the adversary's ability to make $A_i$ very large for certain $i$. I.e. let's assume for now that $A_i = P_i$ for all $i$. 

When $\norm{\p_1 - \p_2} \geq \eps$ (and using our assumption that $\norm{\p_b}_{\infty} \leq O(1/\lambda n)$), a bucketing argument reveals that for some $j \leq O(\log(1/\eps\lambda))$, sampling $i \sim [n]$ uniformly yields $|\p_1[i] - \p_2[i]| \geq \eps 2^j/n$ with probability at least $\gtrsim 2^{-j}$. Thus, for each such $j$, we sample a set $X_j$ of $\approx 2^j$ uniform samples. Then, for each element $X_j$ it suffices to obtain an estimate of $\p_b[i]$ within additive error $\lesssim \eps 2^j/n$. Now, since 
\[
\rho_b[i] = \frac{P_b[i]}{R_b[i]}\approx \p_b[i] \cdot \frac{\lambda m}{R_b[i]} \approx \p_b[i] \cdot \lambda n
\]
(where $R_b[i]$ is known), this translates to a $\lesssim \lambda \eps 2^j$ additive error estimate of $\rho_b[i]$. By Markov's inequality, we can argue that every $i \in X_j$ satisfies $\p_b[i] \lesssim 2^j/n$ and so $\rho_b[i] \lesssim \lambda 2^j$. Then, using a standard bias estimation (\Cref{lem:bias-add}), this costs $\lesssim \frac{\lambda 2^j}{(\lambda \eps 2^j)^2} = \frac{1}{\eps^2 \lambda 2^j}$ queries to obtain. (Each query simulates an independent $\mathrm{Bern}(\rho_b[i])$ trial). Summing over all $j$ yields the desired query complexity. 

\paragraph{Uniformity Testing against Adversarial Contamination.} In the above closeness tester, we used the flattening technique of \cite{diakonikolaskane2016} adapted to obtain a bound on the infinity norm ($\norm{\p_b}_{\infty} \lesssim 1/\lambda n$) using $O(n \log n)$ samples. This was useful for two reasons:
\begin{enumerate}
    \item (Sample complexity) Our approach relies on $P_b[i]/\lambda m$ being a low additive error ($\lesssim \eps/n$) empirical estimate of $\p_b[i]$. To guarantee this, the sample complexity depends on $\p_1[i],\p_2[i]$ (see \Cref{thm:learn}, \Cref{cor:learn}). In particular, to get sample complexity $\widetilde{O}(n \cdot \mathrm{poly}(\eps^{-1},\lambda^{-1}))$, we need a bound $\p_1[i],\p_2[i] \leq \mathrm{poly}(\eps^{-1},\lambda^{-1})/n$. 
    \item (Query complexity) Our bucketing argument allowed us to say that sampling $i \sim [n]$ \emph{uniformly at random} will, for some $j \leq O(\log(1/\eps\lambda))$, guarantee $|\p_1[i]-\p_2[i]| \gtrsim \eps 2^j/n$ with probability $\gtrsim 2^{-j}$. Then, since the test elements $X_j$ are sampled uniformly, we can argue $\p_b[i] \lesssim 2^j/n$ by Markov's inequality. This allowed us to give an upper bound on the coin probability $\rho_b[i]$, crucially helping us bound the number of queries needed for the bias estimation step (\Cref{lem:bias-add}). Without flattening, one will also need to sample test elements from $\p_b$, resulting in possibly much larger $\rho_b[i]$, and leading to more expensive bias estimations.
\end{enumerate}

In the context of \emph{adversarial contamination}, it is not clear how we can use flattening to obtain the properties we need for this approach. However, in \emph{uniformity testing}, the above issues go away since we are only comparing $\p[i]$ against the uniform probability $1/n$. In particular, as $\p[i]$ becomes much larger than $1/n$, the error we can tolerate in the estimate of $\p[i]$ also increases. Concretely, when $\p[i] > 2/n$, we only need a constant factor multiplicative estimate of $\p[i]$. In summary, for uniformity testing we do not rely on flattening to achieve optimal query complexity and this allows for our approach to tolerate adversarial contamination, with some additional technical effort.

\section{Preliminaries}

Let $[n] = \set{1, \dotsc, n}$.
Let $a \ll b$ denote that $a$ is a (large) constant factor less than $b$.
We will use $\p, \q, \rD, \distribution$ to denote non-negative measures (typically probability distributions).
Given measures $\p_1, \dotsc, \p_t$, let $\sum_{i = 1}^{t} c_i \p_i$ denote the measure that has mass $\sum c_i \p_i$.
In particular, $\frac{\p + \q}{2}$ is a balanced mixture of $\p, \q$.
We typically use $\p$ to denote the target distribution (for which we would like to test some property) and $\q$ to denote the noise distribution.

For a given subset $S \subseteq [n]$, let $(\p|S)$ denote the conditional distribution of $\p$ on $S$ and $\p[S]$ the pseudodistribution induced by restricting $\p$ to $S$.
For example, $\p[i]$ is the mass of the $i$-th bucket under $\p$.

Given a multi-set of samples $T \in [n]^{m}$ over $[n]$, we say the frequency of the $i$-th bucket is the number of occurrences of $i$ in $T$.
A $c$-collision is a tuple of $c$ indices $(i_1, \dotsc, i_c)$ such that $T_{i_1} = \dotsc = T_{i_c}$.
If $T$ is sampled from a mixture $\rD = \lambda \p + (1 - \lambda) \q$, a $\p$-collision (resp. $\q$-collision) is a pair of indices $(i, j)$ where $T_i = T_j$ and both are sampled from $\p$ (resp. $\q$).

We give some useful concentration bounds.

\begin{lemma}[Poisson Concentration (see e.g. \cite{CCPoisson}]
    \label{lemma:poisson-concentration}
    Let $X \sim \Poi(\lambda)$.
    Then for any $x > 0$,
    \begin{equation*}
        \max(\Pr(X \geq \lambda + x), \Pr(X \leq \lambda - x)) < e^{-\frac{x^2}{2(\lambda + x)}} \text{.}
    \end{equation*}
\end{lemma}

Note that $\norm{\p}_{c}^{c}$ is the probability a given $c$-tuple of samples collide.
Linearity of expectation then yields the following lemma.

\begin{lemma}
    \label{lemma:c-moment-exp}
    Let $c \geq 2$ be an integer.
    Let $X$ be the number of $c$-collisions given $m$ i.i.d. samples drawn from $\p$.
    Then, $\EX[X] = \binom{m}{c} \norm{\p}_{c}^{c}$.
\end{lemma}

Using Markov's inequality, we can obtain a simple bound on the probability that a $c$-collision occurs when $\EX[X] \ll 1$.

\begin{corollary}
    \label{cor:c-collision-ub}
    Let $c \geq 2$ be an integer and $t > 1$.
    Let $X$ be the number of $c$-collisions given $m$ i.i.d. samples drawn from $p$.
    Then, $\Pr(X \geq t) \leq \frac{\norm{\p}_{c}^{c} \binom{m}{c}}{t} \leq O_{c}(\norm{\p}_{c}^{c} m^{c}/t)$.
\end{corollary}

% \begin{lemma}
%     \label{lemma:poisson-concentration}
%     Let $X \sim \Poi(\lambda)$ be a Poisson random variable with mean $\lambda$.
%     Then
%     \begin{equation*}
%         \Pr(|X - \lambda| > t) < 2 e^{-t^2/(2(\lambda + t))}\text{.}
%     \end{equation*}
% \end{lemma}

% \begin{lemma}
%     \label{lemma:poisson-tail-bound}
%     Let $X \sim \Poi(\lambda)$ be a Poisson random variable with mean $\lambda$.
%     Then
%     \begin{equation*}
%         \Pr(X > t \lambda) < \left( \frac{e \lambda}{t} \right)^{t} e^{- \lambda}\text{.}
%     \end{equation*}
% \end{lemma}

\subsection{Information Theory}

We need the following tools from information theory.

\begin{definition}
    \label{def:entropy}
    Let $X$ be a random variable over domain $\domain$.
    The \emph{entropy} of $X$ is
    \begin{equation*}
        H(X) = \sum_{x \in \domain} \Pr(X = x) \log \frac{1}{\Pr(X = x)}.
    \end{equation*}
\end{definition}

\begin{definition}
    \label{def:k-l-divergence}
    Let $X, Y$ be random variables over domain $\domain$.
    The \emph{KL-divergence} of $X, Y$ is
    \begin{equation*}
        D(X||Y) = \sum_{x \in \domain} \Pr(X = x) \log \frac{\Pr(X = x)}{\Pr(Y = x)}.
    \end{equation*}
\end{definition}

\begin{definition}
    \label{def:mutual-info}
    Let $X, Y$ be random variables over domain $\domain$.
    The \emph{mutual information} of $X, Y$ is
    \begin{equation*}
        I(X: Y) = \sum_{x, y \in \domain} \Pr((X, Y) = (x, y)) \log \frac{\Pr((X, Y) = (x, y))}{\Pr(X = x) \Pr(Y = y)}.
    \end{equation*}
\end{definition}

% We also state the data processing inequality.
\begin{lemma}[Data Processing Inequality]
    \label{lemma:data-processing-inequality}
    Let $X, Y, Z$ be random variables over domain $\domain$ such that $Z$ is independent of $X$ conditioned on $Y$.
    Then
    \begin{equation*}
        I(X: Y) \geq I(X: Z).
    \end{equation*}
\end{lemma}

The following shows that 

\begin{restatable}{lemma}{MutualInfoBound}
    \label{lemma:mutual-info-bound}
    Let $X \sim \Bern(\frac{1}{2})$ and $T$ be a random variable possibly correlated with $X$.
    If there exists a (randomized) function $f$ such that $f(T) = X$ with probability at least $51\%$ probability, then $I(X:T) \geq 2 \cdot 10^{-4}$.
\end{restatable}

The following inequality is an useful tool for bounding the mutual information.

% \begin{claim}[Convexity of KL-divergence]
%     \label{clm:kl-convexity}
%     Let $w \in (0, 1)$, and $p_1, p_2, q_1, q_2$ be  probability distributions over the same domain.
%     Then it holds
%     $$
%     \KL( w p_1 + (1 - w) p_2 || w q_1 + (1-w) q_2 )
%     \leq
%     w  \KL( p_1 || q_1 )
%     + (1 - w) \KL( p_2 || q_2 ).
%     $$
% \end{claim}

% \begin{claim}[Additivity of KL-divergence]
%     \label{clm:kl-additive}
%     Let $w \in (0, 1)$, and $p_1, p_2, q_1, q_2$ be  probability distributions such that $p_1, q_1$ share the same domain and $p_2, q_2$ share the same domain.
%     Then it holds
%     $$
%     \KL( p_1 \otimes p_2 || q_1 \otimes q_2 )
%     =
%     \KL( p_1 || q_1 ) + \KL(p_2 || q_2).
%     $$
% \end{claim}

% \begin{claim}[Mutual Information and KL-divergence]
%     \label{clm:mi-kl-relationship}
%     Let $X$ be an indicator variable, and $H_0, H_1$ be a pair of distributions.
%     Let $H$ be the random variable that follows the distribution of $H_0$ if $X = 0$ and $X_1$ if $X = 1$. 
%     Then it holds
%     $$
%     I(X:H) =  \frac{1}{2} \KL\left( H_0 : \frac{1}{2} \left( H_1+H_0\right) \right) 
%     + \frac{1}{2} \KL\left( H_1: \frac{1}{2} \left( H_1+H_0\right) \right).
%     $$
% \end{claim}

\begin{restatable}[Asymptotic Upper Bound of Mutual Information]{claim}{MIAsymp}
    \label{claim:MI_asymp}
    Let $X$ be an unbiased uniform random bit, and $M$ be a discrete random variable possibly correlated with $X$.
    % such that $\Pr[M=a|X=0] = \Theta(1)\Pr[M=a|X=1],$ for all $a$ within the support of $M$,
    Then the mutual information between the two random variables satisfies that
    \begin{equation*}
        I(X:M) \leq 2 \sum_{a}  \frac{(\Pr[M=a|X=0] - \Pr[M=a|X=1])^2}{\Pr[M=a|X=0] + \Pr[M=a|X=1]} \text{.}
    \end{equation*}
\end{restatable}

    We give the proof of these standard facts for completeness in \Cref{sec:omitted-proofs}.

% \begin{lemma}
%     \label{lemma:binomial-KL}
%     Let $m \geq 0$ be an integer and $0 \leq a < b \leq 1$.
%     Let $X \sim \Binom(m, a)$ and $Y \sim \Binom(m, b)$.
%     Then
%     \begin{equation*}
%         \KL(X||Y), \KL(Y||X)= \bigO{\frac{m (b - a)^2}{\min(a, b, 1 - a, 1 - b)}} \text{.}
%     \end{equation*}
% \end{lemma}

% \begin{proof}
%     By \Cref{clm:kl-additive}, we have
%     \begin{align*}
%         \KL(X||Y) &= m \left( a \log\left(a/b\right) + (1 - a) \log\left(\frac{1 - a}{1 - b}\right) \right) \\
%         &\leq m \left( a\left(\frac{a}{b} - 1\right) + (1 - a)\left(\frac{1-a}{1-b} - 1\right)\right) \\
%         &= \bigO{\frac{m(b - a)^2}{\min(b, 1-b)}} \text{.}
%     \end{align*}
%     The bound on $\KL(Y||X)$ is identical.
% \end{proof}

\subsection{Split Distributions and Flattening}

We require the flattening technique of \cite{diakonikolaskane2016}.

\begin{definition}[Definition 2.4 of \cite{diakonikolaskane2016}]
    \label{def:split-distributions}
    Given a distribution $\rD$ on $[n]$ and a multi-set $S$ of elements of $[n]$, define the split distribution $\rD_{S}$ on $[n + |S|]$ as follows: 
    For $1 \leq i \leq n$, let $a_i$ denote $1$ plus the number of elements of $S$ that are equal to $i$. 
    Thus, $\sum_{i = 1}^{n} a_i = n + |S|$.
    We can therefore associate the elements of $[n + |S|]$ to elements of the set $B = \set{(i,j) : i \in [n], 1 \leq j \leq a_i}$.
    We now define a distribution $\rD_{S}$ with support $B$, by letting a random sample from $\rD_{S}$ be given by $(i, j)$, where $i$ is drawn randomly from $\rD$ and $j$ is drawn randomly from $[a_i]$.
\end{definition}

\begin{fact}[Fact 2.5 of \cite{diakonikolaskane2016}]
    \label{fact:split-preserve-1-norm}
    Let $\p$ and $\q$ be probability distributions on $[n]$, and $S$ a given multi-set of $[n]$. 
    Then: 
    (1) We can simulate a sample from $\p_S$ or $\q_S$ by taking a single sample from $\p$ or $\q$, respectively. 
    (2) It holds $\norm{\p_{S} - \q_{S}}_{1} = \norm{\p - \q}_{1}$.
\end{fact}

Formally, (1) says that there is a function $f := f_{S}$ where $f(x) \sim \p_{S}$ if $x \sim S$.

\begin{lemma}[Lemma 2.6 of \cite{diakonikolaskane2016}]
    \label{lemma:split-reduces-2-norm}
    Let $\p$ be a distribution on $[n]$. 
    Then: 
    (1) For any multi-sets $S \subseteq S'$ of $[n]$, $\norm{\p_{S'}}_{2}^{2} \leq \norm{\p_{S}}_{2}^{2}$, and 
    (2) If $S$ is obtained by taking $\Poi(m)$ samples from $\p$, then $\EX[\norm{\p_{S}}_{2}^{2}] \leq \frac{1}{m}$.
\end{lemma}

\subsection{Bias Estimation}

%\hadley{todo: prove. cite?}

%\begin{lemma} 
%    \label{lem:bias-add} 
%    Given unknown coin probability $\rho \in (0,1)$, $O(\eps^{-2} \log(1/\delta))$ i.i.d. $\mathrm{Bern}(\rho)$ trials suffice to attain an estimate $\hat{\rho}$ such that $\Pr[|\rho-\hat{\rho}| > \eps] < \delta$. 
%\end{lemma}

We make heavy use of the following standard bias estimation result.

\begin{lemma}
    \label{lem:bias-add} 
    Let $q < \frac{1}{2}$.
    Given unknown coin probability $\rho \in (0, q)$, $\bigO{\frac{q}{\eps^2} \log(1/\delta)}$ i.i.d. $\mathrm{Bern}(\rho)$ trials suffice to attain an estimate $\hat{\rho}$ such that $\Pr[|\rho-\hat{\rho}| > \eps] < \delta$. 
\end{lemma}

\begin{proof}
    Let $\hat{\rho}$ be the empirical bias estimate from $m$ trials.
    The result follows immediately from a Chernoff bound.
    In particular,
    \begin{equation*}
        \Pr(|m\hat{\rho} - m\rho| > m\eps) = \Pr\left(|m\hat{\rho} - m\rho| > \frac{\eps}{\rho} m\rho \right) < e^{-\bigOm{\frac{\eps^2}{\rho^2}m\rho}} = e^{-\bigOm{\frac{\eps^2}{\rho}m}} \ll \delta 
    \end{equation*}
    whenever $m \gg \frac{\rho}{\eps^2} \log(1/\delta)$.
\end{proof}

In some cases, we will need a constant-factor approximation of an unknown coin probability. In such cases the above does not suffice, as we would need to set $\eps \approx \rho$, but $\rho$ is unknown. Thus, we will use the following. Note that we have not made an effort to optimize constants in favor of simplicity. 

\begin{lemma}
    \label{lem:bias-add-unknown} 
    Given unknown coin probability $\rho \in (0, 1)$, $O(\rho^{-1})$ i.i.d. $\mathrm{Bern}(\rho)$ trials suffice to attain an estimate $\hat{\rho}$ such that $\Pr[\rho/C \leq \hat{\rho} \leq C\rho ] \geq 1-2/C$, for any $C \geq 1$. 
\end{lemma}
\begin{proof} Perform $C/\rho$ i.i.d. $\mathrm{Bern}(\rho)$ trials. If no heads is seen, output fail, otherwise let $n$ denote the location of the first heads and set $\hat{\rho} = 1/n$. Observe that $n<t$ means that a heads appears in the first $t$ trials, and so
\[
\Pr\left[n < 1/(C\rho)\right] \leq \frac{1}{C\rho} \cdot \rho = \frac{1}{C}
\]
by a union bound, whereas $n > t$ means that the first $t$ trials were entirely tails, and so
\[
\Pr\left[n > C/\rho\right] \leq (1-\rho)^{C/\rho} \leq e^{-C} \leq 1/C
\]
and so by a union bound $\Pr[1/(C\rho) \leq n \leq C/\rho] \geq 1-2/C$, and observe that $1/(C\rho) \leq n \leq C/\rho$ is equivalent to $\rho/C \leq \hat{\rho} \leq C\rho$. \end{proof}

\begin{fact} [Basic Mixture Parameter Estimation] \label{fact:lambda-est} Let $\rD = \lambda \p + (1-\lambda)\q$ where $\lambda \in (0,1)$ is an unknown mixture parameter. There is a procedure using $O(1/\lambda)$ samples and queries which with probability $0.99$ returns $\hat{\lambda} \in (\lambda/40000,\lambda]$. \end{fact}

\begin{proof} By querying one sample from $r$ we have access to a $\lambda$-biased coin. Therefore, by invoking \Cref{lem:bias-add-unknown} with $C = 200$ and dividing the result by $200$ we obtain $\hat{\lambda} \in (\lambda/40000,\lambda]$ with probability $0.99$. \end{proof}

We can also obtain a more accurate estimate.

\begin{lemma} [Mixture Parameter Estimation] \label{lemma:lambda-est-2} Let $\rD = \lambda \p + (1-\lambda) \q$ where $\lambda \in (0,1)$ is an unknown mixture parameter. 
There is a procedure using $O(1/\eta^2 \lambda)$ samples and queries which with probability $0.98$ returns $\hat{\lambda} \in ((1 - 2 \eta)\lambda,\lambda]$.
\end{lemma}

\begin{proof}
    First, compute $\hat{\lambda}_{0} \in (\lambda/40000, \lambda)$ from \Cref{fact:lambda-est}.
    Then, we know $\lambda \leq 40000 \hat{\lambda}_{0}$.
    Now, with \Cref{lem:bias-add} compute an estimate $|\hat{\lambda} - \lambda| \leq \eta \lambda$ using $\bigO{\frac{\hat{\lambda}}{\eta^2 \lambda^2}} = \bigO{\frac{1}{\eta^2 \lambda}}$.
    Note that this step succeeds with high constant probability.
    Finally, note that $\hat{\lambda} - \eta \lambda \leq \lambda$ and
    \begin{equation*}
        (1 - 2 \eta) \lambda \leq \hat{\lambda} - \eta \lambda \leq \lambda \text{.}
    \end{equation*}
\end{proof}

\section{Uniformity Testing} \label{sec:uniformity-UB}

In this section, we present our algorithm for distributionally robust uniformity testing.

\SimpleUniformityAlg*

If $m = \bigO{\frac{\sqrt{n}}{\eps^2 \lambda}}$, we can simply query all samples to obtain $\bigO{\frac{\sqrt{n}}{\eps^2}}$ samples from $\p$ and test uniformity of $\p$ using known uniformity testing algorithms.
Note that such an algorithm uses $\bigTh{\frac{\sqrt{n}}{\eps^2}} = \bigTh{\frac{n}{m \eps^4}}$ queries when $m = o(n)$.
Thus, we assume that $m = \bigOm{\frac{\sqrt{n}}{\eps^2}}$.

% \begin{theorem}
%     \label{thm:simple-uniformity-query}
%     There is a uniformity tester with sample complexity $O(m + \sqrt{n} (\eps^{-2} + (\log \log n)^2) + \eps^{-4})$ and query complexity $O(\max(n m^{-1}, 1) ((\log \log n)^3 + \eps^{-4}))$.
% \end{theorem}

Let $\rD := \lambda \p + (1 - \lambda)\q$ denote the mixture distribution to which we have sample access, where $\p$ is the target distribution for which we would like to test uniformity.
Our algorithm will require the flattening technique of \cite{diakonikolaskane2016}.
Given a distribution $\distribution$ over $[n]$ and function $f: [n] \rightarrow \R$, let $f(\distribution)$ denote the distribution given by $f(x)$ where $x \sim \distribution$.

%\begin{lemma}[Flattening (Fact 2.5 and Lemma 2.6 \cite{diakonikolaskane2016})]
%    \label{lemma:flattening}
%    Fix $m > 0$ an integer.
%    There is an algorithm that given $O(m)$ samples to a distribution $\rD$ over $[n]$ with probability at least $0.99$ outputs function $f: [n] \rightarrow [n + m]$ satisfying:
%    \begin{enumerate}
%        \item $\norm{f(\distribution) - f(\distribution')}_{1} = \norm{\distribution - \distribution'}_{1}$ for any two distributions $\distribution, \distribution'$ over $[n]$.
%        \item $\norm{f\left( \rD \right)}_{2}^{2} \leq \frac{1}{m}$.
%    \end{enumerate}
%\end{lemma}

\begin{restatable}[Mixture Flattening]{lemma}{flattening}
    \label{lemma:flattening}
    Fix an integer $m > 0$.
    There is an algorithm that given $O(m)$ samples to a distribution $\rD = \lambda \p + (1-\lambda) \q$ over $[n]$ with probability at least $0.99$ outputs a random map $f \colon [n] \rightarrow [n + O(m)]$ satisfying:
    \begin{enumerate}
        \item $\norm{f(\distribution) - f(\distribution')}_{1} = \norm{\distribution - \distribution'}_{1}$ for any two distributions $\distribution, \distribution'$ over $[n]$.
        \item $\norm{f\left( \rD \right)}_{2}^{2} \leq \frac{1}{m}$.
        \item $\norm{f\left( \p \right)}_{2}^{2} \leq \frac{1}{\lambda m}$.
    \end{enumerate}
\end{restatable}

\begin{proof}
    Our function $f := f_{S}$ where $S$ consists of $\Poi(1000m)$ iid samples from $\rD$.
    Thus, the first item immediately follows from \Cref{fact:split-preserve-1-norm}.
    For the second item, since \Cref{lemma:split-reduces-2-norm} implies $\EX[\norm{f(\rD)}_{2}^{2}] \leq \frac{1}{1000m}$, Markov's inequality guarantees that the second item holds with probability $0.999$.
    For the third item, note that our multi-set $S$ contains as a subset  $\Poi(1000 \lambda m)$ iid samples from $\p$, we obtain the result using \Cref{lemma:split-reduces-2-norm}.
    
    To bound the sample complexity, we terminate whenever $\Poi(1000m)$ exceeds $10000m = O(m)$.
    By \Cref{lemma:poisson-concentration}), this occurs with probability at most $\exp(-\frac{81000 m^2}{2000m}) < 0.001$.
    A union bound concludes the proof.
\end{proof}

Next, we need a reduction from identity testing to uniformity testing.

\begin{definition}[Identity Testing]
    \label{def:identity-testing}
    Given an explicit distribution $\distribution^*$ over $[n]$ and sample access to distribution $\distribution$ over $[n]$, determine if $\distribution = \distribution^*$ or $\distribution$ is $\varepsilon$-far from $\distribution^*$.
\end{definition}

\begin{theorem}[\cite{goldreich2020uniform}]
    \label{thm:identity-testing-reduction}
    For every distribution $\distribution^*$ over $[n]$ and every $\varepsilon > 0$, there is a reduction $F: [n] \rightarrow [6n]$ that reduces $\varepsilon$-testing equality of unknown distribution $\distribution$ to known $\distribution^*$ to $\varepsilon/3$-testing uniformity of unknown distribution $F(\distribution)$. 
    The same reduction $F$ can be used for all $\varepsilon > 0$.

    Furthermore, for all distributions $\distribution$, $\norm{F(\distribution)}_{2}^{2} \leq 5 \norm{\distribution}_{2}^{2}$.
\end{theorem}

Note that the above theorem immediately implies that our uniformity testing algorithms also give identity testing algorithms.

\begin{proof}
    We need to prove the last statement (the rest is immediate from \cite{goldreich2020uniform}).
    Observe that in our setting we have that $F(\rD) = \lambda F(\p) + (1 - \lambda F(\q)$ so that samples and queries can be directly simulated.
    
    Note that the reduction $F$ is designed in three steps (see \cite{goldreich2020uniform} for details).
    Let $\distribution$ be any distribution.
    First, we transform samples from $\distribution$ to an even mixture of $\distribution_1 := \frac{U_n + \distribution}{2}$ which guarantees every element has mass at least $\frac{1}{2n}$ and at most $\frac{\distribution[i]}{2} + \frac{1}{2n}$ which is at most $\frac{1}{n}$ if $\distribution[i] \leq \frac{1}{n}$ and at most $\distribution[i]$ otherwise.
    Thus, assume $\distribution_1[i] \leq \max(\distribution[i], \frac{1}{n})$.

    Now, assume any distribution $\distribution_1$ satisfies $\distribution_1[i] \geq \frac{1}{2n}$ for all $i$ (and in particular $\distribution^*$).
    To define $\distribution_2$, we transform samples $x \sim \distribution_1$ to output $x$ with probability $\frac{\lfloor6n \cdot \distribution^*[i] \rfloor /6n}{\distribution^*[i]}$ and $n + 1$ otherwise.
    Since $\distribution^*[i] \geq \frac{1}{2n}$, we have $\frac{\lfloor6n \cdot \distribution^*[i] \rfloor /6n}{\distribution^*[i]} \geq \frac{2}{3}$.
    In particular, every $i \in [n]$ has mass $\distribution_2[i] \leq \distribution[i]$ and $[n + 1]$ has some mass at most $\frac{1}{3}$.
    Note that $\distribution_2[n + 1]$ depends only on $\distribution^*$.

    Now, note that $\distribution_2[i]$ is a multiple of $\frac{1}{6n}$ for all $i$ (this is guaranteed by the previous step).
    We define $F(\distribution) := \distribution_3$ by taking every element in $[n + 1]$ with mass $\frac{k}{6n}$ for some integer $k$ and splitting it uniformly into $k$ buckets, which are disjoint for all $i$.
    For any $i \in [n]$, note that $\distribution_3[j] \leq \distribution[i]$ for all $j$ that the $i$-th bucket is split into.
    The bucket $[n + 1]$ which may be heavy is split into buckets each with mass $\frac{1}{6n}$.
    Note that the mass of $[n + 1]$ does not depend on the input distribution $\distribution$ but only the fixed identity distribution $\distribution^*$, so that these buckets have mass $\frac{1}{6n}$ regardless of the input distribution $\distribution$.
    Thus, overall the $\ell_{2}$-norm of $\distribution_3$ for any input distribution $\distribution$ can be bounded by
    \begin{equation*}
        \norm{\distribution_3}_{2}^{2} \leq \frac{1}{6n} + \sum_{i = 1}^{n} \max\left(\distribution[i], \frac{1}{n}\right)^2 \leq \frac{1}{6n} + \sum_{i = 1}^{n} \left(\distribution[i] + \frac{1}{n}\right)^2 \leq \norm{\distribution}_{2}^{2} + \frac{19}{6n} \leq 5 \norm{\distribution}_{2}^{2}
    \end{equation*}
    since any distribution $\distribution$ on $[n]$ has $\ell_{2}$-norm at least $\frac{1}{n}$.
\end{proof}

Given the flattening technique, it suffices to consider distributions with small $\ell_2$ norm (at the cost of collecting enough samples to flatten the distribution appropriately).
We thus design an algorithm for testing uniformity under the assumption that $\rD$ has bounded $\ell_2$-norm.

\begin{lemma}
    \label{lemma:simple-t-query-alg}
    Suppose $\norm{\rD}_{2}^{2} \leq b$. 
    Then, there is an $\lambda$-distributionally robust $\eps$-uniformity tester with $\bigO{\frac{\sqrt{n}}{\eps^2 \lambda}}$ sample complexity and $\bigO{\frac{nb}{\eps^4 \lambda^2} + \frac{1}{\eps^4 \lambda^2}}$ query complexity.
\end{lemma}

\begin{proof}
    First, we use \Cref{lem:bias-add-unknown} to obtain $\hat{\lambda}$ such that $\frac{\lambda}{40000} \leq \hat{\lambda} \leq \lambda$ with probability $0.99$.
    We will design an algorithm that succeeds given a lower bound on $\lambda$.
    Note that this requires $\bigO{\lambda^{-1}}$ samples and queries.
    Thus, in the following assume that $\rD = \lambda \p + (1 - \lambda) \q$ where $\lambda \geq \hat{\lambda}$.
    
    Let $m = \bigO{\frac{\sqrt{n}}{\eps^2}}$.
    Our algorithm takes $\frac{m}{\hat{\lambda}}$ samples from $\rD$.
    Let $m_{\p}$ denote the (unknown) random variable that represents the number of samples drawn from $\p$.
    Let $c(S_{\p})$ denote the (unknown) random variable that represents the number of $(\p, \p)$-collisions in the sample and $c(S_{\rD})$ the (known) number of collisions in the sample.
    Then, the analysis of \cite{DBLP:journals/cjtcs/DiakonikolasGPP19} gives the following guarantee on $c(S_{\p})$.

    \begin{theorem}[\cite{DBLP:journals/cjtcs/DiakonikolasGPP19}]
        \label{thm:uniformity-collision-tester}
        Suppose $m \gg \frac{3200 \sqrt{n}}{\eps^2}$.
        Let $c(S_{\p})$ denote the number of collisions in $m$ samples from $\p$.
        Then, with probability at least $0.999$, the following hold:
        \begin{enumerate}
            \item (Completeness) If $\p = U_n$ is uniform, then $c(S_{\p}) < \binom{m}{2} \frac{1 + 0.1 \eps^2}{n}$.
            \item (Soundness) If $\p$ is $\eps$-far from uniform, then $c(S_{\p}) > \binom{m}{2} \frac{1 + 0.75 \eps^2}{n}$.
        \end{enumerate}
    \end{theorem}

    Due to further estimation errors in our algorithm, we have stated a slightly stronger correctness guarantee than \cite{DBLP:journals/cjtcs/DiakonikolasGPP19}. 
    The details are deferred to the end of this section.

    Our high level goal is therefore to estimate $c(S_{\p})$ and to compare it to $\binom{m_{\p}}{2} \frac{1 + \frac{3}{4} \eps^2}{n}$.
    Note that $\EX[m_{\p}] \geq m$ but we do not know $m_{\p}$ exactly.
    First, we randomly query samples to compute an estimate $\hat{m}_{\p}$.

    \begin{claim}
        \label{clm:m-p-query}
        With $\bigO{\frac{1}{\eps^{4} \lambda^{2}}}$ queries, there is an algorithm that with probability $0.99$ computes an estimate $\hat{m}_{\p}$ such that
        \begin{equation*}
            \frac{2}{n} \left| m_{\p}^2 - \hat{m}_{\p}^{2} \right| \leq 0.01 \frac{\eps^2 m^2}{n} \text{.}
        \end{equation*}
    \end{claim}

    Then, we randomly query collisions to compute an estimate $\hat{c}(S_{\p})$.

    \begin{claim}
        \label{clm:c-Sp-query}
        There exists a universal constant $C > 0$ such that if $m_{\p} = \frac{C \sqrt{n}}{\eps^2}$ there is an algorithm that uses $\bigO{\frac{nb}{\eps^{4} \lambda^{2}}}$ queries, and with probability $0.99$ computes an estimate $\hat{c}(S_{\p})$ such that the following hold:
        \begin{enumerate}
            \item If $\p$ is uniform, then
            \begin{equation*}
                \hat{c}(S_{\p}) \leq \binom{m_{\p}}{2} \frac{1 + 0.1 \eps^{2}}{n}\text{.}
            \end{equation*}

            \item If $\p$ is far from uniform, then 
            \begin{equation*}
                \hat{c}(S_{\p}) - \binom{m_{\p}}{2} \frac{1 + 0.4 \eps^2}{n} \geq \max\left( \frac{m_{\p}^{2} \alpha}{30 n} , \frac{0.3 m_{\p}^2 \eps^2}{n} \right)
            \end{equation*}
            where $\norm{\p}_{2}^{2} = \frac{1 + \alpha}{n}$.
        \end{enumerate}
    \end{claim}

    We are now ready to describe the full uniformity testing algorithm.
    
    \begin{mdframed}
        \begin{enumerate}
            \item Estimate $\hat{\lambda}$ using \Cref{fact:lambda-est}.
            \item Let $m \geq \frac{C n^{1/2}}{\eps^{2}}$ for some sufficiently large constant $C$.
            \item Draw $\frac{m}{\hat{\lambda}}$ samples from $\rD$.
            \item Compute $\hat{m}_{\p}$ using \Cref{clm:m-p-query}.
            \item Compute $\hat{c}(S_{\p})$ using \Cref{clm:c-Sp-query}.
            \item Output accept if $\hat{c}(S_{\p}) < \binom{\hat{m}_{\p}}{2} \frac{1 + 0.4 \eps^2}{n}$.
            Otherwise reject.
        \end{enumerate}
    \end{mdframed}

    We begin by arguing that $0.99m \leq m_{\p}$.
    Since each sample comes from $\p$ with probability $\lambda \geq \hat{\lambda}$ and we take $m/\hat{\lambda}$ samples, a simple Chernoff bound guarantees that
    \begin{equation*}
        \Pr\left( |m_{\p} - m| \leq 0.01m \right) < e^{-\bigOm{m}} \ll 0.01 \text{.}
    \end{equation*}
    Condition on the event that $m_{\p} \geq 0.99m > \frac{C \sqrt{n}}{\eps^2}$ for some large constant $C$.

    Suppose $\p$ is uniform.
    By \Cref{clm:c-Sp-query}, we have that with probability at least $0.99$, 
    \begin{equation*}
        \hat{c}(S_{\p}) \leq \binom{m_{\p}}{2} \frac{1 + 0.1 \eps^2}{n}
    \end{equation*}
    Of course, we do not know $m_{\p}$, but we can estimate it using \Cref{clm:m-p-query}.
    In particular, we have
    \begin{align*}
        \left| \binom{m_{\p}}{2} \frac{1 + 0.4 \eps^2}{n} - \binom{\hat{m}_{\p}}{2} \frac{1 + 0.4 \eps^2}{n} \right| &\leq \frac{2}{n} \left| \frac{(m_{\p}^2 - m_{\p}) - (\hat{m}_{\p}^{2} - \hat{m}_{\p})}{2} \right| \\
        &\leq \frac{2}{n} \left| m_{\p}^2 -\hat{m}_{\p}^{2} \right| \\
        &\leq \frac{0.1 \eps^2 m^2}{n} \text{.}
    \end{align*}
    Using the triangle inequality, we obtain
    \begin{align*}
        \binom{\hat{m}_{\p}}{2} \frac{1 + 0.4 \eps^2}{n} - \hat{c}(S_{\p}) &\geq \left( \binom{m_{\p}}{2} \frac{1 + 0.4 \eps^2}{n} - \hat{c}(S_{\p}) \right) - \left| \binom{m_{\p}}{2} \frac{1 + 0.4 \eps^2}{n} - \binom{\hat{m}_{\p}}{2} \frac{1 + 0.4 \eps^2}{n} \right| \\
        &\geq \left( \binom{m_{\p}}{2} \frac{1 + 0.4 \eps^2}{n} - \hat{c}(S_{\p}) \right) - \frac{0.01 \eps^2 m^2}{n} > 0 \text{.}
    \end{align*}
    Thus, the algorithm outputs accept given the conditioned events, which occur with probability at least $0.9$ by a union bound.

    The soundness case ($\p$ is $\eps$-far from uniform) follows similarly.
    \Cref{clm:c-Sp-query} guarantees the following
    \begin{equation*}
        \hat{c}(S_{\p}) - \binom{m_{\p}}{2} \frac{1 + 0.4 \eps^2}{n} \geq \max \left( \frac{m_{\p}^2 \alpha}{30 n}, \frac{0.3 m_{\p}^2 \eps^2}{n} \right) \text{.}
    \end{equation*}
    In particular, since 
    \begin{equation*}
        0.1 \frac{\eps^2 m^2}{n} \leq \max \left( \frac{m_{\p}^2 \alpha}{30 n}, \frac{0.3 m_{\p}^2 \eps^2}{n} \right)
    \end{equation*}
    the algorithm outputs reject with probability at least $0.9$ as desired (applying the union bound over required events as in the case for uniformity).
    
    We now analyze the complexity.
    The algorithm takes $\bigO{\frac{\sqrt{n}}{\lambda \eps^2}}$ samples and $\bigO{\frac{nb}{\eps^4 \lambda^2} + \frac{1}{\eps^4 \lambda^{2}}}$ queries.
    Note that the samples and queries required to obtain our initial estimate of $\lambda$ do not affect the final complexity.
\end{proof}

We are now ready to conclude the proof of \Cref{thm:simple-uniformity-alg}.

\begin{proofof}{\Cref{thm:simple-uniformity-alg}}
    Fix $m \geq \frac{C \sqrt{n}}{\eps^2 \hat{\lambda}}$ for a sufficiently large constant $C$ as required by \Cref{lemma:simple-t-query-alg}.
    Consider two cases.
    Let $m \leq n$.
    Using \Cref{lemma:flattening} and \Cref{thm:identity-testing-reduction}, we can use $O(m)$ samples to obtain a reduction $F$ such that $\norm{F(\rD)}_{2}^{2} \leq \frac{1}{m}$ and we may consider testing uniformity of $F(\p)$ on the domain $[n + O(m)] = [O(n)]$.
    Now, consider two cases.
    \Cref{lemma:simple-t-query-alg} yields an algorithm using $O(m)$ samples and $\bigO{\frac{n}{m \eps^{4} \lambda^{2}} + \frac{1}{\eps^{4} \lambda^{2}}} = \bigO{\frac{n}{m \eps^{4} \lambda^{2}}}$ queries.

    If $m \geq n$, we instead use $O(n) = O(m)$ samples to obtain a reduction $F$ such that $\norm{F(\rD)}_{2}^{2} \leq \frac{1}{n}$ and we may consider testing uniformity of $F(\p)$ on the domain $[O(n)]$.
    \Cref{lemma:simple-t-query-alg} yields an algorithm using $O(m)$ samples and $\bigO{\frac{n}{n \eps^{4} \lambda^{2}} + \frac{1}{\eps^{4} \lambda^{2}}} = \bigO{\frac{1}{\eps^{4} \lambda^{2}}}$ queries.
\end{proofof}

\subsection{Deferred Proofs for Uniformity Testing}

We begin with the correctness analysis of the uniformity tester of \cite{DBLP:journals/cjtcs/DiakonikolasGPP19}.

\begin{proofof}{\Cref{thm:uniformity-collision-tester}}
    Lemmas 2 and 3 of \cite{DBLP:journals/cjtcs/DiakonikolasGPP19} states that
    \begin{equation}
        \label{eq:p-collision-variance}
        \EX[c(S_{\p})] = \binom{m}{2} \norm{\p}_{2}^{2}\quad,\quad \Var(c(S_{\p})) \leq m^2 \norm{\p}_{2}^{2} + m^3 (\norm{\p}_{3}^{3} - \norm{\p}_{2}^{4}) \text{.}
    \end{equation}

    We begin with the completeness case.
    Since $\norm{\p}_{3}^{3} = \norm{\p}_{2}^{4} = \frac{1}{n^2}$ we have $\Var(c(S_{\p})) \leq m^2/n$.
    By Chebyshev's inequality,
    \begin{equation*}
        \Pr\left( c(S_{\p}) > \binom{m}{2} \frac{1 + \frac{\eps^2}{10}}{n} \right) = \Pr\left( c(S_{\p}) - \EX[c(S_{\p})] > \binom{m}{2} \frac{\eps^2}{10m\sqrt{n}} \frac{m}{\sqrt{n}} \right) \leq \frac{n}{100 \eps^{4} m^2}
    \end{equation*}
    which is less than $0.001$ for $m \geq \frac{10 00\sqrt{n}}{\eps^2}$.
    The soundness case follows directly from the correctness analysis of \cite{DBLP:journals/cjtcs/DiakonikolasGPP19}.
\end{proofof}

We show how to estimate $\hat{m}_{\p}$.

\begin{proofof}{\Cref{clm:m-p-query}}
    Fix some threshold $t$.
    Recall that our algorithm draws $\frac{m}{\hat{\lambda}}$ samples where each sample is independently drawn from $\p$ with probability $\lambda$ and from $\q$ otherwise.
    If we query a random sample, it is drawn from $\p$ with probability $\frac{m_{\p}}{m/\lambda}$.
    We query $q$ random samples.
    Let $\hat{s}_{\p}$ be the number of queried samples drawn from $\p$ so that
    \begin{equation*}
        \EX[\hat{s}_{\p}] = q \frac{m_{\p}}{m/\lambda} = m_{\p} \frac{\lambda q}{m} \text{.}
    \end{equation*}
    Define $\hat{m}_{\p} = \hat{s}_{\p} \frac{m}{\lambda q}$ so that $\hat{m}_{\p}$ is an unbiased estimator of $m_{\p}$.
    By a standard Chernoff bound, we have 
    \begin{align*}
        \Pr\left( \left| \hat{m}_{\p} - m_{\p} \right| > t m \right) &= \Pr\left( \left| \hat{s}_{\p} - s_{\p} \right| > \lambda q t \right) < \exp \left( - 2 t^2 q \lambda^2 \right) < 0.01
    \end{align*}
    for $q = \bigO{t^{-2} \lambda^{-2}}$.
    Then, note that
    \begin{equation*}
        |\hat{m}_{\p}^2 - m_{\p}^2| = (\hat{m}_{\p} + m_{\p})|\hat{m}_{\p} - m_{\p}| \leq 2m^2t
    \end{equation*}
    so that
    \begin{equation*}
        \frac{2}{n} |\hat{m}_{\p}^2 - m_{\p}^2| \leq \frac{4m^2t}{n} \ll \frac{\eps^2m^2}{n}
    \end{equation*}
    for $t \ll \eps^2$.
    In particular, the algorithm requires $q = \bigO{\eps^{-4} \lambda^{-2}}$ queries.
\end{proofof}

We show how to estimate $\hat{c}(S_{\p})$.

\begin{proofof}{\Cref{clm:c-Sp-query}}
    We begin by bounding $c(S_{\rD})$.
    Since $\EX[c(S_{\rD})] = \binom{m/\hat{\lambda}}{2} \norm{\rD}_{2}^{2}$ we have by Markov's inequality that $c(S_{\rD}) \leq 1000 \frac{m^2}{\hat{\lambda}^2} \norm{\rD}_{2}^{2} \leq 1000 \frac{m^2 b}{\hat{\lambda}^2}$ with probability $0.999$.
    Similarly, $\EX[c(S_{\p})] = \binom{m_{\p}}{2} \norm{\p}_{2}^{2}$ and Markov's inequality implies that $c(S_{\p}) = \bigO{m^2 \norm{\p}_{2}^{2}}$ with probability $0.999$.
    We condition on these events.
    Furthermore, note that if we query a random collision, it is a $(\p, \p)$-collision with probability $\frac{c(S_{\p})}{c(S_{\rD})}$.

    Consider the case where $\p$ is uniform.
    Then, from \eqref{eq:p-collision-variance} we have $\Var(c(S_{\p})) \leq m_{\p}^{2}/n$ so that $\sigma := \sqrt{\Var(c(S_{\p}))} \leq \frac{m_{\p}}{\sqrt{n}}$.
    Recall that we condition on the event that $|m_{\p} - m| \leq 0.001 m$.
    By Chebyshev's inequality, we have that
    \begin{equation*}
        \Pr\left( c(S_{\p}) - \binom{m_{\p}}{2} \frac{1}{n} > \binom{m_{\p}}{2} \frac{0.01 \eps^2}{n} \right) = \Pr\left( c(S_{\p}) - \binom{m_{\p}}{2} \frac{1}{n} > \frac{0.004 m_{\p}^2 \eps^2}{n} \right) = \bigO{\frac{\sqrt{n}}{\eps^{2} m_{\p}}} \leq 0.001
    \end{equation*}
    for $m_{\p} = \bigO{\frac{\sqrt{n}}{\eps^2}}$ for some sufficiently large constant.
    Then, Markov's inequality implies that with probability $0.999$ we have that $c(S_{\p}) = \bigO{\frac{m^2}{n}}$ and $c(S_{\rD}) = \bigO{m^2 b \lambda^{-2}}$.
    We query $q$ independently chosen collisions to obtain an estimate $\frac{\hat{c}(S_{\p})}{c(S_{\rD})}$ such that with probability at least $0.999$, 
    \begin{equation*}
        \left| \frac{\hat{c}(S_{\p})}{c(S_{\rD})} - \frac{c(S_{\p})}{c(S_{\rD})} \right| \ll \frac{\eps^2 m^2}{n \cdot c(S_{\rD})} \text{.}
    \end{equation*}
    From \Cref{lem:bias-add}, this requires 
    \begin{equation*}
        \bigO{\frac{n^{2} \cdot c(S_{\rD})^2}{m^{4} \eps^{4}} \cdot \frac{c(S_{\p})}{c(S_{\rD})}} = \bigO{\frac{n^{2} \cdot c(S_{\rD}) \cdot c(S_{\p})}{m^{4} \eps^{4}}}
    \end{equation*}
    queries.
    Applying our bounds on $c(S_{\p}), c(S_{\rD})$, we obtain the query bound $\bigO{\frac{nb}{\lambda^{2} \eps^4}}$.
    Thus, if both estimates are correct we have that
    \begin{align*}
        \binom{m_{\p}}{2} \frac{1 + 0.4 \eps^{2}}{n} - \hat{c}(S_{\p}) 
        &\geq \left| \binom{m_{\p}}{2} \frac{1 + 0.4 \eps^{2}}{n} -\binom{m_{\p}}{2} \frac{1}{n} \right| - \left| \binom{m_{\p}}{2} \frac{1}{n} - c(S_{\p}) \right| - \left| \hat{c}(S_{\p}) - c(S_{\p}) \right| \\
        &\geq \binom{m_{\p}}{2} \frac{0.3 \eps^{2}}{n} \text{.}
    \end{align*}
    Rearranging, we obtain the desired bound.

    Now, consider the case that $\p$ is far from uniform.
    A standard calculation shows that
    \begin{equation*}
        \norm{\p - \frac{1}{n}}_{2}^{2} = \norm{\p}_{2}^{2} - \frac{1}{n} \geq \frac{\eps^2}{n} \text{.}
    \end{equation*}
    In particular, let $\norm{\p}_{2}^{2} = \frac{1 + \alpha}{n}$ where $\alpha \geq \eps^2$.
    If $\alpha = \bigO{1}$, then the we apply the same argument as in the uniform case since we have the same bound on $c(S_{\p})$.
    Thus assume $\alpha = \omega(1)$ (e.g. $\alpha \geq 30$). 
    
    We consider two cases.
    Either $m_{\p}^2 \norm{\p}_{2}^{2} \geq m_{\p}^3 (\norm{\p}_{3}^{3} - \norm{\p}_{2}^{4})$ or vice versa.
    Let $\sigma := \sqrt{\Var(c(S_{\p}))}$ denote the standard deviation of $c(S_{\p})$.

    {\bf Case 1: $m^2 \norm{\p}_{2}^{2} \geq m^3 (\norm{\p}_{3}^{3} - \norm{\p}_{2}^{4})$}

    Since $\Var(c(S_{\p})) \leq 2 m^2 \norm{\p}_{2}^{2}$, we have $\sigma^2 \leq 2 m^2 \norm{\p}_{2}^{2} = \frac{2m_{\p}^2(1 + \alpha)}{n}$.
    In particular, with high constant probability, $c(S_{\p})$ is within $\bigO{\frac{m_{\p}\sqrt{1 + \alpha}}{\sqrt{n}}}$ of its expectation $\binom{m_{\p}}{2} \frac{1 + \alpha}{n}$.
    Note that for large enough $n$ (and therefore large enough $m_{\p}$)
    \begin{equation}
        \label{eq:expectation-gap-far-uniform}
        \EX[c(S_{\p})] - \binom{m_{\p}}{2} \frac{1 + 0.4 \eps^2}{n} = \binom{m_{\p}}{2} \frac{\alpha - 0.4 \eps^2}{n} \geq \frac{m_{\p}^2 \alpha}{10 n} \text{.}
    \end{equation}
    Furthermore, we have
    \begin{equation}
        \label{eq:std-dev-m2}
        \bigO{\frac{m_{\p}\sqrt{1 + \alpha}}{\sqrt{n}}} \ll \frac{m_{\p}^2 \alpha}{10 n}
    \end{equation}
    as long as $m_{\p} = \bigO{\frac{\sqrt{n}}{\alpha}} = \bigO{\frac{\sqrt{n}}{\eps^2}}$ for some sufficiently large constant.
    In particular, by Chebyshev's inequality, the expectation gap \eqref{eq:expectation-gap-far-uniform} and variance bound \eqref{eq:std-dev-m2} we have with high constant probability that
    \begin{equation*}
        c(S_{\p}) - \binom{m_{\p}}{2} \frac{1 + 0.4 \eps^2}{n} \geq \frac{m_{\p}^2 \alpha}{20 n} \text{.}
    \end{equation*}
    Our goal is now to bound 
    \begin{equation}
        \label{eq:collision-est-bound}
        |\hat{c}(S_{\p}) - c(S_{\p})| \ll \frac{m_{\p}^{2} \alpha}{n} \text{.}
    \end{equation}
    Following similar arguments using \Cref{lem:bias-add} this requires 
    \begin{equation*}
        \bigO{\frac{n^2 c(S_{\rD})^{2}}{m_{\p}^{4} \alpha^{2}} \cdot \frac{c(S_{\p})}{c(S_{\rD})}} = \bigO{\frac{n^2 c(S_{\rD})}{m_{\p}^{4} \alpha^{2}} \cdot c(S_{\p})} = \bigO{\frac{n b}{\alpha \lambda^{2}}} = \bigO{\frac{nb}{\eps^2 \lambda^2}}
    \end{equation*}
    queries where we have observed that $c(S_{\p}) = \bigO{m^2 \norm{\p}_{2}^{2}} = \bigO{\frac{m^2 \alpha}{n}}$.
    In particular, with high constant probability, we have $\hat{c}(S_{\p}) - \binom{m_{\p}}{2} \frac{1 + 0.4 \eps^2}{n} \geq \frac{m_{\p}^{2} \alpha}{30 n}$.

    {\bf Case 2: $m^2 \norm{\p}_{2}^{2} \leq m^3 (\norm{\p}_{3}^{3} - \norm{\p}_{2}^{4})$}
    Let $\p[i] = \frac{1}{n} + a[i]$ for some vector $a$.
    We bound the variance as
    \begin{align*}
        \Var(c(S_{\p})) &\leq 2 m_{\p}^3 \left(\norm{\p}_{3}^{3} - \norm{\p}_{2}^{4} \right) \\
        &\leq 2m_{\p}^3 \left( \left( \sum_{i} \left( \frac{1}{n} + a[i] \right)^{3} \right) - \frac{1}{n^2} \right) \\
        &= 2m_{\p}^3 \left( \left( \sum_{i} \frac{1}{n^3} + \frac{3 a[i]}{n^2} + \frac{3 a[i]^2}{n} + a[i]^{3} \right) - \frac{1}{n^2} \right) \\
        &= 2m_{\p}^3 \left( \frac{3 \norm{a}_{2}^{2}}{n} + \norm{a}_{3}^{3} \right) \\
        &\leq 2m_{\p}^3 \left( \frac{3 \norm{a}_{2}^{2}}{n} + \norm{a}_{2}^{3} \right) 
    \end{align*}
    where we have used $\norm{\p}_{2}^{2} \geq \frac{1}{n}$ for all distributions $\p$ and $\sum_{i} a[i] = 0$.
    Now, observe that $\norm{\p}_{2}^{2} = \frac{1 + \alpha}{n}$ so that $\norm{a}_{2}^{2} = \frac{\alpha}{n}$ so that
    \begin{align*}
        \Var(c(S_{\p})) &\leq 2m_{\p}^3 \left( \frac{3 \alpha}{n^2} + \frac{\alpha^{3/2}}{n^{3/2}} \right) \\
        &= \frac{6 m_{\p}^{3} \alpha}{n^{2}} + \frac{2 m_{\p}^{3} \alpha^{3/2}}{n^{3/2}} \text{.}
    \end{align*}
    In particular, we have $\sigma \leq \frac{6 m_{\p}^{3/2} \alpha^{1/2}}{n} + \frac{2 m_{\p}^{3/2} \alpha^{3/4}}{n^{3/4}}$.
    As in Case 1, we hope to bound the standard deviation against the expectation gap \eqref{eq:expectation-gap-far-uniform} i.e.
    \begin{align*}
        \frac{6 m_{\p}^{3/2} \alpha^{1/2}}{n} + \frac{2 m_{\p}^{3/2} \alpha^{3/4}}{n^{3/4}} \ll \frac{m_{\p}^2 \alpha}{n}
    \end{align*}
    which follows as long as $m_{\p} = \bigO{\frac{1}{\alpha} + \frac{\sqrt{n}}{\sqrt{\alpha}}}$ for some sufficiently large constant.
    In particular, $m_{\p} = \bigO{\frac{\sqrt{n}}{\eps^2}}$ suffices.
    The remainder of the argument (estimating $c(S_{\p})$ using $\hat{c}(S_{\p})$) follows identically as in Case 1.
\end{proofof}

\subsection{Adversarially Robust Algorithm for Uniformity Testing}

We present an $\lambda$-adversarially robust $\eps$-uniformity tester with near-optimal sample and query complexity.

\UniformityAdvAlg*

We begin with a high level overview of the algorithm.

\paragraph{Algorithm Overview}
A close inspection of our distributionally robust algorithm shows that the algorithm essentially requires two main ingredients:
\begin{enumerate}
    \item There are thresholds $T_1 \ll T_2$ such that in the completeness case ($\p$ is uniform) there are at most $T_1$ $\p$-collisions while in the soundness case ($\p$ is far from uniform) that are at least $T_2$ $\p$-collisions. 
    This condition is met as long as sufficiently many samples are taken from $\p$ (i.e. $m_{\p} \gg \frac{\sqrt{n}}{\eps^2}$).
    \item $S_{\rD}$ does not produce too many collisions so we can estimate the number of $\p$-collisions without too many queries.
    This condition is met as long as the flattening procedure succeeds (i.e. $\norm{\rD}_{2}^{2} \ll \frac{1}{m}$).
\end{enumerate}
In the adversarial setting, the first condition is easy to ensure by taking at least $\bigO{\frac{\sqrt{n}}{\eps^2 \lambda}}$ samples.
However, the second condition is no longer evident, as the adversarial samples do not have to be chosen iid from some distribution.
Nevertheless, we are able to design an adversarially robust algorithm using a more delicate flattening procedure.
Instead of using the first half of samples to flatten and the second half to test uniformity of the flattened distribution, we will randomly choose which set of samples to flatten with and which set of samples to test uniformity of $\p$ with.

In the distributional setting, we were able to show that $\norm{\rD}_{2}^{2} \leq \frac{1}{m}$ so that $m$ samples produce $O(m)$ collisions.
In the adversarial setting, we would similarly hope to bound the number of collisions as $O(m)$.
If $Y_i$ is the observed frequency of the $i$-th bucket (so $m = \sum Y_i$) this is roughly $\sum_i Y_i^2$ which subject to the constraint is maximized when a single $Y_i$ is maximized.
If $\max_{i} Y_i \leq B$, then we have at most $O(mB)$-collisions.
If we randomly choose each sample to be a flattening sample with probability $\frac{1}{2}$, then we ensure that for large $Y_i > B$, the number of flattening samples in the $i$-th bucket is at least $\Omega(Y_i)$ with high probability (e.g. $1 - e^{-\Omega(B)}$).
Thus, a union bound shows that it suffices to choose $B = O(\log n)$ which yields an upper bound of $O(m \log n)$ collisions in the randomly subsampled test dataset.

We are now ready to prove \Cref{thm:adversarially-robust-uniformity-alg}.

\begin{proofof}{\Cref{thm:adversarially-robust-uniformity-alg}}
    We begin by estimating $\lambda$ using \Cref{fact:lambda-est}.
    A key ingredient will be the following analogue to flattening \Cref{lemma:flattening}.

    \begin{lemma}[Adversarial Flattening]
        \label{lemma:subsample-flattening-alg}
        There is an algorithm that given an $m$ samples from a $(1 - \lambda)$-adversarially contaminated source, with probability $0.99$ outputs a map $F: [n] \rightarrow [n + m]$ and a dataset $S_{\testing} \subset S$ satisfying the following:
        \begin{enumerate}
            \item $F(S_{\testing})$ contains at most $O(m \log n)$-collisions,
            \item $S_{\testing}$ contains at least $\bigOm{\lambda m}$ iid samples from $\p$,
            \item $\norm{\q_1 - \q_2}_{1} = \norm{F(\q_1) - F(\q_2)}_{1}$ for all distributions $\q_1, \q_2$.
        \end{enumerate}  
    \end{lemma}

    Thus, if we fix $m \geq \frac{C' \sqrt{n}}{\eps^2 \hat{\lambda}}$ for sufficiently large constant $C'$, the above lemma yields a dataset $F(S_{\testing})$ with at most $O(m \log n)$-collisions and at least $\frac{C \sqrt{n}}{\eps^2}$ iid samples from $\p$ for $C$ required by \Cref{lemma:simple-t-query-alg}.

    \begin{proofof}{\Cref{lemma:subsample-flattening-alg}}
        Consider an algorithm that takes $m$ iid samples $S$ and evenly distributes (independently with probability $\frac{1}{2}$) the observed samples $S$ into two datasets: $S_{\flatten}, S_{\testing}$.
        Let $Y_i$ denote the frequency of the $i$-th bucket in $S$ and $Y_{i}^{\flatten}, Y_{i}^{\testing}$ denote the frequency of the $i$-th bucket in $S_{\flatten}, S_{\testing}$ respectively.
        We argue that for large $Y_i$, $Y_i^{\flatten} = \Omega(Y_i)$.
        
        \begin{claim}
            \label{clm:subsample-flattening}
            With probability at least $0.999$, $Y_{i}^{\flatten} \geq \frac{Y_{i}}{4}$ for all $Y_i \geq 100 \log n$.
        \end{claim}
    
        \begin{proof}
            Note that $Y_{i}^{\flatten}$ is the sum of $Y_i$ iid Bernoulli random variables with parameter $\frac{1}{2}$. 
            Thus, $\EX[Y_{i}^{\flatten}] = \frac{Y_i}{2}$.
            By a standard Chernoff bound, we have $\Pr\left(Y_{i}^{\flatten} < \frac{Y_{i}}{4}  \right) < \exp \left( - \frac{Y_i}{16} \right) < \frac{0.001}{n}$ for $Y_i \geq 100 \log n$. 
        \end{proof}

        We now describe how to obtain $F$.
        Let $F$ describe the (randomized) map $F_0: [n] \rightarrow [n + |S_{\flatten}|]$ defined as the split distribution induced by $S_{\flatten}$ (see \Cref{def:split-distributions}).
        % That is, associate $[n + |S_{\flatten}|]$ with $\set{(i, j) \given i \in [n], j \in [Y_{i}^{\flatten}]}$ so that $F$ takes input $i$ and outputs a uniform element of $(i, j)$ where $j \in \set{1, \dotsc, |Y_{i}^{\flatten}|}$.
        
        Let $F(S_{\testing})$ denote the dataset obtained by independently applying $F$ to each element of $S_{\testing}$.
        Let $c(F(S_{\testing}))$ denote the number of collisions in $F(S_{\testing})$.
        We bound the number of collisions in $F(S_{\testing})$.
    
        \begin{lemma}
            \label{lemma:subsample-collision-bound}
            With probability $0.999$ over samples $S$, $c(F(S_{\testing})) = O(m \log n)$.
        \end{lemma}
    
        \begin{proof}
            Fix a bucket $i \in [n]$.
            Suppose $Y_i \geq 100 \log n$.
            Two samples in the $i$-th bucket are mapped to the same sub-bucket $j$ with probability $\frac{1}{Y_{i}^{\flatten} + 1}$.
            Since there are $\binom{Y_{i}^{\testing}}{2} \leq Y_{i}^{2}$ pairs in the $i$-th bucket the number of collisions in the $i$-th bucket is in expectation $\bigO{Y_{i}}$.
            Summing over all $n$ buckets, the number of collisions in total is in expectation $\bigO{m}$.
            Thus, Markov's inequality implies that the total number of collisions in $F(S_{\testing})$ is $O(m)$ with probability $0.999$.
    
            It remains to count collisions in buckets with $Y_i \leq 100 \log n$.
            Note that the number of collisions is at most $\sum_{Y_i \leq 100 \log n} \binom{Y_i^{\testing}}{2} = \bigO{m \log n}$ since the sum of squares (constrained on $\sum Y_i \leq m$) is maximized when individual terms are maximized.
        \end{proof}

        Finally, we argue that $S_{\testing}$ contains many samples from $\p$.

        \begin{claim}
            \label{clm:uniformity-sample-size-bound-adv}
            With probability $0.999$, $m_{\p} \geq \frac{C \sqrt{n}}{\eps^2}$ for a sufficiently large constant $C$ required by \Cref{clm:c-Sp-query}.
        \end{claim}
    
        \begin{proof}
            Note that $|S_{\testing}|$ is a sum of $|S| = m$ iid Bernoulli random variables with parameter $\frac{1}{2}$ so that a standard Chernoff bound shows $\frac{m}{2} - 0.001m \leq |S_{\testing}| \leq \frac{m}{2} + 0.001m$ with high probability as $m \geq \frac{\sqrt{n}}{\eps^2 \lambda}$.
            Note that $m_{\p}$ is the sum of (at least) $\lambda m \geq \frac{C \lambda \sqrt{n}}{\hat{\lambda} \eps^2} \geq \frac{C \sqrt{n}}{\eps^2}$ iid Bernoulli random variables with parameter $\frac{1}{2}$ so a Chernoff bound yields $m_{\p} \geq \frac{\lambda m}{4} \geq \frac{C \sqrt{n}}{\eps^2}$ with high probability.
            We conclude with a union bound.
        \end{proof}

        A union bound concludes the proof of \Cref{lemma:subsample-flattening-alg}.
        Note that the final item follows directly from \Cref{fact:split-preserve-1-norm}.
    \end{proofof}

    Let $F(U_n)$ denote the split distribution obtained by taking a sample from $U_n$ and applying $F$.
    Then, let $F^*(x) = F'(F(x))$ where $F'$ is the reduction reducing $\eps$-testing of equality to $F(U_n)$ to $\eps/3$-testing of uniformity on $[6(n + |S_{\flatten}|)]$.

    Now, we argue that we can determine if $\p$ is uniform or far from uniform efficiently with queries.
    As in the distributional case, note that $F(U_n)$ is uniform if $\p$ is uniform while $F(\p)$ is $\eps/3$-far from uniform if $\p$ is $\eps$-far from uniform.
    Furthermore, by applying $F$ independently to each sample from $S_{\testing}$, the uncorrupted iid samples from $\p$ in $S_{\testing}$ are simply uncorrupted iid samples from $F(\p)$ in $F(S_{\testing})$.
    Thus, $F(S_{\testing})$ contains at least $m_{\p} = \bigOm{\lambda m} \geq \frac{C \sqrt{n}}{\eps^2}$ samples from $\p$.
    
    Note that we can estimate $\hat{m}_{\p}$ using \Cref{clm:m-p-query}.
    In particular, this since this is bias estimation up to error $\eps^2 \lambda$ which requires $O(\eps^{-4} \lambda^{-2})$ queries.
    
    It remains to show how to estimate $c(F(S_{\testing})_{\p})$, the number of $\p$-collisions in $F(S_{\testing})$, analogously to \Cref{clm:c-Sp-query}.
    By our lower bound on $m_{\p}$, the assumption of \Cref{clm:c-Sp-query} is satisfied.
    Following similar analysis, we obtain the following sub-routine.

    \begin{claim}
        \label{clm:c-Sp-query-adv}
        There exists a universal constant $C > 0$ such that if $m_{\p} = \frac{C \sqrt{n}}{\eps^2}$ there is an algorithm that uses $\bigO{\frac{n \log n}{\eps^{4} \lambda^{2}}}$ queries, and with probability $0.99$ computes an estimate $\hat{c}$ such that the following hold:
        \begin{enumerate}
            \item If $\p$ is uniform, then
            \begin{equation*}
                \hat{c} \leq \binom{m_{\p}}{2} \frac{1 + 0.1 \eps^{2}}{n}\text{.}
            \end{equation*}

            \item If $\p$ is $\eps$-far from uniform, then 
            \begin{equation*}
                \hat{c} - \binom{m_{\p}}{2} \frac{1 + 0.4 \eps^2}{n} \geq \max\left( \frac{m_{\p}^{2} \alpha}{30 n} , \frac{0.3 m_{\p}^2 \eps^2}{n} \right)
            \end{equation*}
            where $\norm{\p}_{2}^{2} = \frac{1 + \alpha}{n}$.
        \end{enumerate}
    \end{claim}

    \begin{proof}
        Since $m_{\p}$ is sufficiently large, the expectation and variance of $c(F(S_{\testing})_{\p})$ holds exactly as in \Cref{clm:c-Sp-query}.
        In particular, $c(F(S_{\testing})_{\p})$ satisfies the same properties as $c(S_{\p})$ in both completeness and soundness cases.
        
        The query bound to estimate $\hat{c}$ then follows from analogous analysis in \Cref{clm:c-Sp-query} and \Cref{lemma:subsample-collision-bound}.
        For example, in the case $\p$ is uniform, we can obtain an estimate $\hat{c}(F(S_{\testing})_{\p})$ with the same guarantees using
        \begin{equation*}
            \bigO{\frac{n^2 c(F(S_{\testing})) c(F(S_{\testing})_{\p})}{m_{\p}^{4} \eps^{4}}} = \bigO{\frac{n^2 (m \log n) (\lambda^2 m^2 /n)}{(\lambda^{4} m^{4}) \eps^{4}}} = \bigO{\frac{n \log n}{m \lambda^{2} \eps^{4}}} 
        \end{equation*}
        queries, where in the second inequality we have upper bounded $c(F(S_{\testing}))$ from \Cref{lemma:subsample-collision-bound} and $c(F(S_{\testing})_{\p})$ from (the proof of) \Cref{clm:c-Sp-query}, and lower bounded $m_{\p}$ from \Cref{clm:uniformity-sample-size-bound-adv}.
        
        A similar query bound follows for the case when $\p$ is far from uniform.
        For example, if $\alpha \geq 30$, we require
        \begin{equation*}
            \bigO{\frac{n^2 c(F(S_{\testing})) c(F(S_{\testing})_{\p})}{m_{\p}^{4} \alpha^2}} = \bigO{\frac{n^2 (m \log n) (\lambda^2 m^2 \alpha /n)}{(\lambda^{4} m^{4}) \alpha^{2}}} = \bigO{\frac{n \log n}{m \lambda^{2} \alpha}} = \bigO{\frac{n \log n}{m \lambda^{2} \eps^2}} 
        \end{equation*}
        queries.
        Note that we have lost a $O(\log n)$-factor in query complexity as \Cref{lemma:subsample-collision-bound} is weaker than the upper bound on $c(S_{\rD})$ in the distributional contamination setting.
    \end{proof}

    Following the above discussion, we describe the full adversarially robust algorithm.

    \begin{mdframed}
        \begin{enumerate}
            \item Estimate $\hat{\lambda}$ using \Cref{fact:lambda-est}.
            \item Let $m \geq \frac{C n^{1/2}}{\eps^{2} \hat{\lambda}}$ for some sufficiently large constant $C$. Drawn $m$ iid samples, denoted $S$.
            \item Compute $S_{\testing}, F$ with \Cref{lemma:subsample-flattening-alg}. 
            \item In the remaining algorithm, consider testing uniformity against $F(U_n)$.
            \item Compute $\hat{m}_{\p}$ using \Cref{clm:m-p-query} and $\hat{c}$ using \Cref{clm:c-Sp-query-adv}.
            \item Output accept if $\hat{c} < \binom{\hat{m}_{\p}}{2} \frac{1 + 0.4 \eps^2}{n}$.
            Otherwise reject.
        \end{enumerate}
    \end{mdframed}

    Correctness follows from identical analysis as in the distributional contamination setting (\Cref{thm:simple-uniformity-alg}).
    The sample complexity and query complexity bounds follow as above.
\end{proofof}

%\newpage

\section{Closeness Testing} \label{sec:closeness-UB}

In this section, we establish the following upper bound for closeness testing in the verification query model.

%\begin{restatable}{theorem}{closenessUB} \emph{(Closeness Testing with Verification Queries.)} \label{thm:closenessUB}
%Suppose we are given sample access and verification query access to two mixtures $\rD_1 = \lambda \p_1 + (1-\lambda)\q_1$, $r_2 = \lambda \p_2 + (1-\lambda)\q_2$ with unknown mixture parameter $\lambda \in (0,1)$, and arbitrary distributions $\p_1,\p_2,\q_1,\q_2 \colon [n] \to [0,1]$. For all $m \leq O(n)$, there is an algorithm that distinguishes between the case of $\p_1 = \p_2$ and $\norm{\p_1 - \p_2}_1 \geq \eps$ with probability $2/3$ using
%\[
%O\left(m + \frac{n^{2/3}}{\eps^{4/3}\lambda} + \frac{1}{\eps^4 \lambda^3}\right) ~\text{ samples and }~ O\left(\frac{n^2}{m^2} \cdot \frac{1}{\eps^4 \lambda^3}\right) ~\text{ verification queries.}~
%\]
%\end{restatable}

\begin{theorem} [Closeness Testing with Verification Queries]  \label{thm:closenessUB} Suppose we are given sample access and verification query access to two mixtures $\rD_1 = \lambda \p_1 + (1-\lambda)\q_1$, $r_2 = \lambda \p_2 + (1-\lambda)\q_2$ with unknown mixture parameter $\lambda \in (0,1)$, and arbitrary distributions $\p_1,\p_2,\q_1,\q_2 \colon [n] \to [0,1]$. For all $m \leq O(n)$, there is an algorithm that distinguishes between the case of $\p_1 = \p_2$ and $\norm{\p_1 - \p_2}_1 \geq \eps$ with probability $2/3$ using
\[
O\left(m + \frac{n^{2/3}}{\eps^{4/3}\lambda} + \frac{1}{\eps^4 \lambda^3}\right) ~\text{ samples and }~ O\left(\frac{n^2}{m^2} \cdot \frac{1}{\eps^4 \lambda^3}\right) ~\text{ verification queries.}~
\]
\end{theorem}

Our proof works in two steps similarly to standard closeness testing: first, use the classic flattening technique of \cite{diakonikolaskane2016} to "reduce" the $2$-norm of the mixtures $\rD_1,\rD_2$ to $b \approx 1/m$ using $O(m)$ samples up-front, then use a closeness tester for $\ell_2$ distance where the sample complexity improves for small $b$. For the first step we slightly strengthen the flattening lemma, which we also used crucially for our uniformity testing result, to the distributional contamination setting.

\flattening*

Next, we use the following closeness testing result under $\ell_2$ distance. This is the main component of the proof which we obtain by extending the result of \cite{DBLP:journals/cjtcs/DiakonikolasGPP19} to the verification query setting.

\begin{restatable}{theorem}{closenessltwo} \emph{(Tolerant $\ell_2$ Closeness Testing with Verification Queries.)} \label{thm:closeness-l2}
Suppose we are given sample access and verification query access to two mixtures $\rD_1 = \lambda \p_1 + (1-\lambda)\q_1$, $\rD_2 = \lambda \p_2 + (1-\lambda)\q_2$ with unknown mixture parameter $\lambda \in (0,1)$, and arbitrary distributions $\p_1,\p_2,\q_1,\q_2 \colon [n] \to [0,1]$. Moreover, suppose we are provided a number $b \in [1/n,1]$ such that $\norm{\rD_1}_2^2, \norm{\rD_2}_2^2 \leq b$ and $\norm{\p_1}_2^2,\norm{\p_2}_2^2 \leq b/\lambda$. Then, there is an algorithm that distinguishes between the case of $\norm{\p_1 - \p_2}_2 \leq \eps/4$ and $\norm{\p_1 - \p_2}_2 \geq \eps$ with probability $9/10$ using
\[
O\left(\frac{\sqrt{b}}{\eps^{2}\lambda^{3/2}} + \frac{b^2}{\eps^4\lambda^3}\right) ~\text{ samples and }~ O\left( \frac{b^2}{\eps^4\lambda^3} \right) ~\text{ verification queries.}
\]
\end{restatable}

We prove \Cref{thm:closeness-l2} in \Cref{sec:closenessL2}, and now continue to the proof of the main theorem. \\

\begin{proofof}{\Cref{thm:closenessUB}} Our tester is defined as follows.

\begin{mdframed}
    \begin{enumerate}
        \item Let $\min\left(n,\frac{5600 n^{2/3}}{\eps^{4/3}\lambda}\right) \leq m \leq n$. 
        \item Invoke the flattening procedure (\Cref{lemma:flattening}) using $O(m)$ samples from $\frac{1}{2}(\rD_1 + \rD_2)$  to obtain a random map $f$ preserving $\ell_1$ distances (item (1) of \Cref{lemma:flattening}) for which $\norm{f(\frac{1}{2}(\rD_1+\rD_2))}_2^2 \leq \frac{1}{m}$.
        \item Invoke the $\ell_2$-closeness tester (\Cref{thm:closenessUB}) on $f(\rD_1) = \lambda f(\p_1) + (1-\lambda)f(\q_1)$ and $f(\rD_2) = \lambda f(\p_2) + (1-\lambda)f(\q_2)$ using parameters $b := 4/m$ and $\eps' := \eps/\sqrt{n}$, and return the result.
    \end{enumerate}
\end{mdframed}

First, observe that by \Cref{lemma:flattening}, with probability $0.99$, the random map $f$ obtained in line (2) satisfies
\begin{align*} \label{eq:random-map-l2}
    \frac{1}{m} \geq \norm{f\left(\frac{1}{2}(\rD_1+\rD_2)\right)}_2^2 = \norm{\frac{1}{2}\left(f(\rD_1) + f(\rD_2)\right)}_2^2 \geq \frac{1}{4}\left( \norm{f(\rD_1)}_2^2 + \norm{f(\rD_2)}_2^2 \right)
\end{align*}
which implies that $\norm{f(\rD_1)}_2^2, \norm{f(\rD_2)}_2^2 \leq 4/m := b$ where $b$ is the parameter used in line (3). Similarly,
\begin{align*}
    \frac{1}{\lambda m} \geq \norm{f\left(\frac{1}{2}(\p_1+\p_2)\right)}_2^2 = \norm{\frac{1}{2}\left(f(\p_1) + f(\p_2)\right)}_2^2 \geq \frac{1}{4}\left( \norm{f(\p_1)}_2^2 + \norm{f(\p_2)}_2^2 \right)
\end{align*}
and so $\norm{f(\p_1)}_2^2, \norm{f(\p_2)}_2^2 \leq 4/\lambda m = b/\lambda$. 
%Next, observe that 
%\[
%b = \frac{4}{m} \leq \frac{4 \lambda^2 \eps^{4/3}}{ 5600 n^{2/3}} = \frac{\lambda^2}{1400}  (\eps/\sqrt{n})^{4/3} = \frac{\lambda^2 (\eps ')^{4/3}}{1400}
%\]
%thus satisfying the required bound for the use of \Cref{thm:closeness-l2} in line (3). 
Now, if $\p_1 = \p_2$, then $f(\p_1) = f(\p_2)$, and clearly $\norm{f(\p_1)-f(\p_2)}_2^2 = 0$ and so by \Cref{thm:closeness-l2}, the tester will correctly output "close" with probability $9/10$. Next, if $\norm{\p_1 - \p_2}_1 \geq \eps$, then $\norm{f(\p_1) - f(\p_2)}_1 \geq \eps$, and so $\norm{f(\p_1) - f(\p_2)}_2 \geq \eps/\sqrt{n} = \eps'$, thus satisfying the required bound for the use of \Cref{thm:closeness-l2}. Therefore, the tester correctly outputs "far" with probability $9/10$ in this case. All in all, by union bounding the failure probability of flattening and the $\ell_2$-closeness tester, the overall failure probability in both cases is at most $1/100 + 1/10 < 1/3$. This completes the proof of correctness.

We now bound the number of queries. By \Cref{thm:closeness-l2}, and our definition of $\eps' = \eps/\sqrt{n}$ and $b = 4/m$, the query complexity is bounded by
\[
O \left(\frac{b^2}{(\eps')^4\lambda^3}\right) = O\left(\frac{(1/m)^2}{(\eps/\sqrt{n})^4\lambda^3}\right) = O\left(\left\lceil \frac{n^2}{m^2} \right\rceil \cdot \frac{1}{\eps^4 \lambda^3}\right)
\]
as claimed. We now bound the number of samples. By \Cref{thm:closeness-l2} and our definitions $\eps' = \eps/\sqrt{n}$ and $b = O(1/m)$, the total number of samples used is bounded by 
\begin{align} \label{eq:samplecomplexity-closeness}
    O\left(\frac{1}{b} + \frac{\sqrt{b}}{(\eps')^{2}\lambda^{3/2}} + \frac{b^2}{(\eps')^4\lambda^3}\right) = O\left(m + \frac{n}{\sqrt{m} \eps^2 \lambda^{3/2}} + \frac{n^2}{m^2 \eps^4 \lambda^3}\right) \text{.}
\end{align}
\paragraph{Case 1.} If $n = \Omega(\frac{1}{\eps^4 \lambda^3})$, then $n^{1/3} = \Omega(\frac{1}{\eps^{4/3} \lambda})$, which implies that the bound $m = \Omega\left(\frac{n^{2/3}}{\eps^{4/3}\lambda}\right)$ holds by the bounds on $m$ in line (1) of the tester. Observe that in this case the first term in \cref{eq:samplecomplexity-closeness} dominates. Thus, the sample complexity in this case is $O\left(m + \frac{n^{2/3}}{\eps^{4/3}\lambda}\right)$.

\paragraph{Case 2.} If $n = O(\frac{1}{\eps^4 \lambda^3})$, then $m = n$ by the bounds on $m$ in line (1) of the tester. In this case, the second and third terms are bounded by $O\left(\frac{\sqrt{n}}{\eps^2 \lambda^{3/2}} + \frac{1}{\eps^4 \lambda^3}\right) = O(\frac{1}{\eps^4 \lambda^3})$. \end{proofof}

\subsection{\texorpdfstring{$\ell_2$}{l2} Closeness Testing} \label{sec:closenessL2}

In this section we establish our upper bound for closeness testing under $\ell_2$ distance in the verification query model, restated here for convenience.

\closenessltwo*

\begin{proof} We show that by using verification queries we are able to essentially simulate a standard $\ell_2$ closeness tester due to \cite{DBLP:journals/cjtcs/DiakonikolasGPP19} on the distributions $p_1,p_2$. This tester uses $O(\sqrt{b}/\eps^2)$ samples and reports "close" or "far" based on the number of $\p_1$-$\p_1$, $\p_2$-$\p_2$, and $\p_1$-$\p_2$ collisions it sees. For our purposes, it suffices to (a) take enough samples from the mixtures $\rD_1,\rD_2$ so that enough $\p_1,\p_2$ samples are present to simulate this tester, then (b) use verification queries to estimate the number of relevant collisions (those not involving $\q_1,\q_2$). To perform (a), we need at least $\approx m/\lambda$ samples from $\rD_1,\rD_2$, where $m$ is the number of samples needed for the standard closeness tester. %Of course, we do not know $\lambda_1,\lambda_2$, but they can be approximated within a constant factor using $O(1/\lambda)$ samples and queries. 

%\hadley{Add bit about learning $\lambda_1,\lambda_2$ and how approximation is okay}

%\begin{fact} [Mixture Parameter Estimation] \label{fact:lambda-est} There is a procedure using $O(1/\lambda)$ samples and queries which with probability $0.99$ returns $\hat{\lambda}_1 \in (\lambda_1/40000,\lambda_1]$ and $\hat{\lambda}_2 \in (\lambda_2/40000,\lambda_2]$. \end{fact}

%\begin{proof} By querying one sample from $r_1$ we have access to a $\lambda_1$-biased coin. Therefore, by invoking a standard bias estimation procedure (\Cref{lem:bias-add-unknown} with $C = 200$) we obtain $\hat{\lambda}_1 \in (\lambda_1/40000,\lambda]$ with probability $0.995$. Repeating this for $\lambda_2$ and taking a union bound completes the proof. \end{proof}

\paragraph{Notation.} 
For a mixture $\rD = \lambda \p + (1-\lambda) \q$ and set of samples $S_{\rD} \sim \rD$, we use $S_{\p},S_{\q}$ to denote the partition of $S_{\rD}$ into those samples generated by $\p,\q$ respectively. For a multi-set $S \subseteq [n]$ of samples, we use $C(S) = \{(x,y) \in \binom{S}{2} \colon x = y\}$ to denote the set of pairwise collisions in $S$ and let $c(S) := |C(S)|$. For two sets of samples $S_1,S_2$, we use $C(S_1,S_2) = \{(x,y) \in S_1 \times S_2 \colon x = y\}$ to denote the set of pairwise cross-collisions between $S_1,S_2$ and let $c(S_1,S_2) := |C(S_1,S_2)|$. \\

\noindent We start by stating the testing result of \footnote{We give a slightly modified version which follows by their same arguments. The reason is that we require a constant factor gap in the thresholds for the "close" and "far" cases since we will only be able to attain an approximation of the relevant statistics.} \cite{DBLP:journals/cjtcs/DiakonikolasGPP19}. For completeness we give a proof in \Cref{sec:deferred-closeness-standard}.

\begin{theorem} [Tolerant $\ell_2$ Closeness Tester, \cite{DBLP:journals/cjtcs/DiakonikolasGPP19}] \label{thm:closeness-l2-vanilla} Let $\p_1,\p_2 \colon [n] \to [0,1]$ be arbitrary distributions with $\norm{\p_1}_2^2,\norm{\p_2}_2^2 \leq b$ for $b \in (0,1]$. Let $S_{\p_1},S_{\p_2}$ be two sets of $m \geq 250000 \cdot \frac{\sqrt{b}}{\eps^2}$ samples from $\p_1,\p_2$ for sufficiently large constant $c > 0$, and let $Z = c(S_{\p_1}) + c(S_{\p_2}) - \frac{m-1}{m} \cdot c(S_{\p_1},S_{\p_2})$. The following hold.
\begin{enumerate}
    \item If $\norm{\p_1 - \p_2}_2 \leq \eps/4$, then $\Pr\left[Z < {m \choose 2}\frac{\eps^2}{8}\right] \geq 0.99$.
    \item If $\norm{\p_1 - \p_2}_2 \geq \eps$, then $\Pr\left[Z > {m \choose 2}\frac{\eps^2}{2}\right] \geq 0.99$.
\end{enumerate}
\end{theorem}

Before describing the tester, we provide two claims which show how verification queries are used to estimate the number of relevant collisions. We defer their proofs to \Cref{sec:deferred-closeness-source}. The first claim allows us to estimate the number of $\p_1$-$\p_1$ and $\p_2$-$\p_2$ collisions.

\begin{claim} \label{clm:self-collision-est} Let $S_{\rD}$ be a collection of $m/\lambda$ samples drawn from a mixture $\rD = \lambda \p + (1-\lambda)\q$ where $\norm{\rD}_2^2 \leq b$ and $\norm{\p} \leq b/\lambda$. Then, using $O(\frac{b^2}{\eps^4\lambda^3})$ verification queries $Q \subseteq S$ we can compute an estimate $\hat{c}(S_{\p})$ such that 
\[
\Pr_{S_{\rD},Q}\left[|\hat{c}(S_{\p}) - c(S_{\p})| > {m \choose 2} \frac{\eps^2}{100}\right] < 0.01 \text{.}
\]
\end{claim}

The second claim allows us to estimate the number of $\p_1$-$\p_2$ collisions.

\begin{claim} \label{clm:cross-collision-est} Let $S_{\rD_1},S_{\rD_2}$ be two collections of $m/\lambda$ samples drawn from mixtures $\rD_1 = \lambda \p_1 + (1-\lambda) \q_1$ and $\rD_2 = \lambda \p_2 + (1-\lambda) \q_2$ where $\norm{\rD_1}_2^2,\norm{\rD_2}_2^2 \leq b$ and $\norm{\p_1}_2^2,\norm{\p_2}_2^2 \leq b/\lambda$. Then, using $O(\frac{b^2}{\eps^4\lambda^3})$ verification queries $Q \subseteq S_{\rD_1} \cup S_{\rD_2}$ we can compute an estimate $\hat{c}(S_{\p_1},S_{\p_2})$ such that 
\[
\Pr_{S_{\rD_1},S_{\rD_2},Q}\left[|\hat{c}(S_{\p_1},S_{\p_2}) - c(S_{\p_1},S_{\p_2})| > {m \choose 2} \frac{\eps^2}{100}\right] < 0.01 \text{.}
\]
\end{claim}

We now describe our tester for \Cref{thm:closeness-l2}. %Recall $\lambda = \min(\lambda_1,\lambda_2)$. 

\begin{mdframed}
    \begin{enumerate}
    \item Invoke \Cref{fact:lambda-est} to obtain an estimate $\hat{\lambda}$ satisfying $\lambda/40000 < \hat{\lambda} \leq \lambda$ with probability $0.99$ using $O(1/\lambda)$ samples and queries. Let $m = C\left\lceil \frac{\sqrt{b/\hat{\lambda}}}{\eps^2} + \frac{b^2}{\eps^4\hat{\lambda}^2}\right\rceil$ where $C > 1$ is a sufficiently large constant and note that $C\left\lceil \frac{\sqrt{b}}{\eps^2 \sqrt{\lambda}} + \frac{b^2}{\eps^4\lambda^2}\right\rceil < m \leq O(\frac{\sqrt{b}}{\eps^2 \sqrt{\lambda}} + \frac{b^2}{\eps^4\lambda^2})$.
    \item Let $S_{\rD_1},S_{\rD_2}$ be two sets of $m/\lambda$ samples from $\rD_1,\rD_2$ respectively. Let $S_{\p_1} \subset S_{\rD_1}$, $S_{\p_2} \subset S_{\rD_2}$ be the samples which were drawn from $\p_1,\p_2$, respectively. %\Comment{Note that we do not know $S_{p_1},S_{p_2}$.}
    \item Using \Cref{clm:self-collision-est} obtain estimates $\hat{c}(S_{\p_1}),\hat{c}(S_{\p_2})$ for $c(S_{\p_1}),c(S_{\p_2})$. 
    \item Using \Cref{clm:cross-collision-est} obtain an estimate $\hat{c}(S_{\p_1},S_{\p_2})$ for $c(S_{\p_1},S_{\p_2})$. 
    \item Let $\hat{Z} = \hat{c}(S_{\p_1}) + \hat{c}(S_{\p_2}) - \frac{m-1}{m} \cdot\hat{c}(S_{\p_1},S_{\p_2})$.
    \item If $\hat{Z} < {m \choose 2}\frac{\eps^2}{4}$, then return "close". Otherwise, return "far".
\end{enumerate}
\end{mdframed}

We need to show that the estimator $\hat{Z}$ used in line (4) suffices in place of the estimator $Z$ from \Cref{thm:closeness-l2-vanilla}. First, since $\norm{\p_1}_2^2, \norm{\p_2}_2^2 \leq b/\lambda$, it suffices to have $250000 (\sqrt{b/\lambda})/\eps^2 = 250000 \sqrt{b}/(\sqrt{\lambda} \eps^2)$ samples from $\p_1,\p_2$ are to simulate \Cref{thm:closeness-l2-vanilla}. This is the reason for the bound $m \gg \frac{\sqrt{b}}{\eps^2 \sqrt{\lambda}}$ in line (1). The second bound $m \gg b^2/ \eps^4 \lambda^2$ will be needed later in the proof for a different reason.

Let $m_{\p_1} = |S_{\p_1}|, m_{\p_2} = |S_{\p_2}|$ be the random variables denoting the number of samples drawn from $\rD_1,\rD_2$ in line (2) which were generated by $\p_1,\p_2$. Note that $\mathbb{E}[\boldsymbol{m}_{\p_{1}}] = \mathbb{E}[\boldsymbol{m}_{\p_{2}}] = m$ and so by a Chernoff bound
\begin{align} \label{eq:mp1p2-bound}
    \Pr\left[|m_{\p_1} - m| > \sqrt{18m}\right] \leq 2\exp\left(-\big(\sqrt{18/m}\big)^2m/3\right) < \frac{1}{200}
\end{align}
Note that since $m > C$ for sufficiently large constant $C$, we have $\sqrt{18m} < 0.01m$. The same bounds hold for $\boldsymbol{m}_{p_2}$. By a union bound we get the following.
\begin{fact} \label{fact:mp-concentration} Let $\cC$ denote the event that
\[
0.99m < m - \sqrt{18m} \leq m_{\p_1},m_{\p_2} \leq m + \sqrt{18m} < 1.01m\text{.}
\]  
Then, $\Pr[\cC] \geq 0.99$.
\end{fact}

We will now condition on the bounds from \Cref{fact:mp-concentration} holding for $m_{\p_1},m_{\p_2}$. In particular, $m_{\p_1},m_{\p_2}$ are large enough to simulate \Cref{thm:closeness-l2-vanilla} on $\p_1,\p_2$ since $\norm{\p_1}_2^2,\norm{\p_2}_2^2 \leq b / \lambda$. Then, write $S_{\p_1} = S_{\p_1}' \sqcup S_{\p_1}^{\mathrm{extra}}$ and $S_{\p_2} = S_{\p_2}' \sqcup S_{\p_2}^{\mathrm{extra}}$ where $S_{\p_1}'$ is the first $m_{\mathrm{min}} = \min(m_{\p_1},m_{\p_2})$ samples in $S_{\rD_1}$ generated by $\p_1$ and $S_{\p_1}^{\mathrm{extra}}$ is the remaining $\max(0,m_{\p_1}-m_{\p_2})$, and $S_{\p_2}', S_{\p_2}^{\mathrm{extra}}$ are defined similarly. Then, let 
\begin{align} \label{eq:Zt}
    Z = c(S_{\p_1}') + c(S_{\p_2}') - \frac{m_{\mathrm{min}} - 1}{m_{\mathrm{min}}} \cdot c(S_{\p_1}',S_{\p_2}') %~\text{ and }~ t = {m_{\mathrm{min}} \choose 2} \eps^2\text{,}
\end{align}
and recall that our tester uses the estimator,
\begin{align} \label{eq:Zhatthat}
    \hat{Z} = \hat{c}(S_{\p_1}) + \hat{c}(S_{\p_2}) - \frac{m-1}{m} \cdot\hat{c}(S_{\p_1},S_{\p_2}) \text{.} %~\text{ and }~ \hat{t} = {m \choose 2} \eps^2 
\end{align}
Moreover, \Cref{thm:closeness-l2-vanilla} asserts the following: 
\begin{align} \label{eq:vanilla-closeness-guarantee}
    &\norm{\p_1 - \p_2} \leq \frac{\eps}{4} ~\implies~ \Pr\left[Z < {m_{\mathrm{min}} \choose 2}\frac{\eps^2}{8}\right] \geq 0.99 \nonumber \\ 
    &\norm{\p_1 - \p_2} \geq \eps ~\implies~ \Pr\left[Z > {m_{\mathrm{min}} \choose 2}\frac{\eps^2}{2}\right] \geq 0.99
\end{align}
Therefore, it remains to show that the estimator $\hat{Z}$ does not differ too much from $Z$ and the value ${m \choose 2}$ we use to threshold in line (6) does not differ too much from ${m_{\mathrm{min}} \choose 2}$. First, conditioning on the event $\cC$ from \Cref{fact:mp-concentration}, we have $m_{\mathrm{min}} = m-a$ where $|a| \leq 0.01m$ and so
\begin{align} \label{eq:threshold-desc}
    \cC &~\implies~ {m \choose 2} = {m_{\mathrm{min}} \choose 2} + \frac{m_{\mathrm{min}} a + (m_{\mathrm{min}}-1)a + a^2}{2} \nonumber \\
    &~\implies~ \left| {m_{\mathrm{min}} \choose 2} - {m \choose 2} \right| \leq 0.02 m_{\mathrm{min}} \ll {m_{\mathrm{min}} \choose 2}\frac{1}{100}
\end{align}
The following lemma bounds the error between $\hat{Z}$ and $Z$. The first source of error is due to the estimations obtained from \Cref{clm:self-collision-est} and \Cref{clm:cross-collision-est}. The second source of error, which is more painful to deal with, is the fact that we do not know how many samples were generated by $\p_1,\p_2$, and the number of samples are not equal. Nevertheless, they are the same in expectation, and are well-concentrated around their mean. We defer the proof of \Cref{lemma:unequal-error} to \Cref{sec:unequal-error}. (The proof of this lemma is where we use the bound $m \gg b^2 / \eps^4 \lambda^2$.)

\begin{lemma} \label{lemma:unequal-error} Consider $\hat{Z}$ computed in line (4) and $Z$ defined in \cref{eq:Zt}. Then, $\Pr[|\hat{Z} - Z| > {m \choose 2}\frac{\eps^2}{20}] < 0.06$. \end{lemma}

Now, recalling line (6) of the tester and the guarantee from \cref{eq:vanilla-closeness-guarantee}, the tester is correct with probability $0.99$ as long as $|\hat{Z} - Z| \leq {m \choose 2} \frac{\eps^2}{16}$ and $\left| {m_{\mathrm{min}} \choose 2} - {m \choose 2} \right| \leq  {m_{\mathrm{min}} \choose 2}\frac{1}{16}$. Therefore, by \Cref{fact:lambda-est}, \Cref{fact:mp-concentration}, \cref{eq:threshold-desc}, \Cref{lemma:unequal-error}, and a union bound, the tester succeeds with probability $91/100 > 9/10$, as claimed. \end{proof}

\subsubsection{Bounding the Error Due to Unequal \texorpdfstring{$p_1,p_2$}{} Samples: Proof of \texorpdfstring{\Cref{lemma:unequal-error}}{Lemma 5.7}} 
\label{sec:unequal-error}

\begin{proof} First, by definition of $Z,\hat{Z}$ and the triangle inequality, we have
\begin{align} \label{eq:ZZhat}
    &|\hat{Z} - Z| \leq \left|\big(\hat{c}(S_{\p_1}) + \hat{c}(S_{\p_2}) - \frac{m-1}{m} \hat{c}(S_{\p_1},S_{\p_2}) \big) - \big(c(S_{\p_1}') +  c(S_{\p_2}') - \frac{m_{\mathrm{min}} - 1}{m_{\mathrm{min}}} c(S_{\p_1}',S_{\p_2}')\big)\right| \nonumber \\
    &\leq \Big(|\hat{c}(S_{\p_1}) - c(S_{\p_1}')| + |\hat{c}(S_{\p_2}) - c(S_{\p_2}')| + \left|\hat{c}(S_{\p_1},S_{\p_2}) -  c(S_{\p_1}',S_{\p_2}')\right|\Big) + \left|\frac{\hat{c}(S_{\p_1},S_{\p_2})}{m} - \frac{c(S_{\p_1}',S_{\p_2}')}{m_{\mathrm{min}}}\right| \text{.}
\end{align}

\paragraph{Bounding the first term of \cref{eq:ZZhat}.} By the triangle inequality
\begin{align} \label{eq:1stterm-1}
    &\Big(|\hat{c}(S_{\p_1}) - c(S_{\p_1}')| + |\hat{c}(S_{p_2}) - c(S_{p_2}')| + \left|\hat{c}(S_{\p_1},S_{\p_2}) -  c(S_{\p_1}',S_{\p_2}')\right|\Big) \nonumber \\
    &\leq \Big(|\hat{c}(S_{\p_1}) - c(S_{\p_1})| + |\hat{c}(S_{\p_2}) - c(S_{\p_2})| + \left|\hat{c}(S_{\p_1},S_{\p_2}) -  c(S_{\p_1},S_{\p_2})\right|\Big) \nonumber \\
    &+ \Big(|c(S_{\p_1}) - c(S_{\p_1}')| + |c(S_{\p_2}) - c(S_{\p_2}')| + |c(S_{\p_1},S_{\p_2}) -  c(S_{\p_1}',S_{\p_2}')|\Big) \text{.}
\end{align}
Now, let $\cE_1$ denote the following good event:
\begin{align} \label{eq:est1}
    \cE_1\colon ~\max\big(|\hat{c}(S_{\p_1}) - c(S_{\p_1})|,|\hat{c}(S_{\p_2}) - c(S_{\p_2})|, |\hat{c}(S_{\p_1},S_{\p_2}) - c(S_{\p_1},S_{\p_2})|\big) \leq {m \choose 2}\frac{\eps^2}{100} \text{.}
\end{align}
By \Cref{clm:self-collision-est}, \Cref{clm:cross-collision-est}, and a union bound, $\Pr[\cE_1] \geq 0.97$. Conditioned on $\cE_1$, the first term of \cref{eq:1stterm-1} is bounded by at most $3{m \choose 2}\frac{\eps^2}{100}$. We now bound the second term of \cref{eq:1stterm-1}. Note that the following hold:
\begin{itemize}
    \item $c(S_{\p_1}') = c(S_{\p_1}) - (c(S_{\p_1}^{\mathrm{extra}}) + c(S_{\p_1}^{\mathrm{extra}},S_{\p_{1}}'))$,
    \item $c(S_{\p_2}') = c(S_{\p_2}) - (c(S_{\p_2}^{\mathrm{extra}}) + c(S_{\p_2}^{\mathrm{extra}},S_{\p_{2}}'))$, and
    \item $c(S_{\p_1}',S_{\p_2}') = c(S_{\p_1},S_{\p_2}) - \left(c(S_{\p_1}^{\mathrm{extra}},S_{\p_2}) + c(S_{\p_1},S_{\p_2}^{\mathrm{extra}})\right)$.
\end{itemize}
In particular, the second term of \cref{eq:1stterm-1} is upper bounded by 
\begin{align}
    &|c(S_{\p_1}) - c(S_{\p_1}')| + |c(S_{\p_2}) - c(S_{\p_2}')| + |c(S_{\p_1},S_{\p_2}) -  c(S_{\p_1}',S_{\p_2}')| \nonumber \\
    &\leq c(S_{\p_1}^{\mathrm{extra}}) +  c(S_{\p_1}^{\mathrm{extra}},S_{\p_{1}}') + c(S_{\p_2}^{\mathrm{extra}}) + c(S_{\p_2}^{\mathrm{extra}},S_{\p_{2}}') + c(S_{\p_1}^{\mathrm{extra}},S_{\p_2}) + c(S_{\p_1},S_{\p_2}^{\mathrm{extra}})
\end{align}
which is precisely the number of collisions incident on $S_{\p_1}^{\mathrm{extra}} \cup S_{\p_2}^{\mathrm{extra}}$. Let $c^{\mathrm{extra}}$ denote this quantity, which we now bound. By \Cref{fact:mp-concentration}, we have $|S_{\p_1}^{\mathrm{extra}} \cup S_{\p_2}^{\mathrm{extra}}| \leq 2\sqrt{18m}$ and $|S_{\p_1} \cup S_{\p_2}| \leq 2.02m$. Moreover, for any distinct samples $x \in S_{\p_1}^{\mathrm{extra}} \cup S_{\p_2}^{\mathrm{extra}}$ and $y \in S_{\p_1} \cup S_{\p_2}$, the probability that $x,y$ collide is either $\langle \p_1,\p_2\rangle$, $\norm{\p_1}_2^2$, or $\norm{\p_2}_2^2$, all of which are at most $b/\lambda$ by our assumed bounds on $\norm{\p_1}_2^2$, $\norm{\p_2}_2^2$. Therefore, %\hadley{edit constants}
\begin{align} \label{eq:extra-collision-bound-exp}
    \mathbb{E}[c^{\mathrm{extra}}] &\leq 2\sqrt{18m} \cdot 2.02m \cdot \frac{b}{\lambda} < \frac{18 m^{3/2} b}{\lambda} < {m \choose 2}\frac{\eps^2}{20000}
\end{align}
where the last inequality holds since $m > \frac{C b^2}{\eps^4 \lambda^2}$ for sufficiently large constant $C$, i.e. $\sqrt{m}\eps^2 \gg b/\lambda$. Let $\cE_2$ denote the following good event:
\begin{align} \label{eq:extra-collision-whp}
    \cE_2 \colon c^{\mathrm{extra}} \leq {m \choose 2}\frac{\eps^2}{200} ~\text{ which satisfies }~ \Pr\left[\cE_2\right] \geq 0.99
\end{align}
by Markov's inequality. Conditioning on the events $\cC,\cE_1,\cE_2$, the first term of \cref{eq:ZZhat} is bounded by
\begin{align}
    \cC \wedge \cE_1 \wedge \cE_2 ~\implies~ &|\hat{c}(S_{\p_1}) - c(S_{\p_1})| + |\hat{c}(S_{\p_2}) - c(S_{\p_2})| + \left|\hat{c}(S_{\p_1},S_{\p_2}) -  c(S_{\p_1},S_{\p_2})\right| \nonumber \\
    &\leq {m \choose 2}\frac{3\eps^2}{100} + {m \choose 2}\frac{\eps^2}{200} \leq {m \choose 2}\frac{\eps^2}{28} \text{.}
\end{align}

\paragraph{Bounding the second term of \cref{eq:ZZhat}.} First, conditioning on the bounds from \Cref{fact:mp-concentration} and applying the triangle inequality, we have
\begin{align} \label{eq:2ndtermbound}
    \left|\frac{\hat{c}(S_{\p_1},S_{\p_2})}{m} - \frac{c(S_{\p_1}',S_{\p_2}')}{m_{\mathrm{min}}} \right| &= \frac{1}{m_{\mathrm{min}}}\left|\frac{m_{\mathrm{min}}}{m}\hat{c}(S_{\p_1},S_{\p_2}) - c(S_{\p_1}',S_{\p_2}') \right| \nonumber \\
    &\leq \frac{2}{m} \cdot \left|\frac{m_{\mathrm{min}}}{m}\hat{c}(S_{\p_1},S_{\p_2}) - c(S_{\p_1}',S_{\p_2}') \right| \nonumber \\
   &\leq \frac{2}{m} \cdot \left(\left|\frac{m_{\mathrm{min}}}{m}\hat{c}(S_{\p_1},S_{\p_2}) - \frac{m_{\mathrm{min}}}{m} c(S_{\p_1},S_{\p_2}) \right| + \left|\frac{m_{\mathrm{min}}}{m}c(S_{\p_1},S_{\p_2}) - c(S_{\p_1}',S_{\p_2}') \right|\right) \nonumber \\
   &\leq \frac{4}{m}\left|\hat{c}(S_{\p_1},S_{\p_2}) - c(S_{\p_1},S_{\p_2}) \right| + \frac{2}{m}\left|\frac{m_{\mathrm{min}}}{m}c(S_{\p_1},S_{\p_2}) - c(S_{\p_1}',S_{\p_2}') \right|
\end{align}
and we already have a bound on the first term of \cref{eq:2ndtermbound} by conditioning on $\cE_1$ (recall \cref{eq:est1}). For the second term of \cref{eq:2ndtermbound}, recall that $c(S_{\p_1}',S_{\p_2}') = c(S_{\p_1},S_{\p_2}) - (c(S_{\p_1}^{\mathrm{extra}},S_{\p_2}) + c(S_{\p_1},S_{\p_2}^{\mathrm{extra}}))$ and conditioning on the event $\cC$ (from \Cref{fact:mp-concentration}), we have $|m_{\mathrm{min}} - m| \leq \sqrt{18m}$, and so $|(m_{\mathrm{min}} / m) - 1| \leq \sqrt{18/m})$. Thus,
\begin{align} \label{eq:2nd2nd}
    \cC  ~\implies~ &\left|\frac{m_{\mathrm{min}}}{m}c(S_{\p_1},S_{\p_2}) - c(S_{\p_1}',S_{\p_2}') \right| \nonumber \\
    &\leq \left|\big(1+\sqrt{18/ m}\big) \cdot c(S_{\p_1},S_{\p_2}) - c(S_{\p_1},S_{\p_2}) + (c(S_{\p_1}^{\mathrm{extra}},S_{\p_2}) + c(S_{\p_1},S_{\p_2}^{\mathrm{extra}}))\right| \nonumber \\ 
    &\leq \sqrt{\frac{18}{m}} \cdot c(S_{\p_1},S_{\p_2}) + c^{\mathrm{extra}}\text{.}
\end{align}
where $c^{\mathrm{extra}}$ is the random variable denoting the number of collisions incident on $S_{\p_1}^{\mathrm{extra}} \cup S_{\p_2}^{\mathrm{extra}}$. Now, $|S_{\p_1}|, |S_{\p_2}| \leq 1.01m$ by the bound in \Cref{fact:mp-concentration}, and each $x \in S_{\p_1}$, $y \in S_{\p_2}$ form a collision with probability at most $\langle \p_1,\p_2 \rangle \leq b/\lambda$. Thus, $\mathbb{E}[c(S_{\p_1},S_{\p_2})] \leq (1.01m)^2 \cdot (b/\lambda)$. Using this bound, \cref{eq:2nd2nd}, and \cref{eq:extra-collision-bound-exp}, we have
\begin{align} \label{eq:2nd2nd2nd}
     \mathbb{E}\left[\frac{2}{m} \cdot \left|\frac{m_{\mathrm{min}}}{m}c(S_{\p_1},S_{\p_2}) - c(S_{\p_1}',S_{\p_2}') \right| ~\Bigg|~  \cC \wedge \cE_1 \wedge \cE_2 \right] \leq \frac{m^{1/2}b}{\lambda} + \frac{m \eps^2}{10000} \leq \frac{m\eps^2}{5000}
\end{align}
where the second inequality used the fact that $\sqrt{m} \gg b/(\eps^2 \lambda)$. Now, returning to \cref{eq:2ndtermbound}, applying \cref{eq:est1} and \cref{eq:2nd2nd2nd}, we get
\begin{align}
    \mathbb{E}\left[\left|\frac{\hat{c}(S_{\p_1},S_{\p_2})}{m} - \frac{c(S_{\p_1}',S_{\p_2}')}{m_{\mathrm{min}}} \right| ~\Bigg|~  \cC \wedge \cE_1 \wedge \cE_2 \right] &\leq \frac{m\eps^2}{10} \leq {m \choose 2} \frac{\eps^2}{100000}
\end{align}
and so by Markov's
\begin{align} \label{eq:whp-last}
    \Pr\left[\left|\frac{\hat{c}(S_{\p_1},S_{\p_2})}{m} - \frac{c(S_{\p_1}',S_{\p_2}')}{m_{\mathrm{min}}} \right| \leq {m \choose 2} \frac{\eps^2}{1000} ~\Bigg|~  \cC \wedge \cE_1 \wedge \cE_2\right] \geq 0.99 \text{.}
\end{align}

\paragraph{Bounding $|\hat{Z}-Z|$.} Finally, let $\cC$ denote the good event that the bounds from \Cref{fact:mp-concentration} hold. Returning to \cref{eq:ZZhat}, by a union bound and conditioning on the events $\cC,\cE_1,\cE_2$, and using \cref{eq:whp-last}, we have 
\begin{align*}
    \Pr\left[|\hat{Z} - Z| > {m \choose 2} \frac{\eps^2}{20}\right] &\leq \Pr[\neg \cC \vee \neg \cE_1 \vee\neg \cE_2] +  \Pr\left[\left|\frac{\hat{c}(S_{\p_1},S_{\p_2})}{m} - \frac{c(S_{\p_1}',S_{\p_2}')}{m_{\mathrm{min}}} \right| > {m \choose 2} \frac{\eps^2}{1000} ~\Bigg|~  \cC \wedge \cE_1 \wedge \cE_2\right] \nonumber \\
    &\leq 0.06
\end{align*}
and this completes the proof of \Cref{lemma:unequal-error}. \end{proof}

\subsubsection{Estimating Collision Counts with Verification Queries} \label{sec:deferred-closeness-source}

\begin{proofof}{\Cref{clm:self-collision-est}} First, since $\norm{\rD}_2^2 \leq b$ and $\norm{\p}_2^2 \leq b/\lambda$ we have 
\[
\mathbb{E}_{S_{\rD}}[c(S_{\rD})] = \sum_{(x,y) \in {S_{\rD} \choose 2}} \Pr[x = y] \leq {m/\lambda \choose 2}  \cdot b\leq O(m^2 b /\lambda^2)
\]
\[
\mathbb{E}_{S_{\rD}}[c(S_{\p})] \leq \sum_{(x,y) \in {S_{\rD} \choose 2}} \Pr[x,y \sim \p] \cdot \Pr[x=y ~|~ x,y \sim \p] \leq {m / \lambda \choose 2} \cdot \lambda^2 \cdot \frac{b}{\lambda} \leq O(m^2 b /\lambda)
\]
and so by Markov's, we have $c(S_{\rD}) \leq O(b m^2 / \lambda^2)$ and $c(S_{\p}) \leq O(b m^2 /\lambda)$ with probability $0.999$. Thus, we now condition on these bounds. Note that $c(S_{\rD})$ is known to the algorithm simply by counting the number of collisions in $S_{\rD}$, and we have access to a $\mathrm{Bern}(c(S_{\p})/c(S_{\rD}))$ random variable using two verification queries to $S_{\rD}$ as follows: simply pick a uniform random collision $(x,y) \in C(S_{\rD})$ and query these two samples. Observe that this reveals a $\p$-$\p$ collision with probability exactly $c(S_{\p})/c(S_{\rD})$. 

Now, we can use a standard bias estimation procedure (\Cref{lem:bias-add}), to estimate $c(S_{\p})/c(S_{\rD})$. In particular, if we can obtain $\hat{\rho}$ such that 
\begin{align} \label{eq:est-ratio}
    \left|\hat{\rho} - \frac{c(S_{\p})}{c(S_{\rD})}\right| \leq \frac{\eps^2}{100 c(S_{\rD})} {m \choose 2} ~\Longrightarrow~ \left|c(S_{\rD})\hat{\rho} - c(S_{\p})\right| \leq \frac{\eps^2}{100} {m \choose 2}
\end{align}
then we can use the estimator $\hat{c}(S_{\p}) := c(S_{\rD}) \hat{\rho}$ and this satisfies the claim. Thus, we invoke \Cref{lem:bias-add} to obtain $\hat{\rho}$ satisfying the inequality in the LHS of \cref{eq:est-ratio}, using the upper bound $c(S_{\p})/c(S_{\rD}) \leq O(b m^2 / \lambda c(S_{\rD}))$ (by our condition on $c(S_{\p})$ and the fact that the algorithm knows $c(S_{\rD})$ explicitly) and additive error parameter $\frac{\eps^2}{100 c(S_{\rD})} {m \choose 2}$. Plugging these bounds into \Cref{lem:bias-add}, the number of queries is bounded by
\[
O\left(\frac{b m ^2}{\lambda c(S_{\rD})} \cdot \left(\frac{c(S_{\rD})}{\eps^2 {m \choose 2}}\right)^2\right) \leq O\left(\frac{b \cdot c(S_{\rD})}{\eps^4 m^2 \lambda}\right) \leq O\left(\frac{b^2}{\eps^4 \lambda^3}\right)
\]
where the final inequality used our bound $c(S_{\rD}) \leq O(bm^2 /\lambda^2)$. This completes the proof.
\end{proofof}

\vspace{2mm}

\begin{proofof}{\Cref{clm:cross-collision-est}} First, by Cauchy-Schwarz and our bounds $\norm{\rD_1}_2^2,\norm{\rD_2}_2^2 \leq b$ and $\norm{\p_1}_2^2,\norm{\p_2}_2^2 \leq b/\lambda$, we have $\langle \rD_1,\rD_2 \rangle \leq b$ and $\langle \p_1,\p_2 \rangle \leq b/\lambda$. Therefore,
\[
\mathbb{E}_{S_{\rD_1},S_{\rD_2}}[c(S_{\rD_1},S_{\rD_2})] = \sum_{x \in S_{\rD_1}, y \in S_{\rD_2}} \Pr[x = y] = \frac{m^2}{\lambda^2}\langle \rD_1,\rD_2 \rangle \leq O(b m^2 / \lambda^2)
\]
and
\begin{align*}
    \mathbb{E}_{S_{\rD_1},S_{\rD_2}}[c(S_{\p_1},S_{\p_2})] &= \sum_{x \in S_{\rD_1}, y \in S_{\rD_2}} \Pr[x \sim \p_1\text{, } y \sim \p_2] \Pr[x = y ~|~ x \sim \p_1\text{, } y \sim \p_2] \\ 
    &= \frac{m^2}{\lambda^2} \cdot \lambda^2 \cdot \langle \p_1,\p_2 \rangle \leq O(b m^2 /\lambda)
\end{align*}
and so by Markov's we have $c(S_{\rD_1},S_{\rD_2}) \leq O(b m^2 / \lambda^2)$ and $c(S_{\p_1},S_{\p_2}) \leq O(b m^2 / \lambda)$ with probability $0.999$. Thus, we now condition on these events. Note that $c(S_{\rD_1},S_{\rD_2})$ is known to the algorithm explicitly without using verification queries, and we have access to a $\mathrm{Bern}(c(S_{\p_1},S_{\p_2})/c(S_{\rD_1},S_{\rD_2}))$ random variable using two verification queries as follows: simply pick a uniform random collision $(x,y) \in C(S_{\rD_1},S_{\rD_2})$ and query these two samples. Observe that this reveals a $\p_1$-$\p_2$ collision with probability exactly $c(S_{\p_1},S_{\p_2})/c(S_{\rD_1},S_{\rD_2})$. 

Now, we can use a standard bias estimation procedure (\Cref{lem:bias-add}), to estimate $c(S_{\p_1},S_{\p_2})/c(S_{\rD_1},S_{\rD_2})$. In particular, if we can obtain $\hat{\rho}$ such that 
\begin{align} \label{eq:est-ratio-cross}
    \left|\hat{\rho} - \frac{c(S_{\p_1},S_{\p_2})}{c(S_{\rD_1},S_{\rD_2})}\right| \leq \frac{\eps^2}{100 c(S_{\rD_1},S_{\rD_2})} {m \choose 2} ~\Longrightarrow~ \left|c(S_{\rD_1},S_{\rD_2})\hat{\rho} - c(S_{\p_1},S_{\p_2})\right| \leq \frac{\eps^2}{100} {m \choose 2}
\end{align}
then we can use the estimator $\hat{c}(S_{\p_1},S_{\p_2}) := c(S_{\rD_1},S_{\rD_2}) \hat{\rho}$ and this satisfies the claim. Thus, we invoke \Cref{lem:bias-add} to obtain $\hat{\rho}$ satisfying the inequality in the LHS of \cref{eq:est-ratio-cross}, using the upper bound 
\[
\frac{c(S_{\p_1},S_{\p_2})}{c(S_{\rD_1},S_{\rD_2})} \leq O\left(\frac{b m^2}{\lambda c(S_{\rD_1},S_{\rD_2})}\right)
\]
(by our condition on $c(S_{\p_1},S_{\p_2})$ and the fact that the algorithm knows $c(S_{\rD_1},S_{\rD_2})$ explicitly) and additive error parameter $\frac{\eps^2}{100 c(S_{\rD_1},S_{\rD_2})} {m \choose 2}$. Plugging these bounds into \Cref{lem:bias-add}, the number of queries is bounded by
\[
O\left(\frac{b m ^2}{\lambda c(S_{\rD_1},S_{\rD_2})} \cdot \left(\frac{c(S_{\rD_1},S_{\rD_2})}{\eps^2 {m \choose 2}}\right)^2\right) \leq O\left(\frac{b \cdot c(S_{\rD_1},S_{\rD_2})}{\eps^4 m^2 \lambda}\right) \leq O\left(\frac{b^2}{\eps^4 \lambda^3}\right)
\]
where the final inequality used our bound $c(S_{\rD_1},S_{\rD_2}) \leq O(bm^2 /\lambda^2)$. This completes the proof. \end{proofof}

\subsubsection{Deferred Calculations for Standard \texorpdfstring{$\ell_2$}{l2} Closeness Testing} 
\label{sec:deferred-closeness-standard}

\begin{proofof}{\Cref{thm:closeness-l2-vanilla}} Let $\alpha = \norm{\p_1 - \p_2}_2^2$ and let $t = {m \choose 2}\eps^2/2$. Lemmas 3.2 and 3.3 of \cite{DBLP:journals/cjtcs/DiakonikolasGPP19} establish that\footnote{The variance bound comes from Lemma 3.3 along with the fact that $\norm{\p-\q}_4^2 \leq \norm{\p-\q}_2^2 = \alpha$.}
\begin{align} \label{eq:E&V}
    \mathbb{E}[Z] = {m \choose 2} \alpha ~\text{ and }~ \Var[Z] \leq 116m^2 b + 8m^3 \alpha \sqrt{b}\text{.}
\end{align}

We first prove item (1). Suppose that $\alpha \leq \eps^2/16$. Then, using \cref{eq:E&V} and Chebyshev's inequality yields
\begin{align}
    \Pr\left[Z \geq {m \choose 2}\frac{\eps^2}{8}\right] \leq \Pr\left[\left|Z - \mathbb{E}[Z]\right| \geq {m \choose 2}\frac{\eps^2}{16}\right] \leq \frac{\Var[Z]}{({m \choose 2}\frac{\eps^2}{16})^2} \leq \frac{267264 \cdot b}{m^2 \eps^4} + \frac{1152 \sqrt{b}}{m\eps^2} < 1/100 \nonumber
\end{align}
when $m \geq 250000 \sqrt{b}/\eps^2$, as claimed. We now prove item (2). Suppose that $\alpha \geq \eps^2$. Again, using using \cref{eq:E&V} and Chebyshev's inequality yields
\begin{align}
    \Pr\left[Z \leq {m \choose 2}\frac{\eps^2}{2}\right] &\leq \Pr\left[\left|Z - \mathbb{E}[Z]\right| \geq {m \choose 2}\frac{\max(\alpha,\eps^2)}{2}\right] \leq \frac{4\Var[Z]}{({m \choose 2}\max(\alpha,\eps^2))^2} \nonumber \\
    &\leq \frac{4176 \cdot b}{m^2 \cdot \max(\alpha^2,\eps^4)} + \frac{288 \sqrt{b}}{m \cdot \max(\alpha,\frac{\eps^4}{\alpha})} \leq \frac{4176 \cdot b}{m^2 \eps^4} + \frac{288 \sqrt{b}}{m \eps^2} < 1/100 \nonumber
\end{align}
again when $m \geq 250000 \sqrt{b}/\eps^2$, as claimed. \end{proofof}

\section{Sample Query Tradeoff Lower Bounds} \label{sec:LB}
\newcommand{\dom}{m}
\newcommand{\sdom}{d}

In this section, we present our lower bounds.
In particular, we obtain tight (up to constant factors) sample-query trade-offs for bias estimation (\Cref{thm:bias-estimation-lb}), uniformity testing (\Cref{thm:dist-contamination-lb-eps}), and closeness testing (\Cref{thm:dist-contam-closeness-lb}).

\subsection{Bias Estimation}

We begin with a simple bound stating that any bias estimation algorithm essentially needs to query all samples.
A bias estimation algorithm is $\eps$-accurate if given sample access to mixture $\lambda \p + (1-\lambda)\q$ where $\p, \q$ are Bernoulli random variables, can determine the bias of $\p$ up to error $\varepsilon$ with probability at least $2/3$.

\BiasEstimationLB*

\begin{proof}
    We will show that any algorithm that can determine if $\p$ has bias $\frac{1}{2} + \eps$ or $\frac{1}{2} - \eps$ requires $\eps^{-2} \lambda^{-1}$ queries (and therefore $\eps^{-2} \lambda^{-1}$ samples).
    We will reduce from bias estimation (in the standard setting with clean samples).
    Consider the following task:
    Given sample access to two coins $\coin_1, \coin_2$, one with bias $\lambda + 2 \lambda \eps$ and one with bias $\lambda - 2 \lambda \eps$, determine which coin is which.
    It is well known via a standard bias estimation lower bound that this task requires $\bigOm{\lambda^{-1} \eps^{-2}}$ samples (see \Cref{lemma:coin-distinguish-lb}).
    Alternatively, it suffices to determine the bias of $\coin_1$.

    Consider the following instance.
    Fix $\alpha = \frac{\lambda}{1 - \lambda}$.
    Note that $1 - \alpha = \frac{1 - 2 \lambda}{1 - \lambda}$.
    There are two cases:
    \begin{enumerate}
        \item $\p$ is a coin with bias $\frac{1}{2} + \eps$ and $\q$ is a coin with bias $\alpha (\frac{1}{2} - \eps) + (1 - \alpha) \frac{1}{2}$.
        \item $\p$ is a coin with bias $\frac{1}{2} - \eps$ and $\q$ is a coin with bias $\alpha (\frac{1}{2} + \eps) + (1 - \alpha) \frac{1}{2}$.
    \end{enumerate}
    In either case, the mixture $\rD = \lambda \p + (1 - \lambda) \q$ is a fair coin with bias $\frac{1}{2}$, since $(1 - \lambda) \alpha = \lambda$.

    Suppose there is an algorithm $\innerAlg$ using $m$ samples and $q$ queries that determines if $\p$ has bias $\frac{1}{2} - \eps$ or $\frac{1}{2} + \eps$.
    We design an algorithm for bias estimation of $\coin_1$ using $q$ samples.
    Our algorithm begins by providing $m$ samples of a fair coin to $\innerAlg$, denoted $(y_1, \dotsc, y_m)$.
    Let $(Z_1, \dotsc, Z_m)$ denote whether the $i$-th sample is from $\p$ or $\q$.
    Let $(b_1, \dotsc, b_q)$ denote the $q$ samples from $\coin_1, \coin_2$.
    Suppose the algorithm $\innerAlg$ queries indices $(i_1, \dotsc, i_q)$.
    For each query $j \in [q]$, if $y_{i_j} = H$, respond to the query with $\p$ if the $j$-th flip of $\coin_1$ is heads and $\q$ if tails.
    If $y_{i_j} = T$ respond to the query with $\p$ if the $j$-th flip of $\coin_2$ is heads and $\q$ if tails.
    Finally, if $\innerAlg$ outputs $\p$ has bias $\frac{1}{2} + \eps$, return that $\coin_1$ has bias $\lambda + 2 \lambda \eps$.
    Otherwise, if $\innerAlg$ outputs $\p$ has bias $\frac{1}{2} - \eps$, return that $\coin_1$ has bias $\lambda - 2 \lambda \eps$.

    We argue that this algorithm correctly determines the bias of $\coin$.
    Suppose $\p$ has bias $\frac{1}{2} + \eps$.
    Then, we have
    \begin{align*}
        \Pr(Z_i = \p|y_i = H) &= \frac{\Pr(Z_i = \p, y_i = H)}{\Pr(y_i = H)} \\
        &= \frac{\lambda (\frac{1}{2} + \eps)}{\lambda (\frac{1}{2} + \eps) + \lambda(\frac{1}{2} - \eps) + (1 - 2 \lambda) \frac{1}{2}} = \lambda + 2 \lambda \eps \text{.}
    \end{align*}
    Similarly, we have 
    \begin{equation*}
        \Pr(Z_i = \p|y_i = T) = \lambda - 2 \lambda \eps \text{.}
    \end{equation*}
    Therefore, this successfully simulates the case where $\coin_1$ has bias $\lambda + 2 \lambda \eps$ while $\coin_2$ has bias $\lambda - 2 \lambda \eps$.
    In particular, since $\innerAlg$ returns $\p$ has bias $\frac{1}{2} + \eps$, we return that $\coin_1$ has bias $\lambda + 2 \lambda \eps$.

    On the other hand, suppose $\p$ has bias $\frac{1}{2} - \eps$.
    Then, using similar arguments as above, we have
    \begin{align*}
        \Pr(Z_i = \p|y_i = H) &= \lambda - 2 \lambda \eps \\
        \Pr(Z_i = \p|y_i = T) &= \lambda + 2 \lambda \eps\text{.}
    \end{align*}
    and similarly $\Pr(Z_{i} = p | y_{i} = T) = \frac{1}{2} - \eps$.
    In particular, this simulates the case where $\coin_1$ has bias $\lambda - 2 \lambda \eps$ while $\coin_2$ has bias $\lambda + 2 \lambda \eps$.
    Since $\innerAlg$ returns $\p$ has bias $\frac{1}{2} - \eps$, we return that $\coin_1$ has bias $\lambda - 2 \lambda \eps$.

    Finally, it is clear to see that our algorithm for determining the bias of $\coin_1$ has sample complexity $q$.
\end{proof}

% We extend the above lower bound to high dimensional bias estimation.
% A high-dimensional bias estimation algorithm is $\eps$-accurate in $\ell_{2}$ (resp. $\ell_{\infty}$) norm if given sample access to $(p + q)/2$ where $p, q$ are products of $d$ Bernoulli distributions, it can estimate the mean of $p$ up to $\eps$-error in $\ell_{2}$ (resp. $\ell_{\infty}$) norm with probability at least $2/3$.
% Since $\ell_{\infty}$ bias estimation requires $O(\eps^{-2})$ samples, the above lower bound immediately extends to $\ell_{\infty}$ bias estimation.
% Thus, we focus on $\ell_{2}$ mean estimation.

% \begin{theorem}
%     \label{thm:d-bias-estimation-l2-lb}
%     Any $\eps$-accurate bias estimation in $\ell_{2}$ norm requires $\Omega(d/\eps^2)$ samples and $\Omega(d/\eps^2)$ queries.  
% \end{theorem}

% \begin{proof}
    
% \end{proof}

\subsection{Distribution Testing Lower Bound Preliminaries}

We present some preliminaries that will be useful for both uniformity and closeness testing.

\begin{definition}[Poisson Sampling]
    Given a non-negative measure $\p$ over $[n]$ and an integer $m$, the Poisson sampling model samples a number $m' \sim 
    \Poi \left( m \norm{\p}_{1} \right)$, and draws $m'$ samples from $\p / \norm{\p}_{1}$. 
    Define $T \in \R^n$ to be the random vector where $T_i$ counts the number of element $i$ seen.
    We write $\PoiS( m , \p )$ to denote the distribution of the random vector $T$.
    We say $\innerAlg$ is a Poissonized tester with sample complexity $m$ if it takes as input a sample count vector $T \sim \PoiS( m , \p )$.
    \label{def:poi_samp}
\end{definition}

\begin{lemma}
    \label{lem:poisson-sampling}
    Let $\calH_0, \calH_1$ be a meta-distribution over a finite universe $\calX$ satisfying $\norm{\rD}_1 \in (0.5, 2)$.
    Let $\delta \in (0, 1/3)$ and $m$ be a positive integer satisfying $m \geq \log(10/\delta)$. 
    Consider the following two statements:
    \begin{enumerate}[1.]
        \item 
        For any deterministic Poissonized tester $\innerAlg$ 
        with sample complexity $m$,
        if $\innerAlg$ is $\delta$-correct with respect to $\calH_0$ and $\calH_1$, then $\innerAlg$ requires $q$ queries.
        \item 
        For any randomized tester $\innerAlg$ 
        with sample complexity $m' := m/10$,
        if $\innerAlg$ is $\delta/10$-correct with respect to $\calH_0$ and $\calH_1$, then $\innerAlg$ requires $q$ queries.
    \end{enumerate}
    The first statement implies the second statement.
\end{lemma}

\begin{proof}
    Let $\innerAlg$ be a randomized tester that consumes $m'$ samples.
    Consider the negation of the second statement. 
    In particular, assume that $\innerAlg$ is $0.1 \delta$-correct with respect to $\calH_0$ and $\calH_1$.
    We show that this will contradict the first statement.
    
    By Markov's inequality, with probability at least $2/3$ over the choice of the random string $r$, we have that the induced deterministic tester $\innerAlg(;r)$ is $0.3 \delta$-correct with respect to $\calH_0$ and $\calH_1$.
    
    We will now convert the tester into a Poissonized one.
    In particular, consider the Poissonized tester $\bar \innerAlg$ obtained as follows. 
    We first take $k \sim \Poi(m)$ samples from the underlying distribution $\rD / \norm{\rD}_1$.
    If $k \geq m'$, we take the first $m'$ samples, and feed it to $\innerAlg(;r)$.
    If $k < m'$, we simply return reject.
    Since we assume $\norm{\rD}_1 \in (0.5, 2)$ and $m \geq \log(10/\delta)$, 
    it then follows from standard Poisson concentration (e.g. \Cref{lemma:poisson-concentration}) that 
    $k \geq m'$ with probability at least $1 - \delta$.
    In particular, this implies that $\bar \innerAlg$ is a deterministic Poissonized tester with sample complexity $m$ that is $0.4 \delta$-correct with respect to $\calH_0$ and $\calH_{1}$.
    This therefore contradicts the first statement of the lemma.
\end{proof}

\subsection{Uniformity Testing}

In this section, we give our lower bound for uniformity testing.
We begin with a simple argument that resolves the sample-query dependence and shows any $m$-sample algorithm needs $\bigOm{n/m}$ samples.
Then, we develop a framework that resolves the sample-query complexity of uniformity testing.

\subsubsection{Warm Up: Resolving Sample-Query Dependence}

\begin{theorem}
    \label{thm:dist-contamination-lb}
    Any $\frac{1}{2}$-distributionally robust $0.1$-uniformity tester with $m \leq n^{0.99}$ samples requires $\bigOm{\frac{n}{m}}$ queries.
\end{theorem}

We note that the above lower bound can be adjusted to hold for any number of samples $m \leq n^{1 - c}$ for some $c > 0$ at the cost of constant factors in the query complexity.

Our result relies on the following bounds on the number of collisions.

\begin{lemma}
    \label{lemma:collision-concentration-l2}
    Let $m \geq c \sqrt{n}$ for some $c > 0$.
    Let $X$ be the number of collisions given $m$ i.i.d. samples drawn from $\p$ on $[n]$.
    Then for large enough $n$,
    \begin{equation*}
        \Pr \left( X \in (1 \pm 0.0001)\binom{m}{2} \norm{\p}_{2}^{2} \right) > 0.9999 \text{.}
    \end{equation*}
\end{lemma}

For certain applications, we will require the following concentration bound for higher moment collisions.

\begin{lemma}
    \label{lemma:collision-concentration-lc}
    Fix $c \geq 2$ an integer.
    Let $m \geq 2 n^{1 - 1/c}$.
    Let $X$ be the number of $c$-collisions given $m$ i.i.d. samples drawn from $\p$ on $[n]$.
    Then for large enough $n$,
    \begin{equation*}
        \Pr \left( X \in (1 \pm 0.0001)\binom{m}{c} \norm{\p}_{c}^{c} \right) > 0.9999 \text{.}
    \end{equation*}
\end{lemma}

We defer the proofs of these facts to \Cref{sec:omitted-proofs}.

\begin{proof}
    By standard uniformity testing lower bounds, we may assume $m = \Omega(n^{1/2})$, even if we query all samples.
    We will also assume without loss of generality that the algorithm takes $\Poi(m)$ samples, since $\Poi(m)$ is $O(m)$ with high constant probability.
    Under this assumption, the frequency of each bucket is distributed according to independent Poisson processes.

    Let $c^*$ be the smallest integer such that $m \geq 2 n^{1 - 1/c^*}$.
    Then, note that $2 n^{1 - 1/c^*} \leq m < 2n^{1 - 1/(c^* + 1)}$.
    Note that by our assumption on $m$, we have $c^* \leq 100$.

    We construct the hard instance as follows.
    Let $m \ll \sdom \ll n$ where $\sdom$ is a sufficiently large constant multiple of $m$. 
    \begin{definition}
        \label{def:uniformity-hard-instance-simple}
        Let $\calM_{0}$ be the distribution over tuples of non-negative measures $(\p, \q)$ generated as follows: for all $i \in [n]$,
        \begin{equation*}
            (\p[i], \q[i]) = \begin{cases}
                \left( \frac{1}{n}, \frac{1}{2n} + \frac{1}{2\sdom} \right) & w.p. \quad \frac{\sdom}{n} \\
                \left( \frac{1}{n}, \frac{1}{2n} \right) & w.p. \quad \frac{n - \sdom}{2n} \\
                \left( \frac{1}{n}, \frac{1}{2n} \right) & w.p. \quad \frac{n - \sdom}{2n} 
            \end{cases}
        \end{equation*}

        Similarly, let $\calM_{1}$ be the distribution over tuples of non-negative measures $(\p, \q)$ generated as follows: for all $i \in [n]$,
        \begin{equation*}
            (\p[i], \q[i]) = \begin{cases}
                \left( \frac{1}{n}, \frac{1}{2n} + \frac{1}{2 \sdom} \right) & w.p. \quad \frac{\sdom}{n} \\
                \left( \frac{4}{3n}, \frac{1}{6n} \right) & w.p. \quad \frac{n - \sdom}{2n} \\
                \left( \frac{2}{3n}, \frac{5}{6n} \right) & w.p. \quad \frac{n - \sdom}{2n} 
            \end{cases}
        \end{equation*}

        The meta-distribution $\calH_{U}$ is the distribution over random tuples of non-negative measures $(\p, \q)$ generated by choosing $X \sim \Bern(1/2)$ and $(\p, \q) \sim \calM_{X}$.
    \end{definition}
    
    % \chris{continue here.}
    % % \chris{check constants.}
    % Note that such a $d$ exists by our assumption on $m$.
    % Randomly select subsets $S \subset [n]$ by including each bucket with probability $\frac{1}{2}$ and $T \subset [n]$ by including each bucket with probability $\frac{d}{n}$.
    % The adversary then flips a random bit $B$, and if $B = 0$ constructs a no-instance, while if $B = 1$ constructs a yes-instance.
    % In the no-instance, let $\p[i] = \frac{1 + \eps}{n}$ for $i \in S$ and $\frac{1 - \eps}{n}$ for $i \not\in S$.
    % Let $\q[i] = \q_{S}[i] + \q_{T}[i]$ where 
    % \begin{equation*}
    %     \q_{S}[i] = \begin{cases}
    %         \eps/n & i \in S \\
    %         3 \eps/n & i \not\in S
    %     \end{cases}~~~~,~~~~\q_{T}[i] = \begin{cases}
    %         (1 - 2 \eps)/d & i \in T \\
    %         0 & i \not\in T
    %     \end{cases}
    % \end{equation*}
    % In the yes-instance, let $\p$ be uniform on $[n]$ and let $\q[i] = \frac{2 \eps}{n} + q_{T}[i]$ with $\q_{T}$ as above.
    Let $\rD := (\p + \q)/2$ denote the mixture.
    Note that we have (regardless of $X \in \set{0, 1}$),
    \begin{equation}
        \label{eq:simp-unif-mix-dist}
        \rD[i] = \begin{cases}
            \frac{3}{4n} + \frac{1}{4 \sdom} & w.p. \quad \frac{\sdom}{n} \\
            \frac{3}{4n} & w.p. \quad \frac{n - \sdom}{n} 
        \end{cases} \text{.}
    \end{equation}
    Let $H, L_+, L_-$ be a partition of $[n]$ that denotes the set of indices of each type.
    Let $W_i$ be the random variable denoting the type of index $i$.

    \paragraph{Poissonization and Measures}
    Note that $\p, \q, \rD$ may not be distributions on $[n]$ as their total mass may be more or less than $1$.
    We will deal with this in a standard way.
    In particular, defining the distribution $\rD/\norm{\rD}_{1}$ and taking $\Poi(m \norm{\rD}_{1})$ samples from $\rD/\norm{\rD}_{1}$ is equivalent to observing $Y_i \sim \Poi(m \rD[i])$ samples from bucket $i$ independently for all $i$.
    By a Chernoff bound, we can conclude that with high constant probability we have $|H| = \sdom \pm O(\sqrt{\sdom})$ and $|L_+|, |L_-| = \frac{n - \sdom}{2} \pm O(\sqrt{n})$ so that there is a small $\delta < n^{-\Omega(1)}$ such that $1 - \delta \leq \norm{\p}_{1}, \norm{\q}_{1} \leq 1 + \delta$ and so $1 - 2 \delta \leq \norm{\rD}_{1} \leq 1 + 2 \delta$.
    Furthermore, note that since $\norm{\p}_{1}, \norm{\q}_{1}, \norm{\rD}_{1} \in (1 - 2 \delta, 1 + 2 \delta)$, the mass of each element does not change significantly (and in particular neither will higher moment norms of these distributions).
    In particular, for any distribution $\p$, $\norm{\p}_{c}^{c}$ changes by a multiplicative factor of $\norm{\p}_{1}^{c} \leq (1 + \delta)^{c}$ after normalization.
    In our instance, we are concerned with $\delta < n^{-\Omega(1)}$ and $c = O(1)$ so that $(1 + \delta)^{c}$ for constant $c$ is at most an arbitrarily small constant greater than $1$.
    For the remainder of this proof, we will mostly ignore the normalization, commenting on its effect where necessary.

    In either the yes-instance or the no-instance we have $\rD[i] = \frac{3}{4n}$ for $i \not\in H$ and $\rD[i] = \frac{3}{4n} + \frac{1}{4\sdom}$ for $i \in H$.
    In either case the algorithm draws samples from the same distribution $\rD$.
    Thus, without any queries we cannot hope to distinguish $X$ with probability greater than $\frac{1}{2}$.

    To prove a lower bound on the number of queries necessary, we first show that querying singletons yields no information.
    \begin{lemma}
        \label{lemma:query-singleton}
        Consider a bucket $i$ with frequency $Y_i = 1$.
        Let $Z$ denote the whether the query on bucket $i$ reveals $\p$.
        Then $\Pr(Z = \p | X = 0, Y_i = 1) = \Pr(Z = \p | X = 1, Y_i = 1)$.
    \end{lemma}
    
    \begin{proof}
        Note that $Y_i \sim \Poi(m\rD[i])$ is distributed identically to $Y_i = A_i + B_i$ where $A_i \sim \Poi(m \p[i]/2), B_i \sim \Poi(m \q[i] / 2)$ are sampled independently.
        Then,
        \begin{align*}
            \Pr(A_i = 0, B_i = 1 | X = 0) &= \frac{\sdom}{n} \exp\left( -\frac{m}{2n} - \frac{m}{4n} - \frac{m}{4\sdom} \right) \left( \frac{m}{4n} + \frac{m}{4\sdom} \right) \\
            &\quad+ \frac{n - \sdom}{n} \exp\left( -\frac{m}{2n} - \frac{m}{4n}\right) \left( \frac{m}{4n} \right) \\
            \Pr(A_i = 0, B_i = 1 | X = 1) &= \frac{\sdom}{n} \exp\left( -\frac{m}{2n} - \frac{m}{4n} - \frac{m}{4\sdom} \right) \left( \frac{m}{4n} + \frac{m}{4 \sdom} \right) \\
            &\quad+ \frac{n - \sdom}{2n} \exp\left( -\frac{4m}{6n} - \frac{m}{12n}\right) \left( \frac{m}{12n} \right) + \frac{n - \sdom}{2n} \exp\left( -\frac{2m}{6n} - \frac{5m}{12n}\right) \left( \frac{5m}{12n} \right)
        \end{align*}
        so that $\Pr(A_i = 0, B_i = 1 | X = 0) = \Pr(A_i = 0, B_i = 1 | X = 1)$.
        Note that $\Pr(Y_i = 1) = \Pr(Y_i = 1 | X = 0) = \Pr(Y_i = 1 | X = 1)$.
        By the total probability rule, we observe that $\Pr(A_i = 1, B_i = 0 | X = 0) = \Pr(A_i = 1, B_i = 0 | X = 1)$.
        Thus, to conclude, we have
        \begin{align*}
            \Pr(Z = \p |X = 0, Y_i = 1) &= \frac{\Pr(A_i = 1, B_i = 0 | X = 0)}{\Pr(Y_i = 1)} \\
            &= \frac{\Pr(A_i = 1, B_i = 0 | X = 1)}{\Pr(Y_i = 1)} \\
            &= \Pr(Z=\p|X=1, Y_i = 1) \text{.}
        \end{align*}
    \end{proof}

    \Cref{lemma:query-singleton} states that regardless of $B$, a singleton query is equally likely to yield $\p$, and so yields no information. 
    Any query that helps the algorithm distinguish $B$ must therefore be on a bucket with multiple samples.
    
    Next, we argue that by querying collisions, most queries will look the same to the algorithm, and therefore without many queries, the algorithm still cannot hope to distinguish $X$.
    To do so, we will argue that $\rD$ generates significantly more collisions than $\p$.
    In fact, for any frequency count $c \geq 2$, $\rD$ will generate significantly more elements with frequency $c$ than $\p$ will generate collisions at all.
    Thus, since the algorithm cannot distinguish any buckets with the same frequency prior to querying, the algorithm must make many queries to detect any sample of $\p$ among the collisions.

    First, we argue that $\p$ does not produce too many collisions.
    In both the yes and no-instance we note that $\norm{\p}^{2} \leq \frac{(1 + \eps)^2}{n}$ and \Cref{lemma:collision-concentration-l2} ensures that with probability at least $0.9999$, the number of $\p$ collisions is at most $1.0001 \frac{m^2}{2} \frac{(1 + \eps)^2}{n} \leq \frac{1.0001 m^2}{n}$ by assuming $\eps < 0.1$.
    Note that to account for normalization by $\norm{\p}_{1}$, we can instead bound the collisions by say $\frac{1.001 m^2}{n}$.
    % \begin{align*}
    %     \norm{p}_{2}^{2} \leq \frac{1 + \eps^2}{n}
    % \end{align*}
    % Thus, by \Cref{lemma:collision-concentration-l2}, $p$ has at most $\frac{3.001 d}{n^2}$ collisions on $T$ with probability at least $0.9999$.

    % On the other hand, in both the yes and no-instance, we argue that $p$ has few collisions overall.
    % In particular
    % \begin{align*}
    %     \norm{p}_{2}^{2} \leq \frac{1 + \eps^{2}}{n}
    % \end{align*}
    
    Next, we argue that for every $2 \leq c \leq c^*$, $\rD$ has significantly more $c$-collisions than $\p$ has collisions.
    Fix such a $2 \leq c \leq c^* \leq 100$.
    We begin with a lower bound on the number of $c$-collisions.
    
    \begin{lemma}
        \label{lemma:small-domain-unif-lc-ub}
        For any integer $2 \leq c \leq 100$, $\norm{\rD}_{c}^{c} = \Omega(\sdom^{1 - c})$.
    \end{lemma}
    
    \begin{proof}
        Since $\eps < 0.1$ and $|H| = \Omega(\sdom)$,
        \begin{align*}
            \sum_{i} \rD[i]^{c} &\geq |T| \frac{1}{(4\sdom)^{c}} = \Omega(\sdom^{1-c}) \text{.}
        \end{align*}
    \end{proof}

    Now, since $m \geq 2n^{1 - 1/c^*} \geq 2n^{1-1/c}$, \Cref{lemma:collision-concentration-lc} states that the number of $c$-collisions from $\rD$ is between
    \begin{equation*}
        (1 \pm 0.0001) \binom{m}{c} \norm{r}_{c}^{c} = \bigTh{\frac{m^{c} \sdom^{1 - c}}{c!}}
    \end{equation*}
    with high constant probability.
    Applying the union bound over all $2 \leq c \leq c^*$, the above bound holds for all $c$.
    Furthermore, since $\sdom \gg m$ for a sufficiently large constant, we have that the number of $(c + 1)$-collisions is at most a small fraction of the number of $c$-collisions, and therefore the number of buckets with frequency exactly $c$ is also $\bigOm{m^{c} \sdom^{1 - c}} = \bigOm{m}$ since $\sdom = \Theta(m)$.

    Finally, we argue that $\p$ has few samples in the buckets where $\rD$ has collisions.
    For a sample of $\p$ to appear in a bucket where $\rD$ has a collision, $\p$ must collide either with itself or with $\q$ in this bucket.
    We have already bounded the number of collisions $\p$ causes with itself above.
    It remains, therefore, to upper bound the number of collisions $\p$ has with $\q$.
    The probability of a collision between $\p$ and $\q$ in samples $(i, j)$ can be written as
    \begin{align*}
        \sum_{k \in [n]} 2 \Pr(Z_{i} = \p, Y_{i} = k) \Pr(Z_{j} = \q, Y_{j} = k) &\leq \sum_{k \in H} 2 \frac{1}{n} \left( \frac{1}{2n} + \frac{1}{2\sdom} \right) + \sum_{k \not\in H} 2 \frac{10}{18n^2} \\
        &= O(1/n) \text{.}
    \end{align*}
    Above, we have bounded the first term by observing $|T| = O(\sdom)$.
    In particular, using Markov's inequality, we have that with high constant probability, there are at most $O(m^2/n)$ collisions between $\p$ and $\q$.
    Combining our arguments so far, we have shown that there are at most $O(m^2/n)$ samples from $\p$ in buckets with frequency at most $c^*$.
    However, for all $2 \leq c \leq c^*$, $\rD$ produces at least $\Omega(m)$ buckets with frequency $c$.
    In particular, any query of a bucket with frequency $2 \leq c \leq c^*$ finds a sample of $\p$ with probability at most $\bigO{\frac{m^2/n}{m}} = O(m/n)$.
    The following lemma summarizes the above discussion.

    \begin{lemma}
        \label{lemma:frequency-c-p-prob}
        Let $2 \leq c \leq c^*$.
        Any query on a bucket with frequency $c$ yields $\p$ with probability at most $O(m/n)$.
    \end{lemma}

    To conclude the argument, we show that any bucket with frequency $(c^* + 1)$ or greater is unlikely to have samples from $\p$.
    Towards this, we argue that no bucket outside of $H$ produces any $(c^* + 1)$-collisions. 
    
    \begin{lemma}
        \label{lemma:near-uniform-lc-ub}
        For any integer $c \geq 2$, $\norm{\rD[[n] \setminus H]}_{c}^{c} < \left( \frac{3}{4} \right)^{c} \frac{1}{n^{c - 1}}$.
    \end{lemma}

    \begin{proof}
        The proof is a simple computation.
        \begin{align*}
            \sum_{i \not\in H} \rD[i]^{c} = (n - |H|) \left( \frac{3}{4n} \right)^{c} \leq \left( \frac{3}{4} \right)^{c} \frac{1}{n^{c - 1}} \text{.}
        \end{align*}
    \end{proof}
    
    In particular, by \Cref{cor:c-collision-ub}, we have that a $(c^* + 1)$-collision occurs on $[n] \setminus H$ with probability at most 
    \begin{align*}
        \frac{(3/4)^{c^* + 1}}{n^{c^*}} \binom{m}{c^* + 1} &\leq \frac{(3/4)^{c^* + 1}}{n^{c^*}} \frac{m^{c^* + 1}}{(c^* + 1)! (m - c^* - 1)!} \\
        &< \frac{(3/2)^{c^* + 1}}{(c^* + 1)! (m - c^* - 1)!} \\
        &< 0.0001 
    \end{align*}
    for sufficiently large $n$ since $m = \bigOm{\sqrt{n}}$.
    In the second inequality, we have used that $m < 2n^{1 - 1/(c^* + 1)}$ and therefore $m^{c^* + 1} < 2^{c^* + 1} n^{(c^* + 1) - (c^* + 1)/(c^* + 1)} = 2^{c^* + 1} n^{c^*}$.
    Note that here we have explicitly used $m \leq n^{0.99}$ so that $c^* \leq 100$.
    Assuming $m < n^{1 - \delta}$ for any $\delta > 0$ will require the following arguments to proceed assuming that $\p$ does not produce any $\frac{1}{\delta}$-collisions.
    Thus, any bucket with frequency at least $(c^* + 1)$ must be in $H$.
    % \chris{above bound still holds for any $c^* \ll \sqrt{n}$? can we generalize?}

    Consider then a bucket $i \in H$ with empirical frequency $c > c^*$.
    Each query is distributed as a Bernoulli random variable that equals $\p$ with parameter at most $\frac{\p[i]}{2 \rD[i]} = O(\sdom/n) = O(m/n)$.
    Thus, in buckets with frequency at least $(c^* + 1)$, we can again conclude that any query reveals $p$ with probability at most $O(m/n)$.
    This yields the following lemma.

    \begin{lemma}
        \label{lemma:v-large-frequency-p-prob}
        Let $c > c^*$.
        Any query on a bucket with frequency $c$ yields $\p$ with probability at most $O(m/n)$.
    \end{lemma}
    
    Combining our two bounds, any query on a bucket with frequency at least $c \geq 2$ is $\p$ with probability at most $O(m/n)$.
    In particular, without at least $\Omega(n/m)$ queries, all queries by the algorithm reveal $\q$ and therefore the algorithm cannot distinguish $X$ with any non-trivial advantage.
\end{proof}

\subsubsection{Lower Bound for Uniformity Testing}

We prove our main lower bound for uniformity testing.
In addition to obtaining tight dependence on $\eps, \lambda$ and relaxing the assumption on $m$ (in fact, we only require the assumption $m \ll n$) the techniques used here will prove crucial for obtaining an optimal sample-query trade-off for closeness testing.

Consider an algorithm that successfully distinguishes uniform $\p$ from $\p$ that are $\eps$-far from uniform, i.e. the algorithm guesses the hidden random bit $X$ in the hard instance with non-trivial advantage.
It is a standard fact (see e.g. \cite{diakonikolaskane2016}) that any such algorithm must use information that is correlated with $X$.
In particular, if we denote the input of the algorithm (a collection of samples and queries) as $T$, then $I(X:T) = \Omega(1)$ as formalized below.

\MutualInfoBound*

\begin{proof}
    We give a proof of this standard fact for completeness.
    The conditional entropy $H(X|f(T))$ is the expectation over $f(T)$ of $h(q)$ where $h(q)$ is the binary entropy function and $q$ is the probability that $X = f(T)$ given $f(T)$.
    Since $\EX[q] \geq 0.51$ and $h$ is concave, $H(X|f(A)) \leq h(0.51) < 1 - 2 \cdot 10^{-4}$.
    Then, by the data processing inequality
    \begin{equation*}
        I(X:A) \geq I(X:f(A)) \geq H(X) - H(X|f(A)) \geq 2 \cdot 10^{-4}.
    \end{equation*}
\end{proof}

Our lower bound will then follow by bounding the information obtained by an algorithm that makes few queries.

\UniformityTestingLB*

\begin{proof}
    By standard uniformity testing lower bounds (see e.g. \cite{DBLP:journals/tit/Paninski08}), we may assume $m = \Omega(n^{1/2})$, even if we query all samples.
    Consider the following hard instance.
    Define $\alpha := \frac{\lambda}{1 - \lambda}$.
    Since $\lambda \leq \frac{1}{2}$, we have $\alpha \leq 1$.

    \begin{definition}
        \label{def:uniformity-hard-instance}
        Let $\calM_{0}$ be the distribution over tuples of non-negative measures $(\p, \q)$ generated as follows: for all $i \in [n]$,
        \begin{equation*}
            (\p[i], \q[i]) = \begin{cases}
                \left( \frac{1}{n}, \alpha \frac{4 \eps}{n} + \frac{1 - 4 \eps \alpha}{\dom} \right) & w.p. \quad \frac{\dom}{n} \\
                \left( \frac{1}{n}, \alpha \frac{4 \eps}{n} \right) & w.p. \quad \frac{n - \dom}{2n} \\
                \left( \frac{1}{n}, \alpha \frac{4 \eps}{n} \right) & w.p. \quad \frac{n - \dom}{2n} 
            \end{cases}
        \end{equation*}

        Similarly, let $\calM_{1}$ be the distribution over tuples of non-negative measures $(\p, \q)$ generated as follows: for all $i \in [n]$,
        \begin{equation*}
            (\p[i], \q[i]) = \begin{cases}
                \left( \frac{1}{n}, \alpha \frac{4 \eps}{n} + \frac{1 - 4 \eps \alpha}{\dom} \right) & w.p. \quad \frac{\dom}{n} \\
                \left( \frac{1 + 2\eps}{n}, \alpha \frac{2 \eps}{n} \right) & w.p. \quad \frac{n - \dom}{2n} \\
                \left( \frac{1 - 2\eps}{n}, \alpha \frac{6\eps}{n} \right) & w.p. \quad \frac{n - \dom}{2n} 
            \end{cases}
        \end{equation*}

        The meta-distribution $\calH_{U}$ is the distribution over random tuples of non-negative measures $(\p, \q)$ generated by choosing $X \sim \Bern(1/2)$ and $(\p, \q) \sim \calM_{X}$.
    \end{definition}

    Let $h, l_+, l_-$ denote each type, that is let $W_i$ denote the type of bucket $i$ be independently chosen for all $i$ where $\Pr(W_i = h) = \frac{\dom}{n}$ and $\Pr(W_i = l_+) = \Pr(W_i = l_-) = \frac{n - \dom}{2n}$.
    Let $H, L_+, L_-$ denote the number of buckets of each type.
    We begin by arguing the non-negative measures are close to real distributions.

    \begin{lemma}
        \label{lemma:uniform-dist-norm}
        For sufficiently large $n \geq 1$,
        \begin{equation*}
            \Pr_{(\p, \q) \sim \calH_{U}} \left( \norm{\p}_{1}, \norm{\q}_{1} \in (0.9, 1.1) \right) \geq 1 - \frac{1}{\poly(n)} \text{.}
        \end{equation*}
        Furthermore,
        \begin{equation*}
            \Pr_{(\p, \q) \sim \calM_{1}} \left( \tvd(\p, U_n) > \eps \right) \geq 1 - \frac{1}{\poly(n)} \text{.}
        \end{equation*}
    \end{lemma}

    \begin{proof}
        Let $X = 0$. 
        Note $\norm{\p}_1 = 1$ since $\p = U_n$ always.
        Let $H$ denote the number of items of the first type so that $\norm{\q}_1 = 4 \eps \alpha + \frac{|H|}{\dom}(1 - 4 \eps \alpha)$.
        By standard concentration bounds, observe that $|H| \in \dom \pm 0.001\dom$ and $|L_-|, |L_+| \in \frac{n - \dom}{2} \pm 0.001n$ with high probability with the assumption $m = \Omega(n^{1/2})$.
        Note that these conditions are sufficient to guarantee the desired bound on the $\ell_1$ norm of $\p, \q$.
        
        Now, let $X = 1$. 
        By standard concentration bounds, we have $H = \dom \pm 0.001 \dom$ and $L_+, L_- = \frac{n - \dom}{2} \pm 0.001n$ with probability at least $1 - 1/\poly(n)$.
        For sufficiently large $n$, these conditions guarantee that the $\ell_1$-norm of the measures is in $(0.9, 1.1)$ (again using the assumption without loss of generality that $m = \Omega(n^{1/2})$).
        Similarly, they are sufficient to guarantee that $\p \sim \calM_1$ is sufficiently far from uniform.
    \end{proof}

    Thus, when sampling distributions $(\p, \q)$ from $\calH_{U}$, the distributions $\p, \q$ (and therefore the mixture $\rD = \frac{\p + \q}{2}$) satisfy the desired norm conditions with high probability.
    By \Cref{lem:poisson-sampling}, it suffices to consider a deterministic Poissonized tester $\innerAlg$ with sample complexity $m$.
    We will show that if $\innerAlg$ is a $0.1$-correct $\eps$-uniformity tester, $\innerAlg$ requires $\bigOm{n m^{-1} \eps^{-4} \lambda^{-2}}$ queries.

    In fact, we show a stronger lower bound on the number of distinct buckets that the algorithm $\innerAlg$ must query.
    We assume that when the algorithm queries bucket $i$, it receives the source of all samples on bucket $i$.
    Note that this only strengthens the algorithm $\innerAlg$.

    \begin{lemma}
        \label{lemma:uniformity-bucket-query-lb}
        Suppose $\innerAlg$ is a deterministic Poissonized tester with sample complexity $m$ is $0.1$-correct with respect to $\calH_0, \calH_1$.
        Then, $\innerAlg$ queries $\bigOm{\min(n m^{-1} \eps^{-4} \lambda^{-2}, m)}$ distinct buckets.
        This bound holds even when a single query on bucket $i$ reveals the source of every sample in bucket $i$.
    \end{lemma}

    \begin{proof}
        Suppose $nm^{-1}\eps^{-4}\lambda^{-2} = o(m)$.
        Otherwise, $\innerAlg$ can simply query all samples by querying at most $m$ buckets.
        Since $\innerAlg$ is a Poissonized tester, $\innerAlg$ receives independently $Y^{(1)}_i \sim \Poi(m \rD[i])$ samples on bucket $i$, where $\rD[i] = \frac{\p[i] + \q[i]}{2}$.
        Note that we can equivalently sample $Y_i = A_i + B_{i}$ where $A_{i} \sim \Poi(\lambda m \p[i]), B_{i} \sim \Poi((1 - \lambda) m \q[i])$ are sampled independently.
        Of course, the algorithm (without queries) only observes $Y_{i}$.
        If the algorithm queries a bucket $i$, then the algorithm learns $A_{i}, B_{i}$ as well.
        Suppose the algorithm queries a set of buckets $Q \subset [n]$.
        Let $T_i$ denote the information observed by the algorithm from the $i$-th bucket, i.e. $Y_{i}$ if $i \not\in Q$ and $A_{i}, B_{i}$ if $i \in Q$.
        Let $T = (T_1, \dotsc, T_n)$.
        
        Then, the algorithm $\innerAlg$ is a function such that $\innerAlg(T) = X$ with at least $51\%$ probability, and therefore $I(X:T) \geq 2 \cdot 10^{-4}$ by \Cref{lemma:mutual-info-bound}.
        Suppose for contradiction that $\innerAlg$ queries $q = |Q| = o(n m^{-1} \eps^{-4})$ buckets.
        Since $T_i$ are independent conditioned on $X$, we have
        \begin{equation*}
            I(X:T) \leq \sum_{i = 1}^{n} I(X:T_{i}) \text{.}
        \end{equation*}

        For $i \not\in Q$, we claim $I(X:T_{i}) = I(X:Y_{i}) = 0$.
        In particular, for all $i \in [n]$, we have $Y_{i} \sim \Poi(m \rD[i])$ where
        \begin{equation}
            \label{eq:uniformity-mix-dist}
            \rD_{j}[i] = \begin{cases}
                \lambda \frac{1 + 4 \eps}{n} + (1 - \lambda) \frac{1 - 4 \eps \alpha}{\dom} & w.p.\quad \frac{\dom}{n} \\
                \lambda \frac{1 + 4 \eps}{n} & w.p. \quad \frac{n - \dom}{n}
            \end{cases}
        \end{equation}
        regardless of $X$, proving the claim.
        Thus, the samples reveal no information it suffices to bound the information gained from the queries via $\sum_{i \in Q}^{n} I(X:T_{i}) = \sum_{i \in Q}^{n} I(X:A_{i}, B_{i})$.

        Before proceeding with the mutual information bound, we make the simplifying assumption that no buckets produce large collisions. 
        \begin{claim}
            \label{clm:uniformity-p-freq-ub}
            \begin{equation*}
                \Pr\left( \max_{i} Y_i = O(\log n)\right) > 1 - \frac{1}{\poly(n)}
            \end{equation*}
        \end{claim}

        \begin{proof}
            % Our algorithm $\innerAlg$ takes $m' \sim \Poi(m)$ samples.
            % We claim an upper bound on $m'$.
            % From \Cref{lemma:poisson-concentration} we have
            % \begin{equation*}
            %     \Pr(m' > 2m) < \exp\left( - \frac{m^2}{4m} \right) < 0.001 \text{.}
            % \end{equation*}
            % We condition on this event.
            
            % From \Cref{cor:c-collision-ub} and our assumption that $m \leq n^{0.99}$, we obtain that the number of $c$-collisions on buckets of type $l_+, l_-$, is greater than $t$ with probability
            % \begin{equation*}
            %     \frac{1}{t} \binom{m'}{c} \sum_{W_i \in \set{l_1, l_2}} \rD[i]^{c} \leq \frac{1}{t} \binom{m'}{c} \frac{1}{n^{c - 1}} \leq \frac{(2m)^{c}}{t n^{c - 1} \cdot c!} \leq \frac{2^{c} n^{0.99c - c + 1}}{t \cdot c!}
            % \end{equation*}
            % where we simplify with $\frac{1 + 4 \eps}{2n} \leq \frac{1}{n}$ and $m \leq n^{0.99}$.
            % Fixing $t < 1$, note that the above quantity is less than $0.001$ for $c > 100$.
            % We conclude with a union bound over the conditioned event.

            Note that each $Y_i$ is a Poisson variable with parameter at most $m \rD[i] \leq 2$ from \eqref{eq:uniformity-mix-dist}.
            The result follows from \Cref{lemma:poisson-concentration} and a union bound. 
        \end{proof}

        Let $E$ denote the event that the conditions of \Cref{lemma:uniform-dist-norm} and \Cref{clm:uniformity-p-freq-ub} are met.

        \begin{lemma}
            \label{lemma:uniformity-one-bucket-mi}
            For all $i$, we have
            \begin{align*}
                I(X:A_{i}, B_{i}) = \bigO{\frac{\eps^{4} \lambda^{2} m}{n}} \text{.} 
            \end{align*}
        \end{lemma}

        \begin{proof}
            Note $A_i, B_i$ are pairs of Poisson random variables that sum up to $Y_{i} = y$.
            By \Cref{claim:MI_asymp}, we have
            \begin{equation*}
                I(X:A_{i}, B_{i}) = \bigO{\sum_{a, b} \frac{\left( f_{a, b}(0) - f_{a, b}(1) \right)^{2}}{f_{a, b}(0) + f_{a, b}(1)}} 
                = \bigO{\sum_{a + b = y} \frac{\left( f_{a, b}(0) - f_{a, b}(1) \right)^{2}}{f_{a, b}(0) + f_{a, b}(1)}} \text{.} 
            \end{equation*}
            where $f_{a, b}(x) = \Pr(A_{i} = a, B_{i} = b | Y_{i} = y, X = x, E)$ and we note that unless $a + b = y$, $f_{a, b}(x) = 0$.
            We now compute $f_{a, b}(x)$ explicitly.
            \begin{align*}
                f_{a, b}(x) &= \frac{\Pr(A_{i} = a, B_{i} = b ,Y_{i} = y | X = x, E)}{\Pr(Y_{i} = y | X = x, E)} \\
                &= \begin{cases}
                    \frac{\Pr(A_{i} = a, B_{i} = b| X = x, E)}{\Pr(Y_{i} = y | X = x, E)} & \ifT a + b = y \\
                    0 & \otherwise
                \end{cases}
            \end{align*}
            Define $g_{a, b}(x) = \Pr(A_{i} = a, B_{i} = b| X = x, E)$ and $h_{y}(x) = \Pr(Y_{i} = y | X = x, E)$.
            Since samples are independent of $X$, define $h_{y} := h_{y}(0) = h_{y}(1)$.
            Then, we can rewrite 
            \begin{equation*}
                I(X:A_{i}, B_{i})
                = \bigO{\frac{1}{h_{y}} \sum_{a + b = y} \frac{\left( g_{a, b}(0) - g_{a, b}(1) \right)^{2}}{g_{a, b}(0) + g_{a, b}(1)}} \text{.} 
            \end{equation*}

            We require the following technical claims.
            
            \begin{restatable}{claim}{UniformityHLB}
                \label{clm:uniformity-h-lb}
                \begin{equation*}
                    h_{y} = \bigOm{\frac{\dom}{n} \frac{(1 - \lambda)^{y} (1 - 4 \eps \alpha)^{y}}{y!}} \text{.}
                \end{equation*}
            \end{restatable}

            \begin{restatable}{claim}{UniformityGDiffUB}
                \label{clm:uniformity-g-diff-ub}
                If $y = a + b \leq 1$,
                % or $Y_i > 100$
                then $g_{a, b}(0) = g_{a, b}(1)$.
                Otherwise $y = a + b \geq 2$
                % , $Y_i \leq 100$ 
                and 
                \begin{align*}
                    (g_{a, b}(0) - g_{a, b}(1))^{2} &= \bigO{\left( \frac{\lambda^{a + b}}{a! b!} \frac{m^{a + b}}{n^{a + b}} \right)^{2}} \cdot \begin{cases}
                       a^{4} \eps^{4} & b = 0, a \geq 2 \\
                       a^{2} \eps^{4} & b = 1, a \geq 1 \\
                       % \frac{m^{6} \eps^{4} (1 - 2 \eps)^{2b}}{n^{6}} + \frac{(3 \eps m)^{2y}}{n^{2y}} & b \geq 2, a \leq 1 \\
                       % \frac{m^{a + 1} \eps^{4} (1 - 2 \eps)^{2b}}{n^{a + 1}} + \frac{(3 \eps m)^{2y}}{n^{2y}} & b \geq 2, a \geq 2
                       (6 \eps)^{2b}(1 + a^2 \eps^2)^2 & b \geq 2, a \geq 0
                    \end{cases} \text{.}
                \end{align*}
                
            \end{restatable}

            \begin{restatable}{claim}{UniformityGLB}
                \label{clm:uniformity-g-lb}
                For $x \in \set{0, 1}$, 
                \begin{equation*}
                    g_{a, b}(x) = \bigOm{\frac{1}{a!b!}} \begin{cases}
                        \frac{m^{y} \lambda^{y}}{n^{y}} & b = 0 \\
                        \frac{\lambda^{a} (1 - \lambda)^{b} m^{a + 1} (1 - 4 \eps \alpha)^{b}}{n^{a + 1}} & b \geq 1 
                    \end{cases}\text{.}
                \end{equation*}
            \end{restatable}

            We will prove these claims at the end of this section.
            Recall that our goal is to bound,
            \begin{align*}
                I(X:A_i, B_i) &= \bigO{\frac{1}{h_{y}} \sum_{b = 0, a = y - b}^{y} \frac{(g_{a, b}(0) - g_{a, b}(1))^{2}}{g_{a, b}(0) + g_{a, b}(1)}} \text{.}
            \end{align*}
            % Recall that we condition on the event $y \leq 100$.
            
            To bound the summation, we will consider some cases separately.
            Let $\barI_{a, b} := \frac{1}{h_{y}} \frac{(g_{a, b}(0) - g_{a, b}(1))^{2}}{g_{a, b}(0) + g_{a, b}(1)}$.
            We bound a few special values of $\barI_{a, b}$.
            Applying \Cref{clm:uniformity-h-lb}, \Cref{clm:uniformity-g-diff-ub}, and \Cref{clm:uniformity-g-lb}, we obtain
            \begin{align*}
                 \barI_{y, 0} &= \bigO{\frac{\lambda^{2y}m^{2y} y^4 \eps^4}{(y!0!)^2n^{2y}} \cdot \frac{y!0!n^{y}}{m^{y} \lambda^{y}} \cdot \frac{y! n}{m(1 - \lambda)^{y}(1 - 4 \eps \alpha)^{y}}} = \bigO{\frac{\lambda^{y} y^{4} m^{y - 1} \eps^{4}}{n^{y - 1}(1 - 4 \eps \alpha)^{y} (1 - \lambda)^{y}}} = \bigO{\frac{m \eps^{4} \lambda^{2}}{n}}
            \end{align*}
            where we obtain the first inequality by canceling out like terms, then  we obtain the final bound by observing that the quantity is decreasing in $y$ since $m \ll n$ and $y \geq 2$.
            We require $\frac{\lambda m}{n(1 - \lambda)(1 - 4 \eps \alpha)} < 1$ for which it suffices to choose $m \leq n/2$.
            Similarly,
            \begin{align*}
                 \barI_{y - 1, 1} &= \bigO{\frac{\lambda^{2y}m^{2y} a^2 \eps^4}{((y - 1)!1!)^2n^{2y}} \cdot \frac{(y - 1)!1!n^{y}}{m^{y} \lambda^{y - 1} (1 - \lambda)(1 - 4 \eps \alpha)} \cdot \frac{y! n}{m(1 - \lambda)^{y}(1 - 4 \eps \alpha)^{y}}} \\
                 &= \bigO{\frac{y! a^{2} m^{y - 1} \eps^{4} \lambda^{y + 1}}{(y - 1)! 1! n^{y - 1} (1 - \lambda)^{y + 1} (1 - 4 \eps \alpha)^{y + 1}}} \\
                 &= \bigO{\frac{y^{3} m^{y - 1} \eps^{4} \lambda^{y + 1}}{n^{y - 1} (1 - \lambda)^{y + 1} (1 - 4 \eps \alpha)^{y + 1}}} \\
                 &= \bigO{\frac{m \eps^{4} \lambda^{3}}{n}} \text{.}
            \end{align*}
            Here we observed that $\frac{y! a^2}{(y-1)! 1!} \leq y^{3}$, and again used that the quantity is decreasing in $y$.
            In particular, we require $\frac{\lambda m}{(1 - \lambda) (1 - 4 \eps \alpha) n} < \frac{1}{8}$\footnote{This is to ensure the multiplicative increase from the $y^3$ term does not overrule the multiplicative decrease from the left hand side.} for which $m < n/16$ suffices.
            The final bound follows by choosing $y = 2$.
            Then, for $b \geq 2$, we have 
            \begin{align*}
                 \barI_{a, b} &= \bigO{\frac{\lambda^{2y}m^{2y} (6 \eps)^{2b} (1 + a^2 \eps^2)}{(a!b!)^2 n^{2y}} \cdot \frac{a!b!n^{a + 1}}{m^{a + 1} \lambda^{a} (1 - \lambda)^{b}(1 - 4 \eps \alpha)^{b}} \cdot \frac{y! n}{m(1 - \lambda)^{y}(1 - 4 \eps \alpha)^{y}}} \\
                 &= \bigO{\frac{y! \lambda^{2y - a} m^{2y - a - 2} (6 \eps)^{2b} (1 + a^2 \eps^2)}{a! b! n^{2y - a - 2} (1 - \lambda)^{y + b} (1 - 4 \eps \alpha)^{y}}} \\
                 &= \bigO{\frac{2^{y} \alpha^{2b + a} m^{2y - a - 2} (6\eps)^{2b}(1 + a^2 \eps^2)}{n^{2y - a - 2} (1 - 4 \eps \alpha)^{y}}} \\
                 &= \bigO{\frac{m^{y + b - 2}\alpha^{2b + a}}{n^{y + b - 2}} \left( \frac{2}{1 - 4 \eps \alpha} \right)^{y} (6 \eps)^{2b}(1 + a^2 \eps^2)} \\
                 &= \bigO{\frac{m^{y + b - 2}\alpha^{y + b}}{n^{y + b - 2}} \left( \frac{2}{1 - 4 \eps \alpha} \right)^{y} (6 \eps)^{2b}(1 + y^2 \eps^2)} \\
                 &= \bigO{\frac{m^{y}\alpha^{y + 2}}{n^{y}} \left( \frac{2}{1 - 4 \eps \alpha} \right)^{y} \eps^{4} (1 + y^2 \eps^2)} \\
                 &= \bigO{\frac{m^{2} \eps^{4} \lambda^{4}}{n^{2}}} \text{.}
            \end{align*}
            In the first inequality we group like terms.
            in the second we use $\frac{y!}{a!b!} \leq 2^{y}$, $a + b = y$ and $\alpha = \frac{\lambda}{1 - \lambda}$.
            In the third we again group like terms and use $a + b = y$.
            In the fourth, we use $a \leq y$.
            In the fifth, we note that for fixed $y$, the quantity is decreasing in $b$ since $\frac{m}{n} (6 \eps)^{2} \alpha < 1$.
            Thus, we may assume $b = 2$.
            In the final bound, we observe that the quantity is decreasing in $y$ following a similar argument as the previous case, and thus assume $y = b = 2$.
            We conclude by observing $\alpha = O(\lambda)$.

            Since $y = O(\log n)$,
            Then, since $\frac{m^2 \eps^{4} \log n \lambda^{4}}{n^2} = \litO{\frac{m \eps^{4} \lambda^{2}}{n}}$ we have that
            \begin{equation*}
                I(X:A_i, B_i) = \bigO{\barI_{y, 0} + \barI_{y - 1, 1} + \sum_{b = 2}^{y} \barI_{a, b}} = \bigO{\frac{m \eps^{4} \lambda^{2}}{n}} \text{.}
            \end{equation*}
        \end{proof}

        Given \Cref{lemma:uniformity-one-bucket-mi}, we conclude that 
        \begin{equation*}
            I(X:T) = \bigO{q \frac{\eps^4 \lambda^{2} m}{n}} = o(1) 
        \end{equation*}
        which contradicts $I(X:T) \geq 2 \cdot 10^{-4}$.
        Thus, we must have $q = \bigOm{\frac{n}{m \eps^{4} \lambda^{2}}}$.
    \end{proof}

    It remains to prove the omitted claims, which we prove below.
    This concludes the proof of \Cref{thm:dist-contamination-lb-eps}.
\end{proof}

\UniformityHLB*

\begin{proof}
    From \eqref{eq:uniformity-mix-dist} we have
    \begin{align*}
        h_{y} &= \frac{\dom}{n} \exp \left( - \frac{m \lambda (1 + 4 \eps)}{n} - (1 - \lambda)(1 - 4 \eps \alpha) \right)\frac{\left(\frac{m \lambda (1 + 4 \eps)}{n} + (1 - \lambda) (1 - 4 \eps \alpha)\right)^{y}}{y!} \\
        &\quad+\frac{n - \dom}{n} \exp \left( - \frac{m \lambda (1 + 4 \eps)}{n} \right) \frac{\left(\frac{m \lambda (1 + 4 \eps)}{n}\right)^{y}}{y!} \\
        &= \bigOm{\frac{\dom}{n} \frac{(1 - \lambda)^{y}(1 - 4 \eps \alpha)^{y}}{y!}} \text{.}
    \end{align*}
    Note that we obtain the lower bound by observing that $\exp \left( -\frac{m\lambda(1 + 4 \eps)}{2n} - \frac{(1 - \lambda)(1 - 4 \eps)}{2} \right) = \bigOm{1}$ from our assumptions $m \leq n$, $\lambda \leq \frac{1}{2}$ and $\eps < 0.1$.
\end{proof}

\UniformityGDiffUB*

\begin{proof}
    From \Cref{def:uniformity-hard-instance}, we use the definition of the Poisson distribution and write
    \begin{align*}
        g_{a, b}(0) &= \frac{\dom}{n} \exp \left( - \frac{m \lambda}{n} \right) \frac{(\frac{m\lambda}{n})^{a}}{a!} \exp \left( - \frac{4 \eps \lambda m}{n} - (1 - \lambda)(1 - 4 \eps \alpha) \right) \frac{\left(\frac{4 \eps \lambda m}{n} + (1 - \lambda)(1 - 4 \eps \alpha) \right)^{b}}{b!} \\
        &\quad+ \frac{n - \dom}{n} \exp \left( - \frac{m \lambda}{n} \right) \frac{(\frac{m \lambda}{n})^{a}}{a!} \exp \left( - \frac{4 \eps \lambda m}{n}\right) \frac{\left(\frac{4 \eps \lambda m}{n} \right)^{b}}{b!} \\
        &= \frac{\dom}{n} \exp \left(-\frac{m \lambda(1 + 4 \eps)}{n} - (1 - \lambda)(1 - 4 \eps \alpha) \right) \frac{m^{a} \lambda^{a}}{n^{a}} \frac{1}{a!b!} \left( \left(\frac{4 \eps \lambda m}{n} + (1 - \lambda)(1 - 4 \eps \alpha)\right)^{b} \right) \\
        &\quad+ \frac{n - \dom}{n} \exp \left(-\frac{m \lambda (1 + 4 \eps)}{n}\right) \frac{m^{a + b} \lambda^{a + b}}{n^{a + b}} \frac{1}{a!b!} \left( (4\eps)^{b} \right) \text{.}
    \end{align*}
    Note we have used $\lambda = \alpha (1 - \lambda)$ above.
    Following similar computations,
    \begin{align*}
        g_{a, b}(1) &= \frac{\dom}{n} \exp \left( - \frac{\lambda m}{n} \right) \frac{(\frac{m \lambda}{n})^{a}}{a!} \exp \left( - \frac{4 \eps \lambda m}{n} - (1 - \lambda)(1 - 4\eps \alpha) \right) \frac{\left( \frac{4 \eps \lambda m}{n} + (1 - \lambda)(1 - 4 \eps \alpha)\right)^{b}}{b!} \\
        % &\quad+ \frac{\dom}{2n} \exp \left( - \frac{m(1 - \eps)}{2n} \right) \frac{(\frac{m(1 - \eps)}{2n})^{a}}{a!} \exp \left( - \frac{3 \eps m}{2n} - \frac{1 - 2 \eps}{2} \right) \frac{\left(\frac{3 \eps m}{2 n} + \frac{1 - 2 \eps}{2} \right)^{b}}{b!} \\
        &\quad+ \frac{n - \dom}{2n} \exp \left( - \frac{\lambda m(1 + 2\eps)}{n} \right) \frac{(\frac{\lambda m(1 + 2\eps)}{n})^{a}}{a!} \exp \left( - \frac{2\eps \lambda m}{n}\right) \frac{\left(\frac{2\eps \lambda m}{n} \right)^{b}}{b!} \\
        &\quad+ \frac{n - \dom}{2n} \exp \left( - \frac{\lambda m(1 - 2\eps)}{n} \right) \frac{(\frac{\lambda m(1 - 2\eps)}{n})^{a}}{a!} \exp \left( - \frac{6 \eps \lambda m}{n}\right) \frac{\left(\frac{6 \eps \lambda m}{n} \right)^{b}}{b!} \\
        &= \frac{\dom}{n} \exp \left( -\frac{\lambda m(1 + 4 \eps)}{n} - (1 - \lambda)(1 - 4 \eps) \right) \frac{m^{a} \lambda^{a}}{n^{a}} \frac{1}{a!b!} \left( \left(\frac{4\eps \lambda m}{n} + (1 - \lambda) (1 - 4 \eps \alpha)\right)^{b} \right) \\
        &\quad+ \frac{n - \dom}{2n} \exp \left( -\frac{\lambda m(1 + 4 \eps)}{2n} \right) \frac{m^{a + b} \lambda^{a + b}}{n^{a + b}} \frac{1}{a!b!} \left( (1 + 2\eps)^{a} (2\eps)^{b} + (1 - 2\eps)^{a} (6 \eps)^{b} \right)
        \text{.}
    \end{align*}

    If $a + b \leq 1$, note that $g_{a, b}(0) = g_{a, b}(1)$.
    Thus, we may assume that $a + b \geq 2$ (in particular, in the following assume that $y \geq 2$).
    Note that the contribution from the term where $W_i = h$ is $0$ as $A_i, B_i$ are identically distributed in this case.
    % Similarly, if $Y_i > 100$, then conditioned on $E$ we have $W_i = h$ and $g_{a, b}(0) = g_{a, b}(1)$ since the contribution from this term in both cases is identical.
    % Thus, we may also assume that $W_i \in \set{l_+, l_-}$.

    Thus, in order to bound $(g_{a, b}(0) - g_{a, b}(1))^{2}$, we bound the following quantity.
    \begin{align*}
        % \delta_{a, b} = \delta_{a, b}(\eps) &:= \left| 2 \left(\frac{2 \eps m}{n} + (1 - 2 \eps)\right)^{b} - \left( (1 + \eps)^{a} \left(\frac{\eps m}{n} + (1 - 2 \eps)\right)^{b} + (1 - \eps)^{a} \left(\frac{3 \eps m}{n} + (1 - 2 \eps)\right)^{b} \right) \right| \\
        \eta_{a, b} = \eta_{a, b}(\eps) &:=  \left| 2 (4 \eps)^{b} - \left( (1 + 2 \eps)^{a} (2 \eps)^{b} + (1 - 2 \eps)^{a} (6 \eps)^{b} \right)  \right| \text{.}
    \end{align*}

    We begin by showing
    \begin{equation}
        \label{eq:uniformity-eta-bound}
        \eta_{a, b} = \begin{cases}
            \bigO{a^2 \eps^2} & b = 0, a \geq 2 \\
            \bigO{a \eps^2} & b = 1, a \geq 1 \\
            \bigO{(6 \eps)^{b} (1 + a^2 \eps^2))} & b \geq 2
        \end{cases}\text{.}
    \end{equation}
    If $b = 0$, we have $a \geq 2$ and
    \begin{align*}
        \eta_{a, b} = \left| 2 - ((1 + 2\eps)^{a} + (1 - 2\eps^{a})) \right| = \left| 2 - (1 + 2 a \eps + O(a^2 \eps^2) + 1 - 2 a \eps + O(a^2 \eps^2)) \right| = O(a^2 \eps^2) \text{.}
    \end{align*}
    If $b = 1$ and $a \geq 2$ then
    \begin{align*}
        \eta_{a, b} &= \left| 8 \eps - ((1 + 2 \eps)^{a} (2 \eps) + (1 - 2 \eps)^{a} (6 \eps)) \right| = \left| 8 \eps - (2 \eps + 4 a \eps^2 + O(a^2 \eps^3) + 6 \eps - 12 a \eps^2 + O(a^2 \eps^3)) \right| \\
        &= O(a \eps^2) \text{.}
    \end{align*}
    Note that the same bound holds when $a = 1$ (in fact we simply omit the $O(a^2 \eps^3)$ contribution.
    Then, assume $b \geq 2$.
    By the triangle inequality we have
    \begin{align*}
        \eta_{a, b} &\leq 2 (2\eps)^{b} + (1 + \eps)^{a} \eps^{b} + (1 - \eps)^{a} (3 \eps)^{b} \\
        &\leq 2 (2 \eps)^{b} + (3 \eps)^{b} \left( (1 + \eps)^{a} + (1 - \eps)^{a} \right) \text{.}
    \end{align*}
    If $a < 2$, we may bound the above by $O((3 \eps)^{b})$.
    On the other hand, if $a \geq 2$, we observe that
    \begin{align*}
        \eta_{a, b} &\leq 2 (4 \eps)^{b} + (6 \eps)^{b} \left (1 + 2 a \eps + O(a^2 \eps^2) + 1 - 2 a \eps + O(a^2 \eps^2) \right) \\
        &= \bigO{(6 \eps)^{b} (1 + a^2 \eps^{2})} \text{.}
        % = \bigO{(6 \eps )^{b}} \text{.}
    \end{align*}

    Given our bounds on $\eta_{a, b}$, we can conclude
    \begin{align*}
        (g_{a, b}(0) - g_{a, b}(1))^{2} &= \bigO{\left( \frac{1}{a! b!} \frac{m^{a + b} \lambda^{a + b}}{n^{a + b}} \eta_{a, b} \right)^{2}} \text{.}
    \end{align*}

    Applying our bounds on $\eta_{a, b}$ from \eqref{eq:uniformity-eta-bound} we have
    \begin{align*}
        (g_{a, b}(0) - g_{a, b}(1))^{2} &= \bigO{\left( \frac{\lambda^{a + b}}{a! b!} \frac{m^{a + b}}{n^{a + b}} \right)^{2}} \cdot \begin{cases}
           a^{4} \eps^{4} & b = 0, a \geq 2 \\
           a^{2} \eps^{4} & b = 1, a \geq 1 \\
           % \frac{m^{6} \eps^{4} (1 - 2 \eps)^{2b}}{n^{6}} + \frac{(3 \eps m)^{2y}}{n^{2y}} & b \geq 2, a \leq 1 \\
           % \frac{m^{a + 1} \eps^{4} (1 - 2 \eps)^{2b}}{n^{a + 1}} + \frac{(3 \eps m)^{2y}}{n^{2y}} & b \geq 2, a \geq 2
           (6 \eps)^{2b}(1 + a^2 \eps^2)^2 & b \geq 2, a \geq 0
        \end{cases} \text{.}
    \end{align*}
\end{proof}

\UniformityGLB*

\begin{proof}
    Following \Cref{clm:uniformity-g-diff-ub}, we begin with
    \begin{align*}
        g_{a, b}(0) &= \bigOm{\frac{1}{a!b!} \left( \frac{m^{a + 1} \lambda^{a} (1 - \lambda)^{b} (1 - 4 \eps \alpha)^{b}}{n^{a + 1}} +  \frac{m^{y} \lambda^{y} (4 \eps)^{b}}{n^{y}} \right)} \text{.}
    \end{align*}
    Next, we bound
    \begin{align*}
        g_{a, b}(1) &= \bigOm{\frac{1}{a!b!} \left( \frac{m^{a + 1} \lambda^{a} (1 - \lambda)^{b}(1 - 4 \eps \alpha)^{b}}{n^{a + 1}} + \frac{m^{y} \lambda^{y} (2 \eps)^{b}}{n^{y}} \right)} \text{.}
    \end{align*}
    When $b = 0$, we obtain the bound
    \begin{align*}
        g_{a, b}(x) &= \bigOm{\frac{1}{a!b!} \frac{m^{y} \lambda^{y}}{n^{y}}}
    \end{align*}
    and when $b \geq 1$, we obtain the bound
    \begin{align*}
        g_{a, b}(x) &= \bigOm{\frac{\lambda^{a} (1 - \lambda)^{b}}{a!b!} \left( \frac{m^{a + 1} (1 - 4 \eps \alpha)^{b}}{n^{a + 1}} \right)} \text{.}
    \end{align*}
\end{proof}

\subsection{Closeness Testing}

In this section, we give our lower bound for closeness testing.

\ClosenessTestingLB*

\begin{proof}
    By standard closeness testing lower bounds (see e.g. \cite{valiant2008testing}), we may assume $m = \Omega(n^{2/3}\lambda^{-1})$, even if we query all samples.
    This follows as any algorithm must observe at least $\bigOm{n^{2/3}}$ samples from $\p$ which requires $\bigOm{n^{2/3} \lambda^{-1}}$ samples from the mixture $\rD$.

    \paragraph{Hard Instance.}
    Consider the following hard instance.
    As before, define $\alpha := \frac{\lambda}{1 - \lambda} \leq 1$.

    \begin{definition}
        \label{def:closeness-hard-instance}
        Let $\calM_{0}$ be the distribution over tuples of non-negative measures $(\p_1, \q_1, \p_2, \q_2)$ generated as follows: for all $i \in [n]$,
        \begin{equation*}
            (\p_1[i], \q_1[i], \p_2[i], \q_2[i]) = 
            \begin{cases}
                \left( \frac{1 - \eps}{2 \lambda \dom}, 0, \frac{1 - \eps}{2 \lambda \dom}, 0 \right) & w.p. \quad \frac{\lambda \dom}{n} \\
                \left( \frac{1 - \eps}{2(1 - \lambda)\dom}, \frac{1 - \eps\alpha}{(1 - \lambda)\dom}, \frac{1 - \eps}{2(1 - \lambda)\dom}, \frac{1 - \eps\alpha}{(1 - \lambda)\dom} \right) & w.p. \quad \frac{(1 - \lambda) \dom}{n} \\
                \left( \frac{3\eps}{2(n - \dom)}, \frac{\alpha \eps}{2(n - \dom)}, \frac{3\eps}{2(n - \dom)}, \frac{\alpha\eps}{2(n - \dom)} \right)  & w.p. \quad \frac{n - \dom}{2n} \\
                \left( \frac{\eps}{2(n - \dom)}, \frac{3\eps\alpha}{2(n - \dom)}, \frac{\eps}{2(n - \dom)}, \frac{3\eps\alpha}{2(n - \dom)} \right)  & w.p. \quad \frac{n - \dom}{2n}
            \end{cases}
        \end{equation*}

        Similarly, let $\calM_{1}$ be the distribution over tuples of non-negative measures $(\p_1, \q_1, \p_2, \q_2)$ generated as follows: for all $i \in [n]$,
        \begin{equation*}
            (\p_1[i], \q_1[i], \p_2[i], \q_2[i]) = \begin{cases}
                \left( \frac{1 - \eps}{2 \lambda \dom}, 0, \frac{1 - \eps}{2 \lambda \dom}, 0 \right) & w.p. \quad \frac{\lambda \dom}{n} \\
                \left( \frac{1 - \eps}{2(1 - \lambda)\dom}, \frac{1 - \eps\alpha}{(1 - \lambda)\dom}, \frac{1 - \eps}{2(1 - \lambda)\dom}, \frac{1 - \eps\alpha}{(1 - \lambda)\dom} \right) & w.p. \quad \frac{(1 - \lambda) \dom}{n} \\
                \left( \frac{3 \eps}{2(n - \dom)}, \frac{\eps\alpha}{2(n - \dom)}, \frac{\eps}{2(n - \dom)}, \frac{3 \eps\alpha}{2(n - \dom)} \right)  & w.p. \quad \frac{n - \dom}{2n} \\
                \left( \frac{\eps}{2(n - \dom)}, \frac{3 \eps\alpha}{2(n - \dom)}, \frac{3 \eps}{2(n - \dom)}, \frac{\eps\alpha}{2(n - \dom)} \right)  & w.p. \quad \frac{n - \dom}{2n} 
            \end{cases}
        \end{equation*}

        The meta-distribution $\calH_{C}$ is the distribution over random tuples of non-negative measures $(\p_1, \q_1, \p_2, \q_2)$ generated by choosing $X \sim \Bern(1/2)$ and $(\p_1, \q_1, \p_2, \q_2) \sim \calM_{X}$.
    \end{definition}

    Let $h_{\lambda}, h, l_1, l_2$ denote each type, that is let $W_{i}$ denoting the type of bucket $i$ be independently chosen for all $i$ where $\Pr(W_{i} = h_{\lambda}) = \frac{\lambda \dom}{n}, \Pr(W_i = h) = \frac{(1 - \lambda) \dom}{n}$ and $\Pr(W_{i} = l_{1}) = \Pr(W_i = l_2) = \frac{n - \dom}{2n}$.
    We begin by arguing that all measures are close to real distributions.
    Note that by our assumption $m > \frac{1}{\lambda}$, $\frac{1}{\lambda m} < 1$ so all masses are at most $1$.
    
    \begin{lemma}
        \label{lemma:closeness-dist-norm}
        Let $X \in \set{0, 1}$.
        For all $n \geq 1$,
        \begin{equation*}
            \Pr_{(\p_1, \q_1, \p_2, \q_2) \sim \calM_{X}}\left( \norm{\p_1}_{1}, \norm{\q_1}_{1}, \norm{\p_2}_{1}, \norm{\q_2}_{1} \in (0.9, 1.1) \right) \geq 1 - \frac{1}{\poly(n)} \text{.}
        \end{equation*}
    \end{lemma}

    \begin{proof}
        Let $H_{\lambda}, H, L_1, L_2$ denote the number of items in each type.
        By standard concentration bounds, we have $H_{\lambda} = \lambda \dom \pm O(\sqrt{\lambda \dom \log n})$, $H = (1 - \lambda) \dom \pm O(\sqrt{\dom \log n})$ and $L_1, L_2 = \frac{n - \dom}{2} \pm O(\sqrt{n \log n})$ with probability at least $1 - 1/\poly(n)$.
        For sufficiently large $n$, these conditions guarantee that the $\ell_1$-norm of the measures is in $(0.9, 1.1)$ (using the assumption without loss of generality that $m = \Omega(n^{2/3} \lambda^{-1})$). 
        For $\p_1$ when $X = 0$, we have
        \begin{align*}
            \norm{\p_1}_{1} &= H_{\lambda} \frac{1 - \eps}{2 \lambda \dom} + H \frac{1 - \eps}{2 (1 - \lambda) \dom} + L_1 \frac{3 \eps}{2(n - \dom)} + L_2 \frac{\eps}{2(n - \dom)} \\
            &= \frac{1 - \eps}{2} + \bigO{\frac{\sqrt{\lambda \dom \log n}}{\lambda m}} + \frac{1 - \eps}{2} + \bigO{\frac{\sqrt{\dom \log n}}{m}}+ \frac{3 \eps}{4} + \frac{\eps}{4} \pm \bigO{\frac{\sqrt{n \log n}}{n}} \\
            &= 1 \pm o(1) \text{.}
        \end{align*}
        Note that above we have used $\lambda m = \bigOm{n^{2/3}}$ so that all the error terms are $o(1)$ for sufficiently large $n$.
        The bounds on $\p_2, \q_1, \q_2$ (and when $X = 1$) follow similarly.
    \end{proof}

    Thus, \Cref{lemma:closeness-dist-norm} implies that when sampling distributions from the meta-distribution $\calH_{C}$, the distributions (and mixtures $\rD_1, \rD_2$) satisfy the desired norm conditions.
    By \Cref{lem:poisson-sampling}, it suffices to consider a deterministic Poissonized tester $\innerAlg$ with sample complexity $m$.
    We will show that if $\innerAlg$ is $0.1$-correct with respect to $\calH_0, \calH_1$, $\innerAlg$ requires $\bigOm{n^2 m^{-2} \eps^{-4} \lambda^{-3}}$ queries.

    As before, we will show a stronger result.
    \begin{lemma}
        \label{lemma:closeness-bucket-query-lb}
        Suppose $\innerAlg$ is a deterministic Poissonized tester with sample complexity $m$ is $0.1$-correct with respect to $\calH_0, \calH_1$.
        Then, $\innerAlg$ queries $\bigOm{n^2 m^{-2} \eps^{-4} \lambda^{-3}}$ distinct buckets.
        This bound holds even when a single query on bucket $i$ reveals the source of every sample in bucket $i$.
    \end{lemma}

    \begin{proof}
        Given that $\innerAlg$ is a Poissonized tester, $\innerAlg$ receives independently $Y^{(1)}_i \sim \Poi(m \rD_{1}[i]), Y_{i}^{(2)} \sim \Poi(m \rD_{2}[i])$ samples on bucket $i$, where $\rD_{j}[i] = \frac{\p_{j}[i] + \q_{j}[i]}{2}$ for $j \in \set{1, 2}$.
        Note that we can equivalently sample $Y_i^{(j)} = A_i^{(j)} + B_{i}^{(j)}$ where $A_{i}^{(j)} \sim \Poi(\lambda m \p_{j}[i]), B_{i}^{(j)} \sim \Poi((1 - \lambda) m \q_{j}[i])$ are sampled independently.
        Of course, the algorithm (without queries) only observes $Y_{i} := (Y_{i}^{(1)}, Y_{i}^{(2)})$.
        If the algorithm queries a bucket $i$, then the algorithm learns $A_{i} := (A_{i}^{(1)}, A_{i}^{(2)}), B_{i} := (B_{i}^{(1)}, B_{i}^{(2)})$ as well.
        Suppose the algorithm queries a set of buckets $Q \subset [n]$.
        Let $T_i$ denote the information observed by the algorithm from the $i$-th bucket, i.e. $Y_{i}$ if $i \not\in Q$ and $A_{i}, B_{i}$ if $i \in Q$.
        Let $T = (T_1, \dotsc, T_n)$.
        
        Then, the algorithm $\innerAlg$ is a function such that $\innerAlg(T) = X$ with at least $51\%$ probability, and therefore $I(X:T) \geq 2 \cdot 10^{-4}$ by \Cref{lemma:mutual-info-bound}.
        Suppose for contradiction that $\innerAlg$ queries $q = |Q| = o(n^2 m^{-2} \eps^{-4} \lambda^{-2})$ buckets.
        Since $T_i$ are independent conditioned on $X$, we have
        \begin{equation*}
            I(X:T) \leq \sum_{i = 1}^{n} I(X:T_{i}) \text{.}
        \end{equation*}

        For $i \not\in Q$, we claim $I(X:T_{i}) = I(X:Y_{i}) = 0$.
        In particular, for all $i \in [n], j \in \set{1, 2}$, we have $Y_{i}^{(j)} \sim \Poi(m \rD_{j}[i])$ where
        \begin{equation}
            \label{eq:closeness-mix-dist}
            \rD_{j}[i] = \begin{cases}
                \frac{1 - \eps}{2 \dom} & w.p.\quad \frac{\lambda \dom}{n} \\
                \frac{\alpha (1 - \eps)}{2 \dom} + \frac{1 - \eps \alpha}{\dom} & w.p.\quad \frac{(1 - \lambda)\dom}{n} \\
                \frac{2 \lambda \eps}{(n - \dom)} & w.p. \quad \frac{n - \dom}{n}
            \end{cases}
        \end{equation}
        regardless of $X$, proving the claim.
        Thus, the samples reveal no information it suffices to bound the information gained from the queries via $\sum_{i \in Q}^{n} I(X:T_{i}) = \sum_{i \in Q}^{n} I(X:A_{i}, B_{i})$.

        We begin with the simplifying assumption that no single bucket causes large collisions. 
        % This simply follows from Markov's inequality and our assumption that $m \leq n^{0.99}$ implying that the number of $c$-collisions on light buckets is roughly $\frac{m^{c}}{n^{c - 1}}$.
        \begin{lemma}
            \label{lemma:closeness-light-bucket-collision-bound}
            \begin{equation*}
                \Pr \left( \max_{i} Y_i = O(\log n) \right) > 1 - \frac{1}{\poly(n)}
            \end{equation*}
        \end{lemma}

        \begin{proof}
            % As in \Cref{clm:uniformity-p-freq-ub} we have $m' \leq 2m$ with high probability.
            % From \Cref{cor:c-collision-ub} and our assumption that $m \leq n^{0.99}$, we obtain that the number of $c$-collisions on buckets of type $l_1, l_2$, is greater than $t$ with probability
            % \begin{equation*}
            %     \frac{1}{t} \binom{m'}{c} \sum_{W_i \in \set{l_1, l_2}} \rD_{j}[i]^{c} \leq \frac{\eps^{c}}{t} \frac{n}{(n - m)^{c}} \binom{m'}{c} \leq \frac{(2\eps)^c (2m)^{c}}{n^{c - 1} t \cdot c!} \leq \frac{n^{0.99c-c+1} (4 \eps)^{c}}{t \cdot c!}
            % \end{equation*}
            % where we simplify with $n - m \geq n/2$ and $m \leq n^{0.99}$.
            % Fixing $t < 1$, note that the above quantity is less than $0.001$ for $c > 100$.
            % We conclude with a union bound over the conditioned event.

            Note that each $Y_i \sim \Poi(m \rD[i])$ has parameter at most $2$ from \eqref{eq:closeness-mix-dist}. 
            The claim follows as before from \Cref{lemma:poisson-concentration} and the union bound.
        \end{proof}

        Let $E$ denote the event that the conditions of \Cref{lemma:closeness-dist-norm} and \Cref{lemma:closeness-light-bucket-collision-bound} are met.

        \begin{lemma}
            \label{lemma:closeness-one-bucket-mi}
            For all $i$, we have
            \begin{align*}
                I(X:A_{i}, B_{i}) = \bigO{\frac{\eps^{4} \lambda^{3} m^{2}}{n^{2}}} \text{.} 
            \end{align*}
        \end{lemma}

        \begin{proof}
            Note $A_i, B_i$ are pairs of Poisson random variables that sum up to $Y_{i} = (y, z)$.
            By \Cref{claim:MI_asymp}, we have
            \begin{equation*}
                I(X:A_{i}, B_{i}) = \bigO{\sum_{a, b, c, d} \frac{\left( f_{a, b, c, d}(0) - f_{a, b, c, d}(1) \right)^{2}}{f_{a, b, c, d}(0) + f_{a, b, c, d}(1)}} 
                = \bigO{\sum_{a + b = y, c + d = z} \frac{\left( f_{a, b, c, d}(0) - f_{a, b, c, d}(1) \right)^{2}}{f_{a, b, c, d}(0) + f_{a, b, c, d}(1)}} \text{.} 
            \end{equation*}
            where $f_{a, b, c, d}(x) = \Pr(A_{i} = (a, c), B_{i} = (b, d) | Y_{i} = (y, z), X = x, E)$ and we note that unless $a + b = y, c + d = z$, $f_{a, b, c, d}(x) = 0$.
            We now compute $f_{a, b, c, d}(x)$ explicitly.
            \begin{align*}
                f_{a, b, c, d}(x) &= \frac{\Pr(A_{i} = (a, c), B_{i} = (b, d) ,Y_{i} = (y, z)| X = x, E)}{\Pr(Y_{i} = (y, z) | X = x, E)} \\
                &= \begin{cases}
                    \frac{\Pr(A_{i} = (a, c), B_{i} = (b, d)| X = x, E)}{\Pr(Y_{i} = (y, z) | X = x, E)} & \ifT a + b = y, c + d = z \\
                    0 & \otherwise
                \end{cases}
            \end{align*}
            Define $g_{a, b, c, d}(x) = \Pr(A_{i} = (a, c), B_{i} = (b, d)| X = x, E)$ and $h_{y, z}(x) = \Pr(Y_{i} = (y, z) | X = x, E)$.
            Since samples are independent of $X$, define $h_{y, z} := h_{y, z}(0) = h_{y, z}(1)$.
            Then, we can rewrite 
            \begin{equation*}
                I(X:A_{i}, B_{i})
                = \bigO{\frac{1}{h_{y, z}} \sum_{a + b = y, c + d = z} \frac{\left( g_{a, b, c, d}(0) - g_{a, b, c, d}(1) \right)^{2}}{g_{a, b, c, d}(0) + g_{a, b, c, d}(1)}} \text{.} 
            \end{equation*}

            We require the following technical claims.
            
            \begin{restatable}{claim}{ClosenessHLB}
                \label{clm:closeness-h-lb}
                \begin{equation*}
                    h_{y, z} = \bigOm{\frac{\dom}{n} \frac{(1 - \eps \alpha)^{y + z}}{y! z!}} \text{.}
                \end{equation*}
            \end{restatable}

            \begin{restatable}{claim}{ClosenessGLB}
                \label{clm:closeness-g-lb}
                For $x \in \set{0, 1}$, 
                \begin{equation*}
                    g_{a, b, c, d}(x) = \bigOm{\frac{\lambda \dom}{n} \frac{\left(\frac{1 - \eps}{2} \right)^{a + c}}{a!c!}} \text{.}
                \end{equation*}
            \end{restatable}

            \begin{restatable}{claim}{ClosenessGDiffUB}
                \label{clm:closeness-g-diff-ub}
                If $a + b = 0$, $c + d = 0$, 
                % or $Y_{i}^{(1)}, Y_{i}^{(2)} > 100$, 
                then $g_{a, b, c, d}(0) = g_{a, b, c, d}(1)$.
                Otherwise $a + b \geq 1$, $c + d \geq 1$
                % , $Y_{i}^{(1)}, Y_{i}^{(2)} \leq 100$ 
                and 
                \begin{align*}
                    (g_{a, b, c, d}(0) - g_{a, b, c, d}(1))^{2} &= \bigO{\left( \frac{1}{a!b!c!d!} \left( \frac{\eps \lambda m}{n} \right)^{a + b + c + d} \cdot \left( 3^{a + c} + 3^{b + d} - 3^{a + d} - 3^{b + c} \right)\right)^{2}} \text{.}
                \end{align*}
            \end{restatable}

            We will prove these claims at the end of this section.
            Applying \Cref{clm:closeness-g-diff-ub}, we note that if either $Y_i^{(1)}, Y_{i}^{(2)}$ are $0$,
            % or greater than $100$, 
            then $I(X: A_i, B_i) = 0$.
            Thus, we assume $1 \leq Y_i^{(1)}, Y_{i}^{(2)}$.
            Applying \Cref{clm:closeness-h-lb}, \Cref{clm:closeness-g-lb} and \Cref{clm:closeness-g-diff-ub}, we obtain
            \begin{align*}
                I(X:A_i, B_i) &= \bigO{\frac{1}{\frac{m}{n} \frac{(1 - \eps \alpha)^{y + z}}{y!z!}} \sum_{a + b = y, c + d = z} \frac{\left( \frac{1}{a!b!c!d!} \left( \frac{\eps \lambda m}{n} \right)^{a + b + c + d} \cdot \left( 3^{a + c} + 3^{b + d} - 3^{a + d} - 3^{b + c} \right)\right)^{2}}{\frac{\lambda m}{n} \frac{\left(\frac{1 - \eps}{2} \right)^{a + c}}{a!c!}}} \\
                &= \bigO{\left(\frac{m}{n}\right)^{2y + 2z - 2} \frac{2^{y + z} \eps^{2y + 2z} \lambda^{2y + 2z - 1}}{\frac{(1 - \eps)^{2y + 2z}}{y!z!}} \sum_{a + b = y, c + d = z} \frac{1}{a!b!c!d!} \left( 3^{a + c} + 3^{b + d} - 3^{a + d} - 3^{b + c} \right)^{2}} \\
                % &= \bigO{\left(\frac{m}{n}\right)^{2y + 2z - 2} \frac{\eps^{2y + 2z} \lambda^{2y + 2z - 1}}{\frac{(1 - \eps)^{2y + 2z}}{y!z!}} \sum_{a + b = y, c + d = z} \frac{\left( 3^{a} - 3^{b} \right)^{2} \left( 3^{c} - 3^{d} \right)^{2}}{a!b!c!d!}} \\
                &= \bigO{\left(\frac{m}{n}\right)^{2y + 2z - 2} \frac{2^{y + z} \eps^{2y + 2z} \lambda^{2y + 2z - 1}}{(1 - \eps)^{2y + 2z}} y! z! \sum_{a + b = y, c + d = z} \frac{\left( 3^{a} - 3^{b} \right)^{2} \left( 3^{c} - 3^{d} \right)^{2}}{a!b!c!d!}} \text{.}
            \end{align*}
            In the first line  we simplify by collecting like terms, and observe $(1 - \eps \alpha) \geq \frac{1 - \eps}{2}$ and $a + b = y, c + d = z$.
            In the second line, we note that $3^{a + c} + 3^{b + d} - 3^{a + d} - 3^{b + c} = (3^{a} - 3^{b})(3^{c} - 3^{d})$.
            To complete the bound, we require the following technical claim, which we prove at the end of this section.
            
            \begin{restatable}{claim}{ClosenessABCDUB}
                \label{claim:closeness-a-b-c-d-ub}
                \begin{equation*}
                    y!z!\sum_{a + b = y, c + d = z} \frac{\left( 3^{a} - 3^{b} \right)^{2} \left( 3^{c} - 3^{d} \right)^{2}}{a!b!c!d!} = \bigO{18^{y + z}} \text{.}
                \end{equation*}
            \end{restatable}
            
            Applying \Cref{claim:closeness-a-b-c-d-ub} we obtain 
            \begin{align*}
                I(X:A_i, B_i) 
                &= \bigO{\left(\frac{m}{n}\right)^{2y + 2z - 2} 2^{y + z} \eps^{2y + 2z} \lambda^{2y + 2z - 1} \frac{18^{y + z}}{(1 - \eps)^{2y + 2z}}} \\
                &= \bigO{\left(\frac{m}{n}\right)^{2y + 2z - 2} \frac{(6 \eps \lambda)^{2y + 2z}}{(1 - \eps)^{2y + 2z}}} \frac{1}{\lambda} \\
                &= \bigO{\frac{\eps^{4} \lambda^{3} m^{2}}{n^{2}}}
            \end{align*}
            where in the last term we observe that the quantity is decreasing in $y + z$ and use $y + z \geq 2$ since $I(X:A_i, B_i) = 0$ whenever $y = 0$ or $z = 0$ by \Cref{clm:closeness-g-diff-ub}.
        \end{proof}

        Given \Cref{lemma:closeness-one-bucket-mi}, we conclude that 
        \begin{equation*}
            I(X:T) = \bigO{q \frac{\eps^4 \lambda^{3} m^2}{n^2}} = o(1) 
        \end{equation*}
        which contradicts $I(X:T) \geq 2 \cdot 10^{-4}$.
        Thus, we must have $q = \bigOm{\frac{n^2}{m^2 \eps^{4} \lambda^{3}}}$.
    \end{proof}

    It remains to prove the omitted claims, which we prove below.
    This concludes the proof of \Cref{thm:dist-contam-closeness-lb}.
\end{proof}

\ClosenessHLB*

\begin{proof}
    From \eqref{eq:closeness-mix-dist} we have
    \begin{align*}
        h_{y, z} &= \frac{\lambda \dom}{n} \exp \left( -\frac{1 - \eps}{m} \right) \frac{\left( \frac{1 - \eps}{2} \right)^{y + z}}{y!z!} \\
        &\quad+ \frac{(1 - \lambda) \dom}{n} \exp\left(-\alpha (1 - \eps) - 2(1 - \eps \alpha) \right) \frac{\left(\frac{\alpha(1 - \eps)}{2} + (1 - \eps \alpha) \right)^{y + z}}{y!z!} \\
        &\quad+ \frac{n - \dom}{\dom} e^{-4 \eps \lambda m/(n - m)} \frac{(2 \eps \lambda m/(n - m))^{y + z}}{y!z!} \\
        &= \bigOm{\frac{\dom}{n} \frac{(1 - \eps \alpha)^{y + z}}{y! z!}} \text{.}
    \end{align*}
    Above we have used $(1 - \lambda) \geq \frac{1}{2}$, $\alpha (1 - \eps) + 2(1 - \eps \alpha) = O(1)$ and $1 - \eps \alpha = \bigOm{\frac{\alpha(1 - \eps)}{2}}$.
\end{proof}

\ClosenessGLB*

\begin{proof}
    From \Cref{def:closeness-hard-instance}, we compute using $(1 - \lambda) \alpha = \lambda$,
    \begin{align*}
        g_{a, b, c, d}(0) &= \frac{\lambda \dom}{n} \exp \left( - (1 - \eps) \right) \frac{\left( \frac{1 - \eps}{2} \right)^{a + c}}{a!c!} \\
        &\quad+ \frac{(1 - \lambda) \dom}{n} \exp \left(- \alpha (1 - \eps) - 2 (1 - \eps \alpha)\right) \frac{(\frac{\alpha (1 - \eps)}{2})^{a + c}((1 - \eps \alpha))^{b + d}}{a! b! c!d!} \\
        &\quad+ \frac{n - \dom}{2n} \exp \left( \frac{- 3 \eps \lambda m}{(n - m)} \right) \frac{\left( \frac{3 \eps \lambda m}{2(n - m)} \right)^{a + c}}{a! c!} \cdot \exp \left( \frac{- \eps \lambda m}{(n - m)} \right) \frac{\left( \frac{\eps \lambda m}{2(n - m)} \right)^{b + d}}{b! d!} \\
        &\quad+ \frac{n - \dom}{2n} \exp \left( \frac{- 3 \eps \lambda m}{(n - m)} \right) \frac{\left( \frac{3 \eps \lambda m}{2(n - m)} \right)^{b + d}}{b! d!} \cdot \exp \left( \frac{- \eps \lambda m}{(n - m)} \right) \frac{\left( \frac{\eps \lambda m}{2(n - m)} \right)^{a + c}}{a! c!} \\
        &= \bigOm{\frac{\lambda \dom}{n} \frac{\left(\frac{1 - \eps}{2}\right)^{a + c}}{a!c!}} \text{.}
    \end{align*}
    Note that a similar bound holds for $g_{a, b, c, d}(1)$.
\end{proof}

\ClosenessGDiffUB*

\begin{proof}
    Following \Cref{clm:closeness-g-lb}, we write
    \begin{align*}
        g_{a, b, c, d}(1) &= \frac{\lambda \dom}{n} \exp \left( - (1 - \eps) \right) \frac{\left( \frac{1 - \eps}{2} \right)^{a + c}}{a!c!} \\
        &\quad+ \frac{(1 - \lambda) \dom}{n} \exp \left(- \alpha (1 - \eps) - 2 (1 - \eps \alpha)\right) \frac{(\frac{\alpha (1 - \eps)}{2})^{a + c}((1 - \eps \alpha))^{b + d}}{a! b! c!d!} \\
        &\quad+ \frac{n - \dom}{2n} \exp \left( \frac{- 3 \eps \lambda m}{(n - m)} \right) \frac{\left( \frac{3 \eps \lambda m}{2(n - m)} \right)^{a + d}}{a! d!} \cdot \exp \left( \frac{- \eps \lambda  m}{(n - m)} \right) \frac{\left( \frac{\eps \lambda m}{2(n - m)} \right)^{b + c}}{b! c!} \\
        &\quad+ \frac{n - \dom}{2n} \exp \left( -\frac{3 \eps \lambda m}{(n - m)} \right) \frac{\left( \frac{3 \eps \lambda m}{2(n - m)} \right)^{b + c}}{b! c!} \cdot \exp \left( \frac{- \eps \lambda m}{(n - m)} \right) \frac{\left( \frac{\eps \lambda  m}{2(n - m)} \right)^{a + d}}{a! d!} \text{.}
    \end{align*}
    
    We consider separate cases.
    If either $a = b = 0$ or $c = d = 0$, note that $g_{a, b, c, d}(0) = g_{a, b, c, d}(1)$.
    % Similarly, suppose $Y_{i}^{(1)}, Y_{i}^{(2)} > 100$.
    % Then, $g_{a, b, c, d}(0) = g_{a, b, c, d}(1)$ since we have $W_i = h$.
    
    Thus, we may assume $a + b \geq 1$, $c + d \geq 1$ (equivalently $y, z \geq 1$).
    % and $Y_i^{(1)}, Y_{i}^{(2)} \leq 100$.
    Note that the contribution from the heavy elements (i.e. where $W_i \in \set{h_{\lambda}, h}$) appears in both $g_{a, b, c, d}(0), g_{a, b, c, d}(1)$ and therefore cancels out in the difference.
    The difference can then be bounded by
    \begin{align*}
        (g_{a, b, c, d}(0) - g_{a, b, c, d}(1))^{2} &= \left( \frac{n - \dom}{2n} \frac{1}{a!b!c!d!} \exp \left( -\frac{4 \eps \lambda \dom}{n - \dom} \right) \left( \frac{\eps \lambda  m}{2(n - m)} \right)^{a + b + c + d} \right)^{2} \\
        &\quad \cdot \left( 3^{a + c} + 3^{b + d} - 3^{a + d} - 3^{b + c} \right)^{2}
    \end{align*}
    From our assumption $m \leq n/2$ we can simplify the above bound to
    \begin{align*}
        (g_{a, b, c, d}(0) - g_{a, b, c, d}(1))^{2} &= \bigO{\left( \frac{1}{a!b!c!d!} \left( \frac{\eps \lambda m}{n} \right)^{a + b + c + d} \cdot \left( 3^{a + c} + 3^{b + d} - 3^{a + d} - 3^{b + c} \right)\right)^{2}} \text{.}
    \end{align*}
\end{proof}

\ClosenessABCDUB*

\begin{proof}
    We begin by computing
    \begin{align*}
        \sum_{a + b = y} \frac{(3^a - 3^b)^{2}}{a!b!} &= \bigO{\sum_{a < y/2} \frac{y! 9^{a}}{a!(y-a)!}} = \bigO{18^{y}}
    \end{align*}
    since after factoring out $9^{a}$ we are left with a sum of binomial coefficients.
    Similarly, we may bound 
    \begin{equation*}
        \sum_{c + d = z} \frac{(3^c - 3^d)^{2}}{c!d!} = \bigO{18^{z}} \text{.}
    \end{equation*}
    Combining the two bounds yields the desired result.
\end{proof}

\section{Query-Optimal Closeness Testing} \label{sec:closeness-UB-highsample}

In this section, we design a closeness tester with nearly optimal $\widetilde{O}(\eps^{-2}\lambda^{-1})$ query complexity. Our approach relies on taking enough samples to guarantee empirical estimates with small error. In particular, we will rely on the following standard distribution learning bound.

\begin{theorem} [Distribution Learning] \label{thm:learn} Let $\rD \colon [n] \to [0,1]$ be an arbitrary probability distribution, and let $\widetilde{\rD} \colon [n] \to [0,1]$ be an empirical estimate\footnote{I.e., $\widetilde{\rD}[i]$ is the number of samples observed from element $i$ divided by the total number of samples, $m$.} using $m$ samples from $\rD$. Then, for every $i \in [n]$ and $\alpha \in (0,1)$, we have 
\[
\Pr\left[|\rD[i] - \widetilde{\rD}[i]| > \sqrt{\frac{3 \rD[i] \ln (2/\alpha)}{m}}\right] \leq \alpha \text{.}
\] 
\end{theorem}

\begin{proof} First, note that $m \widetilde{\rD}[i]$ follows a $\mathrm{Bin}(m,\rD[i])$ distribution, i.e. $m \widetilde{\rD}[i] = X_1 + \cdots + X_m$ where each $X_j$ is an independent $\mathrm{Bern}(\rD[i])$ random variable. We have $\EX[m \widetilde{\rD}[i]] = m \rD[i]$. Thus, by the Chernoff bound, we have
\[
\delta := \sqrt{\frac{3 \ln (2/\alpha)}{m \rD[i]}} ~\implies~ \Pr[|m \widetilde{\rD}[i] - m \rD[i]| > \delta m \rD[i]] \leq 2 \exp\left(-\delta^2 m\rD[i] / 3 \right) = \alpha
\]
or equivalently, $\Pr[|\widetilde{\rD}[i] - \rD[i]| > \delta \rD[i]] \leq \alpha$, as claimed. \end{proof}

The following is a straightforward corollary for mixtures $\rD = \lambda\p + (1-\lambda)\q$.

\begin{corollary} \label{cor:learn} Let $\p \colon [n] \to [0,1]$ be an arbitrary distribution and let $Z$ be a multi-set of $m$ samples drawn from $\rD = \lambda \p + (1-\lambda) \q$. Let $P[i]$ denote the number of samples which were drawn by $\p$ on element $i$. Then, 
\[
\Pr_{Z}\left[ \left|\frac{P[i]}{\lambda m} - \p[i] \right| > \sqrt{\frac{3 \p[i] \ln(2/\alpha)}{\lambda m}} \right] \leq \alpha \text{.}
\]
\end{corollary}

\begin{proof} Consider a distribution $\p_{\lambda} \colon [n] \cup \{A\} \to [0,1]$ where $\p_{\lambda}[i] = \lambda \p[i]$ and $\p_{\lambda}(A) = 1-\lambda$. For $i \in [n]$, the random variable $P_i/m$ is equivalent to an $m$-sample empirical estimate of $\p_{\lambda}[i] = \lambda \p[i]$. Therefore, by \Cref{thm:learn} we have
\[
\Pr_{Z}\left[\left|\frac{P_i}{m} - \lambda \p[i]\right| > \sqrt{\frac{3 \lambda \p[i] \ln (2/\alpha)}{m}} \right] \leq \alpha
\]
and so dividing by $\lambda$ completes the proof. \end{proof}

We are now ready to prove our main theorem.

%\begin{mdframed}[backgroundcolor=blue!5]
\begin{restatable}{theorem}{closenessUB} \emph{(Query-Optimal Closeness Testing.)} \label{thm:closenessUB-high}
Suppose we have verification query access to two mixtures $\rD_1 = \lambda \p_1 + (1-\lambda)\q_1$, $\rD_2 = \lambda \p_2 + (1-\lambda)\q_2$. There is an algorithm that distinguishes between the case of $\p_1 = \p_2$ and $\norm{\p_1 - \p_2}_1 \geq \eps$ with probability $2/3$ using 
\[
O\left(n\left(\log n + \frac{\log^2(1/\eps \lambda)}{\eps^2 \lambda}\right) \right) ~\text{ samples and }~ O\left(\frac{\log^3 (1/\eps\lambda)}{\eps^2 \lambda}\right) ~\text{ queries.}
\]
\end{restatable}
%\end{mdframed}

\begin{proof} We invoke a flattening procedure which allows us to assume $\rD_b$ and $\p_b$ (for $b \in \{1,2\}$) have bounded infinity norm, then design an efficient tester for such distributions. Our flattening procedure is described in the following theorem, which relies on \Cref{thm:learn} and is proven in \Cref{sec:high-sample-flat}. (It is an immediate corollary of \Cref{thm:strong-mixture-flat}.)

\begin{theorem} [High Sample Mixture Flattening] \label{thm:strong-mixture-flat-cor} Using $O(n \log n)$ samples to a mixture $\rD = \lambda \p + (1-\lambda)\q$ over $[n]$, it is possible to produce a $\ell_1$-distance preserving\footnote{For any two distributions $\cD_1,\cD_2$ over $[n]$, we have $\norm{f(\cD_1) - f(\cD_2)}_1 = \norm{\cD_1 - \cD_2}_1$.} random map $f \colon [n] \to [N]$ such that with probability $0.999$, we have $\norm{f(\rD)}_{\infty} < \frac{2}{N}, \norm{f(\p)}_{\infty} < \frac{2}{\lambda N}$, and $N \leq 3n$. \end{theorem}

We invoke \Cref{thm:strong-mixture-flat-cor} on $\frac{1}{2}(\rD_1 + \rD_2)$ so that with probability $0.999$ we have $N \leq 3n$ and
\begin{align*}
    \frac{2}{N} \geq \norm{f\left(\frac{1}{2}(\rD_1+\rD_2)\right)}_{\infty} = \norm{\frac{1}{2} \left( f(\rD_1)+f(\rD_2) \right)}_{\infty} \geq \frac{1}{2}\max\left(\norm{f(\rD_1)}_{\infty},\norm{f(\rD_2)}_{\infty}\right)
\end{align*}
which implies that $\norm{f(\rD_1)}_2^2, \norm{f(\rD_2)}_2^2 \leq 4/N$. Similarly,
\begin{align*}
    \frac{2}{\lambda N} \geq \norm{f\left(\frac{1}{2}(\p_1+\p_2)\right)}_{\infty} = \norm{\frac{1}{2}\left(f(\p_1) + f(\p_2)\right)}_{\infty} \geq \frac{1}{2}\max\left(\norm{f(\p_1)}_{\infty},\norm{f(\p_2)}_{\infty}\right)
\end{align*}
and so $\norm{f(\p_1)}_2^2, \norm{f(\p_2)}_2^2 \leq 4/\lambda N$. Since $f$ preserves $\ell_1$-distance, it suffices to test closeness of $f(\p_1)$ and $f(\p_2)$. Moreover, we can simulate sample and verification query access to the mixtures 
\[
f(\rD_b) = f(\lambda \p_b + (1-\lambda)\q_b) = \lambda f(\p_b) + (1-\lambda)f(\q_b)
\]
for both $b \in \{1,2\}$. Next, we apply another random map which allows us to assume that the mixtures have reasonable mass on every element, which will simplify our later analysis. This reduction is captured in the following lemma. 

\begin{lemma} \label{lem:unifier} Suppose we have sample and verification query access to arbitrary mixtures $\rD_1 = \lambda \p_1 + (1-\lambda)\q_1$ and $\rD_2 = \lambda \p_2 + (1-\lambda)\q_2$ over $[N]$. Let $g \colon [N] \to [N]$ be a random map such that for every $j \in [N]$, with probability $1/2$, $g(j) = j$, and with probability $1/2$, $g(j)$ is uniformly distributed in $[N]$. Then the following hold.
\begin{enumerate}
    \item $\norm{g(\p_1) - g(\p_2)}_1 = \frac{1}{2} \norm{\p_1 - \p_2}_1$.
    \item We can simulate sample and verification query access to $g(\rD_b) = \lambda g(\p_b) + (1-\lambda)g(\q_b)$ for both $b \in \{1,2\}$. 
\end{enumerate}
\end{lemma}

In particular, the definition of the random map $g$ from \Cref{lem:unifier} guarantees that $g(\cD)[j] \geq 1/2N$ for every $j \in [N]$, and $\norm{g(\cD)}_{\infty} \leq \norm{\cD}_{\infty}$ for any distribution $\cD$. Moreover, $g$ preserves $\ell_1$-distances to a factor of $2$. Thus, it now suffices to establish the following theorem, which we prove in \Cref{sec:closeness-post-flat}. (This constitutes the main effort of the proof.)

%\begin{restatable}{theorem}{closenessUB} \emph{(Closeness Testing Post-Flattening.)} \label{thm:closenessUB-high-post-flattening}
%Suppose we have sample and verification query access to two mixtures $\rD_1 = \lambda \p_1 + (1-\lambda)\q_1$, $\rD_2 = \lambda \p_2 + (1-\lambda)\q_2$ with the promise\footnote{The upper bounds can be guaranteed by flattening \Cref{thm:strong-mixture-flat} using $O(n \log n)$ samples. The lower bound can guaranteed by the reduction \Cref{lem:unifier}.} that $\norm{\rD_b}_{\infty} \leq 4/n$, $\norm{\p_b}_{\infty} \leq 4/\lambda n$, and $\min_{j \in [n]} \rD_b[j] \geq 1/2n$ for both $b \in \{1,2\}$. Then, there is an algorithm that distinguishes between the case of $\p_1 = \p_2$ and $\norm{\p_1 - \p_2}_1 \geq \eps$ with probability $2/3$ using 
%\[
%O\left(\frac{n \log^2(1/\eps \lambda)}{\eps^2\lambda}\right) ~\text{ samples and }~ O\left(\frac{\log^3 (1/\eps\lambda)}{\eps^2 \lambda}\right) ~\text{ queries.}
%\]
%\end{restatable}

\begin{theorem} [Closeness Testing Post-Flattening] \label{thm:closenessUB-high-post-flattening}
Suppose we have sample and verification query access to two mixtures $\rD_1 = \lambda \p_1 + (1-\lambda)\q_1$, $\rD_2 = \lambda \p_2 + (1-\lambda)\q_2$ with the promise\footnote{The upper bounds can be guaranteed by flattening \Cref{thm:strong-mixture-flat} using $O(n \log n)$ samples. The lower bound can guaranteed by the reduction \Cref{lem:unifier}.} that $\norm{\rD_b}_{\infty} \leq 4/n$, $\norm{\p_b}_{\infty} \leq 4/\lambda n$, and $\min_{j \in [n]} \rD_b[j] \geq 1/2n$ for both $b \in \{1,2\}$. Then, there is an algorithm that distinguishes between the case of $\p_1 = \p_2$ and $\norm{\p_1 - \p_2}_1 \geq \eps$ with probability $2/3$ using 
\[
O\left(\frac{n \log^2(1/\eps \lambda)}{\eps^2\lambda}\right) ~\text{ samples and }~ O\left(\frac{\log^3 (1/\eps\lambda)}{\eps^2 \lambda}\right) ~\text{ queries.}
\]
\end{theorem}

In particular, we invoke \Cref{thm:closenessUB-high-post-flattening} using the mixtures $g(f(\rD_b))$ as input and with distance parameter $\eps/2$. The query bound follows by our bound on $N$, and this completes the proof.  \end{proof}

\paragraph{Proof of \Cref{lem:unifier}.} We first observe that item (2) holds. Using the definition of $g$, we can clearly simulate samples to $g(\rD_b)$. Next, suppose we have a sample $x \sim \rD_b$ and we want to simulate a verification query on $g(x)$. We can simply query $x$ and use the oracle's response for $g(x)$.

We now prove item (1). Observe that for any $j \in [N]$, we have $g(\p_b)[j] = \p_b[j]/2 + (1/2N)$. Thus, 
\begin{align}
    \norm{g(\p_1) - g(\p_2)}_1 &= \sum_j |g(\p_1)[j] - g(\p_2)[j]| \nonumber \\
    &= \frac{1}{2}\sum_j |(\p_1[j] + 1/N)  - (\p_2[j] + 1/N)| = \frac{1}{2} \norm{\p_1 - \p_2}_1 \nonumber
\end{align}
and this completes the proof. 

\subsection{Closeness Testing after Flattening: Proof of \texorpdfstring{\Cref{thm:closenessUB-high-post-flattening}}{Theorem 7.6}} \label{sec:closeness-post-flat}

%\begin{mdframed}[backgroundcolor=blue!5]

%\end{mdframed}

The algorithm is defined as follows. Throughout the algorithm's description, let $\gamma = \eps \lambda / 10000$. Using \Cref{lemma:lambda-est-2}, we first obtain an estimate $\hat{\lambda}$ such that $\hat{\lambda} \in [\lambda/2,\lambda]$ using $O(1/\lambda)$ samples and queries which we use in place of $\lambda$ throughout the algorithm's description. We have written italicized comments in teal to aid the reader's intuition for how the analysis coincides with the tester.

\begin{mdframed}
\begin{enumerate}   
    \item \textbf{(Sampling Phase.)} Draw multi-sets $Z^{(1)},Z^{(2)}$ of $m = O(\frac{n \log^2 (1/\eps \lambda )}{ \eps^2 \lambda})$ samples from $\rD_1,\rD_2$. Let $Z^{(1)}_i,Z^{(2)}_i$ denote the sets of samples occurring on element $i$. For $b \in \{1,2\}$, let $P_{b}[i],Q_{b}[i]$ and $R_{b}[i] = P_{b}[i] + Q_{b}[i]$ denote the number of samples in $Z^{(b)}_i$ generated by $\p_b,\q_b$ and $\rD_b$, respectively. 
    \item \textbf{(Selecting Test Elements.)} For each $j \leq \ceil{\log (8/\eps \lambda)} =: k$, let $X_{j}$ be a set of $|X_{j}| = 10 \cdot 2^j$ samples drawn from $U_n$. Let $X = X_1 \cup \cdots \cup X_k$. \Comment{\emph{(By Markov's, with high probability all $i \in X_j$ satisfy $\p_b[i] \leq \widetilde{O}(\frac{2^j}{n})$.)}}
    \item \textbf{(Coin Definition.)} For each $i \in X$ and $b \in \{1,2\}$, let $\rho_b[i] := \frac{P_b[i]}{R_b[i]}$. 
    \Comment{\emph{(We can simulate a $\mathrm{Bern}(\rho_b[i])$ trial by querying a uniform random element of $Z_i^{(b)}$. By the learning error bound, we will be able to guarantee that $|\frac{P_b[i]}{\lambda m} - \p_b[i]| \ll \frac{\eps 2^j}{n}$. By the learning error bound and the fact that $\frac{1}{2n} \leq \rD_b[i] \leq \frac{4}{n}$, we will also have $\frac{R_b[i]}{m} = \Theta(\frac{1}{n})$. This gives us a bound of $\widetilde{O}(\lambda 2^j)$ on $\rho_b[i]$, which helps us bound the query complexity of estimating $\rho_b[i]$.)}} 
    \item \textbf{(Query Phase.)} For $j \leq k$ and $i \in X_{j}$:
    \begin{enumerate}
        \item Obtain an estimate $\tilde{\rho}_b[i]$ such that  $|\tilde{\rho}_b[i] - \rho_b[i]| \leq O(\lambda \eps 2^j)$ with probability $1-\gamma$ by invoking \Cref{lem:bias-add}. \Comment{\emph{(This uses $\widetilde{O}(\frac{\lambda 2^j}{(\lambda \eps 2^j)^2}) = \widetilde{O}(\frac{1}{\lambda 2^j \eps^2})$ queries. Note that this implies $|\frac{R_b[i]}{\lambda m} \tilde{\rho}_b[i] - \frac{P_b[i]}{\lambda m}| = \frac{R_b[i]}{\lambda m}|\tilde{\rho}_b[i] - \rho_b[i]| \leq \frac{R_b[i]}{\lambda m} \cdot O(\lambda \eps 2^j) \leq O(\frac{\eps 2^j}{n})$. This last inequality used the learning error bound on the empirical $\rD_b[i]$ estimate $\frac{R_b[i]}{m}$ and the bound $\norm{\rD_b}_{\infty} \leq 4/n$.)}} 
        \item If $|\frac{R_b[i]}{\lambda m} \tilde{\rho}_1[i] - \frac{R_b[i]}{\lambda m} \tilde{\rho}_2[i]| > \frac{ \eps 2^j}{8n}$, then \textbf{reject}.
    \end{enumerate} 
    \item \textbf{Accept}.
\end{enumerate}
\end{mdframed}

\begin{claim} \label{clm:good-props} With probability $0.999$, for every $j \leq k$, $i \in X_j$, and $b \in \{1,2\}$, we have $\p_b[i] \leq \frac{1000 \cdot 2^j \log (1/\eps \lambda)}{n}$. \end{claim}

\begin{proof} Note that $\EX_{i \sim U_n}[\p_b[i]] = 1/n$ and so by Markov's $\Pr_{i \sim U_n}[\p_b[i] > \frac{t \cdot 2^j}{n}] \leq \frac{1}{t \cdot 2^j}$. Therefore, setting $t := 1000\log(1/\eps\lambda)$, by a union bound we have $\p_b[i] \leq \frac{1000 \cdot2^j \log (1/\eps \lambda)}{n}$ for all $b \in \{1,2\}, j \leq k$, and $i \in X_j$ with probability $0.999$. \end{proof}

Therefore, throughout the proof we will assume the above holds for the elements in $X_1,\ldots,X_k$. Next, we condition on the samples obtained in step (1) yielding small learning error. Here, we crucially use the bound from \Cref{clm:good-props}. %$\norm{\p_1}_{\infty},\norm{\p_2}_{\infty} \leq 4/\lambda n$.

\begin{claim} \label{clm:learning-error-bound} With probability $0.996$, for every $j \leq k$, $i \in X_j$, and $b \in \{1,2\}$, we have $\left|\frac{R_b[i]}{m} - \rD_b[i]\right| \leq \frac{\eps}{200n}$ and $\left|\frac{P_b[i]}{\lambda m} - \p_b[i] \right| \leq \frac{\eps 2^j}{200n}$. \end{claim}

%\begin{claim} \label{clm:learning-error-bound} With probability $0.996$ we have $\left|\frac{R_b[i]}{m} - \rD_b[i]\right| \leq \frac{\eps \lambda}{100n}$ and $\left|\frac{P_b[i]}{\lambda m} - \p_b[i] \right| \leq \frac{\eps}{100n}$ for all $i \in X, b \in \{1,2\}$. \end{claim}

We defer the proof of this claim to \Cref{sec:deferred-closeness} and now condition on the bounds from \Cref{clm:learning-error-bound}. Note that combining \Cref{clm:learning-error-bound} with the fact that $1/2n \leq \rD_b[i] \leq 4/n$, we have $1/3n \leq R_b[i]/m \leq 5/n$.

\paragraph{Query complexity.} We use the above properties to bound the cost of each bias estimation in line (4a). Since $R_b[i]/m \geq 1/3n$ and $\p_b[i] \leq O(\frac{2^j \log (1/\eps \lambda)}{n})$, and using \Cref{clm:learning-error-bound}, we have 
\[
\rho_b[i] = \frac{P_b[i]}{R_b[i]} = \lambda \cdot \frac{P_b[i]/\lambda m}{R_b[i]/m} \leq O\left(\lambda \cdot \frac{\p_b[i] + \eps 2^j/n}{1/n}\right) \leq O(\lambda 2^j \log(1/\eps \lambda))
\]
and thus, by \Cref{lem:bias-add}, estimating $\rho_b[i]$ to additive error $O(\lambda \eps 2^j)$ in line (4a) has query complexity
\[
O\left(\frac{\lambda 2^j \log(1/\eps \lambda) \log(1/\gamma)}{(\lambda \eps 2^j)^2}\right) = O\left(\frac{\log^2(1/\eps\lambda)}{\lambda 2^j \eps^2}\right) \text{.}
\]

In total, the query complexity of the tester is bounded by (using $|X_j| \leq O(2^j)$)
\[
\sum_{j \leq \log(1/\eps \lambda)} |X_j| \cdot O\left(\frac{\log^2(1/\eps \lambda)}{\lambda 2^j \eps^2}\right) = \sum_{j \leq \log(1/\eps \lambda)} O\left(2^j  \cdot\frac{\log^2(1/\eps \lambda)}{\lambda 2^j \eps^2}\right) = O\left(\frac{\log^3(1/\eps \lambda)}{\lambda \eps^2}\right) \text{.}
\]

\paragraph{Proof of correctness.} By our choice of $\gamma \ll \eps \lambda$ and the fact that $|X_1 \cup \cdots \cup X_k| \leq \frac{20}{\eps \lambda}$, we assume by a union bound that the bias estimation for each $i \in X_j$ yields $\widetilde{\rho}_b[i]$ such that $|\widetilde{\rho}_b[i] - \rho_b[i]| \leq \lambda \eps 2^j/1000$, so that 
\begin{align} \label{eq:bias-estimate}
    \left|\frac{R_b[i]}{\lambda m} \tilde{\rho}_b[i] - \frac{P_b[i]}{\lambda m}\right| = \frac{R_b[i]}{\lambda m}|\tilde{\rho}_b[i] - \rho_b[i]| \leq \frac{R_b[i]}{\lambda m} \cdot \frac{\lambda \eps 2^j}{1000} \leq \frac{\eps 2^j}{200n} \text{.}
\end{align}
In the last inequality we used the bound $R_b[i]/m \leq 5/n$.
%\[
%\frac{R_b[i]}{m} \leq \rD_b[i] + \frac{\eps \lambda}{100 n} \leq \frac{3}{n} + \frac{\eps \lambda}{100 n} < \frac{4}{n}
%\]
%by \Cref{clm:learning-error-bound} and the fact that $\norm{\rD_b}_{\infty} \leq 3/n$.

\paragraph{Completeness.}  We now provide the proof of correctness. Suppose that $\p_1 = \p_2$. Then for every $i \in X$ we have $\p_1[i] = \p_2[i]$. By definition of the tester, when step (4b) is reached for some $i \in X_j$, for both $b \in \{1,2\}$ we have by the triangle inequality
\begin{align} \label{eq:C1}
    \left|\frac{R_b[i]}{\lambda m} \tilde{\rho}_b[i] - \p_b[i]\right| &\leq \left|\frac{R_b[i]}{\lambda m} \tilde{\rho}_b[i] - \frac{P_b[i]}{\lambda m}\right| + \left|\frac{P_b[i]}{\lambda m} - \p_b[i]\right| \leq \frac{\eps 2^j}{200n} + \frac{\eps 2^j}{200 n} \leq \frac{\eps 2^j}{100n}
\end{align}
where the first term is bounded using the bias estimation error bound \cref{eq:bias-estimate} and the second term is bounded by our conditioning on small error due to learning, \Cref{clm:learning-error-bound}. Therefore, by \cref{eq:C1} and another application of the triangle inequality we have
\begin{align}
    \left|\frac{R_1[i]}{\lambda m} \tilde{\rho}_1[i] - \frac{R_2[i]}{\lambda m}\tilde{\rho}_2[i] \right| &\leq \left|\frac{R_1[i]}{\lambda m} \tilde{\rho}_1[i] - \p_1[i]\right| + \left|\frac{R_2[i]}{\lambda m} \tilde{\rho}_2[i] - \p_2[i]\right| + |\p_1[i] - \p_2[i]| \leq \frac{\eps 2^j}{50 n}
\end{align}
and so the tester does not reject in step (4b). 

\paragraph{Soundness.} Now suppose that $\norm{\p_1 - \p_2}_1 > \eps$. The following is the main lemma for the analysis.

\begin{lemma} \label{lem:bucket-bias-closeness-strong} Recall that $k = \ceil{\log(8/\eps\lambda)}$ and for each $j \leq k$, $X_j$ is a set of $10 \cdot 2^j$ i.i.d. uniform samples from $[n]$. If $\norm{\p_1 - \p_2}_1 > \eps$, then 
\[
\Pr_{X_1,\ldots,X_k}\left[\exists j \leq k, i \in X_j \colon |\p_1[i] - \p_2[i]| > \frac{\eps 2^j}{4n} \right] \geq 0.999 \text{.}
\]
\end{lemma}

By \Cref{lem:bucket-bias-closeness-strong}, with probability $0.999$, for some $j$, the set $X_{j}$ obtained in line (1) of the tester contains some element $i$ such that $|\p_1[i] - \p_2[i]| > \eps2^j/4n$. Note that we again have that \cref{eq:C1} holds for both $b \in \{1,2\}$ and all $i \in X_j$ using the same argument from the completeness case. 
%\[
%\left|\frac{R_b[i]}{\lambda m} \tilde{\rho}_b[i] - \p_b[i]\right| \leq \frac{\eps 2^j}{100 n} \text{ for both } b\in \{1,2\} \text{ and all } i \in X_j
%\]
Then, by the triangle inequality we have
\begin{align}
    \left|\frac{R_1[i]}{\lambda m} \tilde{\rho}_1[i] - \frac{R_2[i]}{\lambda m}\tilde{\rho}_2[i] \right| &\geq |\p_1[i] - \p_2[i]| - \left(\left|\frac{R_1[i]}{\lambda m} \tilde{\rho}_1[i] - \p_1[i]\right| + \left|\frac{R_2[i]}{\lambda m} \tilde{\rho}_2[i] - \p_2[i]\right|  \right) \nonumber \\
    &\geq \frac{\eps 2^j}{4n} - \frac{\eps 2^j}{100n} - \frac{\eps 2^j}{100 n} > \frac{\eps 2^j}{8n}
\end{align}
which causes the tester to correctly reject in step (5b). This completes the proof of \Cref{thm:closenessUB-high-post-flattening}.

\subsection{Deferred Proofs} \label{sec:deferred-closeness}

\begin{proofof}{\Cref{lem:bucket-bias-closeness-strong}} We prove the lemma using a bucketing argument. Let 
\[
A = \{i \in [n] \colon |\p_1[i] - \p_2[i]| < \eps/2n\} ~\text{ and }~ B = \{i \in [n] \colon \eps/2n < |\p_1[i] - \p_2[i]| \leq 4/\lambda n\} \text{.}
\]
Recall that $\norm{\p_b}_{\infty} \leq 4/\lambda n$, which implies that $A,B$ partition the domain, $[n]$, since $|\p_1[i]-\p_2[i]| > 4/\lambda n$ would imply $\max(\p_1[i],\p_2[i]) > 4/\lambda n$, contradicting our bound on the $\infty$-norm. We further partition $B := B_1 \sqcup \cdots \sqcup B_k$ where $k = \ceil{\log (8/\eps\lambda)}$ and
\[
B_j = \{i \in B \colon 2^j (\eps/4n) \leq |\p_1[i] - \p_2[i]| < 2^{j+1} (\eps/4n)\}\text{ for } j \leq k\text{.}
\]
First, observe that
\[
\eps \leq \sum_i \big|\p_1[i] - \p_2[i] \big| \leq |A| \cdot \frac{\eps}{2n} + \sum_{i \in B} |\p_1[i] - \p_2[i]| \leq \frac{\eps}{2} +  \sum_{i \in B} |\p_1[i] - \p_2[i]| ~\implies~ \sum_{i \in B} |\p_1[i] - \p_2[i]| > \eps/2
\]
Now, using the bucketing, we have
\begin{align} \label{eq:bucketing-cases-closeness}
    \eps/2 < \sum_{j \leq k} \sum_{i \in B_j} |\p_1[i] - \p_2[i]| \leq \sum_{j \leq k} |B_j| 2^{j+1} (\eps/4n) ~\implies~ \sum_{j \leq k} |B_j| 2^j \geq n \text{.} \nonumber 
\end{align}
Observe that
\begin{align}
    \Pr_{X_j}[X_j \cap B_j = \emptyset] = \left(1-\frac{|B_j|}{n}\right)^{10 \cdot 2^j} \leq \exp\left(-10 \cdot \frac{|B_j| 2^j}{n}\right) \text{.} \nonumber 
\end{align}
Combining the previous two inequalities and using independence of the $X_j$-s, we obtain
\begin{align}
    \Pr[\forall j \leq k\colon X_j \cap B_j = \emptyset] &= \prod_{j \leq k} \Pr_{X_j}[X_j \cap B_j = \emptyset] \leq \exp\left(-\frac{10}{n} \sum_{j\leq k} |B_j|2^j \right) \leq e^{-10} < .0001 \text{.} \nonumber 
\end{align}
Finally, observe that this implies
\begin{align}
    \Pr_{X_1,\ldots,X_k}\left[\exists j \leq k, i \in X_j \colon |\p_1[i] - \p_2[i]| > \frac{2^j \eps}{4n} \right] &\geq \Pr[\exists j \leq k \colon X_j \cap B_j \neq \emptyset] \nonumber \\
    &= 1-\Pr[\forall j \leq k\colon X_j \cap B_j = \emptyset] \geq 0.999\text{.} \nonumber 
\end{align}
This completes the proof. \end{proofof}

\vspace{1em}

\begin{proofof}{\Cref{clm:learning-error-bound}} By \Cref{cor:learn}, for each $i \in [n]$, we have
\[
\Pr_{Z}\left[ \left|\frac{P_b[i]}{\lambda m} - \p_b[i] \right| \leq \sqrt{\frac{3 \p_b[i] \ln(2/\alpha)}{\lambda m}} \right] \geq 1-\alpha \text{.}
\]
Since, $X_j \leq O(2^j)$ where $j \leq k = O(\log(1/\eps \lambda))$, we can set $\alpha := 1/1000|X|$ and get that this inequality holds for all $i \in X_j$ by a union bound with probability $1-\frac{1}{1000k}$. Again by a union over all $j \leq k$, we get that this holds for all $j$ with probability $0.999$. Then, recall that from \Cref{clm:good-props}, for $i \in X_j$, we have $\p_b[i] \leq \frac{1000 2^j \log(1/\eps \lambda)}{n}$ for all $i \in X_j$. Then, using and $m = O(\frac{n \log^2(1/\eps \lambda)}{\eps^2 \lambda})$ (the sample complexity used in line (1) of the tester), we have
\begin{align}
    \left|\frac{P_b[i]}{\lambda m} - \p_b[i] \right| &\leq \sqrt{\frac{3 \p_b[i] \ln(2/\alpha)}{\lambda m}} \leq \sqrt{\frac{3000 \cdot 2^j \log^2(1/\eps\lambda)}{\lambda n \cdot m}} \leq \frac{\eps 2^j }{200 n}\text{.}
\end{align}
Note that last inequality holds as long as $m \geq \frac{Cn \log^2(1/\eps\lambda)}{\eps^2\lambda 2^j}$ for a sufficiently large constant $C$. This is of course maximized when $j=1$, and so this bound is satisfied for the choice of $m$ in the tester (line (1)).
Similarly, by \Cref{thm:learn} we have for each $i \in [n]$
\[
\Pr_{Z}\left[ \left|\frac{R_b[i]}{ m} - \rD_b[i] \right| \leq \sqrt{\frac{3 \rD_b[i] \ln(2/\alpha)}{m}} \right] \geq 1-\alpha
\]
and by a union bound with probability $0.999$ this holds for every $i \in X_1 \cup \cdots \cup X_k$. Then, using our bound $\norm{\rD_b}_{\infty} \leq \frac{4}{n}$ and $m \geq \frac{Cn \log^2(1/\eps\lambda)}{\eps^2\lambda}$ yields
\[
\left|\frac{R_b[i]}{ m} - \rD_b[i] \right| \leq \sqrt{\frac{3 \rD_b[i] \ln(2/\alpha)}{m}} \leq \sqrt{\frac{12 \ln(2/\alpha)}{n m}} \leq \frac{\eps \sqrt{\lambda}}{200n} \leq \frac{\eps}{200 n} \text{.}
\]
Taking a union bound over both inequalities for both $b \in \{1,2\}$ completes the proof. \end{proofof}

\section{Query-Optimal Uniformity Testing with Adversarial Contamination} \label{sec:highsample}

Note that the closeness tester established in \Cref{sec:closeness-UB-highsample} implies a uniformity tester with the same sample and query complexity which works against distributional contamination. In this section, we show that for uniformity testing we can achieve similar query complexity even against \emph{adversarial} contamination.
%In this section, we design a uniformity tester using $\widetilde{O}(n \cdot \eps^{-2} \lambda^{-1})$ samples and $\widetilde{O}(\eps^{-2} \lambda^{-1})$ queries. In particular, this query complexity is nearly optimal, even given unbounded samples, in light of our lower bound for bias estimation (\Cref{thm:bias-estimation-lb}). Moreover, our algorithm works in the following adversarial setting:
That is, with probability $\lambda$ the sample is drawn according to the target distribution $\p$, and with probability $1-\lambda$ the sample is chosen by an adversary. Moreover, all of the adversarial samples can be chosen after all of the $\p$-samples are fixed. We prove the following theorem.

\begin{theorem} \label{thm:n-sample-1-query-adversary} For all $\lambda, \eps \in (0,1)$, there is a $\lambda$-adversarially robust $\eps$-uniformity tester using 
\[
O\left(\frac{n \log (1/\eps)}{\eps^2\lambda} \right) \text{ samples and } O\left(\frac{\log^4(1/\eps)}{\eps^2\lambda}\right) \text{ verification queries.} 
\]
\end{theorem}

In the rest of the section, we use $\rD_{A,\p,\lambda}$ to denote an oracle, which when sampled, produces a sample from $\p$ with probability $\lambda$, and produces an adversarially chosen element of $[n]$ with probability $1-\lambda$.

%\hadley{Move content of this section to preliminaries?}
\subsection{Preliminaries}

Before proving \Cref{thm:n-sample-1-query-adversary}, we provide the needed claims for estimating coin bias. 

\begin{claim} \label{clm:bias-est-mult} Let $\rho \in (\beta,1)$ be an unknown coin probability, where $\beta$ is known. Then $O(\frac{\ln (1/\gamma)}{\beta \delta^2})$ i.i.d. $\mathrm{Bern}(\rho)$ trials suffice to attain a $(1\pm\delta)$-approximation of $\rho$ with probability $1-\gamma$. \end{claim}

\begin{proof} Consider taking $m$ i.i.d. trials $X_1,\dots,X_m$ and let $X = \sum_{j\leq m} X_j$ and $\mu := \EX[X] = \rho m$. The output of the algorithm will be $X/m$. By the Chernoff bound,
\[
\Pr[X/m \notin [(1\pm \delta)\rho m]] = \Pr[|X-\mu| > \delta \mu] \leq 2\exp(-\delta^2 (\rho m) / 3)
\]
and the RHS is at most $\gamma$ when $m = \frac{50\ln(1/\gamma)}{\beta \delta^2}$. \end{proof}

%\begin{claim} \label{clm:bias-est-add} Let $\rho \in (0,1)$ be an unknown coin probability. Then $O(\frac{1}{\alpha^2})$ i.i.d. $\mathrm{Bern}(\rho)$ trials suffice to attain a $\pm \delta$ additive approximation of $\rho$ with probability $99/100$. \end{claim}

\begin{claim} \label{clm:bias-est-small-case} Let $\rho \in (0,1)$ be an unknown coin probability, and let $\beta$ be known. Then $O(\frac{\ln(1/\gamma)}{\beta})$ i.i.d. $\mathrm{Bern}(\rho)$ trials suffice to distinguish $\rho \geq \beta$ from $\rho < \beta/8$ with probability $1-\gamma$. \end{claim}

\begin{proof} Let $X_1,\ldots,X_m$ be i.i.d. $\mathrm{Bern}(\rho)$ trials and let $X$ denote their sum. Then $\mu := \Exp[X] = \rho m$. We will conclude that $\rho = \beta$ if $X > \beta m/2$. Otherwise we will conclude that $\rho < \beta/8$. If $\rho \geq \beta$ it suffices to attain a constant factor approximation of $\rho$ using $O(\frac{\ln (1/\gamma)}{\beta})$ trials with \Cref{clm:bias-est-mult}. Now suppose $\rho < \beta/8$. We will use the following statement of the Chernoff bound: for any $\delta > 0$, we have
%\[
%\Pr[X > \mu + \delta\mu] \leq \left(\frac{e^{\delta}}{1+\delta)^{(1+\delta)}}\right)^{\mu} \leq \exp(-\delta^2 \mu / (2 + \delta))
%\]
\[
\Pr[X > \mu + \delta\mu] \leq \exp\left(-\frac{\delta^2 \mu}{2 + \delta}\right)
\]
In particular, if $\rho < \beta/8$, then
\[
\Pr[X > \beta m/2] \leq \Pr[X > \rho m + (3\beta/8\rho) \rho m] \leq \exp\left(-\frac{(3\beta/8\rho)^2 \cdot \rho m}{2 + (3\beta/8\rho)}\right) \text{.}
\]
%\[
%\Pr[X > \beta m/2] \leq \Pr[X > \rho m + (3\beta/8\rho) \rho m] \leq \left(\frac{e^{3\beta/8\rho}}{1+(3\beta/8\rho))^{(1+(3\beta/8\rho))}}\right)^{\rho m}
%\]
Since $\rho < \beta/8$, we have $3\beta/8\rho > 3$ and so the RHS above is at most
\[
\exp\left(-\frac{(3\beta/8\rho)^2 \cdot \rho m}{2 \cdot(3\beta/8\rho)}\right) = \exp\left(-\frac{3\beta m}{16}\right) \leq \gamma
\]
for $m \geq c \ln(1/\gamma) /\beta$ for sufficiently large constant $c$. \end{proof}

\begin{claim} [$(\beta,\delta,\gamma)$-Bias Estimation] \label{clm:bias-est} Let $\rho \in (0,1)$ be an unknown coin probability, and let $\beta \in (0,1)$. Then there is an algorithm using $O(\frac{\ln(1/\gamma)}{\beta \delta^2})$ i.i.d. $\mathrm{Bern}(\rho)$ trials with the following guarantee.
\begin{enumerate}
    \item If $|\rho - \beta| < \delta\beta /2$, then the algorithm outputs "close" with probability at least $1-\gamma$.
    \item If $|\rho - \beta| > \delta \beta$, then the algorithm outputs "far" with probability $1-\gamma$.
\end{enumerate}
\end{claim}

\begin{proof} Take $m = O(\frac{\ln(1/\gamma)}{\beta \delta^2} + \frac{\ln(1/\gamma)}{\beta})$ samples where the first term is enough to attain a $(1 \pm \delta)$ multiplicative approximation of $\rho$ when $\rho > \beta/8$ using \Cref{clm:bias-est-mult} and the second term is enough to detect when $\rho < \beta/8$ \Cref{clm:bias-est-small-case} (this does not require any lower bound on $\rho$). In both cases the empirical estimate is enough to distinguish the two cases with probability $1-\gamma$. \end{proof}

Additionally, we will use the following corollary of \Cref{thm:learn} for adversarial mixtures.

\begin{corollary} \label{cor:learn-adversary} Let $X \subseteq [n]$ and let $\p \colon [n] \to [0,1]$ be an arbitrary distribution. Consider a multi-set $Z$ of $m = O(\frac{n \ln |X|}{\lambda \eps^2})$ samples drawn from $\rD_{A,\p,\lambda}$. Let $P[i]$ denote the number of samples which were drawn by $\p$ and which occurred on element $i$. Then, 
\[
\Pr_{Z}\left[ \forall i \in X \colon \left|\frac{P[i]}{m} - \lambda \p[i] \right| \leq \frac{\eps \lambda}{100} \sqrt{\frac{\p[i]}{n}} \right] \geq 0.99 \text{.}
\]
\end{corollary}

\begin{proof} Consider a distribution $\p_{\lambda} \colon [n] \cup \{A\} \to [0,1]$ where $\p_{\lambda}[i] = \lambda \p[i]$ and $\p_{\lambda}(A) = 1-\lambda$. For $i \in [n]$, the random variable $P[i]/m$ is equivalent to an $m$-sample empirical estimate of $\p_{\lambda}[i] = \lambda \p[i]$. Therefore, by \Cref{thm:learn} we have
\[
\Pr_{Z}\left[\left|\frac{P[i]}{m} - \lambda \p[i]\right| > \sqrt{\frac{3 \lambda \p[i] \ln (2/\alpha)}{m}} \right] \leq \alpha
\]
and so setting $m = \frac{C n \ln |X|}{\lambda \eps^2}$ for sufficiently large constant $C > 0$, $\alpha = 1/(100|X|)$, and taking a union bound completes the proof. \end{proof}

\subsection{Query-Optimal Uniformity Testing: Proof of \texorpdfstring{\Cref{thm:n-sample-1-query-adversary}}{Theorem 8.1}} \label{sec:uniformity-UB-highsample}

%\hadley{Go over this section and simplify notation}

The algorithm is defined as follows. Throughout the algorithm's description, let $\gamma = \eps^2/1000$. Using \Cref{lemma:lambda-est-2}, we first obtain an estimate $\hat{\lambda}$ such that $\hat{\lambda} \in [\lambda/2,\lambda]$ using $O(1/\lambda)$ samples and queries which we use in place of $\lambda$ throughout the algorithm's description.

\begin{mdframed}
\begin{enumerate}   
    \item \textbf{(Sampling Phase.)} Draw a multi-set $Z$ of $m = O(\frac{n \ln (1/\eps)}{\lambda\eps^2})$ samples from $\rD_{A,\p,\lambda}$. Let $Z_i$ denote the set of samples occurring on element $i$. Let $P[i]$ and $A[i]$ denote the number of such samples generated by $\p$ and by the adversary, respectively, and let $R[i] = P[i] + A[i]$. 
    \Comment{\emph{(Note: $P[i],A[i]$ are not known to the tester, but $R[i]$ is.)}}
    \item \textbf{(Selecting Test Elements.)} Draw a set $X_{\p}$ of $|X_{\p}| = \frac{100}{\eps}$ samples from $\p$. \Comment{\emph{(This can be done using $O(\frac{1}{\lambda \eps})$ samples and queries with high probability.)}} For each $j \leq \ceil{\log (4/\eps)} = :k$, draw a set $X_{U,j}$ of $|X_{U,j}| = 10 \cdot 2^j \log (1/\eps)$ samples drawn from $U_n$. \Comment{\emph{(This requires no queries.)}} Let $X = X_{\p} \cup X_{U,1} \cup \cdots \cup X_{U,k}$. 
    \item \textbf{(Coin Definition.)} For each $i \in X$, let $\rho[i] := \frac{P[i]}{R[i]}$. \Comment{\emph{(Using one query to a uniform random sample in $Z_i$, we have access to a $\rho[i]$-probability coin. Recall that the denominator $R[i]$ is known to the tester.)}} 
    \item \textbf{(Query Phase 1.)} For each $i \in X_{\p}$: 
    \begin{enumerate}
        \item Invoke \Cref{clm:bias-est-small-case} with $\beta = \frac{\lambda \eps}{100}$ and error probability $\gamma$ on $\rho[i]$ using queries to uniform random elements of $Z_i$. If it returns "$\rho[i] < \beta/8$", then skip the next step and continue the for-loop.
        \item Invoke \Cref{clm:bias-est-mult} with $\beta = \frac{\lambda\eps}{20}, \delta = \frac{1}{200}$ on $\rho[i]$ using queries to uniform random elements of $Z_i$ to obtain a multiplicative estimate $\widetilde{\rho}[i]$ satisfying $(1-\frac{1}{200})\rho[i] \leq \widetilde{\rho}[i] \leq (1+\frac{1}{200})\rho[i]$ with probability $1-\gamma$. If $\widetilde{\rho}[i] > \frac{3\lambda m}{2 R[i] n}$, then \textbf{reject}.
    \end{enumerate}
    
    \item \textbf{(Query Phase 2.)} For $j \leq k$ and $i \in X_{U,j}$:
    \begin{enumerate}
        \item If $\frac{R[i]}{m} > \frac{2^{j} \log(1/\eps)}{n}$, then skip the next step and continue the for-loop.
        \item Invoke $(\beta,\delta,\gamma)$-bias estimation (\Cref{clm:bias-est}) with $\beta = \frac{\lambda m}{R[i] n}, \delta = \frac{\eps 2^j}{100}$ on $\rho[i]$, using queries to uniform random elements of $Z_i$. If this outputs "far", then reject. Otherwise continue. 
    \end{enumerate}
    \item \textbf{Accept}.
\end{enumerate}
\end{mdframed}

\paragraph{Sample and query complexity.} All samples from $\rD_{A,\p,\lambda}$ are drawn in steps (1) and (2), which use $O(\frac{n \ln (1/\eps)}{\lambda\eps^2})$ and $O(\frac{1}{\lambda\eps})$ samples respectively. 

We now analyze the number of queries. Step (2) makes $O(\frac{1}{\lambda \eps})$ queries. By \Cref{clm:bias-est-small-case}, step (4a) uses $O(\frac{\ln (1/\eps)}{\lambda \eps})$ queries. By \Cref{clm:bias-est-mult}, step (4b) then also uses $O(\frac{\ln (1/\eps)}{\lambda \eps})$ queries. Thus, given that the set $X_{\p}$ drawn in step (2) contains $O(1/\eps)$ elements, the total number of queries made in step (4) is bounded by $O(\frac{\log(1/\eps)}{\lambda \eps^2})$.

We now analyze the number of queries made in step (5). Step (5a) makes no queries. Now note that by virtue of step (5a), step (5b) is only ever invoked with $\beta = \frac{\lambda m}{R[i] n} \geq \frac{\lambda}{2^{j+1} \log(1/\eps)}$. Thus, using \Cref{clm:bias-est} and $\delta = \frac{\eps 2^j}{100}$, the number of queries made by an iteration of step (5b) is at most 
\[
O\left(\frac{\ln(1/\gamma)}{\beta \delta^2}\right) = O\left(\frac{\log^2(1/\eps)}{2^j \lambda\eps^2}\right) \text{.}
\]
Now, since the set $X_{U,j}$ drawn in step (2) is of size $O(2^j \log (1/\eps))$, the total number of queries made across all iterations of step (5b) for a fixed $j$ is bounded by $O(\frac{\log^3(1/\eps)}{\lambda \eps^2})$. Then, summing over all $j \leq k = O(\log (1/\eps))$, we see that the query complexity of the tester is bounded by $O(\frac{\log^4(1/\eps)}{\lambda\eps^2})$, as claimed.

\paragraph{Tester analysis.} First, by \Cref{cor:learn}, since $|X| \leq O(1/\eps^2)$, we have
\begin{align} \label{eq:good-Z}
    \Pr_{Z,X}\left[\forall i \in X \colon \left|\frac{P[i]}{m} - \lambda \p[i]\right| \leq \frac{\lambda\eps}{100} \sqrt{ \frac{\p[i]}{n}}\right] \geq 0.99 \text{.}
\end{align}
We will denote the event inside the above probability by $\cE_{X,Z}$ and will henceforth assume this holds for $X$ and $Z$. Note that the tester does not have knowledge of what $P[i]/m$ is, but it is a well-defined quantity given a fixed $Z$. Now, note that using $\rho[i] := P[i]/R[i]$ we have
\begin{align} \label{eq:good-Z-2}
     \cE_{X,Z} ~\implies~ \forall i \in X \colon \left|\rho[i] - \frac{\lambda \p[i] m }{R[i]}\right| \leq \frac{\lambda m \eps}{100 R[i]} \sqrt{ \frac{\p[i]}{n}}\text{.}
\end{align}
Throughout the rest of the analysis we will crucially rely on the condition that the inequality in the RHS of \cref{eq:good-Z-2} holds.

\paragraph{Completeness.} Suppose that $\p$ is uniform. Then for every $i \in X$ we have $\p[i] = 1/n$, and conditioning on $\cE_{X,Z}$ implies
\begin{align} \label{eq:uniform-case}
\left|\rho[i] - \frac{\lambda m}{R[i] n}\right|\leq \frac{\lambda m \eps}{100 R[i] n} ~\implies~ \frac{\lambda m }{R[i] n} \left(1 - \frac{\eps}{100}\right) \leq \rho[i] \leq \frac{\lambda m}{R[i] n} \left(1 + \frac{\eps}{100}\right) \text{.}
\end{align}
Therefore by \Cref{clm:bias-est-mult}, the estimate $\widetilde{\rho}[i]$ produced in step (4b) will satisfy $\widetilde{\rho}[i] \leq \frac{3\lambda m}{2R[i] n}$ with probability at least $1-\gamma$. Similarly, each iteration of step (5b) will output "close" with probability at least $1-\gamma$. Thus, each of these steps reject with probability at most $\gamma$. Step (4b) is invoked $100/\eps$ times and step (5b) is invoked $10k \log (1/\eps) \sum_{j=1}^k 2^j \leq \frac{10 \log^2 (1/\eps)}{\eps} \leq \frac{100}{\eps^2}$. Thus, by a union bound, the probability that any of these steps reject is at most $\frac{200}{\eps^2} \cdot \gamma \leq \frac{1}{5}$. Therefore, the tester accepts in this case with probability at least $0.99 \cdot \frac{4}{5} > \frac{2}{3}$, as claimed. 

\paragraph{Soundness.} Now suppose that $\p$ is $\eps$-far from uniform, and let $B = \{i \in [n] \colon |\p[i] - 1/n| > \eps/2n\}$. Further partition $B := B_1 \sqcup \cdots \sqcup B_k$ where $k = \ceil{\log (4/\eps)}$ and
\[
B_j = \{i \in B \colon 2^j (\eps/4n) \leq |\p[i] - 1/n| < 2^{j+1} (\eps/4n)\}\text{ for } j < k\text{, and } B_{k} := \{i \in B \colon \p[i] \geq 2/n\} \text{.}
\]

The following is the main lemma for the analysis.

\begin{lemma} \label{lem:bucket-bias} If $\norm{\p - U_n}_1 > \eps$, then one of the two following items hold.
\begin{enumerate}
    \item $\Pr_{i \sim \p}\left[\p[i] > \frac{2}{n} \text{ and } R[i]/m \leq \frac{20}{\eps} \p[i]\right] \geq \frac{\eps}{8}$.
    \item For some $j^{\ast} < k$, $\Pr_{i \sim U_n}\left[i \in B_{j^{\ast}}\text{ and } R[i]/m \leq \frac{2^{j^{\ast}} \log(1/\eps)}{n}\right] \geq \frac{1}{2 \cdot 2^{j^{\ast}} \log (1/\eps)}$.
    %\item There exists $j^{\ast}$ such that $|B_{j^{\ast}}| > \frac{n}{2^{j^{\ast}}\log(1/\eps)}$.
\end{enumerate}
\end{lemma}

We defer the proof of \Cref{lem:bucket-bias} to \Cref{sec:pf-bucket-bias}. We now complete the tester analysis and the proof of \Cref{thm:n-sample-1-query-adversary} using \Cref{lem:bucket-bias}.

\paragraph{Case 1:} Suppose item (1) of \Cref{lem:bucket-bias} holds. We will argue that some iteration of step (4b) will reject with high probability. Suppose $i \in X$ satisfies $\p[i] > 2/n$ and $R[i]/m \leq \frac{20}{\eps} \p[i]$. This implies that
\[
\rho[i] \geq \frac{\lambda \p[i] m}{R[i]} \left(1 - \frac{\eps}{100} \sqrt{\frac{1}{\p[i] n}} \right) \geq \frac{\lambda \eps}{20} \left(1 - \frac{\eps}{100} \right) \geq \frac{\lambda \eps}{30}
\]
implying that such an $i$ will pass the check in step (4a) and proceed to step (4b), in which we obtain an estimate $\widetilde{\rho}[i]$ which by \Cref{clm:bias-est-mult} satisfies 
\[
\widetilde{\rho}[i] \geq \left(1-\frac{1}{200}\right) \rho[i] \geq \left(1-\frac{1}{200}\right) \frac{\lambda \p[i] m}{R[i]} \left(1-\frac{\eps}{100} \right) > \frac{3\lambda m}{2R[i] n} \text{~(since } \p[i] > 2/n\text{)}
\]
with probability at least $1-\gamma$, causing the tester to reject. 

Now, since we are assuming item (1) of \Cref{lem:bucket-bias} holds, a sample from $\p$ yields such an element with probability at least $\eps/8$. Therefore, the set $X_{\p}$ drawn in step (2) contains such an element with probability at least $1-(1-\eps/8)^{100/\eps} \geq 0.99$. Thus, the rejection probability in this case is at least $(0.99)^2 \cdot (1-\gamma) \gg 2/3$. 

\paragraph{Case 2.} Suppose item (2) of \Cref{lem:bucket-bias} holds. We will argue that some iteration of step (5b) with $j = j^{\ast}$ will reject with high probability. Suppose $i \in B_{j^{\ast}}$ and $R[i]/m \leq \frac{2^{j^{\ast}} \log (1/\eps)}{n}$. First, by the bound on $R[i]/m$, this $i$ passes the check in step (5a) and thus continues to step (5b). Now recall that $i \in B_{j^{\ast}}$ means that 
\[
\frac{\eps 2^{j^{\ast}}}{4n} \leq \left|\p[i] - \frac{1}{n}\right| \leq \frac{\eps 2^{j^{\ast} + 1}}{4n}
\]
and so we have
\begin{align} \label{eq:bias-rho}
    \left| \frac{\lambda \p[i] m}{R[i]} - \frac{\lambda m}{R[i] n} \right| = \frac{\lambda m}{R[i]} \left|\p[i] - \frac{1}{n}\right| \geq \frac{\lambda m}{R[i]} \cdot \frac{\eps 2^{j^{\ast}}}{4n}
\end{align}
and using our conditioning on $\cE_{Z,X}$ yields
\begin{align}
    \p[i] \leq \frac{1}{n}\left(1 + \eps 2^{j^{\ast}-1}{}\right) ~\implies~ \left|\rho[i] - \frac{\lambda \p[i] m}{R[i]}\right| \leq \frac{\lambda \eps m}{100R[i]} \sqrt{\frac{\p[i]}{n}} \leq \frac{\lambda \eps m}{100R[i] n }\sqrt{1+\eps2^{j^{\ast}-1}} \leq \frac{\lambda \eps m}{100 R[i] n } \text{.}
\end{align} 
Combining the above and using the triangle inequality yields
\begin{align}
    \left|\rho[i] - \frac{\lambda m}{R[i] n}\right| \geq \left| \frac{\lambda \p[i] m}{R[i]} - \frac{\lambda m}{R[i] n} \right| - \left|\rho[i] - \frac{\lambda \p[i] m}{R[i]}\right| \geq \frac{\lambda \eps \cdot m 2^{j^{\ast}}}{4R[i] n} - \frac{\lambda\eps m}{100 R[i] n } \geq \frac{\lambda \eps \cdot m 2^{j^{\ast}}}{10R[i] n}
\end{align}
and therefore the call to \Cref{clm:bias-est} in step (5b) (with $\beta = \frac{\lambda m}{R[i] n}$ and $\delta = \frac{\eps 2^{j^{\ast}}}{100}$) for this $i$ outputs "far" with probability at least $1-\gamma$. 

Finally, since we are assume item (2) of \Cref{lem:bucket-bias} holds, a uniform random sample from $[n]$ yields such an $i$ with probability at least $\frac{1}{2^{j^{\ast}+1} \log (1/\eps)}$. Thus, the set $X_{U,j^{\ast}}$ drawn in step (2) contains such an element $i$ with probability at least
\[
1 - \left(1 - \frac{1}{2^{j^{\ast}+1} \log (1/\eps)}\right)^{10 \cdot 2^{j^{\ast}} \log (1/\eps)} \geq 0.99
\]
Therefore, the overall rejection probability in this case is at least $(0.99)^2 \cdot (1-\gamma) \gg 2/3$. This completes the proof of \Cref{thm:n-sample-1-query-adversary}.

\subsection{Proof of \texorpdfstring{\Cref{lem:bucket-bias}}{Lemma 8.6}} \label{sec:pf-bucket-bias}

\begin{proof} First, observe that
\[
\eps \leq \sum_i \big|\p[i] - 1/n \big| \leq |\overline{B}| \cdot \eps / 2n + \sum_{i \in B} |\p[i] - 1/n| \leq \eps/2 +  \sum_{i \in B} |\p[i] - 1/n| ~\implies~ \sum_{i \in B} |\p[i] - 1/n| > \eps/2
\]
Now, using the bucketing, we have
\begin{align} \label{eq:bucketing-cases}
    \eps/2 < \sum_{i \in B} |\p[i] - 1/n| = \sum_{i \in B_k} |\p[i] - 1/n| + \sum_{i \in B_1 \cup \cdots \cup B_{k-1}} |\p[i] - 1/n| \leq \sum_{i \in B_k} \p[i] +  \sum_{j < k} |B_j| 2^j (\eps/4n)
\end{align}
\paragraph{Case 1:} Suppose the first term in the RHS of \cref{eq:bucketing-cases} is at least $\eps/4$. This implies 
\begin{align} \label{eq:p-large}
    \Pr_{i \sim \p}[\p[i] > 2/n] \geq \eps/4\text{.}
\end{align}
%Recall $K_i = \widetilde{p}_{\lambda}[i] + a[i]$. 
Since we are conditioning on $\cE_{X,Z}$, we have
\begin{align} \label{eq:p-tilde}
    \p[i] > 2/n ~\implies~ \frac{P[i]}{m} \leq \lambda \p[i] + \frac{\lambda \eps}{100}\sqrt{\frac{\p[i]}{n}} = \lambda \p[i] \left(1 + \frac{\eps}{100} \sqrt{\frac{1}{\p[i]n}}\right) \leq 1.01 \cdot\lambda \p[i]
\end{align}
and so to show that item (1) holds, it suffices to show that $\Pr_{i \sim \p}[\p[i] > 2/n \text{ and } A[i] \leq \frac{16}{\lambda\eps} P[i]] \geq \eps/8$. We now show that this holds. To start, note that
\begin{align} \label{eq:case1-and}
    \Pr_{i \sim \p}\left[\p[i] > 2/n \text{ and } A[i] \leq \frac{16}{\lambda\eps} \cdot P[i]\right] &= \Pr_{i \sim \p}[\p[i] > 2/n] \cdot \Pr_{i \sim \p}\left[A[i] \leq \frac{16}{\lambda\eps} \cdot P[i] ~\Big|~ \p[i] > 2/n\right] \nonumber \\
    &\geq \frac{\eps}{4} \cdot \left(1 - \Pr_{i \sim \p} \left[\frac{A[i]}{P[i]} > \frac{16}{\lambda\eps} ~\Big|~ \p[i] > 2/n\right]\right)
\end{align}
where we used \cref{eq:p-large}. Now, we have
\begin{align}  \label{eq:case1-condition}
    \EX_{i \sim \p} \left[\frac{A[i]}{P[i]} ~\Big|~ \p[i] > 2/n\right] &= \sum_{i \leq n} \frac{A[i]}{P[i]} \cdot \Pr_{j \sim \p} [j = i ~|~ \p[i] > 2/n] \nonumber \\
    &= \sum_{i \leq n} \frac{A[i]}{P[i]} \cdot \frac{\Pr_{j \sim \p} [j = i \wedge \p[i] > 2/n]}{\Pr_{j \sim \p}[\p[j] > 2/n]}  \nonumber \\
    &\leq \frac{4}{\eps} \cdot \sum_{i \leq n} \frac{A[i]}{P[i]} \cdot \p[i] \cdot \mathbf{1}(\p[i] > 2/n) \nonumber \\
    &\leq \frac{4}{\eps} \cdot \max_{i \colon \p[i] > 2/n} \frac{\p[i]}{P[i]} \cdot \sum_{i \leq n} A[i] \leq \frac{4}{\eps} \cdot \max_{i \colon \p[i] > 2/n} \frac{\p[i]}{P[i]/m}
\end{align}
where in the last step we used $\sum_{i \leq n} A[i] \leq m$. Now, since $\cE_{X,Z}$ holds, if $\p[i] > 2/n$, then we have
\[
\frac{\p[i]}{P[i]/m} = \frac{\p[i]}{\lambda \p[i] + ((P[i]/m)-\lambda \p[i])} \leq \frac{\p[i]}{\lambda \p[i] - \frac{\lambda \eps}{100}\sqrt{\frac{\p[i]}{n}}} = \frac{1}{\lambda\left(1 - \frac{\eps}{100}\sqrt{\frac{1}{\p[i]n}}\right)} \leq \frac{2}{\lambda}
\]
and so plugging this back into \cref{eq:case1-condition} yields
\begin{align} \label{eq:case1-condition2}
    \EX_{i \sim \p} \left[\frac{A[i]}{P[i]} ~\Big|~ \p[i] > 2/n\right] \leq \frac{8}{\lambda \eps} ~\implies \Pr_{i \sim \p}\left[ \frac{A[i]}{P[i]} > \frac{16}{\lambda\eps} ~\Big|~ \p[i] > 2/n\right] \leq \frac{1}{2} ~\text{(by Markov's)}\text{.}
\end{align}
Returning now to \cref{eq:case1-and}, we have
\[
\Pr_{i \sim \p}\left[\p[i] > 2/n \text{ and } A[i] \leq \frac{16}{\lambda \eps} \cdot P[i]\right] \geq \frac{\eps}{8}
\]
as claimed.

\paragraph{Case 2:}  Suppose the first term in the RHS of \cref{eq:bucketing-cases} is at most $\eps/4$. Then, we have
\[
\eps/4 < \sum_{j < k} |B_j| 2^j (\eps/4n) ~\implies~ \sum_{j < k} |B_j| 2^j > n
\]
and so there exists $j^{\ast} < k$ where 
\begin{align} \label{eq:j*}
|B_{j^{\ast}}| > \frac{n}{2^{j^{\ast}}\log(1/\eps)} ~\Longleftrightarrow~ \Pr_{i \sim U_n} \left[i \in B_{j^{\ast}}\right] > \frac{1}{2^{j^{\ast}}\log(1/\eps)}\text{.}
\end{align}
Next, %recall that $K_i = \frac{P[i] + A[i]}{m}$ where $m$ is the number of samples and $P[i],A[i]$ are the number of occurrences of $i$ due to $\p$ and the adversary, respectively. Thus, clearly, 
by definition we have $\sum_{i \leq n} R[i] = m$ and so $\EX_{i \sim U_n}[R[i]/m] = 1/n$. Thus, Markov's inequality implies
\[
    \Pr_{i \sim U_n}\left[R[i]/m > \frac{2 \cdot 2^{j^{\ast}} \log (1/\eps)}{n}\right] \leq \frac{1}{2 \cdot 2^{j^{\ast}} \log (1/\eps)}
\]
and so combining with \cref{eq:j*} using a union bound obtains
\[
\Pr_{i \sim U_n}\left[i \in B_{j^{\ast}} \text{ and } R[i]/m \leq \frac{2 \cdot 2^{j^{\ast}} \log (1/\eps)}{n}\right] \geq \frac{1}{2 \cdot 2^{j^{\ast}} \log (1/\eps)}
\]
and thus item (2) of the claim holds. This completes the proof. \end{proof}

%\newpage

\section{Uniformity Testing with Complete Mixture Knowledge} \label{sec:uniformity-UB-MK}

With \Cref{thm:simple-uniformity-alg} we have obtained an algorithm with a smooth sample-query trade-off, using $O(m)$ samples and $O(n / m)$ queries.
In this section, we show that when the algorithm is given perfect knowledge of the mixture $\rD$ there is an algorithm with near optimal sample and query complexity with respect to the domain size $n$.
\Cref{thm:dist-contamination-lb-eps} shows that knowledge of the mixture is necessary to obtain this sample and query complexity.
Note that while $\rD$ is known to the algorithm initially, the contamination parameter $\lambda$ remains unknown.

\begin{theorem}[Formal \Cref{thm:uniformity-alg-mixture-informal}]
    \label{thm:dist-mix-knowledge-alg}
    Given an explicit description of the mixture $\rD := \lambda \p + (1 - \lambda) \q$, there is a $\lambda$-distributionally robust $\eps$-uniformity tester with  $\bigtO{\frac{\sqrt{n}}{\eps^{9/2} \lambda^{7/2}} + \frac{1}{\eps^{12} \lambda^{6}}}$ samples and $\bigO{\frac{\log^{7} n}{\alpha^{6} \lambda^{8}} + \frac{\log^{12} n}{\alpha^{12} \lambda^{6}}}$ queries.
\end{theorem}

We note that our dependence on $n$ is optimal in both sample and query complexity up to polylogarithmic factors.
We leave it as future work to improve the dependence on $\eps, \lambda$.

\paragraph{Algorithm Overview}
We begin with a high level overview of our algorithm.
Recall the fundamental sub-routine of \Cref{lemma:simple-t-query-alg} which tests uniformity of $\p$ using $\bigO{\frac{nb}{\lambda^2 \eps^{4}}}$ queries where $\norm{\rD}_{2}^{2} \leq b$.
A slight generalization of this algorithm uses $\bigO{\frac{b \cdot \polylog(n)}{\poly(\eps \lambda)}}$ queries if we are instead given the promise $\norm{\rD}_{2}^{2} \leq b \norm{\p}_{2}^{2}$.

To see this, consider the case where $\p$ is near uniform, i.e. $\norm{\p}_{2}^{2} \leq \frac{O(1)}{n}$.
Then, we have $\norm{\rD}_{2}^{2} \leq \frac{O(b)}{n}$ so that \Cref{lemma:simple-t-query-alg} yields $\bigO{\frac{b}{\lambda^2 \eps^{4}}}$ queries.
On the other hand, when $\p$ has large $\ell_{2}$-norm, we hope to distinguish the case $\norm{\p}_{2}^{2} \geq \alpha$ from the case $\norm{\p}_{2}^{2} \leq \frac{\alpha}{2}$.
In full generality, we need to repeat this by halving $\alpha$ from $O(1)$ to $\frac{1}{n}$, and so we can also assume without loss of generality that for example, $\norm{\p}_{2}^{2} \leq 2 \alpha$.
Towards this goal, we want to estimate the number of $\p$-collisions up to accuracy $\ll \lambda^2 m^2 \alpha$ among the total $m^2 \norm{\rD}_{2}^{2} = O(m^2 b \alpha)$ collisions.
Using \Cref{lem:bias-add}, the analysis of \Cref{lemma:simple-t-query-alg} shows that this can be done with query complexity
\begin{equation*}
    \bigO{\left( \frac{m^2 b \alpha}{\lambda^2 m^2 \alpha} \right)^{2} \cdot \frac{\lambda^2 m^2 \alpha}{m^2 b \alpha}} = \bigO{\frac{b}{\lambda^2}} \text{.}
\end{equation*}

With this generalization in hand, we can test the uniformity of $\p$ by partitioning the domain such that we can invoke the above algorithm on each part.
We thus partition the domain into parts where $2^{-k} \leq \rD[i] \leq 2^{-k+1}$ for all $k \leq \log n$.
If $\p$ is $\eps_0$-far from uniform then it must be $\eps := \eps_0/\log n$-far from uniform on at least one part.
It is easy to see that $\norm{\rD}_{2}^{2} \leq 2^{-k+1}$ on this part of the domain.
To apply the above algorithm, it remains to show that $\p$ if $\eps$-far from uniform on this part, $\norm{\p}_{2}^{2}$ is large.
Since (1) $\sum_{i} \p[i] \geq \eps$ is large, (2) $\p[i] \leq \frac{1}{\lambda} \rD[i] \leq \frac{2^{-k+1}}{\lambda}$ is capped for each individual $i$, and (3) there are at most $2^{k}$ buckets $i$ in each part, we can show that there are many large $\p[i]$ and in particular obtain a bound $\norm{\p}_{2}^{2} \geq \poly(\eps \lambda) 2^{-k}$.
In particular, $\norm{\rD}_{2}^{2} \leq \frac{\norm{\p}_{2}^{2}}{\poly(\eps \lambda)}$ allows us to invoke the above algorithm whose query complexity depends only polylogarithmically on $n$, as $\eps = \eps/\log n$.

\begin{proof}
    % Let $\rD = \lambda \p + (1 - \lambda) \q$.
    % As in the standard setting, we begin by obtaining a lower bound on $\lambda$ with $O(1/\lambda)$ samples and queries with \Cref{fact:lambda-est}.
    % Let $\lambda$ be the obtained estimate that satisfies $\lambda/40000 \leq \lambda \leq \lambda^*$.
    Assume without loss of generality that $\lambda \leq \frac{1}{2}$.
    We begin by obtaining an estimate $\hat{\lambda}$ such that $\frac{2}{3} \lambda \leq \hat{\lambda} \leq \lambda$ via \Cref{lemma:lambda-est-2}.

    For any distribution $\rD$, we define $S^{\rD}_{i} = \set{i \given 2^{-i} < \rD[i] \leq 2^{-i + 1}}$ for $1 \leq i \leq \log n$ and $S_{\log n + 1}^{\rD} = \set{i \given \rD[i] \leq \frac{1}{n}}$.
    Next, define $\overline{S}^{\rD}_{i} = \bigcup_{j \geq i} S_{j}^{\rD}$ and $S^{\rD}_{[i:j]} = \bigcup_{i \leq k \leq j} S_{k}^{\rD}$.
    Note that all these sets (without any samples or queries) can be constructed given an explicit description of $\rD$.
    The construction of these sets is crucially where we exploit mixture knowledge, and in fact without mixture knowledge these sets cannot be constructed without observing many samples.
    We will make heavy use of the following sample bound.
    \begin{claim}
        \label{claim:sample-bound-mix}
        For all $m \gg \log(1/\delta)$, if $m/\hat{\lambda}$ samples are taken from $\rD$, then with probability at least $1 - \delta/10$, $\frac{\lambda m}{2 \hat{\lambda}} \leq m_{\p} \leq \frac{2 \lambda m}{\hat{\lambda}}$.
        In particular, $\frac{m}{2} \leq m_{\p} \leq 3m$.
    \end{claim}
    
    \begin{proof}
        Note that $\EX[m_{\p}] = m \frac{\lambda}{\hat{\lambda}} \geq m$.
        For large enough $C$, a standard concentration bound shows that $\frac{m \lambda}{2 \hat{\lambda}} \leq m_{\p} \leq 2 m \frac{\lambda}{\hat{\lambda}}$ with high probability whenever $m \gg \log(1/\delta)$.
        Then, note that $1 \leq \frac{\lambda}{\hat{\lambda}} \leq \frac{3}{2}$.
    \end{proof}

    \subsection{Testing Distributions with Heavy Buckets}

    We begin with a subroutine to distinguish $\p$ with large $\ell_2$-norm.
    The algorithm at a high level is similar to \Cref{lemma:simple-t-query-alg} (i.e. randomly query collisions to estimate the number of $\p$-collisions on $S$).
    
    \begin{lemma}
        \label{lemma:dist-mix-knowledge-sub-alg}
        Let $1 \leq k < \log n + 3 \log \alpha + \log \hat{\lambda} - 3$ and $\alpha, \delta > 0$.
        There is an algorithm making $\bigO{\frac{n^{1/2} \log(1/\delta)}{\alpha^{9/2} \lambda^{9/2}}}$ samples and $\bigO{\frac{\log(1/\delta)}{\alpha^{6} \lambda^{8}}}$ queries satisfying the following guarantees:
        \begin{enumerate}
            \item (Completeness) if $\p$ is uniform, then the algorithm accepts with probability at least $1 - \delta$.
            \item (Soundness) if $\alpha < \sum_{i \in S_{k}^{\rD}, \p[i] > \frac{1}{n}} \p[i] - \frac{1}{n}$, then the algorithm rejects with probability at least $1 - \delta$.
        \end{enumerate}
    \end{lemma}

    \begin{proof} 
        Fix $m = \frac{C n^{1/2} \log(1/\delta)}{\alpha^{9/2} \hat{\lambda}^{7/2}} \geq \frac{C n^{1/2} \log(1/\delta)}{\alpha^{9/2} \lambda^{7/2}}$ for some sufficiently large constant $C$ and take $m/\hat{\lambda}$ samples.
        As before, let $m_{\p}, m_{\q}$ denote the number of samples drawn from $\p, \q$.
        Let $c(\p)$ denote the number of $\p$-collisions among the $m$ samples.
        For any subset $S \subset [n]$, let $c(\p; S)$ denote the number of $\p$-collisions among the observed samples within buckets in $S$.
        Analogously define $c(\q), c(\rD), c(\q; S), c(\rD; S)$.

        We prove some properties of the distributions $\p, \q$ on the buckets in $S_{k}^{\rD}$.
        We begin with a lower bound on the $2$-norm of $\p$.
        
        \begin{lemma}
            \label{lemma:p-l2-lb}
            In the soundness case, $\norm{\p[S_{k}^{\rD}]}_{2}^{2} \geq \alpha^3 \lambda 2^{-k + 1}$ and $|S_{k}^{\rD}| \geq \alpha \lambda 2^{k - 1}$.
        \end{lemma}
    
        \begin{proof}
            Recall that in the soundness case, we have
            \begin{equation*}
                \sum_{i \in S_{k}^{\rD}, \p[i] > \frac{1}{n}} \p[i] - \frac{1}{n} > \alpha \text{.}
            \end{equation*}
            Let $s$ denote the number of $i \in S_{k}^{\rD}$ with $\frac{1}{n} < \p[i] \leq \alpha 2^{- k + 1}$.
            Let $b$ denote the number of remaining elements in $S_{k}^{\rD}$ with $\p[i] > \alpha 2^{-k + 1}$.
            Then,
            \begin{align*}
                \sum_{i \in S_{k}^{r}, \p[i] > \frac{1}{n}} \p[i] &\leq s \alpha 2^{-k + 1} + b 2^{-k - 2} \\
                &\leq 2 \alpha + \frac{b}{\lambda} 2^{- k + 1}
            \end{align*}
            where we have used $s \leq 2^{-k}$ as each bucket in $S_{k}^{\rD}$ has mass at least $2^{-k}$ and $\p[i] \leq \frac{1}{\lambda} \rD[i] \leq \frac{2^{- k + 1}}{\lambda}$ for $i \in S_{k}^{\rD}$.
            On the other hand, we have
            \begin{align*}
                \sum_{i \in S_{k}^{\rD}, \p[i] > \frac{1}{n}} \p[i] \geq \alpha \text{.}
            \end{align*}
            Combining the above equations, we have
            \begin{align*}
                \frac{b}{\lambda} 2^{- k + 1} &\geq \alpha \\
                b &\geq \frac{\lambda \alpha}{2^{- k + 1}} \text{.}
            \end{align*}
            Thus, there are at least $b$ elements in $S_{k}^{\rD}$ with $\p[i] \geq \alpha 2^{- k + 1}$. 
            % These elements are in $S_{k + \log(\log n/\varepsilon) + 2}^{p}$.
            % In particular, we have
            % \begin{equation*}
            %     \left| S_{k}^{r} \cap S_{k + \log(\log n/\varepsilon) + 2}^{p} \right| \geq \frac{\varepsilon 2^{k}}{16 \log n}\text{.}
            % \end{equation*}
            This allows us to conclude with
            \begin{equation*}
                \sum_{i \in S_{k}^{\rD}} \p[i]^{2} \geq \left( \frac{\alpha \lambda}{2^{- k + 1}} \right) \left( \alpha^2 2^{- 2k + 2} \right) \geq \alpha^3 \lambda 2^{-k + 1} \text{.}
            \end{equation*}
        \end{proof}

        We prove a simple upper bound on the $\ell_{2}$-norm of both $\p, \q$.
        \begin{lemma}
            \label{lemma:q-l2-ub}
            $\norm{\p[S_{k}^{\rD}]}_{2}^{2} \leq \frac{2^{-k + 1}}{\lambda}\quad,\quad \norm{\q[S_{k}^{\rD}]}_{2}^{2} \leq \frac{2^{-k + 1}}{\lambda}$.
        \end{lemma}
    
        \begin{proof}
            Note that for each $i \in S_{k}^{\rD}$, $\p[i] \leq \frac{1}{\lambda} \rD[i] \leq \frac{2^{-k + 1}}{\lambda}$.
            Then, via Holder's inequality, we have
            \begin{equation*}
                \norm{\p[S_{k}^{\rD}]}_{2}^{2} \leq \norm{\p[S_{k}^{\rD}]}_{1} \norm{\p[S_{k}^{\rD}]}_{\infty} \leq \frac{2^{-k + 1}}{\lambda} \text{.}
            \end{equation*}
            An identical argument holds for $\q$.
        \end{proof}

        \begin{lemma}
            \label{lemma:tester-concentration-bounds}
            Suppose $m_{\p} = \bigO{\frac{\sqrt{n}}{\alpha^{9/2} \lambda^{7/2}}}$ for some sufficiently large constant.
            Then, the following hold:
            \begin{enumerate}
                \item If $\p$ is uniform, $c(\p; S_{k}^{\rD}) \leq c(\p) \leq \binom{m_{\p}}{2} \frac{1 + 0.01}{n}$ with probability $0.999$.
                \label{it:p-unif-collision-ub-2}
                \item In the soundness case, $c(\p; S_{k}^{\rD}) \geq \binom{m_{\p}}{2} (0.99 \alpha^3) \lambda 2^{-k + 1} \geq \binom{m_{\p}}{2} (0.99 \alpha^3) \hat{\lambda} 2^{-k + 1}$ with probability $0.999$.
                \label{it:p-far-collision-lb-2}
                \item $c(\rD; S_{k}^{\rD}) = \bigO{m_{\p}^2 2^{-k} \lambda^{-3}}$ and $\bigO{c(\p; S_{k}^{\rD})} = \bigO{m_{\p}^2 2^{-k} \lambda^{-1}}$ with probability at least $0.999$.
                \label{it:prob-collision-ub-2}
            \end{enumerate}
        \end{lemma}
    
        \begin{proof}
            Towards \Cref{it:p-unif-collision-ub-2} and \Cref{it:p-far-collision-lb-2}, we note that 
            \begin{equation}
                \label{eq:p-collision-var-S-k}
                \EX[c(\p; S_{k}^{\rD})] = \binom{m_{\p}}{2} \norm{\p[S_{k}^{\rD}]}_{2}^{2}\quad,\quad \Var(c(\p; S_{k}^{\rD})) \leq m_{\p}^{2} \norm{\p[S_{k}^{\rD}]}_{2}^{2} + m_{\p}^{3} \norm{\p[S_{k}^{\rD}]}_{3}^{3} \text{.}
            \end{equation}
            
            Suppose $\p$ is uniform.
            Recall that $c(\p; S_{k}^{\rD}) \leq c(\p)$ since the number of $\p$-collisions among elements in $S_{k}^{\rD}$ is at most the number of $\p$-collisions.
            Then, applying \eqref{eq:p-collision-variance} we have that $\EX[c(\p)] = \binom{m_{\p}}{2} \frac{1}{n}$ and $\Var(c(\p)) \leq m_{\p}^2/n$.
            Thus, $\sigma := \sqrt{\Var(c(\p))} \leq \frac{m_{\p}}{\sqrt{n}}$.
            With Chebyshev's inequality (see e.g. \Cref{clm:c-Sp-query}), we obtain that with probability $0.999$, $c(\p) \leq \binom{m_{\p}}{2} \frac{1 + 0.01}{n}$ for $m_{\p} = \bigO{\sqrt{n}}$ for some sufficiently large constant factor.

            Now, consider the soundness case.
            From \Cref{lemma:p-l2-lb} we have $\norm{\p[S_{k}^{\rD}]}_{2}^{2} \geq \alpha^3 \lambda 2^{-k + 1}$.
            Thus, the expectation can be bounded as
            \begin{equation*}
                \EX[c(\p; S_{k}^{\rD})] \geq \binom{m_{\p}}{2} \alpha^3 \lambda 2^{-k + 1} \text{.}
            \end{equation*}

            To bound the variance, we consider two separate cases.
            
            {\bf Case 1: $m_{\p}^2 \norm{\p[S_{k}^{\rD}]}_{2}^{2} \geq m_{\p}^3 \norm{\p[S_{k}^{\rD}]}_{3}^{3}$}

            From \Cref{lemma:q-l2-ub} we have $\Var(c(\p; S_{k}^{\rD})) \leq 2 m_{\p}^{2} \norm{\p[S_{k}]^{\rD}}_{2}^{2} \leq 2 m_{\p}^{2} 2^{-k+1} \lambda^{-1}$.
            In particular, $\sigma \leq 2 m_{\p} 2^{-(k-1)/2} \lambda^{-1/2}$.
            By Chebyshev's inequality, observe
            \begin{align*}
                \Pr \left( c(\p; S_{k}^{\rD}) < \binom{m_{\p}}{2} (0.99 \alpha^3) \lambda 2^{-k + 1} \right) &\leq \Pr \left( c(\p; S_{k}^{\rD}) < 0.99 \EX \left[ c(\p; S_{k}^{\rD}) \right] \right) \\
                &< \bigO{\left( \frac{m_{\p} 2^{-(k-1)/2} \lambda^{-1/2}}{m_{\p}^2 \alpha^3 \lambda 2^{-k + 1}} \right)^2} \\
                &< \bigO{\left( \frac{2^{k/2}}{m_{\p} \alpha^3
                \lambda^{1/2}} \right)^2} < 0.001
            \end{align*}
            for $m_{\p} = \bigO{\frac{2^{k/2}}{\alpha^3 \lambda}} = \bigO{\frac{\sqrt{n}}{\alpha^{3/2} \lambda}}$ since 
            \begin{equation*}
                2^{k/2} = \bigO{\sqrt{n} \alpha^{3/2}} \text{.}
            \end{equation*}

            {\bf Case 2: $m_{\p}^2 \norm{\p[S_{k}^{\rD}]}_{2}^{2} \leq m_{\p}^3 \norm{\p[S_{k}^{\rD}]}_{3}^{3}$}

            From \Cref{lemma:q-l2-ub} we have $\Var(c(\p; S_{k}^{\rD})) \leq 2 m_{\p}^{3} \norm{\p[S_{k}]^{\rD}}_{3}^{3} \leq 2 m_{\p}^{3} \norm{\p[S_{k}]^{\rD}}_{3}^{3} \leq 2 m_{\p}^{3} 2^{-1.5k + 1.5} \lambda^{-3/2}$.
            In particular, $\sigma \leq 2 m_{\p}^{3/2} 2^{-0.75k + 0.75} \lambda^{-0.75}$.
            By Chebyshev's inequality, observe
            \begin{align*}
                \Pr \left( c(\p; S_{k}^{\rD}) < \binom{m_{\p}}{2} (0.99 \alpha^3) \lambda 2^{-k + 1} \right) &\leq \Pr \left( c(\p; S_{k}^{\rD}) < 0.99  \EX \left[ c(\p; S_{k}^{\rD}) \right] \right) \\
                &= \bigO{\left( \frac{m_{\p}^{3/2} 2^{-0.75k + 0.75} \lambda^{-0.75}}{m_{\p}^2 \alpha^3 \lambda 2^{-k + 1}} \right)^2} \\
                &= \bigO{\left( \frac{2^{k/4}}{m_{\p}^{1/2} \alpha^3 \lambda^{7/4}} \right)^2} \\
                &= \bigO{\frac{2^{k/2}}{m_{\p} \alpha^{6} \lambda^{7/2}}} < 0.001
            \end{align*}
            for $m_{\p} = \bigO{\frac{\sqrt{n}}{\alpha^{9/2} \lambda^{7/2}}}$.
            Since $\hat{\lambda} \leq \lambda$, the same bound holds if the threshold is defined using our estimate $\hat{\lambda}$ instead.

            Finally, towards \Cref{it:prob-collision-ub-2} we apply Markov's inequality we bound the total number of collisions by 
            \begin{align*}
                \bigO{c(\rD; S_{k}^{\rD})} &= \bigO{m^{2} \lambda^{-2} \norm{\rD[S_{k}^{\rD}]}_{2}^{2}} = \bigO{m^{2} 2^{-k} \lambda^{-1}} = \bigO{m_{\p}^2 2^{-k} \lambda^{-3}} \\
                \bigO{c(\p; S_{k}^{\rD})} &= \bigO{m^{2} \norm{\rD[S_{k}^{\rD}]}_{2}^{2}} = \bigO{m^{2} 2^{-k} \lambda^{-1}} = \bigO{m_{\p}^2 2^{-k} \lambda^{-1}} \text{.}
            \end{align*}

        \end{proof}
    
        We now describe analogous sub-routines to \Cref{clm:m-p-query} and \Cref{clm:c-Sp-query}.

        \begin{claim}
            \label{clm:m-p-query-mix}
            With $\bigO{\lambda^{-2}}$ queries, there is an algorithm that with probability $0.99$ computes an estimate $\hat{m}_{\p}$ such that
            \begin{equation*}
                \left| \binom{m_{\p}}{2} - \binom{\hat{m}_{\p}}{2} \right| \alpha^{3} \hat{\lambda} 2^{-k+1} \leq 0.01 \binom{m_{\p}}{2} \alpha^{3} \hat{\lambda} 2^{-k + 1} \text{.}
            \end{equation*}
        \end{claim}

        \begin{proof}
            We compute
            \begin{equation*}
                 \left| \binom{m_{\p}}{2} - \binom{\hat{m}_{\p}}{2} \right| \leq \frac{1}{2} \left| m_{\p}^{2} +m_{\p} - \hat{m}_{\p}^{2} +\hat{m}_{\p} \right| \leq (m_{\p} + \hat{m}_{\p}) |m_{\p} - \hat{m}_{\p}| = \bigO{m |m_{\p} - \hat{m}_{\p}|} \text{.}
            \end{equation*}
            In the final inequality, we have used $m_{\p}, \hat{m}_{\p} = O(m)$ from \Cref{claim:sample-bound-mix}.
            Since $\frac{m_{\p}^{2}}{4} \leq \binom{m_{\p}}{2}$, it suffices to obtain an estimate such that $|m_{\p} - \hat{m}_{\p}| \ll \frac{m_{\p}^2}{m} \ll m_{\p}$.
            We can do this by randomly choosing samples and querying them, thus accessing a coin with bias $\frac{m_{\p}}{m/\hat{\lambda}}$ which we must estimate with accuracy $\ll \frac{m_{\p}}{m/\hat{\lambda}} \ll \lambda$.
            The query bound follows from an application of \Cref{lem:bias-add}.
        \end{proof}

        \begin{claim}
            \label{clm:c-Sp-query-mix}
            There exists a universal constant $C > 0$ such that if $m_{\p} \geq \frac{C \sqrt{n}}{\alpha^{9/2} \lambda^{7/2}}$ there is an algorithm that uses $\bigO{\frac{1}{\alpha^{6} \lambda^{8}}}$ queries, and with probability $0.99$ computes an estimate $\hat{c}(\p; S_{k}^{\rD})$ such that the following hold:
            \begin{enumerate}
                \item If $\p$ is uniform, 
                \begin{equation*}
                    \hat{c}(\p; S_{k}^{\rD}) \leq \binom{m_{\p}}{2} 0.2  \alpha^3 \lambda 2^{-k + 1} \leq \binom{m_{\p}}{2} 0.3  \alpha^3 \hat{\lambda} 2^{-k + 1} \text{.}
                \end{equation*}
                \item In the soundness case, 
                \begin{equation*}
                    \hat{c}(\p; S_{k}^{\rD}) \geq \binom{m_{\p}}{2} 0.7  \alpha^3 \hat{\lambda} 2^{-k + 1} \text{.}
                \end{equation*}
            \end{enumerate}
        \end{claim}

        \begin{proofof}{\Cref{clm:c-Sp-query-mix}}
            Note that $0.2 \alpha^3 \hat{\lambda} 2^{-k + 1} \geq \frac{1 + 0.01}{n}$ as long as $k < \log n + 3 \log \alpha + \log \hat{\lambda} - 3$, which holds by assumption.
            Thus, in both cases, following \Cref{lemma:tester-concentration-bounds} it suffices to estimate $c(\p; S_{k}^{\rD})$ up to error additive error $\binom{m_{\p}}{2} \frac{\alpha^3 \hat{\lambda} 2^{-k + 1}}{10}$ using queries.
            Following similar arguments to \Cref{clm:c-Sp-query}, this requires query complexity
            \begin{equation*}
                \bigO{\frac{2^{2k} c(\rD; S_{k}^{\rD})^2}{m_{\p}^{4} \alpha^{6} \hat{\lambda}^{2}} \cdot \frac{c(\p; S_{k}^{\rD})}{c(\rD; S_{k}^{\rD})}} = \bigO{\frac{2^{2k} c(\rD; S_{k}^{\rD}) c(\p; S_{k}^{\rD})}{m_{\p}^{4} \alpha^{6} \hat{\lambda}^{2}}} = \bigO{\frac{1}{\alpha^{6} \lambda^{8}}} \text{.}
            \end{equation*}
            In the last inequality, we have used \Cref{it:prob-collision-ub-2} to bound both $c(\rD; S_{k}^{\rD})$ and $c(\p; S_{k}^{\rD})$.
        \end{proofof}

        We now present our algorithm.

        \begin{mdframed}
            \begin{enumerate}
                \item Let $m := \frac{C n^{1/2}}{\alpha^{9/2} \hat{\lambda}^{7/2}}$.
                \item Repeat the following for $1 \leq t \leq T = C \log(1/\delta)$.
                \begin{enumerate}
                    \item Draw $\frac{m}{\hat{\lambda}}$ fresh samples from $\rD$.
                    \item Compute $\hat{m}_{\p}$ using \Cref{clm:m-p-query-mix}.
                    \item Compute $\hat{c}\left(\p; S_{k}^{\rD}\right)$ using \Cref{clm:c-Sp-query-mix}.
                    \item Let $Z_{t} \gets \ind{\hat{c}\left(\p; S_{k}^{\rD}\right) < \binom{\hat{m}_{\p}}{2} 0.5 \alpha^3 \hat{\lambda} 2^{-k + 1}}$.
                \end{enumerate}
                \item If $\frac{1}{T} \sum Z_{t} > \frac{1}{2}$, accept. Otherwise, reject.
            \end{enumerate}
        \end{mdframed}

        Fix an iteration $t$.
        By \Cref{claim:sample-bound-mix}, we have that $m_{\p} \geq 0.99 m$ so that \Cref{lemma:tester-concentration-bounds} applies.
        If $\p$ is uniform, then we claim $Z_t = 0$ with probability at most $0.1$ by a union bound.
        To see this, we may union bound over the correctness of the estimates $\hat{m}_{\p}, \hat{c}(\p; S_{k}^{\rD})$ and obtain
        \begin{align*}
            \binom{\hat{m}_{\p}}{2} 0.5 \alpha^3 \hat{\lambda} 2^{-k + 1} - \hat{c}(\p; S_{k}^{\rD}) &\geq \left( \binom{m_{\p}}{2} 0.5 \alpha^3 \hat{\lambda} 2^{-k+1} - \hat{c}(\p; S_{k}^{\rD}) \right)  - \left| \binom{\hat{m}_{\p}}{2}  - \binom{m_{\p}}{2} \right| 0.5 \alpha^3 \hat{\lambda} 2^{-k+1} \\
            &\geq \binom{m_{\p}}{2} 0.2 \alpha^3 \hat{\lambda} 2^{-k+1} - 0.01 \cdot \binom{m_{\p}}{2} \alpha^3 \hat{\lambda} 2^{-k+1} \\
            &> 0 \text{.}
        \end{align*}
        Similarly, in the soundness case we have $Z_t = 1$ with probability at most $0.1$.
        Then, by taking a majority vote, we observe that the algorithm fails with probability $e^{-\Omega(T)} < \delta$ for $T = O(\log(1/\delta))$.
        The sample and query complexities are immediate from the algorithm.
    \end{proof}

    \subsection{Testing Distributions with Light Buckets}
    
    The next sub-routine we need is an algorithm that can detect if $\p$ is far from uniform for some $k > \log n + 3 \log \alpha + \log \hat{\lambda} - 3$.

    \begin{lemma}
        \label{lemma:dist-mix-knowledge-large-k-alg}
        Let $\alpha > 0$ and $k_0 = \log n + 3 \log \alpha + \log \hat{\lambda} - 3$ and let $S = \overline{S_{k_0}^{\rD}}$.
        
        There is an algorithm using $\bigO{\frac{\sqrt{n}}{\alpha^{4} \lambda^{2}}}$ samples and $\bigO{\frac{1}{\alpha^{12} \lambda^{6}}}$ queries satisfying the following:
        \begin{enumerate}
            \item (Completeness) if $\p$ is uniform, then the algorithm accepts with probability at least $0.999$.
            \item (Soundness) if $\alpha < \sum_{i \in S, \p[i] > 1/n} p_{i} - \frac{1}{n}$, then the algorithm rejects with probability at least $0.999$.
        \end{enumerate}
    \end{lemma}

    \begin{proof}
        Since, all elements are in $\overline{S_{k}^{\rD}}$, $\norm{\rD[S]}_{2}^{2} \leq \norm{\rD[S]}_{\infty} \leq \frac{16}{n \alpha^3 \hat{\lambda}}$.
        For all $i$, write $\p[i] = \frac{1}{n} + a[i]$ for some vector $a$.
        Then, we have for all $i \in S$ that $\p[i] \leq \frac{\rD[i]}{\lambda} \leq \frac{16}{n \alpha^3 \hat{\lambda}^{2}}$ so that $a[i] \leq \frac{16}{n \alpha^3 \hat{\lambda}^{2}}$.

        We begin with two simplifying assumptions.
        \begin{claim}
            \label{clm:p-size-prob-lb}
            In the soundness case, $|S| \geq \frac{n \alpha^{4} \hat{\lambda}^{2}}{16}$ and $\p[S] \geq \alpha$.
            Furthermore, there is an algorithm with $\bigO{\log(1/\delta) \alpha^{-2}}$ queries that outputs accept with probability $1 - \delta$ if either $|S| < \frac{n \alpha^{4} \hat{\lambda}}{16}$ or $\p[S] < \alpha/2$.
        \end{claim}

        \begin{proof}
            In the soundness case, $\sum_{i \in S} a[i] \geq \alpha$ implies $|S| \geq \frac{\alpha}{\max_i a[i]} \geq \frac{n \alpha^{4} \hat{\lambda}^{2}}{16}$.
            Similarly, $\p[S] \geq \p[S^+] \geq \alpha$.

            An algorithm can check for the size of $S$ directly since we are given $S$ via the mixture $\rD$.
            Similarly, we can estimate $\p[S]$ up to additive error $\alpha/6$ with $\bigO{\alpha^{-2}}$ samples and queries via \Cref{lem:bias-add}, and output accept if this estimate is at most $\frac{2 \alpha}{3}$.
        \end{proof}

        Note that the above algorithm does not output accept in the soundness case.
        Thus, let us assume $|S| \geq \frac{n \alpha^{4} \hat{\lambda}^{2}}{16}$ and $\p[S] \geq \alpha/2$.
        Let $S^+ = \set{i \in S \given \p[i] > \frac{1}{n}}$.
        We consider two cases.

        \paragraph{Case 1: $\p[S] \leq |S|/n$}

        When $\p[S] \leq |S|/n$, we observe that $(\p|S)$ is uniform (resp. far from uniform) if and only if $\p$ is.
        In particular, we test uniformity of $(\p|S)$ via sample access to the mixture $(\rD|S)$ using \Cref{lemma:simple-t-query-alg}.
        
        In the completeness case, we have $(\p|S)$ is uniform.
        In the soundness case, note that 
        \begin{equation*}
            (\p|S)[i] = \frac{\p[i]}{\p[S]} \geq \frac{n}{|S|} \p[i]
        \end{equation*}
        so that 
        \begin{equation*}
            \sum_{i \in S^{+}} (\p|S)[i] - (U_n|S)[i] = \sum_{i \in S^{+}} (\p|S)[i] - \frac{1}{|S|} \geq \sum_{i \in S^{+}} \frac{n}{|S|} \p[i] - \frac{n}{|S|} \frac{1}{n} \geq \frac{n}{|S|} \alpha \text{.}
        \end{equation*}
        In particular, we have from \Cref{clm:p-size-prob-lb} that
        \begin{equation*}
            \norm{(\p|S) - (U_n|S)}_{1} \geq \sum_{i \in S^{+}} (\p|S)[i] - (U_n|S)[i] \geq \frac{n}{|S|} \alpha \geq \alpha \text{.}
        \end{equation*}
        
        Furthermore, 
        $\p[S] \geq \p[S^{+}] \geq \alpha$, so that a single sample from $\rD$ falls in $S$ and is from $\p$ with probability $\frac{1}{\lambda \alpha}$.
        If we take $\frac{2m}{\lambda \alpha}$ samples, then we have at least $m$ samples from $\p$ in $S$ with probability at least $1 - \exp(-\Omega(m)) > 1 - \delta$ if $m \gg \log(1/\delta)$.
        Note that since the algorithm knows $S$ (as we are given knowledge of the distribution $\rD$), we will only consider samples on $S$ and assume that we are sampling from the conditional distribution $(\rD|S)$.

        We first show that the mixture $(\rD|S) := \eta (\p|S) + (1 - \eta) (\q|S)$ is not too contaminated.
        \begin{lemma}
            \label{lemma:eta-lb}
            $\eta \geq \frac{\lambda \alpha}{\rD[S]} \geq \lambda \alpha$.
        \end{lemma}

        \begin{proof}
            Fix a sample $Y \sim (\rD|S)$ and let $Z$ denote its source.
            \begin{align*}
                (\rD|S)[i] = \Pr(Y = i | Y \in S) &= \Pr(Z = \p, Y = i| Y \in S) + \Pr(Z = \q, Y = i| Y \in S) \\
                &= \Pr(Y = i | Z = \p,  Y \in S) \Pr(Z = \p| Y \in S) \\
                &\quad+ \Pr(Y = i|Z=\q,  Y \in S)\Pr(Z = \q| Y \in S) \\
                &= (\p|S)[i] \Pr(Z = \p| Y \in S) + (\q|S)[i] \Pr(Z = \q| Y \in S) \text{.}
            \end{align*}
            Thus $\eta$ is the probability that a sample in $S$ came from $\p$.
            Thus,
            \begin{align*}
                \Pr(Z = \p|Y \in S) &= \frac{\Pr(Z = \p, Y = S)}{\Pr(Y \in S)} \geq \frac{\lambda \alpha}{\rD[S]} \text{.}
            \end{align*}
        \end{proof}

        Furthermore, note that for all $i \in S$, we have $\rD[S] \geq \lambda \p[S] \geq \lambda \alpha$.
        This allows us to bound $(\rD|S)[i] = \frac{\rD[i]}{\rD[S]} \leq \frac{16}{n \alpha^3 \hat{\lambda} \rD[S]} \leq \frac{16}{n \alpha^{4} \hat{\lambda}^2}$ and 
        Then, directly applying \Cref{lemma:simple-t-query-alg} to the mixture $(\rD|S)$ we obtain a uniformity tester using 
        \begin{equation*}
            \bigO{\frac{\sqrt{n}}{\alpha^{2} \eta}} = \bigO{\frac{\sqrt{n}}{\alpha^{3} \lambda}}
        \end{equation*}
        samples and 
        \begin{equation*}
            \bigO{\frac{n\frac{1}{n \alpha^{4} \lambda^{2}}}{\alpha^{4} \eta^{2}}} = \bigO{\frac{1}{\alpha^{10} \lambda^{4}}}
        \end{equation*}
        queries.
        Recall that we incur an additional $\alpha \lambda$ factor in the sample complexity to obtain samples from $(\rD|S)$ (yielding $\bigO{\frac{\sqrt{n}}{\alpha^{4} \lambda^{2}}}$ samples).

        \paragraph{Case 2: $\p[S] \geq |S|/n$}

        Consider now the second case.
        When $\p[S] \geq |S|/n$, we note that $\p$ produces many collisions on the subset $S$ if and only if $\p$ is far from uniform.
        We run the algorithm of \Cref{lemma:dist-mix-knowledge-large-k-alg} but with different parameters which we state below.
        Note that we only run the algorithm once $(T = 1)$ since we only require high constant probability of correctness.
        
        In the uniform case, we have $\norm{U_n[S]}_{2}^{2} = \frac{|S|}{n^2}$.
        In the soundness case, we have
        \begin{equation*}
            \norm{\p[S] - U_n[S]}_{1} = \sum_{i \in S} \left| \p[i] - \frac{1}{n} \right| \geq \alpha 
        \end{equation*}
        so that $\norm{\p[S] - U_n[S]}_{2}^{2} \geq \alpha^2/|S|$.
        Furthermore, we have
        \begin{equation*}
            \frac{\alpha^2}{|S|} \leq \norm{\p[S] - U_n[S]}_{2}^{2} = \norm{\p[S]}_{2}^{2} + \frac{|S|}{n^2} - \frac{2\p[S]}{n} \leq \norm{\p[S]}_{2}^{2} - \frac{|S|}{n^2}
        \end{equation*}
        since $\p[S] \geq \frac{|S|}{n}$.
        Rearranging, we obtain
        \begin{equation*}
            \norm{\p[S]}_{2}^{2} \geq \frac{\alpha^2}{|S|} + \frac{|S|}{n^2} = \frac{\alpha^2}{|S|} + \norm{U_n[S]}_{2}^{2} \text{.}
        \end{equation*}
        
        Suppose we take $m/\hat{\lambda}$ samples.
        Let $c(\p; S)$ denote the number of $\p$-collisions in $S$.
        From the above discussion, we expect an expectation gap in the number of collisions $c(\p; S)$ of approximately $\bigTh{\frac{\alpha^2 m^2}{|S|}}$.
        Following \eqref{eq:p-collision-var-S-k}, we have
        \begin{equation*}
            \Var(c(\p)) \leq m_{\p}^{2} \norm{\p[S]}_{2}^{2} + m_{\p}^{3} \norm{\p[S]}_{3}^{3} \text{.}
        \end{equation*}
        Our goal now is to follow the framework of \Cref{lemma:dist-mix-knowledge-sub-alg}.

        Analogously to \Cref{lemma:tester-concentration-bounds}, we obtain the following result.
        \begin{lemma}
            \label{lemma:large-k-tester-concentration}
            Suppose $m_{\p} = \bigO{\frac{\sqrt{n}}{\alpha^{7/2} \lambda} + \frac{1}{\alpha^{10} \lambda^{4}}}$ for some sufficiently large constant.
            Then, the following hold:
            \begin{enumerate}
                \item $\left|c(\p; S) - \binom{m_{\p}}{2} \norm{\p[S]}_{2}^{2} \right| \leq  \frac{0.01 \alpha^2 m^2}{n}$ with probability $0.999$.
                \item $c(\rD; S) = \bigO{\frac{m^{2}}{n \alpha^3 \lambda^3}}$ and $c(\p; S) = \bigO{\frac{m^{2}}{n \alpha^3 \lambda^{2}}}$ with probability $0.999$.
            \end{enumerate}
        \end{lemma}

        \begin{proof}
            Since $\norm{\p[S]}_{2}^{2} \leq \max_{i \in S} \p[i] \leq \max_{i \in S} \frac{\rD[i]}{\lambda} \leq \frac{16}{n \alpha^{3} \hat{\lambda}^{2}}$ and
            \begin{equation*}
                \norm{\p[S]}_{3}^{3} = \sum_{i \in S} \p[i]^3 \leq (\max_{i \in S} \p[i])\norm{\p[S]}_{2}^{2} \leq \frac{256}{n^2 \alpha^{6} \hat{\lambda}^{4}} \text{,}
            \end{equation*}
            we have
            \begin{equation*}
                \Var(c(\p; S)) \leq m_{\p}^{2} \frac{16}{n \alpha^{3} \hat{\lambda}^{2}} + m_{\p}^{3} \frac{256}{n^2 \alpha^{6} \hat{\lambda}^{4}} \text{.}
            \end{equation*}
            By Chebyshev's inequality, we have
            \begin{equation*}
                \Pr \left( \left|c(\p; S) - \binom{m_{\p}}{2} \norm{\p[S]}_{2}^{2} \right| > \frac{0.01 \alpha^2 m^2}{n} \right) < 0.001
            \end{equation*}
            as long as $m_{\p} = \bigO{\frac{\sqrt{n}}{\alpha^{7/2} \lambda} + \frac{1}{\alpha^{10} \lambda^{4}}}$ for some sufficiently large constant.
            Note that here we have used $m_{\p} = \Theta(m)$ from \Cref{claim:sample-bound-mix}.

            As in \Cref{lemma:tester-concentration-bounds}, we have $c(\rD; S) = \bigO{(m/\lambda)^2 \norm{\rD[S]}_{2}^{2}} = \bigO{\frac{m^2}{n \alpha^3 \lambda^{3}}}$ and $c(\p; S) = \bigO{m^2 \norm{\p[S]}_{2}^{2}} = \bigO{\frac{m^2}{n \alpha^3 \lambda^{2}}}$ with probability $0.999$ with Markov's inequality.
        \end{proof}

        Following \Cref{lemma:dist-mix-knowledge-sub-alg}, our goal is to obtain estimates $\hat{m}_{\p}$, $\hat{c}(\p; S)$ satisfying:
        \begin{enumerate}
            \item $\left| \hat{m}_{\p} - m_{\p} \right| \ll \frac{|S| \alpha^2 m}{n}$.
            \item $\left| \hat{c}(\p; S) - c(\p; S) \right| \ll \frac{m^{2} \alpha^{2}}{n}$
        \end{enumerate}
        
        Towards the first estimate, we note that by picking a random sample and querying it, we flip a coin with bias $\frac{m_{\p}}{m/\hat{\lambda}}$ which we must estimate up to accuracy $\ll \frac{|S| \alpha^2 m}{n} \frac{\hat{\lambda}}{m} = \frac{|S| \alpha^2 \hat{\lambda}}{n}$.
        From \Cref{lem:bias-add}, this requires 
        \begin{equation*}
            \bigO{\frac{n^2}{\alpha^{4} |S|^2 \lambda^2}} = \bigO{\frac{1}{\alpha^{12} \lambda^{6}}}
        \end{equation*}
        queries since $|S| = \bigOm{n \alpha^{4} \lambda^{2}}$.
        To use this estimate, we will observe that 
        \begin{align}
            \label{eq:m-error-bound}
            \left| \binom{m_{\p}}{2} - \binom{\hat{m}_{\p}}{2} \right| \left( \frac{|S|}{n^2} + \frac{\alpha^2}{|S|} \right) &\leq (m_{\p} + \hat{m}_{\p}) |m_{\p} - \hat{m}_{\p}| \left( \frac{|S|}{n^2} + \frac{\alpha^2}{|S|} \right) = \bigO{\frac{2}{|S|} m |m_{\p} - \hat{m}_{\p}|} \ll \frac{\alpha^2 m^{2}}{n} \text{.}
        \end{align}
        In the first inequality, we use $a^2 - b^2 = (a + b)(a - b)$, in the second we note that $|S| \leq n$ implies that $\frac{|S|}{n} + \frac{\alpha^2}{|S|} \leq \frac{2}{|S|}$ and $m_{\p} = O(m)$, and in the final inequality we use our error bound on $\hat{m}_{\p}$.

        Towards the second estimate, we pick a random collision and query it, which yields a $\p$-collision with probability $\frac{c(\p; S)}{c(\rD; S)}$.
        To estimate this up to accuracy $\ll \frac{m^2 \alpha^2}{c(\rD; S) n}$ we require
        \begin{equation*}
            \bigO{\frac{n^2 c(\rD; S)^2}{m^{4} \alpha^{4}} \cdot \frac{c(\p; S)}{c(\rD; S)}} = \bigO{\frac{n^2 c(\p; S) c(\rD; S)}{\alpha^{4} m^{4}}} = \bigO{\frac{1}{\alpha^{10} \lambda^{5}}}
        \end{equation*}
        queries by bounding the number of collisions with \Cref{lemma:large-k-tester-concentration}.
        
        Given the estimates, we return $\ind{\hat{c}(\p; S) < \binom{\hat{m}_{\p}}{2} \left( \frac{|S|}{n^2} + \frac{\alpha^2}{2|S|} \right)}$.
        When $\p$ is uniform, the algorithm accepts as
        \begin{align*}
            \hat{c}(\p; S) &\leq c(\p; S) + \frac{0.01 \alpha^2 m^2}{n} \\
            &\leq \binom{m_{\p}}{2} \frac{|S|}{n^2} + \frac{0.02 \alpha^2 m^2}{n} \\
            &\leq \binom{\hat{m}_{\p}}{2} \frac{|S|}{n^2} + \frac{0.03 \alpha^2 m^2}{n} \\
            &\leq \binom{\hat{m}_{\p}}{2} \frac{|S|}{n^2} + \binom{\hat{m}_{\p}}{2} \frac{0.25 \alpha^2}{|S|} \text{.}
        \end{align*}
        In the first inequality, we apply the error bound on $\hat{c}(\p; S)$, in the second \Cref{lemma:large-k-tester-concentration}, in the third the error bound on $\hat{m}_{\p}$ and \eqref{eq:m-error-bound}, and in the final inequality we use $m \leq 2 m_{\p}$ from \Cref{claim:sample-bound-mix} and $|S| \leq n$.
        In the soundness case, the algorithm rejects (following similar arguments) as
        \begin{align*}
            \hat{c}(\p; S) &\geq c(\p; S) - \frac{0.01 \alpha^2 m^2}{n} \\
            &\geq \binom{m_{\p}}{2} \left( \frac{|S|}{n^2} + \frac{\alpha^2}{|S|} \right) - \frac{0.02 \alpha^2 m^2}{n} \\
            &\geq \binom{\hat{m}_{\p}}{2} \left( \frac{|S|}{n^2} + \frac{\alpha^2}{|S|} \right) - \frac{0.03 \alpha^2 m^2}{n} - \\
            &\geq \binom{\hat{m}_{\p}}{2} \left( \frac{|S|}{n^2} + \frac{\alpha^2}{|S|} \right) - \binom{\hat{m}_{\p}}{2} \frac{0.25 \alpha^2}{|S|} \text{.}
        \end{align*}

        The algorithm requires $\bigO{\frac{\sqrt{n}}{\alpha^{7/2} \lambda} + \frac{1}{\alpha^{10} \lambda^{4}}}$ samples and $\bigO{\frac{1}{\alpha^{12} \lambda^{6}}}$ queries.

        \paragraph{Putting It All Together}
        Our final algorithm for \Cref{lemma:dist-mix-knowledge-large-k-alg} is to run both algorithms for the above cases and output accept if they both accept.
        If either of the algorithms reject, we reject.
        
        For correctness, note that if $\p$ is uniform, then $\p[S] = \frac{|S|}{n}$ so that the conditions of both cases are satisfied.
        In particular, by a union bound over both cases, both algorithms will accept, as desired.
        On the other hand, if $\p$ is far from uniform, then at least one of the conditions are satisfied, and one algorithm will reject, so we reject, as desired.
        The correctness probability follows again from a union bound over both cases.
        The sample and query complexity bounds follow from summing over both cases.
    \end{proof}

    We are now ready to present the final algorithm for \Cref{thm:dist-mix-knowledge-alg}.

    \begin{mdframed}
        \begin{enumerate}
            \item Estimate $\frac{2 \lambda}{3} \leq \hat{\lambda} \leq \lambda$ using \Cref{lemma:lambda-est-2}.
            \item Set $\alpha = \frac{\varepsilon}{1 + \log n}$ and $k_0 \gets \log n + 3 \log \alpha + \log \hat{\lambda} - 3$.
            \item For all $1 \leq k \leq k_0$:
            \begin{enumerate}
                \item Run the sub-routine of \Cref{lemma:dist-mix-knowledge-sub-alg} with $\delta \ll \frac{1}{k_0}$. 
                If it outputs reject, reject. Otherwise continue.
            \end{enumerate}
            \item Run the sub-routine of \Cref{lemma:dist-mix-knowledge-large-k-alg}. 
            Return the output.
        \end{enumerate}
    \end{mdframed}

    We argue that our algorithm is correct.
    If $\p$ is uniform, then every sub-routine returns accept, and we return accept. 
    This fails to occur with at most small constant probability by a union bound.

    If $\p$ is far from uniform, we claim at least one iteration will reject.
    This follows from our setting of $\alpha$ and the following lemma.
    \begin{claim}
        \label{claim:far-index}
        Suppose $\p$ is $\eps$-far from uniform.
        There is an index $1 \leq k \leq 1 + \log n$ such that $\frac{\varepsilon}{1 + \log n} < \sum_{i \in S_{k}^{r}, \p[i] > \frac{1}{n}} \p[i] - \frac{1}{n}$.
    \end{claim}
    
    \begin{proof}
        By definition of total variation distance,
        \begin{equation*}
            \varepsilon < \sum_{\p[i] > \frac{1}{n}} \p[i] - \frac{1}{n} \text{.}
        \end{equation*}
        By partitioning $[n]$ into $S_{k}^{\rD}$, we have
        \begin{equation*}
            \varepsilon < \sum_{k = 1}^{\log n + 1} \sum_{i \in S_{k}^{r}, \p[i] > \frac{1}{n}} \p[i] - \frac{1}{n} \text{.}
        \end{equation*}
        In particular, there exists some index $k$ where
        \begin{equation*}
            \frac{\varepsilon}{\log n + 1} < \sum_{i \in S_{k}^{r}, \p[i] > \frac{1}{n}} \p[i] - \frac{1}{n}\text{.}
        \end{equation*}
    \end{proof}
    
    We thus bound the sample complexity via \Cref{lemma:dist-mix-knowledge-sub-alg} and \Cref{lemma:dist-mix-knowledge-large-k-alg} to obtain by $\bigtO{\frac{\sqrt{n}}{\alpha^{9/2} \lambda^{7/2}} + \frac{1}{\alpha^{12} \lambda^{6}}}$ and the query complexity by $\bigO{\frac{\log n}{\alpha^{6} \lambda^{8}} + \frac{1}{\alpha^{12} \lambda^{6}}}$.
    We complete the analysis by substituting $\alpha = \frac{\varepsilon}{1 + \log n}$ and noting that $\alpha \geq \frac{\eps}{2 \log n}$.
\end{proof}

\section*{Acknowledgments}

We would like to thank Nathaniel Harms, Sihan Liu, and anonymous reviewers for helpful comments and feedback on earlier versions of this work.

\newpage 

\bibliographystyle{alpha}
\bibliography{biblio}

\appendix

\newpage

\section{Omitted Proofs}
\label{sec:omitted-proofs}

\subsection{Probability and Information Theory}

\MutualInfoBound*

\begin{proof}
    The conditional entropy $H(X|f(T))$ is the expectation over $f(T)$ of $h(q)$ where $h(q)$ is the binary entropy function and $q$ is the probability that $X = f(T)$ given $f(T)$.
    Since $\EX[q] \geq 0.51$ and $h$ is concave, $H(X|f(T)) \leq h(0.51) < 1 - 2 \cdot 10^{-4}$.
    Then, by the data processing inequality
    \begin{equation*}
        I(X:T) \geq I(X:f(T)) \geq H(X) - H(X|f(T)) \geq 2 \cdot 10^{-4} \text{.}
    \end{equation*}
\end{proof}

\MIAsymp*

\begin{proof}
        We give a proof of this standard fact for completeness.
        % Denote $\alpha := \Pr[M=a,X=1], \beta := \Pr[M=a,X=0]$ for simplicity. 
        Note $\Pr[X=1]=\Pr[X=0] = 1/2$.
        % follows from our assumption, and therefore $\frac{(\beta-\alpha)^2}{\beta+\alpha} =\Theta\left(\frac{(\beta-\alpha)^2}{\alpha}\right)=  \Theta\left(\frac{(\beta-\alpha)^2}{\beta}\right)$. 
        By definition, 
        \begin{align*}
            I(X:M) &= \sum_a \sum_{i=0,1}\Pr[X=i,M=a] \log \left(\frac{\Pr[X=i,M=a]}{\Pr[X=i] \Pr[M=a]}\right)\\
            &= \sum_a \sum_{i=0,1}\Pr[X=i,M=a] \log \left(\frac{2 \Pr[X=i,M=a]}{\Pr[M=a]}\right) \text{.}
            % &= \sum_a \left(\beta \log \left(\frac{2\beta}{\alpha + \beta}\right)+ \alpha \log \left(\frac{2\alpha}{\alpha + \beta}\right) \right)\\
            % &\text{rearranging}\\
            % & = \sum_a \left(\beta\log \left(\frac{1}{1+\frac{\alpha-\beta}{2\beta}}\right)+ \alpha \log \left(\frac{1}{1-\frac{\alpha-\beta}{2\alpha}}\right)\right).
        \end{align*}
        Fix an $a$.
        Then, we have (using $\log(x) \leq x - 1$ for $x > 0$),
        \begin{align*}
            \quad &\sum_{i} \Pr[X=i,M=a] \log \left(\frac{2 \Pr[X=i,M=a]}{\Pr[M=a]}\right) \\
            \leq &\sum_{i} \Pr[X=i,M=a]  \left(\frac{2 \Pr[X=i,M=a]}{\Pr[M=a]} - 1 \right) \\
            = &\frac{1}{\Pr[M = a]} \sum_{i} 2 \Pr[X=i,M=a]^2 - \Pr[X = i, M = a] \Pr[M = a] \\
            = &\frac{1}{\Pr[M = a]} \sum_{i} \Pr[X=i,M=a]^2 - \Pr[X = i, M = a] \Pr[X = 1 - i, M = a] \\
            = &\frac{(\Pr[X = 0, M = a] - \Pr[X = 1, M = a])^2}{\Pr[M = a]} \\
            % = &\frac{\frac{1}{4} (\Pr[M = a | X = 0] - \Pr[M = a | X = 1])^2}{\frac{1}{2} \left( \Pr[M = a | X = 0] + \Pr[M = a | X = 1] \right)} \\
            = & \frac{2(\Pr[M = a | X = 0] - \Pr[M = a | X = 1])^2}{\Pr[M = a | X = 0] + \Pr[M = a | X = 1]} \text{.}
        \end{align*}
        In the second equality we used $\Pr[M = a] = \Pr[X = 0, M = a] + \Pr[X = 1, M = a]$ and in the last step we used
        \[
        \Pr[X = 0, M = a] + \Pr[X = 1, M = a] = \frac{1}{2}\big(\Pr[M = a | X = 0] + \Pr[M = a | X = 1]\big)
        \]
        since $X$ is an unbiased coin. Applying the above bound to all $a$ concludes the proof. %, noting that the constant does not depend on $a$.
        % Denote $A:=\frac{\alpha-\beta}{2\beta}$ and $B:=\frac{\alpha-\beta}{2\alpha}$, then by Taylor expansion of $\log\left(\frac{1}{1+A}\right) $ and $\log\left(\frac{1}{1-B}\right) $ we have that 
        % \begin{align*}
        %     I(X:M) &= \Theta(1) \sum_a \left(\beta\left(\sum_{n=1}^{\infty} (-1)^n\frac{A^n}{n}\right)+\alpha \left(\sum_{n=1}^{\infty} \frac{B^n}{n}\right)\right)\nonumber\\
        %     &= \Theta(1)\sum_a \left(\sum_{\substack{n=3\\ n\text{ odd}}} \frac{1}{n}(\alpha B^n-\beta A^n) + \sum_{\substack{n=2\\ n\text{ even}}} \frac{1}{n}(\beta A^n+\alpha B^n)\right) \label{IXMub}\\
        %     &\le \Theta(1) \sum_a  \sum_{n=2}(\alpha B^n+\beta A^n) \\
        %     &\le \Theta(1) \sum_{a} (\alpha B^2 + \beta A^2) \\
        %     &= \sum_aO\left(\frac{(\beta-\alpha)^2}{\alpha}\right)\nonumber
        % \end{align*}
        % as desired.
    \end{proof}

Next, we prove the concentration bounds required in \Cref{thm:dist-contamination-lb}.

\begin{proofof}{\Cref{lemma:collision-concentration-l2}}
    From \Cref{lemma:c-moment-exp}, the expected number of collisions is $\binom{m}{2} \norm{\p}_{2}^{2}$.
    A standard calculation (see e.g. \cite{Canonne:Survey:ToC}) yields
    \begin{align*}
        \Var(X) &\leq \binom{m}{2}^{2} \left( \frac{2 \norm{\p}_{2}^{2}}{m(m-1)} + \frac{4 \norm{\p}_{3}^{3}}{m} \right) \\
        &= \frac{m(m-1) \norm{\p}_{2}^{2}}{2} + \norm{\p}_{3}^{3} m (m - 1)^2 \\
        &\leq \binom{m}{2} \norm{\p}_{2}^{2} \left( 1 + 2 \norm{\p}_{2} (m- 1) \right) \text{.}
    \end{align*}
    Since $\norm{\p}_{2} \geq \frac{1}{\sqrt{n}}$ and $m \geq c \sqrt{n}$, we have that $1 + 2 \norm{\p}_{2} (m - 1) \leq (2 + c) \norm{\p}_{2} (m - 1)$.
    Thus, $\Var(X) \leq O(m^3 \norm{\p}_{2}^{3})$.
    By Chebyshev's Inequality, we set $\sigma \leq O(m^{1.5} \norm{\p}_{2}^{1.5})$ and obtain 
    \begin{equation*}
        \Pr \left( X \in (1 \pm 0.0001)\binom{m}{2} \norm{\p}_{2}^{2} \right) > 0.9999
    \end{equation*}
    for sufficiently large $n$.
\end{proofof}

\begin{proofof}{\Cref{lemma:collision-concentration-lc}}
    Let $Y_{1}, \dotsc, Y_{m}$ denote the individual samples and $X = \sum_{i_1, \dotsc, i_c} \ind{Y_{i_1} = \dotsc = Y_{i_c}}$ be the number of $c$-collisions where $i_1, \dotsc, i_c$ are distinct.
    As before, the expected number of collisions is $\EX[X] = \binom{m}{c} \norm{\p}_{c}^{c}$.
    We proceed to bound the variance, in particular $\EX[X^2]$ which we may write as
    \begin{align*}
        \EX[X^2] &= \EX\left[ \left( \sum_{i_1, i_2, \dotsc, i_c} \ind{Y_{i_1} = \dotsc = Y_{i_c}} \right)^2\right] \\
        &= \EX\left[ \sum_{i_1, i_2, \dotsc, i_c, j_1, \dotsc, j_c} \ind{Y_{i_1} = \dotsc = Y_{i_c}} \ind{Y_{j_1} = \dotsc = Y_{j_c}} \right] \text{.}
    \end{align*}
    We now need to do some case analysis.
    Either there is the case $i_1, \dotsc, i_c, j_1, \dotsc j_c$ are all distinct, or there is some overlap in the tuples $i_1, \dotsc, i_c$ and $j_1, \dotsc, j_c$. 
    Note that $i_1, \dotsc, i_c$ are always distinct (similarly for $j_1, \dotsc, j_c$).
    In the former case, all indices are distinct and independent and we obtain
    \begin{equation*}
        \EX\left[ \sum_{i_1, i_2, \dotsc, i_c, j_1, \dotsc, j_c} \ind{Y_{i_1} = \dotsc = Y_{i_c}} \ind{Y_{j_1} = \dotsc = Y_{j_c}} \right] = \binom{m}{2c} \norm{p}_{c}^{2c} \text{.}
    \end{equation*}
    Note that this is at most $\EX[X]^2$ so we can freely subtract this term from the variance calculation.
    In the latter case, since there is overlap between $i_1, \dotsc, i_c$ and $j_1, \dotsc, j_c$ (of which there are at most $2c - 1$ distinct indices) we can upper bound the desired quantity as
    \begin{equation*}
        \Var(X) \leq \EX\left[ \sum_{i_1, i_2, \dotsc, i_c, j_1, \dotsc, j_c} \ind{Y_{i_1} = \dotsc = Y_{i_c} = Y_{j_1} = \dotsc = Y_{j_c}} \right] = O_{c}\left( \sum_{k = c}^{2c - 1} \binom{m}{k} \norm{\p}_{k}^{k} \right) \text{.}
    \end{equation*}
    In the above, we note that there are $O_{c}(1)$ distinct cases of overlap between the two sets of indices.
    We upper bound the sum in each case by noting  that when there are $c \leq k \leq 2c - 1$ distinct, there are at most $O_{c}(\binom{m}{k})$ summands (choices of indices $i_1, \dotsc, i_c, j_1, \dotsc, j_c$) and the probability they all coincide is $\norm{\p}_{k}^{k} \leq \norm{\p}_{c}^{k}$.
    Now, observe that for every $\p$, $\norm{\p}_{c} \geq \frac{1}{n^{(c - 1)/c}}$ so that $m \norm{\p}_{c} \geq 2$.
    Therefore,
    \begin{equation*}
        \Var(X) = O_{c}\left( \sum_{k = c}^{2c - 1} \binom{m}{k} \norm{\p}_{k}^{k} \right) = O_{c}\left( \binom{m}{2c - 1} \norm{\p}_{c}^{2c - 1}  \right)
    \end{equation*}
    so the standard deviation is at most $O_{c}\left( m^{c-1/2} \norm{p}_{c}^{c - 1/2} \right)$ which for sufficiently large $n$ (depending on $c$ only) is asymptotically less than $\EX[X]$.
    The result then follows by applying the Chebyshev inequality.
\end{proofof}

Finally, we give the sample lower bound for the problem given in \Cref{thm:bias-estimation-lb}.

\begin{lemma}
    \label{lemma:coin-distinguish-lb}
    Suppose an algorithm is given sample access to two coins $\coin_1, \coin_2$ with biases $a < b$ and is asked to determine which coin has which bias.
    Then, any algorithm that succeeds with probability at least $\frac{2}{3}$ requires $\bigOm{\frac{\min(a, b, 1 - a, 1 - b)}{(b - a)^2}}$ samples.
\end{lemma}

\begin{proof}
    Let $\innerAlg$ denote such an algorithm which takes $m$ samples.
    Consider an adversary that flips a uniform bit $X$ and if $X = 0$, assigns $\coin_1$ to have bias $a$, and otherwise assigns $\coin_1$ to have bias $b$.
    Let $T = (T_1, T_2)$ denote the samples observed by algorithm $\innerAlg$.
    From \Cref{lemma:mutual-info-bound} we have $I(X:T) \geq 2 \cdot 10^{-4}$.
    On the other hand, since $T_1, T_2$ are independent conditioned on $X$, we can upper bound $I(X:T) \leq I(X:T_1) + I(X:T_2)$.

    We claim $I(X:T_1) = \bigO{\frac{m(b-a)^2}{\min(a,b,1-a,1-b)}}$ (see Lemma 3.9 of \cite{hopkins2024replicability}).
    A similar argument proceeds for $I(X:T_2)$.
    Note that $T_1$ consists of $m$ iid samples of $\Bern(a)$ if $X = 0$ and $\Bern(b)$ if $X = 1$.
    In particular, applying the independence of each sample conditioned on $X$ and \Cref{claim:MI_asymp}, we obtain
    \begin{equation*}
        I(X:T_1) = \bigO{m \frac{(b - a)^2}{\min(a, b, 1 - a, 1 - b)}}
    \end{equation*}
    as desired.
    Combined with our lower bound on $I(X:T)$ for any accurate algorithm, this yields a lower bound on $m$.
\end{proof}

\subsection{Flattening with Bounded Infinity Norm} \label{sec:high-sample-flat}

\begin{theorem} [High Sample Mixture Flattening] \label{thm:strong-mixture-flat} Using $O(n \log n)$ samples to a mixture $\rD = \lambda \p + (1-\lambda)\q$ over $[n]$, it is possible to produce a random map $f \colon [n] \to [N]$ such that the following hold.
\begin{itemize}
    \item For any two distributions $\cD_1,\cD_2$ over $[n]$, we have $\norm{f(\cD_1) - f(\cD_2)}_1 = \norm{\cD_1 - \cD_2}_1$.
    \item For any constant $\delta \in (0,1)$, $\norm{f(\rD)}_{\infty} < \frac{1+\delta}{n}, \norm{f(\p)}_{\infty} < \frac{1+\delta}{\lambda n}$, and $N \leq (2 + \delta) n$ with probability $0.999$.
\end{itemize}
\end{theorem}

\begin{proof} By \Cref{thm:learn}, using $O(n \log n)$ samples, we can learn $\widetilde{\rD}[1],\ldots,\widetilde{\rD}[n]$ satisfying $|\widetilde{\rD}[i] - \rD[i]| \leq \delta \sqrt{\frac{\rD[i]}{n}}$ for any constant $\delta \in (0,1/2)$ with probability $0.999$. In particular,
\begin{align} \label{eq:max}
\rD[i] - \delta\sqrt{\frac{\rD[i]}{n}} \leq \widetilde{\rD}[i] \leq \rD[i] + \delta\sqrt{\frac{\rD[i]}{n}} ~\implies~ \rD[i] \left(1 - \delta\sqrt{\frac{1}{\rD[i]n}}\right) \leq \widetilde{\rD}[i] \leq \rD[i] \left(1 + \delta\sqrt{\frac{1}{\rD[i]n}}\right) \text{.}
\end{align}
%Since $\rD$ is a probability distribution, we clearly have $\max_i \rD[i] \geq 1/n$. Note the following:
Therefore,
\begin{itemize}
    \item When $\rD[i] \geq 1/n$, the RHS of \cref{eq:max} implies $(1-\delta)\rD[i] \leq \widetilde{\rD}[i] \leq (1 + \delta) \rD[i]$. %\leq 1.1 \norm{\rD}_{\infty}$.
    \item When $\rD[i] \leq 1/n$, the LHS of \cref{eq:max} implies $\widetilde{\rD}[i] \leq \rD[i] + \frac{\delta}{n} \leq \frac{1 + \delta}{n}$.
    %\item When $\rD[i] \leq 1/n$, the LHS of \cref{eq:max} implies $\rD[i] - \frac{1}{10n}  \leq \widetilde{\rD}[i] \leq \rD[i] + \frac{1}{10n} < \frac{1.1}{n}$.
\end{itemize}
 %Therefore, $\norm{\widetilde{\rD}}_{\infty} \leq 1.01\norm{\rD}_{\infty}$. 
 We now define the random map $f \colon [n] \to [N]$ as follows. For each $i \in [n]$, let $B_i$ be a domain of size $m_i := \lceil \widetilde{\rD}[i]n \rceil$. Then, we define $f(i)$ to be a uniform random element of $B_i$. Note that $f$ maps $[n]$ to $B_1 \cup \cdots \cup B_n$ which (using the above upper bounds for the $\widetilde{\rD}[i]$-s) is of size 
 \begin{align} \label{eq:num-buckets}
    N := \sum_{i \leq n} m_i &\leq  \sum_{i \leq n} (1 + \widetilde{\rD}[i] n) \leq n\left(1 + \sum_{i \colon \rD[i] \geq 1/n} (1+\delta)\rD[i] +  \sum_{i \colon \rD[i] < 1/n} \left(\rD[i] + \frac{\delta}{n}\right)  \right) \leq 2(1+\delta)n \text{.}
 \end{align}
We now show $1$-norms are preserved under $f$. For a distribution $\cD$, the probability mass given to $j \in B_i$ is exactly $\cD[i]/m_i$. Thus,
\begin{align} \label{eq:norm-preserve}
    \norm{f(\cD_1) - f(\cD_2)}_1 = \sum_{i \leq n} \sum_{j \in B_i} \left| \frac{\cD_1[i]}{m_i} - \frac{\cD_2[i]}{m_i} \right| = \sum_{i \leq n} |\cD_1[i] - \cD_2[i]| = \norm{\cD_1 - \cD_2}_1 \text{.} 
\end{align}
Next, we bound the infinity norm of $f(\rD)$ and $f(\p)$. For $j \in B_i$, the probability mass of $f(\rD)$ on $j$ is $\rD[i]/m_i = \rD[i]/\lceil \widetilde{\rD}[i]n \rceil$ and so $\norm{f(\rD)}_\infty = \max_i \frac{\rD[i]}{\lceil n \widetilde{\rD}[i]\rceil}$, $\norm{f(\p)}_\infty = \max_i \frac{\p[i]}{\lceil n \widetilde{\rD}[i]\rceil}$. We get the following bounds:
\begin{itemize}
    \item When $\rD[i] \geq 1/n$, we have $\widetilde{\rD}[i] \geq (1-\delta)\rD[i]$ (as shown above) and so 
    \[
    \frac{\rD[i]}{\lceil n \widetilde{\rD}[i]\rceil} \leq \frac{\rD[i]}{n \cdot (1-\delta) \rD[i]} = \frac{1}{(1-\delta)n} ~\text{ and }~
    \frac{\p[i]}{\lceil n \widetilde{\rD}[i]\rceil} \leq \frac{\p[i]}{n \cdot (1-\delta) \rD[i]} \leq \frac{1}{(1-\delta) \lambda n}
    \]
    \item When $\rD[i] \leq 1/n$, by the use of the ceiling function the denominator is at least $1$, and so 
    \[
    \frac{\rD[i]}{\lceil n \widetilde{\rD}[i]\rceil} \leq \rD[i] \leq \frac{1}{n} ~\text{ and }~
    \frac{\p[i]}{\lceil n \widetilde{\rD}[i]\rceil} \leq \p[i] \leq \frac{\rD[i]}{\lambda} \leq \frac{1}{\lambda n}
    \]
\end{itemize}
and so $\norm{f(\rD)}_\infty \leq \frac{1}{(1-\delta)n} < \frac{(1+2\delta)}{n}$ (since $\delta < 1/2$) and $\norm{f(\p)}_\infty \leq \frac{1+2\delta}{\lambda n}$, as claimed. \end{proof} 

\subsection{Reduction to Equal Mixture Parameters}

In \Cref{sec:closeness-UB} and \Cref{sec:closeness-UB-highsample}, we described closeness testers for distributions $\p_1,\p_2$ under sample and verification query access to mixtures $\rD_1 = \lambda \p_1 + (1-\lambda)\q_1$ and $\rD_2 = \lambda \p_2 + (1-\lambda)\q_2$. In particular, we assumed a common mixture parameter $\lambda$ in both distributions, which may be unrealistic in certain applications. However, in this section we give a simple reduction showing that this assumption is without loss of generality, at the cost of extra polynomial factors of $O(\lambda^{-1})$ in the sample and query complexity. We remark there may exist a more efficient reduction that avoids this increase.

\begin{lemma} \label{lem:unequal-mixture-params} Suppose we are given sample and query access to two mixtures $\rD_1 = \lambda_1 \p_1 + (1-\lambda_1)\q_1$ and $\rD_2 = \lambda_2 \p_2 + (1-\lambda_2)\q_2$ with unknown (and possible unequal) mixture parameters $\lambda_1,\lambda_2 \in (0,1)$. Then, there exist distributions $\q_1',\q_2'$ for which we can simulate sample and query access to $\rD_1 = \lambda_1\lambda_2 \p_1 + (1-\lambda_1\lambda_2)\q_1'$ and $\rD_2 = \lambda_1\lambda_2 \p_2 + (1-\lambda_1\lambda_2)\q_2'$. In particular, to simulate one sample costs one sample to the original mixture, and to simulate one query costs two queries and one sample to the original mixtures. \end{lemma}

%As a result, we obtain the following corollary by combining \Cref{thm:closenessUB} and \Cref{lem:unequal-mixture-params}.

%\begin{corollary} \label{cor:closenessUB-unequal-params}
%Suppose we are given sample access and verification query access to two mixtures $\rD_1 = \lambda_1 \p_1 + (1-\lambda_1)\q_1$, $\rD_2 = \lambda_2 \p_2 + (1-\lambda_2)\q_2$ with unknown (and possible unequal) mixture parameters $\lambda_1, \lambda_2 \in (0,1)$, and arbitrary distributions $\p_1,\p_2,\q_1,\q_2 \colon [n] \to [0,1]$. There is an algorithm that distinguishes between the case of $\p_1 = \p_2$ and $\norm{\p_1 - \p_2}_1 \geq \eps$ with probability $2/3$ using
%\[
%O\left(\left(m + \frac{n^{2/3}}{\eps^{4/3}}\right)\frac{1}{(\lambda_1\lambda_2)^2}\right) ~\text{ samples and }~ O\left(\left\lceil \frac{n^2}{m^2} \right\rceil \cdot \frac{1}{\eps^4 (\lambda_1\lambda_2)^4}\right) ~\text{ verification queries.}~
%\] 
%\end{corollary}

\begin{proof} First, write
\begin{align}
    \rD_1 = \lambda_1 \p_1 + (1-\lambda_1) \q_1 &= (\lambda_1 \lambda_2 + (\lambda_1 - \lambda_1\lambda_2))\p_1 + (1- (\lambda_1 \lambda_2 + (\lambda_1 - \lambda_1\lambda_2)))\q_1 \nonumber \\
    &= \lambda_1\lambda_2 \p_1 + (1-\lambda_1\lambda_2) \left(\frac{(\lambda_1 - \lambda_1\lambda_2)\p_1 + (1- (\lambda_1 \lambda_2 + (\lambda_1 - \lambda_1\lambda_2)))\q_1}{1-\lambda_1\lambda_2}\right) \nonumber \\
    &= \lambda_1\lambda_2 \p_1 + (1-\lambda_1\lambda_2)\q_1' = \rD_1'
\end{align}
where
\[
\q_1' := \frac{(\lambda_1 - \lambda_1\lambda_2)\p_1 + (1- (\lambda_1 \lambda_2 + (\lambda_1 - \lambda_1\lambda_2)))\q_1}{1-\lambda_1\lambda_2} \text{.}
\]
Clearly $\rD_1 = \rD_1'$ and so samples are trivial to simulate. We now show how to simulate queries to $\rD_1'$. Suppose we sample $x \sim \rD_1'$ and wish to simulate a query on $x$. Then, administer a query on $x$ using the $\rD_1$ verification query oracle. If the oracle returns "irrelevant", then we return "irrelevant". If the oracle returns "relevant", then we return "relevant" with probability $\lambda_2$ and "irrelevant" with probability $1-\lambda_2$. To do this, we need access to a $\lambda_2$-probability coin. Crucially, we can simulate such a coin with one sample and one query to the $\rD_2$ verification query oracle. All in all, we have
\[
\Pr_{x \sim \rD_1'}[x = i \text{ and } \rD_1'\text{-oracle returns } \p_1] = \Pr_{x \sim \rD_1'}[x = i \text{ and } x \sim \p_1] \cdot \Pr[\rD_1\text{-oracle returns } \p_1 ~|~ x=i] = \lambda_1\p_1(i) \cdot \lambda_2 \text{.}
\]
Similarly, we let $\rD_2' = \lambda_1\lambda_2 \p_2 + (1-\lambda_1\lambda_2)\q_2'$ where
\[
\q_2' := \frac{(\lambda_1 - \lambda_1\lambda_2)\p_2 + (1- (\lambda_1 \lambda_2 + (\lambda_1 - \lambda_1\lambda_2)))\q_2}{1-\lambda_1\lambda_2}
\]
and the simulation proof for the $\rD_2'$ verification query oracle is exactly analogous. \end{proof}

\end{document}